\documentclass[aps,prx,twocolumn,10pt,superscriptaddress,floatfix,nofootinbib,longbibliography]{revtex4-2}

\usepackage{graphicx}
\usepackage{xcolor}
\usepackage{hyperref}
\hypersetup{colorlinks=true,linkcolor=blue,citecolor=blue,urlcolor=cyan}
\usepackage{amsmath}
\usepackage{amsfonts}
\usepackage{amsthm}
\usepackage{latexsym}
\usepackage{amssymb}
\usepackage{bm}
\usepackage{enumitem}
\usepackage{setspace}
\usepackage{mathtools}
\usepackage[caption=false]{subfig}
\usepackage{graphicx}
\usepackage{xcolor}
\usepackage{array}
\usepackage{tabularx}
\usepackage{multirow}
\usepackage{tikz}
\usepackage{framed}

\DeclarePairedDelimiter{\ket}{\lvert}{\rangle}
\DeclarePairedDelimiterX\braket[2]{\langle}{\rangle}{#1\,\delimsize\vert\,\mathopen{}#2}
\DeclarePairedDelimiterX\dyad[2]{\lvert}{\rvert}{#1\delimsize\rangle\!\delimsize\langle#2}
\DeclarePairedDelimiterX\projector[1]{\lvert}{\rvert}{#1\delimsize\rangle\!\delimsize\langle#1}
\DeclarePairedDelimiterX\mel[3]{\langle}{\rangle}{#1\,\delimsize\vert\,\mathopen{}#2\,\delimsize\vert\,\mathopen{}#3}
\DeclarePairedDelimiter{\norm}{\lVert}{\rVert}

\DeclareMathOperator\Tr{Tr}
\DeclareMathOperator\supp{supp}

\DeclareMathOperator\wt{wt}

\DeclareMathOperator\conv{conv}
\DeclareMathOperator\id{id}

\DeclareMathOperator\pr{pr}
\DeclareMathOperator\polylog{polylog}
\DeclareMathOperator\poly{poly}

\newtheorem{lemma}{Lemma}
\newtheorem{theorem}{Theorem}
\newtheorem{proposition}{Proposition}
\newtheorem{definition}{Definition}
\newtheorem{corollary}{Corollary}

\begin{document}

\title{Fault Tolerant Quantum Phases of Matter}

\author{Colin V. Coane}
\email{colin.coane@yale.edu}
\affiliation{Department of Physics, Yale University, New Haven, CT, 06511, USA}
\affiliation{Yale Quantum Institute, Yale University, New Haven, CT, 06511, USA}

\author{Shouzhen Gu}
\affiliation{Yale Quantum Institute, Yale University, New Haven, CT, 06511, USA}
\affiliation{Department of Applied Physics, Yale University, New Haven, CT, 06511, USA}

\author{Aleksander Kubica}
\affiliation{Yale Quantum Institute, Yale University, New Haven, CT, 06511, USA}
\affiliation{Department of Applied Physics, Yale University, New Haven, CT, 06511, USA}

\begin{abstract}
    We introduce a classification of quantum phases of mixed states based on the ability to preserve and fault tolerantly transfer information. We formulate phase equivalence between sets of states that encode logical quantum information and are connected by shallow local channel circuits assisted by local measurements and global classical communication of measurement outcomes, where all quantum operations, including measurements, must be robust to noise. We show that the existence of a recovery threshold for local noise is a property of an entire phase. We prove that any two sets of states with nonzero thresholds that are related by finite-depth local channel circuits belong to the same fault tolerant phase. Our approach leads to a classification of subsystem codes without favoring any state of the gauge subsystem, and we demonstrate phase equivalence between 2D and 3D color codes using dimensional jumps within the 3D gauge color code. Lastly, we study how syndrome structure affects fault tolerance in circuits with measurements. We show that qudit stabilizer codes with extended syndromes permit single-shot state preparation, whereas a broad class of shallow circuits measuring point-like syndromes, including state preparation for 2D topological codes and non-Abelian topological orders with solvable anyon theories, is not fault tolerant.
\end{abstract}

\maketitle

\section{Introduction} \label{intro}

Classifying and understanding phases of matter is a central goal of quantum many-body physics~\cite{zeng2019}. Tools and insights from quantum information theory have aided greatly in this task, most notably providing information theoretic diagnostics~\cite{horodecki2009,hamma2005,kitaev2006,levin2006,wolf2008,lee2013} to probe the properties of highly entangled states and a circuit-based definition of gapped quantum phases~\cite{chen2010}. This definition, asserting two gapped ground states are in the same phase if they are related by a shallow local unitary circuit, allows one to rigorously classify pure states based on their entanglement properties~\cite{hastings2005,bravyi2006,bravyi2010,bravyi2011,hastings2013,zeng2015,chiu2016,haah2016}.

Recent works have extended this classification to different settings, notably to phases of mixed states~\cite{degroot2022,coser2019,ma2023,rakovszky2024,sang2024rg,ellison2025,sang2024markov,sang2025reversibility,yang2025,negari2024,ma2025} related by local channel (LC) circuits (denoted LC phases), and to phases of states connected by local projective measurements and global classical communication of measurement outcomes~\cite{piroli2021,tantivasadakarn2024,tantivasadakarn2023,ren2025,bravyi2022,verresen2022,lu2022,lu2023,piroli2024,smith2024,malz2024,buhrman2024}; these settings are motivated by physical systems and experimental demonstrations. Mixed states arise naturally in driven open quantum systems~\cite{diehl2008,verstraete2009,pastawski2011,chirame2025}, at finite temperature~\cite{hastings2011,zhou2025,ma2025temperature,lu2020,chen2024symmetry}, and when pure states are subjected to local decoherence~\cite{bao2023,fan2024,chen2024separability,lee2025spt,lee2025ci,zhang2025stability}. Furthermore, since many quantum devices allow mid-circuit measurement~\cite{iqbal2024,iqbal2024nonabelian,chen2025nishimori}, it is useful to understand which exotic states can be prepared using quantum-local (QL) circuits, which, by definition~\cite{bombin2015singleshot}, are constructed from LC circuits, local measurements, and global classical feedforward of measurement outcomes to condition subsequent LC circuits applied to the system.

However, all operations on real quantum devices suffer from noise. While LC phases can model local decoherence acting on a static quantum state as an LC circuit, neither of the above classifications concern noise acting within a circuit which connects two states. Measurement errors may be particularly damaging as flipped measurement outcomes can be instantaneously transmitted nonlocally via classical communication and may incorrectly condition adaptive operations. In the worst case, a single measurement error may induce a large residual error on the system by conditioning the wrong operation.

\begin{figure*}[ht!]
    \centering
    \includegraphics[width=0.99\textwidth]{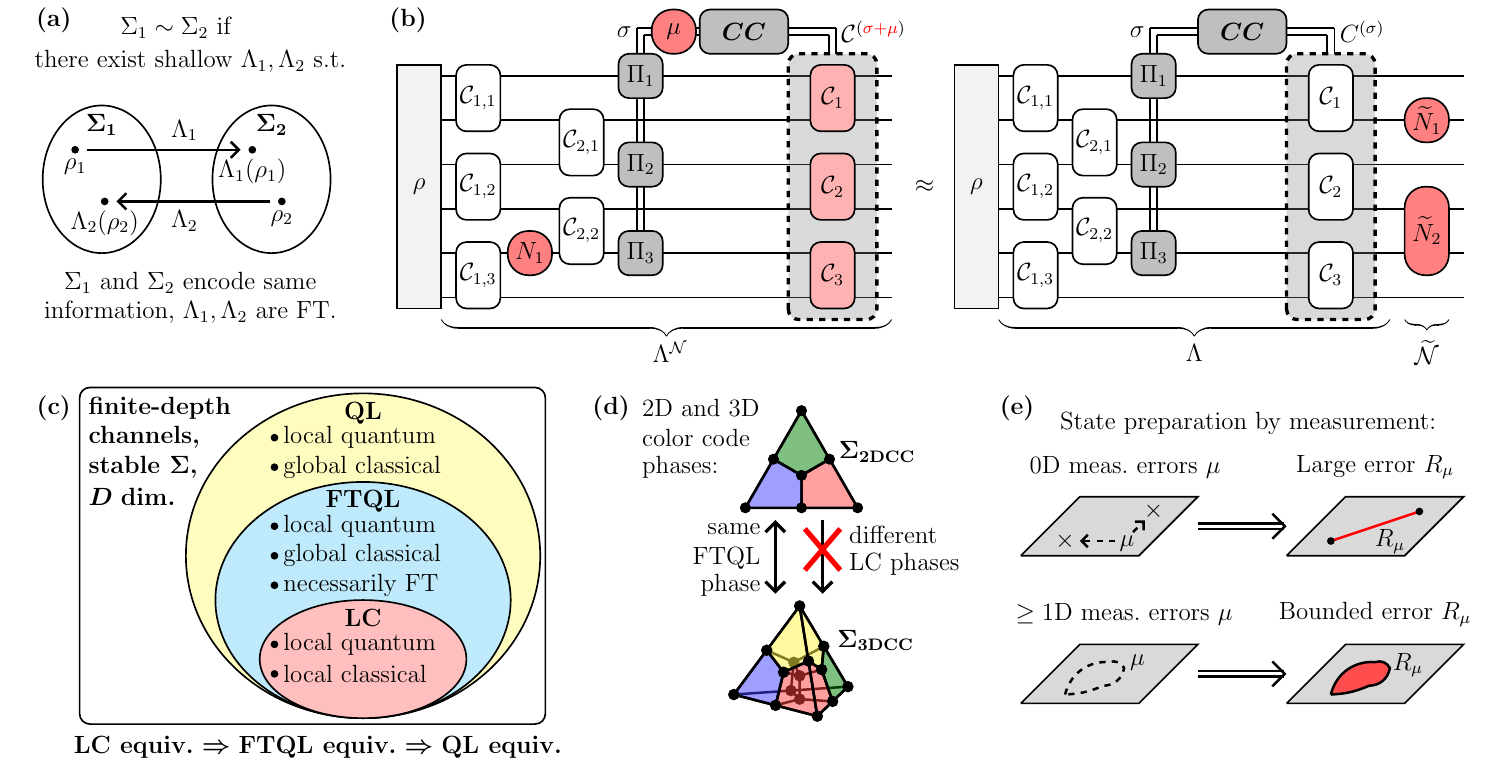}
    \caption{(a) Definition of fault tolerant (FTQL) phase equivalence (Definition~\ref{def:ft-phase-equivalence}). Two convex sets of states $\Sigma_{1}$ and $\Sigma_{2}$ are in the same phase if there exist shallow quantum-local (QL) circuits $\Lambda_{1}$ and $\Lambda_{2}$ which map between states in $\Sigma_{1}$ and $\Sigma_{2}$, $\rho_{i}$ and $\Lambda_{i}(\rho_{i})$ represent the same encoded information, and $\Lambda_{i}$ is fault tolerant.
    (b) A QL circuit $\Lambda$ consists of layers of locally reversible local channel (LC) circuits $\mathcal{C}_{i}$, local measurements $\Pi_{i}$, and locally reversible LC circuits $\mathcal{C}^{(\sigma)}$ conditioned on global processing of measurement outcomes $\sigma$. A noisy circuit $\Lambda^{\mathcal{N}}$ has qudit noise $N_{i}$ and measurement noise $\mu$ inserted between layers. The action of this circuit on a state $\rho$ can be mimicked by applying a noiseless channel $\Lambda$ followed by effective residual noise $\widetilde{\mathcal{N}}$.
    (c) Relationships between constraints on allowed operations in existing local channel (LC) and quantum-local (QL) phase classifications, compared with our proposed definition (FTQL). Here, we restrict to finite-depth circuits acting on convex sets of states which are stable according to Definition~\ref{def:stable-phase}. In this setting, if two stable convex sets are in the same LC phase and related by finite-depth LC circuits, then they are also in the same FTQL phase.
    (d) When defined on compatible triangular and tetrahedral lattices embedded in 3D space, the code spaces of the 2D and 3D color codes, $\Sigma_{\mathrm{2DCC}}$ and $\Sigma_{\mathrm{3DCC}}$, are in the same FTQL phase as one can fault tolerantly code switch between them. However, they are in different LC phases as one cannot map all states in $\Sigma_{\mathrm{2DCC}}$ to states in $\Sigma_{\mathrm{3DCC}}$ with a shallow LC circuit.
    (e) Fault tolerant state preparation of topological codes by measuring stabilizers and applying corrections is not possible in depth $O(\log\log(n))$ when syndromes are pointlike, $0$D objects, but it is possible in finite depth when syndromes form extended, $\geq 1$D objects.}
    \label{fig:overview}
\end{figure*}

In this work, we introduce a classification of quantum phases where the ability to protect and fault tolerantly transfer quantum information using noisy, shallow, QL circuits is the defining principle. Our approach differs from previous phase classifications in its focus on active, fault tolerant information transfer rather than equivalence of states with respect to noiseless shallow circuits or passive preservation under noise.
Concretely, we define phase equivalence with respect to \textit{convex sets} of states which encode logical quantum information, for example code spaces of quantum error correcting (QEC) codes, rather than individual states. We say two sets of states are in the same fault tolerant (FTQL) phase if they are connected by two shallow QL circuits where all quantum operations, including measurements, are robust to noise and information transferred by each circuit is recoverable.
Fault tolerant phases then characterize whether robust, two-way information transfer between convex sets can be implemented with shallow circuits.

We define stability of a convex set as the existence of a nonzero recovery threshold for local noise and we show that stability is both a property of an entire phase and a requirement for fault tolerant information transfer to be possible. We show stability of the code spaces of quantum low-density parity check (LDPC) codes~\cite{breuckmann2021} whose distances scale polynomially with the system size. We also prove that no convex set where information is encoded in a geometrically local region is stable.

When one restricts to finite-depth (rather than shallow) circuits, LC phase equivalence between two stable convex sets implies FTQL phase equivalence. To show this, we bound how noise spreads in LC circuits and prove that all finite-depth, locally reversible LC circuits connecting two stable convex sets are fault tolerant.

Our definition of fault tolerant phases also lets us classify phases of subsystem codes without favoring any state of the gauge subsystem.
This is done by identifying different gauge-fixed states which represent the same encoded information.
We use this approach, along with code switching and dimensional jumps within the 3D gauge color code~\cite{bombin2015singleshot,bombin2016} to prove that code spaces of the 2D and 3D stabilizer color codes are in the same FTQL phase;
similar conclusions can be reached for the toric code family, since the toric and color code families are closely related in any dimension~\cite{kubica2015unfolding}.
We then show that these two codes are in different LC phases, which means stable FTQL phases are strictly coarser than LC phases when one considers only finite-depth circuits.

Lastly, we investigate what properties allow or prohibit fault tolerance in shallow circuits with measurements. Here, we study qudit stabilizer codes whose syndromes form extended objects such as loops and show that single-shot state preparation and QEC are possible, generalizing results in Ref.~\cite{bombin2015singleshot} to qudit codes and non-Pauli noise. This highlights the difference between state preparation and information transfer, as one can fault tolerantly prepare a distinct code state from a product state, yet the code space is not in the same phase as a set of unencoded qudits. We also study lattice models which host pointlike excitations and prove that a wide range of shallow state preparation and QEC circuits, where Hamiltonian terms are measured and unitary corrections are applied conditioned on measurement outcomes, are not fault tolerant. While this does not prohibit fault tolerant state preparation with a shallow circuit, it rules out a range of existing methods for state preparation of 2D topological codes and non-Abelian topological orders with solvable anyon theories, as given in Refs.~\cite{piroli2021,tantivasadakarn2024,tantivasadakarn2023,bravyi2022,ren2025,verresen2022}.

The rest of the paper is structured as follows. In Section~\ref{prev-defs} we review previous phase definitions and motivate how noise causes issues by analyzing noisy state preparation of the 2D toric code. In Section~\ref{ft-phases} we define fault tolerant phase equivalence and the noise models we use. In Section~\ref{ft-phase-properties} we define and discuss stability, argue there is no single ``trivial'' fault tolerant phase, and compare fault tolerant phases with previous phase definitions. In Section~\ref{cc-phases} we study phases of various code spaces of $D$-dimensional gauge color codes. In Sections~\ref{robust-meas-feedback} and~\ref{non-robust-circuits} we study properties of circuits with measurements which either allow or prohibit fault tolerance. Our main results are summarized in Figure~\ref{fig:overview}.

\section{Previous definitions of phase equivalence} \label{prev-defs}

In this paper, we consider systems of qudits with local dimension $d$ arranged on regular lattices embedded in $D$-dimensional Euclidean space. We consider quantum circuits acting on these systems constructed from layers of local channel gates and projective measurements. A range-$r$ local channel gate is a quantum channel supported within a ball of radius $\leq r$ on the system lattice. One may add a finite number of unentangled qudits to the system at each lattice site or trace out arbitrary, possibly entangled qudits as part of each local channel gate. A range-$r$ measurement gate is a projective measurement of an operator supported within a ball of radius $\leq r$ on the system lattice. An $(r,\ell)$ circuit consists of $\ell$ layers of range-$r$ gates with disjoint support in each layer. We define the range of an $(r,\ell)$ circuit to be the maximum range of any circuit element times its depth, i.e., $r\ell$. This quantity may capture the spread of information through the circuit. When classifying quantum phases, we will implicitly consider families of systems of growing size $n$, where $n$ is the number of qudits, so we can make statements about how properties scale with $n$ and discuss the thermodynamic limit $n \rightarrow \infty$. For instance, we say a circuit is \textit{finite-depth} if $r\ell = O(1)$ and \textit{shallow} if $r\ell = O(\polylog(n))$. Additionally, we adopt the convention that $\rho_{1} \approx \rho_{2}$ means the trace distance between the two states asymptotically decays at least polynomially with $n$.
Specifically,
\begin{equation}
    \label{eq:trace-distance}
    \rho_{1} \approx \rho_{2} \Leftrightarrow   \norm{\rho_{1} - \rho_{2}}_{1} = 1/\Omega(\poly(n)).
\end{equation}
The above choices ensure the range of a shallow circuit scales slowly compared to $n$ and the trace distance between ``equal'' states vanishes quickly as $n$ increases.

Mixed-state phases may be defined by asserting that two states are in the same phase if they can be transformed back and forth into each other with a pair of shallow local channel (LC) circuits~\cite{coser2019,ma2023,sang2024rg,rakovszky2024,ellison2025}. However, we will apply a more restrictive definition here where all local channel gates applied must be \textit{locally reversible}~\cite{sang2024markov,sang2025reversibility}. This definition ensures each individual gate in a circuit preserves locality and does not disturb long-range correlations, resolving issues which arise when local reversibility is not required~\cite{sang2025reversibility}. Here, an $(r,\ell)$ LC circuit decomposes into $\ell$ layers, $\mathcal{C} = \mathcal{C}_{\ell} \circ \ldots \circ \mathcal{C}_{1}$, where each layer $\mathcal{C}_{i} = \bigotimes_{\alpha} \mathcal{C}_{i,\alpha}$ consists of range-$r$ local channel gates $\mathcal{C}_{i,\alpha}$ with disjoint support. The circuit $\mathcal{C}$ is locally reversible with respect to the set of states $\Sigma$ if and only if there exists an $(r^\prime,\ell)$ LC circuit $\widetilde{\mathcal{C}} = \widetilde{\mathcal{C}}_{1} \circ \ldots \circ \widetilde{\mathcal{C}}_{\ell}$ with $r^\prime\ell = O(r\ell)$ such that the action of each gate $\mathcal{C}_{i,\alpha}$ in $\mathcal{C}$ is reversed by a corresponding gate $\widetilde{\mathcal{C}}_{i,\alpha}$ in $\widetilde{\mathcal{C}}$ at every layer of $\mathcal{C}$ for any $\rho \in \Sigma$. More precisely, denoting the intermediate state $\rho^{(i)} \coloneqq \mathcal{C}_{i} \circ \ldots \circ \mathcal{C}_{1} (\rho)$ for each $i \in [1,\ell]$ and $\rho^{(0)} \coloneqq \rho$, we say $\mathcal{C}$ is locally reversible if for every $\rho \in \Sigma$,
\begin{equation}
    \label{eq:local-reversibility}
    \sum_{i,\alpha} \norm*{\widetilde{\mathcal{C}}_{i,\alpha} \circ \mathcal{C}_{i,\alpha} \left(\rho^{(i-1)}\right) - \rho^{(i-1)}}_{1} = 1/\Omega(\poly(n)) .
\end{equation}
This bound implies $\widetilde{\mathcal{C}} \circ \mathcal{C} (\rho) \approx \rho$ and both 
\begin{align}
    \widetilde{\mathcal{C}}_{i} \circ \mathcal{C}_{i} \circ \mathcal{C}_{i-1} \circ \ldots \circ \mathcal{C}_{1} (\rho) &\approx \mathcal{C}_{i-1} \circ \ldots \circ \mathcal{C}_{1} (\rho) , \\
    \widetilde{\mathcal{C}}_{i,\alpha} \circ \mathcal{C}_{i,\alpha} \circ \mathcal{C}_{i-1} \circ \ldots \circ \mathcal{C}_{1} (\rho) &\approx \mathcal{C}_{i-1} \circ \ldots \circ \mathcal{C}_{1} (\rho) , \nonumber
\end{align}
for every individual layer $\mathcal{C}_{i}$ and for every individual gate $\mathcal{C}_{i,\alpha}$~\cite{sang2024markov}. We now present the formal definition of LC phases we adopt in this work.
\begin{definition}[LC phases]
\label{def:lc-phase}
    Two states $\rho_1$ and $\rho_2$ are in the same mixed-state phase (LC phase) if there exists an $(r,\ell)$ locally reversible local channel circuit $\mathcal{C}$ with $r\ell = O(\polylog(n))$ such that $\mathcal{C}(\rho_1) \approx \rho_2$ (thus $\widetilde{\mathcal{C}}(\rho_2) \approx \rho_1$ as well by contractivity of the trace distance).
\end{definition}

States within the same LC phase possess the same macroscopic entanglement structure and share invariant properties such as topological entanglement negativity~\cite{lee2013,fan2024}, mixed-state anomaly~\cite{lessa2025anomaly,wang2025anomaly,lessa2025lre}, strong-to-weak spontaneous symmetry breaking~\cite{lessa2025swssb,luo2025}, and topological degeneracy~\cite{sang2025reversibility,yang2025}. Additionally, since noise can often be modeled as an LC circuit~\cite{sang2024rg}, states in an LC phase may preserve quantum information encoded in global degrees of freedom under sufficiently weak noise~\cite{sang2024rg,sang2024markov,fan2024,li2025} or in some cases at sufficiently low, but nonzero, temperature~\cite{hastings2011,lu2020,bergamaschi2025rapidmixing}, contingent on whether the circuit modeling the applied noise is locally reversible.

Various ``measurement-equivalent'' or ``LOCC'' phase definitions have been introduced and studied in previous works to classify states which can be transformed into each other when one allows for nonlocal classical communication~\cite{piroli2021,tantivasadakarn2023,friedman2023feedback}. A simple, yet powerful, class of circuits unifies these definitions, namely \textit{quantum-local} (QL) circuits~\cite{bombin2015singleshot}. QL circuits are constructed by composing locally reversible LC circuits with deterministic measurement-feedback circuits, where one first measures a set of operators then applies a locally reversible LC circuit globally conditioned on measurement outcomes.

More specifically, in an $(r,\ell)$ measurement-feedback circuit, first $\ell_{1}$ layers of range-$r$ projective measurements are made, measuring operators $O_{i,\alpha}$ in each layer $i$. Then, an $(r,\ell_{2})$ locally reversible LC circuit is applied conditioned on the measurement outcomes of $O_{i,\alpha}$. Here, $\ell_{1} + \ell_{2} = \ell$. Denoting the set of all measured operators $\mathcal{I}$, this circuit can be written as
\begin{equation}
    \mathcal{K} (\rho) = \sum_{\sigma} \mathcal{C}^{(\sigma)} \left( \Pi_{\sigma}^{\mathcal{I}} \rho {\Pi_{\sigma}^{\mathcal{I}}}^{\dagger} \right)
\end{equation}
where $\sigma$ is a list of measurement outcomes $\sigma_{i,\alpha}$
labeling eigenvalues of $O_{i,\alpha}$, $\Pi_{\sigma}^{\mathcal{I}} = \prod_{i=1}^{\ell_{1}} \prod_{\alpha} \Pi_{\sigma_{i,\alpha}}^{O_{i,\alpha}}$ is a projector representing the ordered sequence of measurements made which sequentially projects onto the $\sigma_{i,\alpha}$ eigenspace of each $O_{i,\alpha}$, and $\mathcal{C}^{(\sigma)}$ is an $(r,\ell_{2})$ locally reversible LC circuit conditioned on $\sigma$. We say $\mathcal{K}$ is \textit{deterministic} with respect to a set of states $\Sigma$ and a subsystem $\mathcal{A}$ if for all $\rho \in \Sigma$, 
\begin{equation}
    \mathcal{C}^{(\sigma)} \left( \Pi_{\sigma}^{\mathcal{I}} \rho {\Pi_{\sigma}^{\mathcal{I}}}^{\dagger} \right) \approx c_{\sigma} \rho^{\prime}_{\mathcal{A}} \otimes \rho^{\prime}_{\sigma} , \quad \sum_{\sigma} c_{\sigma} = 1 ,
\end{equation}
i.e., each measurement ``shot'' yields the same state $\rho^{\prime}_{\mathcal{A}} \in \mathcal{A}$ of $\mathcal{A}$ irrespective of $\sigma$. This means $\mathcal{K} (\rho) \approx \rho^{\prime}_{\mathcal{A}} \otimes \sum_{\sigma} c_{\sigma}\rho^{\prime}_{\sigma}$. This captures circuits used for QEC in both stabilizer and subsystem codes, where the state of the gauge subsystem differs between shots.

An $(r,\ell)$ QL circuit $\Lambda$ then takes the form
\begin{equation}
    \label{eq:ql-circuit}
    \Lambda = \mathcal{K}^{(b)} \circ \mathcal{C}^{(b)} \circ \ldots \circ \mathcal{K}^{(1)} \circ \mathcal{C}^{(1)} ,
\end{equation}
where $\mathcal{C}^{(a)}$ are $(r,\ell_{a})$ locally reversible LC circuits, $\mathcal{K}^{(a)}$ are deterministic $(r,\ell_{a}^{\prime})$ measurement-feedback circuits, and $\sum_{a=1}^{b} \left(\ell_{a} + \ell_{a}^{\prime}\right) = \ell$. We now present the definition of quantum-local phases we use throughout this work.

\begin{definition}[QL phases]
\label{def:ql-phase}
    Two states $\rho_1$ and $\rho_2$ are in the same quantum-local phase (QL phase) if there exist $(r_{i},\ell_{i})$ quantum-local circuits, $\Lambda_1$ and $\Lambda_2$, with $r_i \ell_i = O(\polylog(n))$, $i=1,2$, such that $\Lambda_1(\rho_1) \approx \rho_2$ and $\Lambda_2(\rho_2) \approx \rho_1$.
\end{definition}
The classical nonlocality present in a QL circuit is powerful and can efficiently rearrange the macroscopic entanglement structure of a system~\cite{friedman2023locality,lee2022,zhu2023}. As such, invariants such as topological entanglement negativity no longer distinguish QL phases. Entanglement growth is still bounded~\cite{piroli2021,lu2022,lu2023}, so distinct phases do exist, but the set of states in each phase is greatly enlarged.

\subsection{Issues with definitions under noise} \label{noise-issues}

While QL phases allow one to classify states with respect to operations physically realizable in quantum devices, they do not characterize the ability for a system to preserve encoded information under noise. For example, 2D toric code ground states are in the same QL phase as $\ket{0}^{\otimes n-2} \otimes \ket{\Psi}$, where $\ket{\Psi}$ is an arbitrary two-qubit state~\cite{piroli2021}. However, if one applies independent and identically distributed (i.i.d.)~bit flip noise with strength $p$ to each qubit of $\ket{0}^{\otimes n-2} \otimes \ket{\Psi}$, $\ket{\Psi}$ will be corrupted with probability $p(2-p)$, while if one applies this noise to each qubit of the encoded toric code state $\ket{\overline{\Psi}}$, if $p < p_{c} \approx 0.11$ is below the threshold $p_{c}$, the state $\ket{\overline{\Psi}}$ can be recovered~\cite{dennis2002,sang2024markov}. This shows that states within the same QL phase may have different levels of robustness to noise.

In addition to this, circuits connecting states within a QL phase may be fragile to noise, especially to measurement errors. Here we will study errors which flip the observed outcome of a measurement, as opposed to the weak measurements or measurements of incorrect operators discussed in Refs.~\cite{zhu2023,lee2022}. Even one or two flipped measurement outcomes in a QL circuit may induce a large residual error on the system, since channels $\mathcal{C}^{(\sigma)}$ are globally conditioned on these outcomes. To demonstrate how i.i.d.~measurement noise may influence macroscopic properties of a post-measurement state, consider preparation of the logical $\ket{\overline{0}\overline{0}}$ state of the 2D toric code from the product state $\ket{0}^{\otimes n}$. The 2D toric code is defined on a square lattice tiling a torus with linear size $L$ and $n = 2L^2$ qubits on its edges. Code states are stabilized by products of $Z_j$ surrounding each vertex $v$, $A_{v} = \prod_{j \ni v} Z_j$, and products of $X_j$ surrounding each plaquette $w$, $B_{w} = \prod_{j \in w} X_j$. Its $Z$-type logical operators are formed by products of $Z_j$ along noncontractible closed loops $C_1$ and $C_2$ on the dual lattice, $\bar{Z}_a = \prod_{j \in C_a} Z_j$. Eigenstates of all $A_{v}$, $B_{w}$, and $\bar{Z}_a$, $\left\{ \ket{\vec{m},\vec{e},\bar{\ell}_1 \bar{\ell}_2} \right\}$, form a basis for the system Hilbert space. In this basis, $B_{w}$ have eigenvalues $(-1)^{e_w}$, $A_{v}$ have eigenvalues $(-1)^{m_v}$, and $\bar{Z}_{a}$ have eigenvalues $(-1)^{\bar{\ell}_a}$. To prepare the encoded state $\ket{\overline{0}\overline{0}} = \ket{\vec{0},\vec{0},\overline{0}\overline{0}}$ from the product state $\ket{0}^{\otimes n}$, one can measure all $B_{w}$ to obtain a configuration of excitations $\vec{e}$, then apply strings $\xi(\vec{e})$ of $Z_j$ on the dual lattice with endpoints at plaquettes $w$ with $e_{w} = 1$ to pair up and annihilate excitations,
\begin{align}
    &\sum_{\vec{e}} R_{\vec{e}} \Pi_{\vec{e}}^{\mathcal{S}_{X}} \dyad{0}{0}^{\otimes n} \Pi_{\vec{e}}^{\mathcal{S}_{X}} R_{\vec{e}}^{\dagger} = \dyad{\overline{0}\overline{0}}{\overline{0}\overline{0}} , \\ &\Pi_{\vec{e}}^{\mathcal{S}_{X}} = \prod_{w} \left( \frac{I + (-1)^{e_w} B_{w}}{2} \right) , \quad R_{\vec{e}} = \prod_{j \in \xi(\vec{e})} Z_{j} . \nonumber
\end{align}
However, if each measurement outcome is flipped with probability $p$, the vector of observed excitations is modified to be $\vec{e} \rightarrow \vec{e} + \vec{\mu}$, where $\mu_{w} \in \{0,1\}$ denotes whether or not a measurement error occurs. In this case, the noisy final state is proportional to
\begin{equation}
    \sum_{\vec{\mu}: \vert \vec{\mu} \vert = 0 \bmod 2}\left(\frac{p}{1-p}\right)^{\vert \mu \vert} \dyad{\vec{0},\vec{\mu},\overline{0} \overline{0}}{\vec{0},\vec{\mu},\overline{0} \overline{0}} .
\end{equation}
This state is qualitatively different from $\ket{\overline{0}\overline{0}}$ and is in a different LC phase for any $p > 0$. We show this in Appendix~\ref{rg-derivation} using the mixed-state renormalization group procedure developed in Ref.~\cite{sang2024rg} and by explicitly constructing a shallow LC circuit which prepares the above state from the product state $\ket{0}^{\otimes n}$. 

These examples demonstrate that QL phases do not capture whether states within a given phase, or the circuits connecting them, are robust against errors. This is important when studying how encoded quantum information can be reliably manipulated and transferred in noisy systems. LC phases suffice when all operations are local, and we will demonstrate in Section~\ref{lc-finite-depth} that all finite-depth locally reversible LC circuits are fault tolerant. However, this is no longer guaranteed in circuits with measurements and careful treatment is needed to assess robustness of states and of circuits to noise. Lastly, while LC phases of a resource state in $D+1$ dimensions can be used to diagnose whether information transferred through a noisy circuit acting on a $D$-dimensional state is recoverable~\cite{negari2024}, this diagnostic does not assess the specific phase of the $D$-dimensional output state.

\section{Fault tolerant quantum phases} \label{ft-phases}

To define quantum phases which preserve encoded information under noisy application of QL circuits, we must specify how information is encoded in a quantum many-body system, how noise acts within the QL circuits applied, and how to quantify whether encoded information is preserved or not. To discuss quantum information encoded within a phase, we will establish phase equivalence between \textit{convex sets} of states, rather than between individual states. For example, one can discuss phase equivalence of the code spaces of two QEC codes as opposed to equivalence between a specific code state of the first code and a specific code state of the second code.

Informally, we say that two convex sets of states are in the same fault tolerant phase if there exist two shallow QL circuits which (1) map each state in one set to a state in the other set, and (2) are both fault tolerant with respect to a parametrized class of noise acting on the two circuits. By fault tolerant, we loosely mean that for all noise distributions below a given strength, errors within a circuit do not propagate too badly to the output state and the information transferred by applying the circuit to a set of states is recoverable. In Section~\ref{encoded-info}, we will precisely discuss how to encode information and how to relate convex sets to each other in the noiseless case. We will then discuss classes of noise and how to insert noise in the QL circuits which map between two convex sets of states in Section~\ref{noise-char}. Lastly, in Section~\ref{ft-phase-defs} we will explicitly state the model of fault tolerance we use for a QL circuit mapping between two convex sets and present a formal definition of fault tolerant phase equivalence.

\subsection{Encoding information} \label{encoded-info}

Quantum information must be encoded in a subset of states, not just in a single state. To allow mixtures of states which represent encoded information, we will define phase equivalence with respect to convex sets of states. This differs from definitions which characterize state equivalence, however convex sets appear in prior classifications when discussing phase invariants. For example, the ground state subspaces of local, gapped Hamiltonians have the same degeneracy if representative states from each subspace are in the same LC phase~\cite{chen2010}. Similarly, the information convex sets of two topologically ordered mixed states are isomorphic if the two states are in the same LC phase~\cite{shi2020,yang2025}. To specify the information encoded within a phase here, we introduce an abstract \textit{logical} Hilbert space $\mathcal{H}_{\mathrm{log}}$, a convex subset of states $\Omega$ in this space which contains the encoded information, and maps between the physical Hilbert space $\mathcal{H}$ and $\Omega$ which specify how information is encoded in $\mathcal{H}$. The encoded information $\Omega$ can be any convex set. For example, it could be the set of all (mixed) states of a qubit, or it could be a classical bit and all convex combinations of $0$ and $1$, as depicted in Figure~\ref{fig:convex-set-examples}. There are multiple inequivalent ways to embed $\Omega$ in $\mathcal{H}$, however each embedding encodes the same information. To establish notation, given a Hilbert space $\mathcal{H}$ we will denote $\mathbb{S}(\mathcal{H})$ as the space of mixed states in $\mathcal{H}$.

\begin{definition}
\label{def:logical-information}
    Given a physical Hilbert space $\mathcal{H}$, consider a convex set of states $\Sigma \subseteq \mathbb{S}(\mathcal{H})$, a logical Hilbert space $\mathcal{H}_{\mathrm{log}}$ with $\dim \mathcal{H}_{\mathrm{log}} \leq \dim \mathcal{H}$, and a CPTP map $\mathcal{D} : \mathbb{S}(\mathcal{H}) \rightarrow \mathbb{S}(\mathcal{H}_{\mathrm{log}})$, called a decoding map for $\Sigma$. We define the encoded information of $\Sigma$, $\Omega \subseteq \mathbb{S}(\mathcal{H}_{\mathrm{log}})$, to be
    \begin{equation}
        \Omega = \mathcal{D} \left[ \Sigma \right] = \left\{ \mathcal{D}(\rho) : \rho \in \Sigma \right\} .
    \end{equation}
    Additionally, we say $\mathcal{E} : \Omega \rightarrow \Sigma$ is an encoding map for $\Sigma$ if $\mathcal{D} \circ \mathcal{E} \vert_{\Omega} = \id \vert_{\Omega}$ and for all $\bar{\varrho}_{1} , \bar{\varrho}_{2} \in \Omega$,
    \begin{equation}
        \norm{\mathcal{E}(\bar{\varrho}_{1}) - \mathcal{E}(\bar{\varrho}_{2})}_{1} = \norm{\bar{\varrho}_{1} - \bar{\varrho}_{2}}_{1} .
    \end{equation}
\end{definition}

\begin{figure}[t!]
    \centering
    \includegraphics[width=0.45\textwidth]{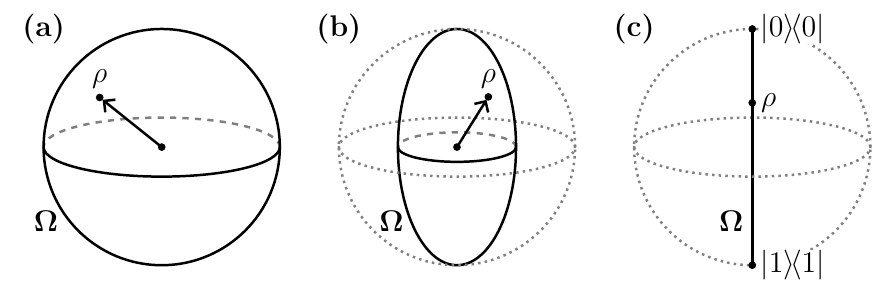}
    \caption{Examples of convex sets $\Omega$ used to encode logical information, embedded in the space of single-qubit states $\mathbb{S}\left(\mathbb{C}^{2}\right)$. (a) $\Omega \cong \mathbb{S}\left(\mathbb{C}^{2}\right)$ is the Bloch ball and encodes one logical qubit. (b) $\Omega$ is a subregion of the Bloch ball. (c) $\Omega$ is a classical probability distribution.}
    \label{fig:convex-set-examples}
\end{figure}

For logical states $\bar{\varrho}_{1},\bar{\varrho}_{2} \in \Omega$, we adopt the convention
\begin{equation}
    \bar{\varrho}_{1} \approx \bar{\varrho}_{2} \Leftrightarrow   \norm{\bar{\varrho}_{1} - \bar{\varrho}_{2}}_{1} = 1/\Omega(\poly(n)) ,
\end{equation}
where $n$ is the number of physical qudits of $\mathcal{H}$, not $\mathcal{H}_{\mathrm{log}}$. In Definition~\ref{def:logical-information}, the decoding map $\mathcal{D}$ specifies the information encoded in $\Sigma$ and establishes a correspondence between each physical state $\rho \in \Sigma$ and a logical state $\mathcal{D}(\rho) = \bar{\varrho} \in \Omega$. Since $\Sigma$ is convex and $\mathcal{D}$ is a linear map, $\Omega$ is a convex set as well. However, $\Sigma$ and $\Omega$ are not necessarily isomorphic, and it is possible for $\mathcal{D}$ to map multiple states $\rho_{i} \in \Sigma$ to the same $\bar{\varrho} \in \Omega$. Consequently, the image of an encoding map $\mathcal{E}$ is a subset of $\Sigma$ isomorphic to $\Omega$, but need not be the entire set $\Sigma$. Different encoding maps $\mathcal{E}^{\alpha}$ may also map $\Omega$ to different subsets of $\Sigma$. Each $\Sigma$ also may represent information encoded in different ways. For example, if two decoding maps have the same image $\Omega$ but map states in $\Sigma$ to different logical states in $\Omega$, these decoding maps may represent different logical bases of $\Sigma$. In this case, a map from $\Sigma$ to itself but with different initial and final decoding maps would represent a logical operation. For clarity, we will always specify either the specific decoding map $\mathcal{D}$ or the encoded information $\Omega$ corresponding to $\Sigma$.

Given a convex set $\Sigma_{1}$ and a quantum circuit $\Lambda_{1}$, we say $\Lambda_{1} : \Sigma_{1} \rightarrow \Sigma_{2}$ if for every state $\rho_{1} \in \Sigma_{1}$, $\Lambda_{1}(\rho_{1}) \approx \rho_{2} \in \Sigma_{2}$. If $\Sigma_{1}$ and $\Sigma_{2}$ contain the same encoded information, $\Omega \cong \mathcal{D}_{1} \left[ \Sigma_{1} \right] \cong \mathcal{D}_{2} \left[ \Sigma_{2} \right]$, then the circuit $\Lambda_{1}$ may be reversible by some other circuit $\Lambda_{2} : \Sigma_{2} \rightarrow \Sigma_{1}$, where for every $\rho_{2} \in \Sigma_{2}$, $\Lambda_{2}(\rho_{2}) \approx \rho_{1} \in \Sigma_{1}$. We say that $\Lambda_{1}$ and $\Lambda_{2}$ reverse each other with respect to $\Omega$ if
\begin{align}
    \label{eq:decoding-equivalence}
    \forall \rho_{1} \in \Sigma_{1} : \mathcal{D}_{2} \circ \Lambda_{1} (\rho_{1}) &\approx \mathcal{D}_{1} (\rho_{1}) , \nonumber \\
    \forall \rho_{2} \in \Sigma_{2} : \mathcal{D}_{1} \circ \Lambda_{2} (\rho_{2}) &\approx \mathcal{D}_{2} (\rho_{2}) .
\end{align}
Eq.~(\ref{eq:decoding-equivalence}), $\Lambda_{1} : \Sigma_{1} \rightarrow \Sigma_{2}$, and $\Lambda_{2} : \Sigma_{2} \rightarrow \Sigma_{1}$ also imply
\begin{align}
    \label{eq:decoding-equivalence-alternate}
    \forall \rho_{1} \in \Sigma_{1} : \mathcal{D}_{1} \circ \Lambda_{2} \circ \Lambda_{1} (\rho_{1}) &\approx \mathcal{D}_{1} (\rho_{1}) , \nonumber \\
    \forall \rho_{2} \in \Sigma_{2} : \mathcal{D}_{2} \circ \Lambda_{1} \circ \Lambda_{2} (\rho_{2}) &\approx \mathcal{D}_{2} (\rho_{2}) ,
\end{align}
as shown in Appendix~\ref{short-proofs}. However, if $\Sigma_{1}$ and $\Sigma_{2}$ contain different encoded information, $\mathcal{D}_{1} \left[ \Sigma_{1} \right] \ncong \mathcal{D}_{2} \left[ \Sigma_{2} \right]$, then it is impossible to construct quantum circuits $\Lambda_{1}$ and $\Lambda_{2}$ which satisfy Eq.~(\ref{eq:decoding-equivalence}). Therefore, a necessary condition for convex sets $\Sigma_{1}$ and $\Sigma_{2}$ equipped with respective decoding maps $\mathcal{D}_{1}$ and $\mathcal{D}_{2}$ to be in the same phase is that they correspond to the same encoded information $\Omega$.

One important case of the above arises when $\mathcal{D}_{1}$ and $\mathcal{D}_{2}$ are isometries on $\Sigma_{1}$ and $\Sigma_{2}$, namely $\norm{\rho - \sigma}_{1} = \norm{\mathcal{D}_{i}(\rho) - \mathcal{D}_{i}(\sigma)}_{1}$ for all $\rho,\sigma \in \Sigma_{i}$, $i=1,2$. Since $\mathcal{D}_{i}$ uniquely maps each $\rho_{i} \in \Sigma_{i}$ to $\mathcal{D}_{i} (\rho_{i}) \in \Omega$ in this case, $\mathcal{D}_{2} \circ \Lambda_{1} (\rho_{1}) \approx \mathcal{D}_{2} (\rho_{2})$ if and only if $\Lambda_{1} (\rho_{1}) \approx \rho_{2}$, and $\mathcal{D}_{1} \circ \Lambda_{2} (\rho_{2}) \approx \mathcal{D}_{1} (\rho_{1})$ if and only if $\Lambda_{2} (\rho_{2}) \approx \rho_{1}$. Also, since $\Sigma_{1}$ and $\Sigma_{2}$ are both isomorphic to $\Omega$ here, we may establish equivalence or inequivalence at the level of states in $\Sigma_{1}$ and $\Sigma_{2}$ and do not need to invoke any decoding maps. However, the general construction in Definition~\ref{def:logical-information} and discussed above is useful as it emphasizes equivalence of encoded information and allows for different decoding maps such as those used to treat subsystem QEC codes, as we will discuss in Section~\ref{cc-phases}.

Lastly, a second special case of the above occurs when $\Sigma = \{\rho\}$ contains a single state, which means $\Omega = \mathcal{D} \left[ \Sigma \right]$ contains a single element. We will discuss equivalence between individual states further in Section~\ref{ft-state-prep}.

\subsection{Adding noise} \label{noise-char}

Before discussing how to model noise acting within a quantum circuit, it is useful to first review how noise acts on a static quantum state, $\rho$. We model environmental noise acting on $\rho$ with a quantum channel
\begin{equation}
    \mathcal{N} (\rho) = \sum_{i} N_{i}\rho N_{i}^{\dagger} , \quad \sum_{i} N_{i}^{\dagger} N_{i} = I ,
\end{equation}
where we sum over a set of Kraus operators $\left\{ N_{i} \right\}$
describing the channel $\mathcal{N}$. Often it is convenient to refer to a channel by its Kraus operators, $\mathcal{N} = \left\{ N_i \right\}$. In this notation, channel composition may be written as $\mathcal{N} \circ \mathcal{M} = \left\{ N_i \right\} \circ \left\{ M_j \right\} = \left\{ N_i M_j \right\}$. In this paper, we will mainly consider noise channels $\mathcal{N}$ which admit a decomposition as a convex sum of channels $\Phi_{\alpha}$, where various constraints may be placed on allowed channels in the set $\{ \Phi_{\alpha} \}$. We call this type of noise \textit{general channel noise}, and here $\mathcal{N}$ takes the form
\begin{equation}
    \label{eq:qudit-noise}
    \mathcal{N} (\rho) = \sum_{\alpha} \pr(\alpha) \Phi_{\alpha} (\rho) , \quad \sum_{\alpha} \pr(\alpha) = 1 .
\end{equation}
We may interpret $\mathcal{N}$ as the environment applying a specific channel $\Phi_{\alpha}$ to the system with probability $\pr(\alpha)$. Examples of constraints one may impose include restricting $\Phi_{\alpha}$ to be unitary operators, $\Phi_{\alpha}(\rho) = U_{\alpha} \rho U_{\alpha}^{\dagger}$, or further restricting to Pauli noise, where each $U_{\alpha} = \prod_{q} X_{q}^{a_{q}} Z_{q}^{b_{q}}$ is some Pauli operator. Pauli noise is often studied as it is easy to analyze, however it does not always capture how a system interacts with a generic environment and does not always remain Pauli when propagated through a generic quantum circuit. Because of this, we will usually consider general channel noise, Eq.~(\ref{eq:qudit-noise}), and specialize to Pauli noise only when explicitly specified.

Since one does not generally have control over the specific noise channels $\mathcal{N}$ which act on a system, we will study parametrized classes of qudit noise channels, $\mathbb{N}_{\tau}$. Generally, $\mathbb{N}_{\tau}$ is a set of quantum channels, $\mathcal{N} \in \mathbb{N}_{\tau}$, where the parameter $\tau \in [0,1)$ loosely quantifies the noise strength for all $\mathcal{N} \in \mathbb{N}_{\tau}$. Before enumerating properties which all classes $\mathbb{N}_{\tau}$ we study must satisfy, it is necessary to first introduce the notion of a $\lambda$-bounded probability distribution and the class of local stochastic noise~\cite{gottesman2014}.
\begin{definition}
    \label{def:tau-bound}
    Consider a probability distribution $\pr(A)$ over subsets $A \subseteq B$ of some set $B$. $\pr(A)$ is $\lambda$-bounded if there exists $\lambda \in [0,1)$ such that for all $A \subseteq B$,
    \begin{equation}
        \sum_{A^\prime : A^\prime \supseteq A} \pr(A^\prime) \leq \lambda^{\vert A \vert} .
    \end{equation}
\end{definition}

The class of local stochastic noise, $\mathbb{L}_{\lambda}$, is the set of all error channels $\mathcal{N}$ which admit a decomposition of the form of Eq.~(\ref{eq:qudit-noise}) such that the probability distribution over the supports of channels $\Phi_{\alpha}$,
\begin{align}
    \label{eq:local-stochastic-pauli}
    \pr(A) = \sum_{\alpha : \supp \Phi_{\alpha} = A} \pr(\alpha) ,
\end{align}
is $\lambda$-bounded. $\mathbb{L}_{\lambda}$ is useful as it contains physically realistic noise channels, where errors may be correlated but the likelihood of a specific error $\Phi_{\alpha}$ decreases with its weight. In some situations it is useful to consider more general classes of noise than $\mathbb{L}_{\lambda}$, such as when local stochastic noise does not strictly remain local when propagated through a QL circuit and when enlarging the class of noise considered makes proofs more tractable. However, all classes we discuss will contain $\mathbb{L}_{\lambda}$ as a subclass to maintain physical relevance, and will satisfy the following properties.

\begin{definition}
    \label{def:noise-class}
    A parametrized class of noise $\mathbb{N}_{\tau}$ is a set of quantum channels $\mathcal{N} \in \mathbb{N}_{\tau}$ for which the following properties hold for $\tau \in (0,1)$.
    \begin{enumerate}
        \item For all $\kappa \leq \tau$, $\mathbb{N}_{\kappa} \subseteq \mathbb{N}_{\tau}$.
        \item There exists $\lambda > 0$ such that $\mathbb{L}_{\lambda} \subseteq \mathbb{N}_{\tau}$.
        \item For $\tau_{1},\tau_{2} > 0$, $\mathbb{N}_{\tau_{1}} \circ \mathbb{N}_{\tau_{2}} \subseteq \mathbb{N}_{\tau_{1} + \tau_{2}}$, i.e., if $\mathcal{N}_{1} \in \mathbb{N}_{\tau_{1}}$, $\mathcal{N}_{2} \in \mathbb{N}_{\tau_{2}}$, then $\mathcal{N}_{1} \circ \mathcal{N}_{2} \in \mathbb{N}_{\tau_{1} + \tau_{2}}$.
    \end{enumerate}
\end{definition}
Here, property 1 ensures $\tau$ sufficiently characterizes the intensity of noise, property 2 ensures classes are compatible with local stochastic noise, and property 3 ensures that channels compose in a controlled way. We verify in Appendix~\ref{short-proofs} that $\mathbb{L}_{\lambda}$ satisfies the above properties. In certain situations, one may want to restrict classes $\mathbb{N}_{\tau}$ to Pauli classes which contain only Pauli noise. When studying Pauli classes, one must also restrict $\mathbb{L}_{\lambda}$ used in Definition~\ref{def:noise-class} to local stochastic Pauli noise, where the probability distribution of the supports of Pauli operators are $\lambda$-bounded. However, Pauli classes must still satisfy the other two properties of Definition~\ref{def:noise-class} as stated.

Now, consider a QL circuit $\Lambda$. To construct a noisy version of $\Lambda$, we insert general channel noise before and after each layer of local channel gates in $\Lambda$ and we insert measurement errors $\mu$, with probability $\pr(\mu)$, immediately before measurement outcomes are classically processed and a channel is conditionally applied. If a measured operator has $k_{\alpha}$ distinct eigenvalues, denote the true measurement outcome $\sigma_{\alpha} \in \mathbb{Z}_{k_{\alpha}}$. An error $\mu_{\alpha}$ then changes the reported noisy measurement outcome to $\sigma_{\alpha} + \mu_{\alpha} \bmod k_{\alpha}$. A representative example of a QL circuit $\Lambda$ is depicted in Figure~\ref{fig:noisy-channel}(a) and a specific pattern of noise inserted in $\Lambda$ is depicted in Figure~\ref{fig:noisy-channel}(b).

Since any QL circuit $\Lambda$ can be written in the form of Eq.~(\ref{eq:ql-circuit}), a noisy realization of $\Lambda$, $\Lambda^{\mathcal{N}}$, may be written as
\begin{equation}
    \Lambda^{\mathcal{N}} = \mathcal{K}^{(b) ; \mathcal{N}} \circ \mathcal{C}^{(b) ; \mathcal{N}} \circ \ldots \circ \mathcal{K}^{(1) ; \mathcal{N}} \circ \mathcal{C}^{(1) ; \mathcal{N}} .
\end{equation}
Each noisy LC circuit $\mathcal{C}^{(a) ; \mathcal{N}}$ decomposes into alternating layers of local channel gates $\mathcal{C}^{(a)}_{i}$ and noise channels $\mathcal{N}_{i}$,
\begin{equation}
    \mathcal{C}^{(a) ; \mathcal{N}} = \mathcal{N}_{\ell+1} \circ \mathcal{C}^{(a)}_{\ell} \circ \mathcal{N}_{\ell} \circ \ldots \circ \mathcal{C}^{(a)}_{1} \circ \mathcal{N}_{1} .
\end{equation}
Each noisy measurement-feedback circuit $\mathcal{K}^{(a) ; \mathcal{N}}$ contains both measurement noise, which affects which specific feedback channel is applied, and qudit noise, acting on the feedback channel itself. Similar to above,
\begin{align}
    \mathcal{K}^{(a) ; \mathcal{N}} (\rho) &= \sum_{\sigma,\mu} \pr(\mu) \mathcal{C}^{(\sigma+\mu);\mathcal{N}}
    \left( \Pi_{\sigma}^{\mathcal{I}} \rho {\Pi_{\sigma}^{\mathcal{I}}}^{\dagger} \right) , \\
    \mathcal{C}^{(\sigma+\mu) ; \mathcal{N}} &= \mathcal{N}_{\ell+1} \circ \mathcal{C}^{(\sigma+\mu)}_{\ell} \circ \mathcal{N}_{\ell} \circ \ldots \circ \mathcal{C}^{(\sigma+\mu)}_{1} \circ \mathcal{N}_{1} . \nonumber
\end{align}
Here we use a noise model where errors are uncorrelated in time. This will ensure transitivity of the phase definition given in the following section and allows one to decompose circuits and analyze individual LC circuits $\mathcal{C}^{(a)}_{i}$ and measurement-feedback circuits $\mathcal{K}^{(a)}$. It is straightforward to introduce time correlations by summing over patterns of noise inserted in $\Lambda$ of the form depicted in Figure~\ref{fig:noisy-channel}(b). However, not all the arguments made in this paper carry over to this setting. Most notably, transitivity of fault tolerance will no longer always hold.

\begin{figure}[t!]
    \centering
    \includegraphics[width=0.45\textwidth]{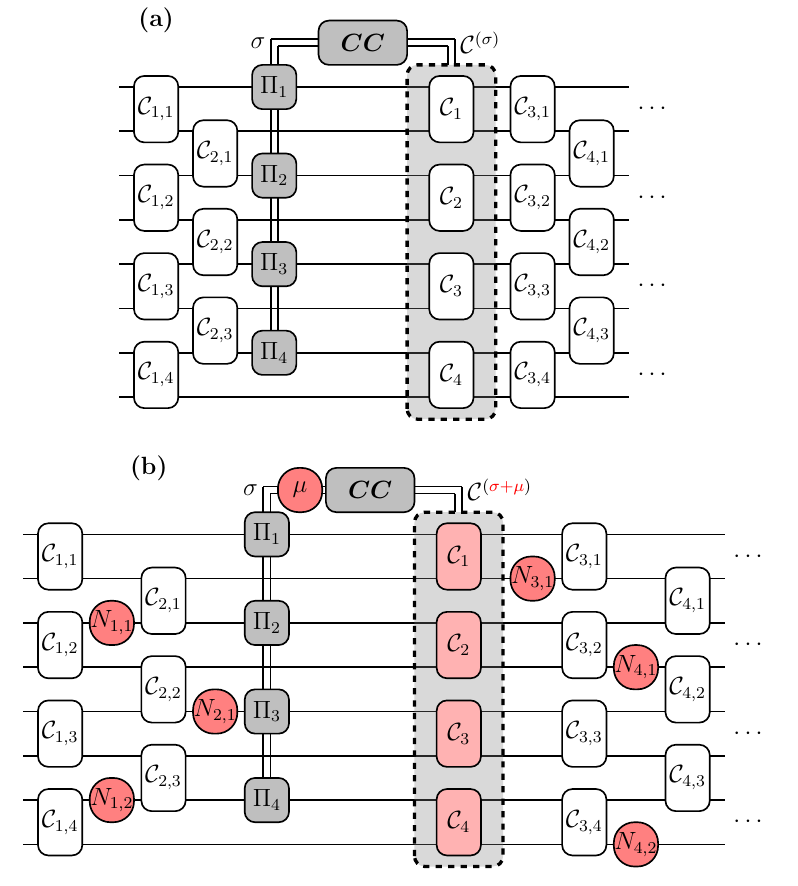}
    \caption{(a) Example of a noiseless quantum-local (QL) circuit. A QL circuit $\Lambda$ consists of layers of locally reversible local channel gates, local measurements, global classical processing of measurement outcomes and feedback applied in the form of local unitary operators. (b) Specific instance of a noisy QL circuit with a specific pattern of qudit ($N_{i}$) and measurement ($\mu$) noise inserted between layers.}
    \label{fig:noisy-channel}
\end{figure}

Given a QL circuit $\Lambda$, we can define parametrized classes of noisy realizations of $\Lambda$, $\Lambda^{(\mathbb{N}_{\tau}, \mathbb{M}_{\zeta})}$. Here, $\mathbb{N}_{\tau}$ is a class of qudit noise, as given in Definition~\ref{def:noise-class}, and $\mathbb{M}_{\zeta}$ is a class of measurement error distributions. A specific noisy circuit $\Lambda^{\mathcal{N}}$ is an element of the class $\Lambda^{(\mathbb{N}_{\tau}, \mathbb{M}_{\zeta})}$, $\Lambda^{\mathcal{N}} \in \Lambda^{(\mathbb{N}_{\tau}, \mathbb{M}_{\zeta})}$, if each layer of qudit noise $\mathcal{N}_{i} \in \mathbb{N}_{\tau}$ and each distribution of measurement errors $\left\{ \mu , \pr(\mu) \right\}_{\mu} \in \mathbb{M}_{\zeta}$. In this work we only consider local stochastic measurement noise, where $\mathbb{M}_{\zeta}$ is the set of measurement error distributions $\left\{ \mu , \pr(\mu) \right\}_{\mu}$ such that for all $\left\{ \mu , \pr(\mu) \right\}_{\mu} \in \mathbb{M}_{\zeta}$, the probability distribution over the supports of $\mu$,
\begin{equation}
    \label{eq:local-stochastic-meas}
    \pr(A) = \sum_{\mu : \supp \mu = A} \pr(\mu) ,
\end{equation}
is $\zeta$-bounded as stated in Definition~\ref{def:tau-bound}.

\subsection{Fault tolerant quantum phases} \label{ft-phase-defs}

Given a convex set of states $\Sigma$ and a noise channel $\mathcal{N}$ acting on states $\rho \in \Sigma$, we wish to quantify whether information encoded in $\Sigma$ is preserved or corrupted by $\mathcal{N}$. To do so, we will introduce a recovery channel $\mathcal{R}$ associated with $\Sigma$, which we will use as an information-theoretic tool to assess whether recovery of information from a corrupted state is possible.
\begin{definition}
\label{def:recovery-channel}
    A quantum channel $\mathcal{R}$ is a recovery channel for the convex set of states $\Sigma$ and decoding map $\mathcal{D}$ if the following conditions hold.
    \begin{enumerate}
        \item For all states $\rho \in \mathbb{S}(\mathcal{H})$, $\mathcal{R}(\rho) \in \Sigma$.
        \item For all states $\rho \in \Sigma$, $\mathcal{D} \circ \mathcal{R} (\rho) \approx \mathcal{D}(\rho)$.
    \end{enumerate}
\end{definition}
We will place no restrictions on whether $\mathcal{R}$ is realizable with a shallow QL circuit or not, and we will always assume it is noiseless. Also, many $\mathcal{R}$ may exist for a given $\Sigma$ and may output different recovered states $\mathcal{R} \circ \mathcal{N} (\rho)$ when one attempts to undo the action of $\mathcal{N}$ on a state $\rho \in \Sigma$. We can quantify the effectiveness of a given recovery channel $\mathcal{R}$ against some noise channel $\mathcal{N}$ using the trace distance $\norm{\mathcal{D}(\rho) - \mathcal{D} \circ \mathcal{R} \circ \mathcal{N}(\rho)}_{1}$. A useful diagnostic is whether $\mathcal{N}$ is $\mathcal{R}$-recoverable, as defined below.
\begin{definition}
\label{def:recoverable}
    A quantum channel $\mathcal{N}$ acting on a convex set of states $\Sigma$ with decoding map $\mathcal{D}$ is $\mathcal{R}$-recoverable if
    \begin{equation}
        \label{eq:recoverable-channel}
        \sup_{\left\{\rho \in \Sigma\right\}} \norm{\mathcal{D}(\rho) - \mathcal{D} \circ \mathcal{R} \circ \mathcal{N}(\rho)}_{1} = 1/\Omega(\poly(n)) .
    \end{equation}
    We say that a class of noise $\mathbb{N}_{\tau}$ is $\mathcal{R}$-recoverable if every $\mathcal{N} \in \mathbb{N}_{\tau}$ is $\mathcal{R}$-recoverable.
\end{definition}
If $\mathcal{N}$ is $\mathcal{R}$-recoverable, then $\mathcal{D} \circ \mathcal{R} \circ \mathcal{N} (\rho) \approx \mathcal{D}(\rho)$ for all $\rho \in \Sigma$.
We now state the definition of fault tolerance which we use throughout this work.
\begin{definition}[Fault tolerant circuit]
\label{def:ft-channel}
    A quantum circuit $\Lambda : \Sigma_{1} \rightarrow \Sigma_{2}$ is fault tolerant, given decoding map $\mathcal{D}_{2}$ for $\Sigma_{2}$, against a class of noise $(\mathbb{N}_{\tau}, \mathbb{M}_{\zeta})$ if there exist critical noise strengths $\tau^{\star},\zeta^{\star} > 0$ independent of the system size $n$ and a continuous function $\eta(\tau,\zeta)$ with $\eta(0,0) = 0$ such that for all $\tau < \tau^{\star}$, $\zeta < \zeta^{\star}$, and $\mathcal{N} \in (\mathbb{N}_{\tau}, \mathbb{M}_{\zeta})$, the following conditions hold.
    \begin{enumerate}
        \item There exists a channel $\widetilde{\mathcal{N}} \in \mathbb{N}_{\eta}$, where $\eta = \eta(\tau,\zeta)$, such that $\Lambda^{\mathcal{N}} (\rho_1) \approx \widetilde{\mathcal{N}} \circ \Lambda (\rho_1)$ for all $\rho_1 \in \Sigma_{1}$.
        \item There exists a recovery channel $\mathcal{R}_{2}$ for $\Sigma_{2}$ such that $\widetilde{\mathcal{N}}$ is $\mathcal{R}_{2}$-recoverable for all $\mathcal{N} \in (\mathbb{N}_{\tau}, \mathbb{M}_{\zeta})$.
    \end{enumerate}
\end{definition}

In property 1 of Definition~\ref{def:ft-channel}, the residual noise channel $\widetilde{\mathcal{N}}$ may depend on $\Lambda$ and $\mathcal{N}$ but is independent of the state $\rho_1 \in \Sigma_{1}$ which $\Lambda^{\mathcal{N}}$ is applied to. The requirement that $\eta(\tau,\zeta)$ is continuous ensures that the residual noise grows in a controllable way when propagated through the circuit $\Lambda$. Property 2 then ensures logical information transferred by $\Lambda$ will be recovered successfully with high probability. For $\rho_{2} \approx \Lambda (\rho_1)$, these properties imply
\begin{align}
    \label{eq:ft-implications}
    \widetilde{\mathcal{N}} \circ \Lambda (\rho_1) \approx \widetilde{\mathcal{N}}(\rho_{2}) \quad &\text{(property 1)} \\
    \mathcal{D}_{2}(\rho_{2}) \approx \mathcal{D}_{2} \circ \mathcal{R}_{2} \circ \Lambda^{\mathcal{N}} (\rho_1) \quad &\text{(property 2)} \nonumber
\end{align}
as shown in Appendix~\ref{short-proofs}.
Also, the composition of two fault tolerant channels is fault tolerant, as stated by the following lemma and proven in Appendix~\ref{short-proofs}.

\begin{lemma}
\label{lemma:ft-composition}
    Suppose $\Lambda_{1} : \Sigma_{1} \rightarrow \Sigma_{2}$ and $\Lambda_{2} : \Sigma_{2} \rightarrow \Sigma_{3}$ are both fault tolerant, given decoding maps $\mathcal{D}_{2}$ for $\Sigma_{2}$ and $\mathcal{D}_{3}$ for $\Sigma_{3}$, against the class $(\mathbb{N}_{\tau}, \mathbb{M}_{\zeta})$, with minimum critical noise strengths $\tau^{\star}$ and $\zeta^{\star}$ for both channels. Then, the channel $\Lambda_{2} \circ \Lambda_{1} : \Sigma_{1} \rightarrow \Sigma_{3}$ is fault tolerant, given $\mathcal{D}_{3}$, against $(\mathbb{N}_{\tau}, \mathbb{M}_{\zeta})$ with modified critical noise strengths inferred from $\zeta < \zeta^{\star}$ and $\tau + \eta_{1} < \tau^{\star}$, where $\eta_{1} = \eta_{1}(\tau,\zeta)$ is the residual noise strength for $\Lambda_{1}$ given in Definition~\ref{def:ft-channel}.
\end{lemma}

Definition~\ref{def:ft-channel} may be used to establish a partial order on convex sets of states based on the ability to fault tolerantly transfer information from one convex set to another. We say that $\Sigma_{1} \succeq \Sigma_{2}$ with respect to the class of noise $(\mathbb{N}_{\tau}, \mathbb{M}_{\zeta})$ if there exists an $(r,\ell)$ QL circuit $\Lambda$ with $r\ell = O(\polylog(n))$ such that for all $\rho_{1} \in \Sigma_{1}$,
\begin{equation}
    \Lambda(\rho_{1}) \approx \rho_{2} \in \Sigma_{2} , \quad \mathcal{D}_{2} \circ \Lambda (\rho_{1}) \approx \mathcal{D}_{1} (\rho_{1}) ,
\end{equation}
and $\Lambda : \Sigma_{1} \rightarrow \Sigma_{2}$ is fault tolerant against $(\mathbb{N}_{\tau}, \mathbb{M}_{\zeta})$ given $\mathcal{D}_{2}$.
We may then define fault tolerant phase equivalence as the case when this information transfer is reversible.
\begin{definition}[FTQL phases]
\label{def:ft-phase-equivalence}
    Two convex sets of states $\Sigma_{1}$ and $\Sigma_{2}$ are in the same fault tolerant phase (FTQL phase) with respect to decoding maps $\mathcal{D}_{1}$ and $\mathcal{D}_{2}$ and the noise class $(\mathbb{N}_{\tau}, \mathbb{M}_{\zeta})$, $\Sigma_{1} \sim \Sigma_{2}$, if for $i=1,2$ there exists an $(r_{i},\ell_{i})$ quantum-local circuit $\Lambda_{i}$ with $r_i \ell_i = O(\polylog(n))$, such that
    \begin{enumerate}
        \item $\Lambda_{1}$ and $\Lambda_{2}$ reverse each other, meaning $\Lambda_{1} : \Sigma_{1} \rightarrow \Sigma_{2}$, $\Lambda_{2} : \Sigma_{2} \rightarrow \Sigma_{1}$, and Eq.~(\ref{eq:decoding-equivalence}) is satisfied,
        \item $\Lambda_{1}$ and $\Lambda_{2}$ are both fault tolerant for $(\mathbb{N}_{\tau}, \mathbb{M}_{\zeta})$.
    \end{enumerate}
\end{definition}

Definition~\ref{def:ft-phase-equivalence} essentially states that $\Sigma_{1}$ and $\Sigma_{2}$ are in the same phase if one can fault tolerantly transfer information from $\Sigma_{1}$ to $\Sigma_{2}$ and transfer the same information back from $\Sigma_{2}$ to $\Sigma_{1}$, guaranteeing this information is preserved. In terms of the partial order discussed above, $\Sigma_{1} \sim \Sigma_{2}$ if $\Sigma_{1} \succeq \Sigma_{2}$, $\Sigma_{2} \succeq \Sigma_{1}$, and $\Lambda_{1}$ and $\Lambda_{2}$ reverse each other. As discussed in Section~\ref{encoded-info}, a necessary condition for $\Sigma_{1} \sim \Sigma_{2}$ is that they encode the same information $\Omega$. Because of this, we will generally discuss phase equivalence of convex sets $\Sigma$ with respect to a specific set of encoded information $\Omega$, as two convex sets are always in different phases otherwise.

We comment that Definition~\ref{def:ft-phase-equivalence} is a valid equivalence relation as it is reflexive, symmetric, and transitive. Symmetry follows by construction, transitivity is proven in Lemma~\ref{lemma:ft-composition}, and reflexivity follows by defining the identity channel, $\mathbb{I}$, as a QL circuit with zero gates and zero layers. Since there are no locations where noise can be inserted in $\mathbb{I}$, any noisy realization $\mathbb{I}^{\mathcal{N}} = \mathbb{I}$ and it is always fault tolerant by definition. Therefore, any convex set $\Sigma$ is fault tolerantly connected to itself by $\mathbb{I}$. 
Lastly, one may see similarities between Definition~\ref{def:ft-phase-equivalence} and code switching in QEC. We will elaborate on these similarities and other properties in Section~\ref{ft-phase-properties}.

\subsection{Equivalence of individual states} \label{ft-state-prep}

If one chooses $\Sigma = \{\rho\}$ to be a single state, then every $\Omega = \mathcal{D} \left[ \Sigma \right]$ contains a single element. Definition~\ref{def:ft-phase-equivalence} then reduces to an equivalence relation between individual states. The recovery channel $\mathcal{R} (\cdot) = \rho \Tr\left [\cdot \right]$ will return $\rho$ irrespective of the state it is applied to, meaning property 2 of Definition~\ref{def:ft-channel} holds for any circuit $\Lambda$ connecting states $\rho_{1}$ and $\rho_{2}$. Therefore, FTQL state equivalence concerns whether there exist shallow QL circuits connecting $\rho_{1}$ and $\rho_{2}$ where noise does not propagate uncontrollably, specifically whether property 1 of Definition~\ref{def:ft-channel} is satisfied. State equivalence is similar to the noiseless phase classifications given by Definitions~\ref{def:lc-phase} and~\ref{def:ql-phase} and highlights the difference between state preparation, which depends only on equivalence between states, and transfer of encoded information, which depends on equivalence between convex sets of states. Equivalence between two states $\rho_{1} \in \Sigma_{1}$ and $\rho_{2} \in \Sigma_{2}$ does not necessarily imply the two convex sets $\Sigma_{1}$ and $\Sigma_{2}$ are equivalent. Even if one can prepare individual states which span $\Sigma_{2}$ from states in $\Sigma_{1}$, if the preparation methods differ one cannot necessarily prepare an arbitrary state in $\Sigma_{2}$ from some state in $\Sigma_{1}$. For example, in Section~\ref{q-form-tc}, we will show that the $\dyad{\overline{0}}{\overline{0}}^{\otimes k}$ state of the 4D toric code is equivalent to the product state, $\dyad{0}{0}^{\otimes n}$. However, we show in Section~\ref{trivial-phases} that the code space of any topological code, such as the 4D toric code, is not in the same phase as any set of $k$ unencoded qubits. We will study instances and properties of encoded information transfer in Sections~\ref{ft-phase-properties} and~\ref{cc-phases} and of state preparation in Sections~\ref{robust-meas-feedback} and~\ref{non-robust-circuits}.

\section{Stability and properties of FTQL phases} \label{ft-phase-properties}

While Definition~\ref{def:ft-phase-equivalence} applies to all convex sets of states $\Sigma$, not all $\Sigma$ are useful for preserving encoded information. To be useful, $\Sigma$ must at least be stable, which means here that the information encoded in $\Sigma$ is in principle recoverable when subjected to local stochastic noise with strength below some nonzero threshold.
\begin{definition}
\label{def:stable-phase}
    A convex set of states $\Sigma$ with decoding map $\mathcal{D}$ is stable if there exists a recovery channel $\mathcal{R}$ for $\Sigma$
    and critical noise strength $\lambda^{\star} > 0$ such that the noise class $\mathbb{L}_{\lambda^{\star}}$ is $\mathcal{R}$-recoverable.
\end{definition}

Since here every $\mathcal{N} \in \mathbb{L}_{\lambda^{\star}}$ must be $\mathcal{R}$-recoverable for the same $\mathcal{R}$, stability is a stronger requirement than asking whether a recovery channel exists for each $\mathcal{N} \in \mathbb{L}_{\lambda^{\star}}$, as the latter allows one to use different recovery channels for different $\mathcal{N}$. Also recall $\mathbb{L}_{\lambda} \subseteq \mathbb{L}_{\lambda^{\star}}$ for all $\lambda < \lambda^{\star}$ by definition. It turns out that stability here is a property of an entire phase, not just a specific convex set $\Sigma$, i.e., if $\Sigma$ is stable then all convex sets in the same phase as $\Sigma$ are stable as well. In fact, if a phase contains more than one distinct convex set, then each convex set in this phase is stable, as shown by the following lemma.

\begin{lemma}
\label{lemma:stable-consistency}
    Consider a convex set of states $\Sigma_{1}$ with decoding map $\mathcal{D}_{1}$. Suppose there exists $\Sigma_{2}$ such that $\Sigma_{1} \sim \Sigma_{2}$ but $\mathbb{I} : \Sigma_{1} \not\rightarrow \Sigma_{2}$ (i.e., one cannot connect the two sets with the identity). Then $\Sigma_{1}$ is stable.
\end{lemma}
\begin{proof}
    Suppose $\Sigma_{1} \sim \Sigma_{2}$. Suppose $\mathcal{D}_{2}$ is a decoding map for $\Sigma_{2}$. Suppose the channels $\Lambda_{1} : \Sigma_{1} \rightarrow \Sigma_{2}$ and $\Lambda_{2} : \Sigma_{2} \rightarrow \Sigma_{1}$ satisfy the conditions of Definition~\ref{def:ft-phase-equivalence} for some class $(\mathbb{N}_{\tau}, \mathbb{M}_{\zeta})$. As mentioned above, $\Lambda_{1} \neq \mathbb{I}$ cannot be the identity channel, therefore noise can be inserted in $\Lambda_{1}$. Let $\mathcal{R}_{2}$ be a recovery channel for $\Sigma_{2}$, satisfying the conditions of Definition~\ref{def:ft-channel} for $(\mathbb{N}_{\tau}, \mathbb{M}_{\zeta})$. 
    Since $\mathbb{L}_{\lambda} \subseteq \mathbb{N}_{\tau}$ for some $\lambda$, the noisy channel $\Lambda_{1} \circ \mathcal{N} \in \Lambda_{1}^{(\mathbb{N}_{\tau}, \mathbb{M}_{\zeta})}$ for all $\mathcal{N} \in \mathbb{L}_{\lambda}$. Since $\Lambda_{1}$ is fault tolerant against $(\mathbb{N}_{\tau}, \mathbb{M}_{\zeta})$, there exists $\lambda^{\star} > 0$ such that for all $\mathcal{N} \in \mathbb{L}_{\lambda^{\star}}$,
    \begin{equation}
        \label{eq:stochastic-recovery}
        \mathcal{D}_{2} \circ \mathcal{R}_{2} \circ \Lambda_{1} \circ \mathcal{N} (\rho_{1}) \approx \mathcal{D}_{2} \circ \Lambda_{1}(\rho_{1})
    \end{equation}
    for all $\rho_{1} \in \Sigma_{1}$. However, since $\mathcal{R}_{2} \circ \Lambda_{1} \circ \mathcal{N} (\rho_{1}) \in \Sigma_{2}$, applying Eq.~(\ref{eq:decoding-equivalence}) to both sides of Eq.~(\ref{eq:stochastic-recovery}) implies
    \begin{equation}
        \mathcal{D}_{1} \circ \Lambda_{2} \circ \mathcal{R}_{2} \circ \Lambda_{1} \circ \mathcal{N} (\rho_{1}) \approx \mathcal{D}_{1}(\rho_{1})
    \end{equation}
    for all $\rho_{1} \in \Sigma_{1}$. This means the noise class $\mathbb{L}_{\lambda^{\star}}$ acting on $\Sigma_{1}$ is $\mathcal{R}_{1}$-recoverable with $\mathcal{R}_{1} = \Lambda_{2} \circ \mathcal{R}_{2} \circ \Lambda_{1}$. Therefore, $\Sigma_{1}$ is stable according to Definition~\ref{def:stable-phase}.
\end{proof}

The contrapositive of Lemma~\ref{lemma:stable-consistency} is also important. Specifically, if $\Sigma_{1}$ is not stable, then no $\Sigma_{2}$ exists such that $\Sigma_{1} \sim \Sigma_{2}$ and $\mathbb{I} : \Sigma_{1} \not\rightarrow \Sigma_{2}$. Stability is most directly interpreted as a necessary condition for $\Sigma$ to be a reliable quantum memory. Stable phases $\Sigma$ will preserve encoded information when idle even in the presence of noise, which is not the case when $\Sigma$ is not stable. As we will discuss in the following section, many QEC codes are stable. Lastly, we note that stability does not matter for state preparation. Recall from Section~\ref{ft-state-prep} that if $\Sigma = \{\rho\}$, one can choose a recovery channel $\mathcal{R} (\cdot) = \rho \Tr\left [\cdot \right]$, which returns $\rho$ irrespective of the state it is applied to. In this case, any arbitrary noise channel is $\mathcal{R}$-recoverable.

\subsection{Stabilizer codes represent stable phases} \label{qec-relation}

One useful class of examples occurs when the convex set $\Sigma$ corresponds to the set of logical states of some quantum error correcting (QEC) code. A QEC code is defined by a subspace $\mathcal{V} \subseteq \mathcal{H}$, where states $\ket{\psi} \in \mathcal{V}$ are called logical pure states. The set of logical mixed states $\Sigma_{\mathcal{V}}$ is the convex hull of pure states in $\mathcal{V}$, 
\begin{equation}
    \label{eq:conv-set-stabilizer}
    \Sigma_{\mathcal{V}} = \conv \left( \left\{ \dyad{\psi}{\psi} : \ket{\psi} \in \mathcal{V} \right\} \right) .
\end{equation}
Let the decoding map $\mathcal{D}_{\mathcal{V}}$ for $\Sigma_{\mathcal{V}}$ be an isometry which maps each state $\rho \in \Sigma_{\mathcal{V}}$ to a distinct logical state $\mathcal{D}_{\mathcal{V}} (\rho) \in \Omega = \mathcal{D}_{\mathcal{V}} \left[ \Sigma_{\mathcal{V}} \right]$, hence $\Sigma_{\mathcal{V}} \cong \Omega$. If system sites are $d$-dimensional qudits and $\dim \mathcal{V} = d^{k}$, then $\Sigma_{\mathcal{V}}$ encodes $k$ qudits worth of information. One widely studied class of QEC codes is that of qudit stabilizer codes~\cite{gottesman1997}, where the code space $\mathcal{V}$ is defined as the mutual $+1$ eigenspace of an Abelian group of Pauli operators $\mathcal{S}$, called the stabilizer group. Given that single-qudit Pauli operators satisfy $Z_{q}X_{q} = \omega X_{q}Z_{q}$ where $\omega = \exp(2 \pi i / d)$, $\mathcal{S} \subset \mathcal{P}$ is an Abelian subgroup of the $n$-qudit Pauli group,
\begin{equation}
    \mathcal{P} = \left\langle \omega^{a} \bigotimes_{q=1}^{n} P_{q} : P_{q} \in \left\langle X_{q},Z_{q} \right\rangle , a \in \mathbb{Z}_{d} \right\rangle ,
\end{equation}
where $\omega^{a} I \notin \mathcal{S}$ for all $a \neq 0 \bmod d$. $\mathcal{S} = \left\langle \mathcal{S}_{0} \right\rangle$ can be constructed by a generating set $\mathcal{S}_{0}$. If $\vert \mathcal{S} \vert = d^{n-k}$, then $\dim \mathcal{V} = d^{k}$. Logical operators, those which change the system state within $\Sigma_{\mathcal{V}}$, can be formed from linear combinations of Pauli operators in the centralizer of $\mathcal{S}$ in $\mathcal{P}$, $\mathcal{Z}(\mathcal{S}) \subset \mathcal{P}$. All $S_{i} \in \mathcal{S}$ act as the logical identity on $\Sigma_{\mathcal{V}}$.

Properties of stabilizer codes are well-understood in the setting of LC phases and there exist physical mechanisms which realize these phases. Specifically, given a code with stabilizer group $\mathcal{S}$, code space $\mathcal{V}$, and a generating set $\mathcal{S}_{0}$ for $\mathcal{S}$, one can define a gapped Hamiltonian
\begin{equation}
    \label{eq:stabilizer-hamiltonian}
    H = - \frac{1}{2} \sum_{S \in \mathcal{S}_{0}} \left(S + S^{\dagger}\right)
\end{equation}
whose ground state subspace is $\mathcal{V}$ and whose excited states correspond to configurations of Pauli errors acting on $\mathcal{V}$. If $H$ is topologically ordered, where loosely speaking $\mathcal{S}_{0}$ is translationally invariant and generators in $\mathcal{S}_{0}$ are geometrically local, then its ground state degeneracy depends on the topology of the system lattice and remains constant under local perturbations~\cite{bravyi2010,bravyi2011}. Moreover, this degeneracy generalizes to the topological degeneracy of the convex set $\Sigma_{\mathcal{V}}$, which is also invariant under shallow locally reversible LC circuits~\cite{sang2025reversibility,yang2025}.

Many stabilizer codes are stable according to Definition~\ref{def:stable-phase}. Specifically, the code spaces $\Sigma_{\mathcal{V}}$ of all quantum low-density parity check (LDPC) codes with distance polynomial in the system size, $d_{\mathcal{V}} = \Omega (n^{a})$ for some $a > 0$, are stable. By LDPC, we mean codes whose stabilizer admits a generating set $\mathcal{S}_{0}$ where every generator is supported on a finite number of qudits and every qudit is in the support of a finite number of generators.
\begin{lemma}
    \label{lemma:ldpc-stability}
    Suppose $\Sigma_{\mathcal{V}}$ corresponds to the code space of an LDPC qudit stabilizer code with distance $d_{\mathcal{V}} = \Omega(n^{a})$ for some $a > 0$ and isometric decoding map $\mathcal{D}_{\mathcal{V}}$, as described above. Then $\Sigma_{\mathcal{V}}$ is stable.
\end{lemma}
Lemma~\ref{lemma:ldpc-stability} follows directly from a simple generalization of Theorem 3 of Ref.~\cite{gottesman2014}, which proves a threshold exists for QEC of any Pauli noise channel in $\mathbb{L}_{\lambda}$ for qubit LDPC codes. For completeness we provide this argument in Appendix~\ref{gen-stochastic-stability}. When $\mathcal{V}$ is a CSS code and $\mathcal{S}_{0} = \mathcal{S}^{X}_{0} \oplus \mathcal{S}^{Z}_{0}$ decomposes into $X$-type generators $\mathcal{S}^{X}_{0}$ and $Z$-type generators $\mathcal{S}^{Z}_{0}$, we conjecture stability according to Definition~\ref{def:stable-phase} implies perturbative stability of the Hamiltonian, Eq.~(\ref{eq:stabilizer-hamiltonian}), to local terms such as $-h_{X} \sum_{q} X_{q}$ and $-h_{Z} \sum_{q} Z_{q}$. This is because Ref.~\cite{li2025perturbative} provides evidence connecting perturbative stability of the above Hamiltonian to recoverability of i.i.d.~Pauli noise, and i.i.d.~Pauli noise of strength $\leq \lambda$ is in the noise class $\mathbb{L}_{\lambda}$.

When convex sets of code states of two QEC codes are in the same phase, $\Sigma_{\mathcal{V}} \sim \Sigma_{\mathcal{V}^{\prime}}$, one can view phase equivalence as the ability to perform fault tolerant code switching with a shallow circuit~\cite{paetznick2013,anderson2014,bombin2015gauge,kubica2015transversal}, where one changes the physical system state while retaining encoded logical information. In Section~\ref{cc-phases} we will demonstrate how one can use code switching to realize phase equivalence between different topological stabilizer codes through the lens of gauge fixing within a topological subsystem code.

Lastly, while stabilizer codes are particularly useful and easy to study as examples of phases of matter, the construction in Eq.~(\ref{eq:conv-set-stabilizer}) is not specific to these codes. Generally one can construct a convex set of mixed states $\Sigma_{\mathcal{V}}$ as the convex hull of all pure states in an arbitrary subspace $\mathcal{V} \subseteq \mathcal{H}$ and similarly define a decoding map $\mathcal{D}_{\mathcal{V}}$ such that the encoded information is isomorphic to this convex hull, $\Omega = \mathcal{D}_{\mathcal{V}} \left[ \Sigma_{\mathcal{V}} \right] \cong \Sigma_{\mathcal{V}}$. Interesting examples include those when $\mathcal{V}$ corresponds to the ground state subspace of a local Hamiltonian similar to Eq.~(\ref{eq:stabilizer-hamiltonian}), such as quantum double models~\cite{kitaev2003}, twisted quantum double models~\cite{hu2013}, and string-net models~\cite{levin2005}.

\subsection{No distinct FTQL trivial phase exists} \label{trivial-phases}

In the classification of LC and QL phases, one important phase is the trivial phase, which is often defined as the phase which includes product states~\cite{sang2024rg}. However, since LC and QL phases are defined with respect to individual states rather than convex sets of states, they may exhibit different behavior than FTQL phases. One subtlety when dealing with convex sets of states is that, while any single product state may be transformed into any other product state with a range-$1$ LC circuit, one cannot transform a convex set of product states into any other arbitrary convex set of product states. For example, if one wants to transform the state of a single physical qudit on lattice site $i$ to the state of a qudit on lattice site $j$, where sites $i$ and $j$ are a distance $t$ apart, one needs an LC circuit with range $r\ell \geq t$. This cannot be done with a shallow LC circuit if $t$ is comparable with the system size, for example for $t = \Omega(n^{a})$ with $a>0$. This is not an issue with QL circuits as one can teleport a state an arbitrary distance in finite range, however.

The above discussion motivates a different approach to triviality when studying FTQL phases, specifically one where a trivial convex set encodes information in a localized region or multiple localized regions on the system lattice. Often a convex set of this form can be obtained from a convex set of product states by applying a finite-depth LC circuit $\mathcal{C}$.

\begin{definition}
\label{def:trivial-subspace}
    A convex set of states $\Sigma$ is FTQL trivial if there exists a partition of the system lattice into two disjoint regions $A$ and $B$, where $\vert A \vert = O(1)$ is independent of $n$ and there exists a fixed state $\sigma_{B}$ supported on $B$, $\supp \sigma_{B} = B$, such that for all $\rho \in \Sigma$, $\rho \approx \rho_{A} \otimes \sigma_{B}$ where $\supp \rho_{A} = A$.
\end{definition}

As discussed above, while individual states in a FTQL trivial convex set may be in the trivial LC phase, different trivial convex sets which encode the same information may not be connected to each other by shallow LC circuits. Moreover, there is no single ``trivial'' FTQL phase. This is because trivial convex sets cannot protect any encoded information, as this information can be modified by a channel acting locally on region $A$. To be precise, no trivial convex set $\Sigma$ which contains more than one distinct state is stable, as discussed in the following lemma. Since this means $\Sigma$ is essentially in a phase by itself, the notion of the trivial LC or QL phase does not carry over to the FTQL classification. Instead, one should just regard trivial convex sets as examples of convex sets which cannot protect encoded information.

\begin{lemma}
\label{lemma:trivial-no-ft}
Consider a FTQL trivial convex set $\Sigma$ with decoding map $\mathcal{D}$ which distinguishes between at least two distinct states $\rho_{1},\rho_{2} \in \Sigma$ with $\mathcal{D} (\rho_{1}) \not\approx \mathcal{D} (\rho_{2})$. Then $\Sigma$ is not stable.
\end{lemma}
\begin{proof}
    First, express $\rho_{1},\rho_{2} \in \Sigma$ as
    \begin{equation}
        \rho_{1} \approx \rho_{A} \otimes \sigma_{B} , \quad
        \rho_{2} \approx \kappa_{A} \otimes \sigma_{B} .
    \end{equation}
    Suppose $\mathcal{D} (\rho_{1}) \not\approx \mathcal{D} (\rho_{2})$,  therefore $\norm{\mathcal{D} (\rho_{1}) - \mathcal{D} (\rho_{2})}_{1}$ decays slower than $1/O(\poly(n))$. Now, consider a noise channel of the form
    \begin{align}
        \mathcal{N}(\rho) &= (1-p)\rho + p \Phi_{A}(\rho) ,
    \end{align}
    where $\supp \Phi_{A} \subseteq A$ and $ \Phi_{A}(\rho_{A}) = \kappa_{A}$. Here, $\mathcal{N} \in \mathbb{L}_{\tau}$ for $\tau = p^{1/\vert A \vert}$. Since $\Phi_{A}(\rho_{1}) \approx \rho_{2}$, applying $\mathcal{N}$ to $\rho_{1}$ yields
    \begin{equation}
        \mathcal{N}(\rho_{1}) \approx (1-p)\rho_{1} + p\rho_{2} \in \Sigma ,
    \end{equation}
    therefore acting with an arbitrary recovery operator $\mathcal{R}$ for $\Sigma$ and with the decoding map $\mathcal{D}$ on $\mathcal{N}(\rho_{1})$ yields the trace distance
    \begin{equation}
        \norm{\mathcal{D} \circ \mathcal{R} \circ\mathcal{N}(\rho_{1}) - \left[(1-p) \mathcal{D}(\rho_{1}) + p \mathcal{D}(\rho_{2})\right]}_{1} = \varepsilon
    \end{equation}
    where $\varepsilon = 1/\Omega(\poly(n))$. Applying the reverse triangle inequality to the above yields
    \begin{equation}
        \norm{\mathcal{D} \circ \mathcal{R} \circ \mathcal{N} ( \rho_{1}) - \mathcal{D} (\rho_{1})}_{1} \geq p\norm{\mathcal{D} (\rho_{1}) - \mathcal{D} (\rho_{2})}_{1} - \varepsilon .
    \end{equation}
    Since $p$ is independent of $n$ and $\norm{\mathcal{D} (\rho_{1}) - \mathcal{D} (\rho_{2})}_{1}$ decays slower than $1/O(\poly(n))$, $\norm{\mathcal{D} \circ \mathcal{R} \circ \mathcal{N} ( \rho_{1}) - \mathcal{D} (\rho_{1})}_{1}$ decays slower than $1/O(\poly(n))$ as well. Since the above channel is not $\mathcal{R}$-recoverable for any $\tau > 0$ and any $\mathcal{R}$, the class $\mathbb{L}_{\tau}$ acting on $\Sigma$ is not $\mathcal{R}$-recoverable and $\Sigma$ is not stable.
\end{proof}

Since Lemmas~\ref{lemma:stable-consistency} and~\ref{lemma:trivial-no-ft} imply no stable convex set is in the same FTQL phase as a trivial convex set, an immediate corollary to Lemmas~\ref{lemma:stable-consistency},~\ref{lemma:ldpc-stability}, and~\ref{lemma:trivial-no-ft} is that no LDPC stabilizer code $\Sigma_{\mathcal{V}}$ with distance polynomial in $n$ is in the same phase as a trivial convex set. For example, $\Sigma_{\mathcal{V}}$ is not in the same phase as any set of unencoded qudits $\Sigma$, where $\rho = \rho_{k} \otimes \dyad{0}{0}^{\otimes n - k}$ for all $\rho \in \Sigma$.

\begin{corollary}
    \label{cor:ldpc-code-nontrivial}
    Suppose $\Sigma_{\mathcal{V}}$ is the code space of an LDPC qudit stabilizer code with distance $d_{\mathcal{V}} = \Omega(n^{a})$ for some $a > 0$ and isometric decoding map $\mathcal{D}_{\mathcal{V}}$. Suppose $\Sigma$ is FTQL trivial and has decoding map $\mathcal{D}$ which distinguishes between at least two distinct states in $\Sigma$. Then $\Sigma_{\mathcal{V}}$ and $\Sigma$ are not in the same FTQL phase.
\end{corollary}

Lastly, to further emphasize that ``triviality'' is a property of a specific convex set and there is no single ``trivial FTQL phase'', we highlight that a convex set of product states can be nontrivial according to Definition~\ref{def:trivial-subspace}. One example of a convex set of product states which protects encoded information is
\begin{equation}
    \Sigma_{\mathrm{rep}} = \conv \left( \dyad{0}{0}^{\otimes n} , \dyad{1}{1}^{\otimes n} \right) ,
\end{equation}
which encodes a classical bit and all convex combinations of the form $(1-x)\dyad{0}{0}^{\otimes n} + x\dyad{1}{1}^{\otimes n}$. Here $\dyad{0}{0}^{\otimes n}$ and $\dyad{1}{1}^{\otimes n}$ are in trivial LC and QL phases and $(1-x)\dyad{0}{0}^{\otimes n} + x\dyad{1}{1}^{\otimes n}$ are classical states which exhibit no quantum correlations~\cite{hastings2011,sang2024rg}. However, $\Sigma_{\mathrm{rep}}$ is stable and information encoded in $\Sigma_{\mathrm{rep}}$ is protected as a channel $\Phi$ must be supported on all $n$ sites to transform one state in $\Sigma_{\mathrm{rep}}$ to a different state also in $\Sigma_{\mathrm{rep}}$.

\subsection{Connections to previous phase definitions} \label{relation-prev-defs}

We now explicitly discuss how LC and QL phases relate to FTQL phases. Definitions~\ref{def:lc-phase} and~\ref{def:ql-phase} concern individual states, so to compare LC and QL phases to FTQL phases we will say two convex sets $\Sigma_{1}$ and $\Sigma_{2}$ are in the same QL phase if property 1 of Definition~\ref{def:ft-phase-equivalence} is satisfied, and we will say $\Sigma_{1}$ and $\Sigma_{2}$ are in the same LC phase if property 1 of Definition~\ref{def:ft-phase-equivalence} is satisfied specifically by a pair of locally reversible LC circuits. To set notation, when given two phase classifications $A$ and $B$, $A \subseteq B$ will indicate phase equivalence in classification $A$ implies phase equivalence in classification $B$. Essentially, $B$ is a coarser phase classification than $A$. Strict inclusion, $A \subsetneq B$, indicates that there exist equivalent pairs of convex sets in classification $B$ which are not equivalent in classification $A$.

\begin{figure}[t!]
    \centering
    \includegraphics[width=0.45\textwidth]{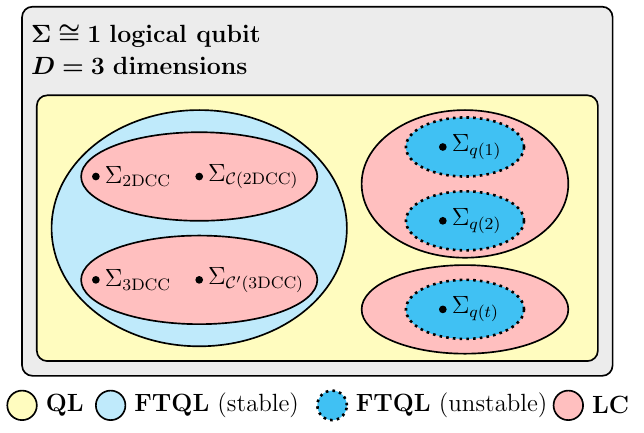}
    \caption{Representative phase diagram comparing equivalence and inequivalence of convex sets of states which encode a single logical qubit according to fault tolerant (FTQL, our definition), local channel (LC), and quantum-local (QL) phase classifications. Points label distinct convex sets and colored regions label all convex sets within the same phase. Stable FTQL phases have solid borders and lighter shading, while unstable FTQL phases have dotted borders and darker shading. We denote code spaces of the 2D and 3D color codes $\Sigma_{\mathrm{2DCC}}$ and $\Sigma_{\mathrm{3DCC}}$ respectively, 2D and 3D color codes deformed by finite-depth locally reversible LC circuits $\Sigma_{\mathrm{\mathcal{C}(2DCC)}}$ and $\Sigma_{\mathrm{\mathcal{C}^{\prime}(3DCC)}}$, and the sets of unencoded single-qubit states on qubit $i$, $\Sigma_{q(i)}$. We show two neighboring sites $i=1,2$ and a site $i=t$ which is a distance $\Omega(n^{a})$  away from site $i=1$ for some $a>0$. This diagram is constructed by applying Lemmas~\ref{lemma:stable-consistency},~\ref{lemma:ldpc-stability},~\ref{lemma:trivial-no-ft} and~\ref{lemma:3dgcc-phases} and Theorems~\ref{theorem:ft-equivalence-lc} and~\ref{theorem:colorcode-lc-phases}.}
    \label{fig:phase-diagram}
\end{figure}

First, LC phase equivalence implies QL phase equivalence as all LC circuits are QL circuits. However, states or sets of states in the same QL phase are not necessarily in the same LC phase. For example, topological CSS code states, such as 2D toric code or color code states, are in the same QL phase as product states since they can be encoded and decoded in finite depth~\cite{lodyga2015}, but are not in the same LC phase as product states because preparing a code state from a product state with an LC circuit requires $r \ell = \Omega(\poly(n))$~\cite{konig2014}. This means $LC \subsetneq QL$.

\begin{figure*}[t!]
    \centering
    \includegraphics[width=0.99\textwidth]{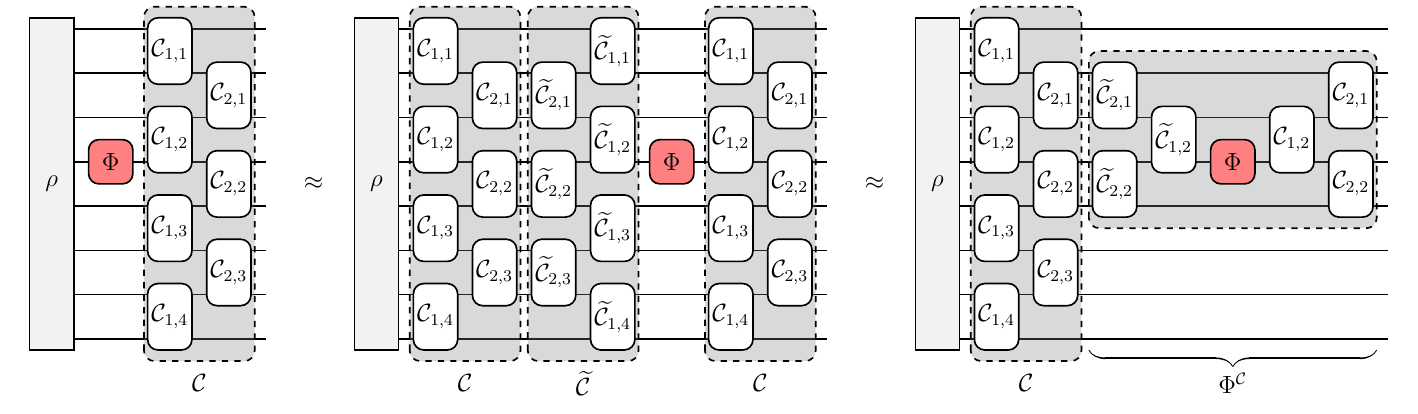}
    \caption{Growth of a channel $\Phi$ after application of a locally reversible local channel circuit $\mathcal{C}$. Local reversibility allows one to propagate $\Phi$ to the end of the circuit, however the support of the effective channel $\Phi^{\mathcal{C}}$ will be enlarged.}
    \label{fig:loc-rev-growth}
\end{figure*}

For FTQL phases, $FTQL \subseteq QL$ as the channels used in FTQL phase equivalence are necessarily QL circuits. Moreover, Lemma~\ref{lemma:trivial-no-ft} implies $FTQL \subsetneq QL$ as, for example, the code spaces of topological CSS codes are no longer in the same phase as any set of unencoded qubits. A full classification of how FTQL phases relate to LC phases is more complicated, but we can discuss two limiting cases here. First, if we restrict our attention to stable convex sets and allow only finite-depth circuits with $r\ell = O(1)$ as opposed to shallow circuits with $r\ell = O(\polylog(n))$, we find
\begin{equation}
    \label{eq:stable-phase-inclusions}
    LC \subsetneq FTQL \subseteq QL \quad \left(\text{stable, } r\ell = O(1)\right) .
\end{equation}
To show this, we prove in Section~\ref{lc-finite-depth} that two stable sets of states $\Sigma_{1}$ and $\Sigma_{2}$ related by finite-depth, locally reversible LC circuits are in the same FTQL phase, therefore $LC \subseteq FTQL$. Then, in Section~\ref{cc-phases}, we show that code spaces of the 2D and 3D color codes, $\Sigma_{\mathrm{2DCC}}$ and $\Sigma_{\mathrm{3DCC}}$, are in the same FTQL phase but in different LC phases, demonstrating strict inclusion, $LC \subsetneq FTQL$. We expect $FTQL \subseteq QL$ to be strict as well.

Second, if we restrict our attention to trivial convex sets, then Lemma~\ref{lemma:trivial-no-ft} implies
\begin{equation}
    \label{eq:trivial-phase-inclusions}
    FTQL \subsetneq LC \subsetneq QL \quad \left(\text{trivial}\right) .
\end{equation}
A representative phase diagram for 1 encoded qubit ($\Omega = \conv \left( \left\{ \dyad{\psi}{\psi} : \ket{\psi} \in \mathbb{C}^{2} \right\} \right)$) in $D=3$ spatial dimensions is shown in Figure~\ref{fig:phase-diagram} and contains examples discussed previously and in the following sections.

\subsection{Finite-depth LC equivalence implies FTQL equivalence} \label{lc-finite-depth}

One important property of locally reversible LC circuits is that the growth of the support of an operator applied before a circuit is constrained by circuit range and the system lattice. Each local channel gate acts on a geometrically local region of sites within a ball of radius $r$, which limits the spread of errors. Given a channel $\Phi$ and a locally reversible LC circuit $\mathcal{C}$, we may construct a dressed channel $\Phi^{\mathcal{C}}$ which mimics propagating $\Phi$ to the end of the circuit $\mathcal{C}$. Explicitly, $\Phi^{\mathcal{C}}$ may be defined as~\cite{sang2025reversibility}
\begin{equation}
    \label{eq:lc-conjugation}
    \Phi^{\mathcal{C}} = \mathcal{C}_{\mathrm{red}} \circ \Phi \circ \widetilde{\mathcal{C}}_{\mathrm{red}} ,
\end{equation}
where $\mathcal{C}_{\mathrm{red}}$ is a reduced channel which can be constructed recursively,
\begin{align}
    \label{eq:lc-reduced}
    \mathcal{C}_{\mathrm{red}} &= \mathcal{C}_{\mathrm{red},\ell} \circ \ldots \circ \mathcal{C}_{\mathrm{red},1} , \\ \mathcal{C}_{\mathrm{red},1} &= \bigotimes_{\alpha : \supp \mathcal{C}_{1,\alpha} \cap \supp \Phi \neq \emptyset} \mathcal{C}_{1,\alpha} \nonumber \\
    \mathcal{C}_{\mathrm{red},i} &= \bigotimes_{\alpha : \supp \mathcal{C}_{i,\alpha} \cap \supp \mathcal{C}_{\mathrm{red},i-1} \neq \emptyset} \mathcal{C}_{i,\alpha} \quad \forall i \in [2,\ell] , \nonumber
\end{align}
and $\widetilde{\mathcal{C}}_{\mathrm{red}}$ is the reversal channel for $\mathcal{C}_{\mathrm{red}}$, constructed accordingly from gates in $\widetilde{\mathcal{C}}$. This means for all states $\rho$ which $\mathcal{C}$ is locally reversible with respect to~\cite{sang2025reversibility},
\begin{equation}
    \mathcal{C} \circ \Phi (\rho) \approx \mathcal{C} \circ \Phi \circ \widetilde{\mathcal{C}} \circ \mathcal{C} (\rho) \approx \Phi^{\mathcal{C}} \circ \mathcal{C} (\rho) .
\end{equation}
However, the support of $\Phi^{\mathcal{C}}$ is bounded since all gates $\mathcal{C}_{i,\alpha}$ and $\widetilde{\mathcal{C}}_{i,\alpha}$ are geometrically local. Denoting the weight of an operator $\wt (\cdot) = \vert \supp (\cdot) \vert$, we find
\begin{equation}
    \label{eq:lc-bound}
    \wt \Phi^{\mathcal{C}} \leq \gamma (r \ell)^{D} \wt \Phi .
\end{equation}
This is schematically depicted in Figure~\ref{fig:loc-rev-growth}. This bound on channel growth is sufficient to show that if two convex sets of states $\Sigma_{1}$ and $\Sigma_{2}$ are both stable and are connected by locally reversible LC circuits with $r\ell = O(1)$, then $\Sigma_{1}$ and $\Sigma_{2}$ are in the same FTQL phase. Before proving phase equivalence, we will first bound the strength of the residual noise after applying a locally reversible LC circuit subjected to local stochastic noise. This bound is given by the following lemma.

\begin{lemma}
\label{lemma:lc-localnoise}
    Consider a system on a $D$-dimensional lattice. Suppose the channel $\Lambda : \Sigma_{1} \rightarrow \Sigma_{2}$ is an $(r_{1},\ell)$ locally reversible LC circuit and its reversal circuit $\widetilde{\Lambda} : \Sigma_{2} \rightarrow \Sigma_{1}$ is an $(r_{2},\ell)$ LC circuit, where $\max (r_{1},r_{2}) = r$. Suppose noise $\mathcal{N} \in \mathbb{L}_{\tau}$ acts between each layer of $\Lambda$, where $\tau \leq \left(2ez\right)^{-\gamma (r\ell)^{D}} / \ell$ and $z$ is the maximum number of nearest neighbors for any system site. Then
    \begin{equation}
        \Lambda^{\mathcal{N}} (\rho_{1}) \approx \widetilde{\mathcal{N}}^{\Lambda}
        \circ \Lambda (\rho_{1}) \quad \forall \rho_{1} \in \Sigma_{1} ,
    \end{equation}
    where $\widetilde{\mathcal{N}}^{\Lambda} \in \mathbb{L}_{\eta}$ is $\eta$-bounded with
    \begin{equation}
        \label{eq:lc-noise-bound}
        \eta \leq \frac{(\ell\tau)^{1 / \gamma (r\ell)^{D}}}{\left(2ez\right)^{-1} - (\ell\tau)^{1 / \gamma (r\ell)^{D}}} .
    \end{equation}
\end{lemma}
We prove Lemma~\ref{lemma:lc-localnoise} in Appendix~\ref{lc-proof}. The central idea of the proof is to propagate noise in each layer to the end of the circuit $\Lambda$ to obtain $\widetilde{\mathcal{N}}^{\Lambda}$, bound the growth of the support of each noise channel similar to Eq.~(\ref{eq:lc-bound}), then bound the probability distribution over supports of these enlarged channels in $\widetilde{\mathcal{N}}^{\Lambda}$ using a combinatorial argument. Using Lemma~\ref{lemma:lc-localnoise}, we can prove FTQL phase equivalence between $\Sigma_{1}$ and $\Sigma_{2}$ if they are both stable and connected by finite-depth locally reversible LC circuits.

\begin{theorem}
\label{theorem:ft-equivalence-lc}
    Consider a system on a $D$-dimensional lattice and two convex sets $\Sigma_{i}$ equipped with decoding maps $\mathcal{D}_{i}$, $i=1,2$, such that $\mathcal{D}_{1}\left[\Sigma_{1}\right] = \mathcal{D}_{2}\left[\Sigma_{2}\right] = \Omega$. Suppose $\Sigma_{i}$ are both stable and have recovery channels $\mathcal{R}_{i}$ with critical noise strengths $\eta_{i}^{\star}$. Suppose $\Lambda : \Sigma_{1} \rightarrow \Sigma_{2}$ is an $(r_{1},\ell)$ locally reversible LC circuit and its reversal circuit $\widetilde{\Lambda} : \Sigma_{2} \rightarrow \Sigma_{1}$ is an $(r_{2},\ell)$ LC circuit, where $\max (r_{1},r_{2}) = r$, $r\ell = O(1)$, and $\Lambda$ and $\widetilde{\Lambda}$ reverse each other with respect to $\Omega$, satisfying Eq.~(\ref{eq:decoding-equivalence}). Then, $\Sigma_{1}$ and $\Sigma_{2}$ are in the same FTQL phase, $\Sigma_{1} \sim \Sigma_{2}$, with respect to $\mathbb{L}_{\tau}$, $\mathcal{D}_{1}$, and $\mathcal{D}_{2}$.
\end{theorem}
\begin{proof}
    From Lemma~\ref{lemma:lc-localnoise}, both $\Lambda , \widetilde{\Lambda}$ satisfy condition 1 of Definition~\ref{def:ft-channel} with residual noise $\widetilde{\mathcal{N}}^{\Lambda} , \widetilde{\mathcal{N}}^{\widetilde{\Lambda}} \in \mathbb{L}_{\eta}$. Since $\Sigma_{i}$ are both stable with critical noise strengths $\eta_{i}^{\star}$, given residual channels $\widetilde{\mathcal{N}}^{\Lambda} , \widetilde{\mathcal{N}}^{\widetilde{\Lambda}} \in \mathbb{L}_{\eta}$, we may rewrite Eq.~(\ref{eq:lc-noise-bound}) to find that $\eta < \eta_{i}^{\star}$ when
    \begin{equation}
        \label{eq:lc-tau-thresh}
        \tau < \tau_{i}^{\star} = \frac{1}{\ell} \left[ \eta_{i}^{\star} \left(2ez\right)^{-1} / \left(\eta_{i}^{\star} + 1\right) \right]^{\gamma (r\ell)^{D}} .
    \end{equation}
    Since $r\ell$, $\eta_{i}^{\star}$, $\gamma$, and $z$ are independent of $n$, so is $\tau_{i}^{\star}$. Furthermore, since each $\Sigma_{i}$ is stable and $\mathbb{L}_{\eta}$ is $\mathcal{R}_{i}$-recoverable for $\eta < \eta_{i}^{\star}$, both $\widetilde{\mathcal{N}}^{\Lambda}$ and $\widetilde{\mathcal{N}}^{\widetilde{\Lambda}}$ satisfy condition 2 of Definition~\ref{def:ft-channel}. Therefore, $\Lambda$ and $\widetilde{\Lambda}$ are both fault tolerant and $\Sigma_{1} \sim \Sigma_{2}$ with respect to $\mathbb{L}_{\tau}$, $\mathcal{D}_{1}$, and $\mathcal{D}_{2}$.
\end{proof}

Theorem~\ref{theorem:ft-equivalence-lc} implies that if two stable sets of states $\Sigma_{1}$ and $\Sigma_{2}$ are in the same LC phase and connected by finite-depth locally reversible LC circuits, they are also in the same FTQL phase. This demonstrates the inclusion $LC \subseteq FTQL$ given in Eq.~(\ref{eq:stable-phase-inclusions}). While the above proof relies on $r\ell = O(1)$ to ensure the critical noise strength $\tau_{i}^{\star}$ remains independent of $n$, Eq.~(\ref{eq:lc-noise-bound}) is a conservative upper bound on $\eta$. It may be possible to obtain better bounds on $\eta$ with different techniques than those used in Lemma~\ref{lemma:lc-localnoise} or to apply different circuits, so if $\Sigma_{1}$ and $\Sigma_{2}$ are connected by LC circuits whose range grows with $n$, they may still be in the same FTQL phase. Additionally, one only needs to confirm $\Sigma_{1}$ is stable for Theorem~\ref{theorem:ft-equivalence-lc} to hold. This is because $\mathbb{L}_{\eta}$ acting on $\Sigma_{1}$ is $\mathcal{R}_{1}$-recoverable for all $\eta < \eta_{1}^{\star}$, and one can construct a recovery channel $\mathcal{R}_{2} = \Lambda \circ \mathcal{R}_{1} \circ \widetilde{\Lambda}$ similar to that in Lemma~\ref{lemma:stable-consistency} to show $\mathbb{L}_{\tau}$ acting on $\Sigma_{2}$ is $\mathcal{R}_{2}$-recoverable for $\tau < \tau_{1}^{\star}$ satisfying Eq.~(\ref{eq:lc-tau-thresh}). However, a recovery channel with a higher threshold may often exist in practice.

\section{Phases of topological subsystem codes} \label{cc-phases}

\subsection{Challenges with specifying phases} \label{subsystem-phase-challenges}

Phases corresponding to topological stabilizer codes are well-understood, however those corresponding to topological subsystem codes are far less understood and cannot be approached in the same way as stabilizer codes. Subsystem codes~\cite{kribs2005,poulin2005,kribs2006} can be viewed as a generalization of stabilizer codes and are defined by a subgroup $\mathcal{G} \subset \mathcal{P}$ of the Pauli group, called the gauge group. The center of $\mathcal{G}$, up to phases, forms a stabilizer group, $\mathcal{S} = \left(\mathcal{Z}(\mathcal{G}) \cap \mathcal{G} \right) / \langle e^{i \theta} \rangle \subseteq \mathcal{G}$. As with stabilizer codes, one defines the code space $\mathcal{V} \subseteq \mathcal{H}$ as the mutual $+1$ eigenspace of $\mathcal{S}$, however the gauge group $\mathcal{G}$ induces a tensor product structure on $\mathcal{V}$. Specifically, $\mathcal{G}$ partitions $\mathcal{V} = \mathcal{A} \otimes \mathcal{B}$ into a logical subsystem $\mathcal{A}$ and a gauge subsystem $\mathcal{B}$, where gauge operators $G \in \mathcal{G}$ take the form $G = I_{\mathcal{A}} \otimes G_{\mathcal{B}}$ and act trivially on $\mathcal{A}$ but possibly nontrivially on $\mathcal{B}$. This decomposition of $\mathcal{V}$ does not correspond to physical regions of the system lattice. Here, the encoded logical state does not depend on the state of the gauge subsystem. Logical Pauli operators may be either bare logical operators which act nontrivially only on $\mathcal{A}$, $L_{\mathcal{A}} \otimes I_{\mathcal{B}} \in \mathcal{Z}(\mathcal{G})$, or dressed logical operators which may act nontrivially on $\mathcal{B}$ as well, $L_{\mathcal{A}} \otimes G_{\mathcal{B}} \in \mathcal{Z}(\mathcal{S})$. Observe that stabilizer codes are special cases of subsystem codes where $\mathcal{S} = \mathcal{G}$, hence $\mathcal{B}$ is trivial.

Some issues arise when naively constructing convex sets corresponding to the encoded information in a subsystem code. A logical pure state of a subsystem code splits into a logical state $\ket{\psi} \in \mathcal{A}$ and an arbitrary mixed state of the gauge subsystem $g \in \mathbb{S}(\mathcal{B})$, namely $\rho_{\psi} = \dyad{\psi}{\psi} \otimes g$. The convex hull of all logical pure states corresponding to a specific gauge state $g$,
\begin{equation}
    \label{eq:conv-set-gauge}
    \Sigma_{g} = \conv \left( \left\{ \dyad{\psi}{\psi} \otimes g : \ket{\psi} \in \mathcal{A} \right\} \right) ,
\end{equation}
contains a distinct representation of each logical state in $\mathbb{S}(\mathcal{A})$. One way of viewing each $\Sigma_{g}$ is similar to the code space of a stabilizer code in Section~\ref{qec-relation}, defining a decoding map $\mathcal{D}_{g}$ for each $g$ as an isometry such that $\Omega = \mathcal{D}_{g} \left[ \Sigma_{g} \right] \cong \Sigma_{g}$. In this way, $\Omega$ is the same for each $\Sigma_{g}$. However, this approach treats the subsystem code as a collection of many different gauge-fixed convex sets $\Sigma_{g}$ which all correspond to the same encoded information $\Omega$, rather than a single object. Also, the union $\bigcup_{g} \Sigma_{g}$ is not necessarily convex so one cannot just construct a larger convex set starting from all individual $\Sigma_{g}$.

Furthermore, there is no distinct way to construct a Hamiltonian for a subsystem code. First, it is unclear what terms should be included in the Hamiltonian, since the code has an Abelian stabilizer group but non-Abelian gauge group. One possible Hamiltonian takes the form
\begin{equation}
    H = - \frac{1}{2} \sum_{S \in \mathcal{S}_{0}} J_{S} \left(S + S^{\dagger}\right) - \frac{1}{2} \sum_{G \in \mathcal{G}_{0}} J_{G} \left(G + G^{\dagger}\right) ,
\end{equation}
where $\mathcal{S}_{0}$ and $\mathcal{G}_{0}$ are geometrically local generating sets for $\mathcal{S}$ and $\mathcal{G}$, and $J_{S}$, $J_{G}$ are coupling strengths~\cite{bombin2015gauge,kubica2018ungauging,ellison2023,li2024phase}. If one sets all $J_{G} = 0$, $J_{S} = 1$, then the Hamiltonian is just the stabilizer code Hamiltonian for $\mathcal{S}_{0}$, Eq.~(\ref{eq:stabilizer-hamiltonian}). However, this Hamiltonian treats gauge operators $G \in \mathcal{G}$ as logical operators and the ground state degeneracy might be much larger than the number of encoded logical qudits. A more natural choice, setting $J_{G} = J_{S} = 1$, yields a non-commuting Hamiltonian which is generally not solvable. Properties of such Hamiltonians, such as their ground state degeneracy, structure of excitations, or even whether they are gapped, are generally not known except in special cases~\cite{li2024phase,bridgeman2024}.
Other choices of $J_{S}$ and $J_{G}$ may yield solvable Hamiltonians, but this also may favor some specific state of the gauge qudits and it is not a priori clear what a good choice of coupling strengths is.

\subsection{Phase equivalence for subsystem codes} \label{subsystem-phases}

The issues discussed above suggest that phases of subsystem codes are not sufficiently captured by either constructing a distinct local Hamiltonian for the code or by only studying gauge-fixed convex sets $\Sigma_{g}$. Instead, one way to treat a subsystem code as a single object and study its phase is to consider all code states equally by using a suitably chosen decoding map. First, define the set $\Sigma_{\mathcal{G}}$ to be the convex hull of all pure states in $\mathcal{V} = \mathcal{A} \otimes \mathcal{B}$,
\begin{equation}
    \label{eq:subsystem-codespace}
    \Sigma_{\mathcal{G}} = \conv \left( \left\{ \dyad{\Psi}{\Psi} : \ket{\Psi} \in \mathcal{A} \otimes \mathcal{B} \right\} \right) .
\end{equation}
Then, define a decoding map
$\mathcal{D}_{\mathcal{G}}(\cdot) = (\mathcal{D}_{\mathcal{A}} \otimes \Tr_{\mathcal{B}}) (\cdot)$ which traces out the gauge subsystem. Let $\mathcal{H}_\mathcal{A}$ be a logical Hilbert space where $\mathcal{H}_\mathcal{A} \cong \mathcal{A}$, let $\mathcal{D}_{\mathcal{A}}$ be an isometry $\mathcal{D}_{\mathcal{A}} : \mathbb{S}(\mathcal{A}) \rightarrow \mathbb{S}(\mathcal{H}_\mathcal{A})$ which acts only on the logical subsystem, and let $\Tr_{\mathcal{B}}$ trace out subsystem $\mathcal{B}$. Therefore $\mathcal{D}_{\mathcal{G}} \left[ \Sigma_{\mathcal{G}} \right] = \mathbb{S}(\mathcal{H}_\mathcal{A}) = \Omega$. Note that $\mathcal{D}_{\mathcal{G}}$ may also be used as a decoder for each convex set $\Sigma_{g}$ defined previously, namely one can choose $\mathcal{D}_{g} = \mathcal{D}_{\mathcal{G}}$ restricted to $\Sigma_{g}$ for any $g$ and obtain the same encoded information $\Omega$. Therefore $\mathcal{D}_{\mathcal{G}} \left[ \Sigma_{\mathcal{G}} \right] \cong \mathcal{D}_{g} \left[ \Sigma_{g} \right] \cong \Omega$ for all $g$. However, $\Sigma_{\mathcal{G}} \supseteq \bigcup_{g} \Sigma_{g}$ contains all states of all $\Sigma_{g}$, but also linear combinations of states in disjoint sets $\Sigma_{g}$. For example, states where the logical and gauge subsystems are entangled, such as
\begin{equation}
    \rho_{\Psi} = \dyad{\Psi}{\Psi} , \quad \ket{\Psi} = \frac{1}{\sqrt{\vert \mathcal{A} \vert}}\sum_{i} \ket{\overline{i}}_{\mathcal{A}} \otimes \ket{i}_{\mathcal{B}} ,
\end{equation}
are in the set $\Sigma_{\mathcal{G}}$ but are not in any single $\Sigma_{g}$, since each $\Sigma_{g}$ contains only states which are separable with respect to the decomposition $\mathcal{V} = \mathcal{A} \otimes \mathcal{B}$. While states such as $\rho_{\Psi}$ are not technically part of the subsystem code, since one often assumes code states take the form $\rho_{\mathcal{A}} \otimes \rho_{\mathcal{B}}$~\cite{poulin2005}, applying $\mathcal{D}_{\mathcal{G}}$ to $\rho_{\Psi}$ still yields a state $\mathcal{D}_{\mathcal{G}} \left( \rho_{\Psi} \right) \in \Omega$, albeit a logical mixed state, $\mathcal{D}_{\mathcal{G}} \left( \rho_{\Psi} \right) \propto I_{\Omega}$. Furthermore, all logical states of the subsystem code are treated on the same footing by studying the set $\Sigma_{\mathcal{G}}$ and decoding map $\mathcal{D}_{\mathcal{G}}$, and no specific gauge state $g$ is favored.

This approach is also useful for studying the relationship between phases of an entire subsystem code $\Sigma_{\mathcal{G}}$ and phases of gauge-fixed convex sets $\Sigma_{g}$, which may correspond to code spaces of stabilizer codes. $\Sigma_{\mathcal{G}}$ and $\Sigma_{g}$ correspond to the same encoded information $\Omega$, and transforming between $\Sigma_{g}$ and $\Sigma_{\mathcal{G}}$ or between two $\Sigma_{g}$ with different gauge states $g$ amounts to performing code switching, as mentioned in Section~\ref{qec-relation}. One subtlety here is that although every $\Sigma_{g} \subseteq \Sigma_{\mathcal{G}}$, this does not immediately imply $\Sigma_{g}$ and $\Sigma_{\mathcal{G}}$ are in the same FTQL phase or that $\Sigma_{g_{1}}$ and $\Sigma_{g_{2}}$ corresponding to different gauge states $g_{1} \neq g_{2}$ are in the same FTQL phase. While one can always transform from any $\Sigma_{g}$ to $\Sigma_{\mathcal{G}}$ by doing nothing, it is possible for states in some $\Sigma_{g_{1}}$ to only be obtainable by a deep circuit from states in $\Sigma_{g_{2}}$, for example if states in $\Sigma_{g_{1}}$ are volume-law entangled but states in $\Sigma_{g_{2}}$ are only area-law entangled. Since an $(r,\ell)$ QL circuit can only increase the entanglement between a region $R$ of the lattice and its complement $R^{c}$ by an amount proportional to its boundary, $O\left( r\ell \vert \partial R \vert \right)$~\cite{piroli2021,lu2022,lu2023}, no shallow QL circuit can convert an area-law entangled state into a volume-law entangled state, therefore $\Sigma_{g_{1}} \nsim \Sigma_{g_{2}}$ and $\Sigma_{g_{1}} \nsim \Sigma_{\mathcal{G}}$. This may occur in the contrived example where gauge qubits are physical qubits on the system lattice. For example, consider a periodic, cubic lattice in $D=3$ with a 2D toric code placed on a 2D slice of the lattice and gauge qubits placed on all links not contained within this slice. In this case, there are $O(n)$ gauge qubits which can host both area-law entangled states or volume-law entangled states. Furthermore, even if one can change the state of the gauge subsystem with a shallow circuit, this circuit may not be fault tolerant. While most examples we consider in this work will not have these issues, it is important to recognize that $\Sigma_{g}$ are not automatically in the same phase as each other or as $\Sigma_{\mathcal{G}}$.

In the rest of this section, we will study phase properties of the gauge color code~\cite{bombin2015gauge,kubica2015transversal,Kubicathesis} as a representative example of the structure of phases related to subsystem codes $\Sigma_{\mathcal{G}}$ and gauge-fixed convex sets $\Sigma_{g}$. While we focus on the color code here, results in the following subsections are generally applicable to other topological subsystem codes such as the subsystem toric code~\cite{kubica2022,bridgeman2024}. Specifically, we demonstrate that often properties of the underlying subsystem code allow FTQL phase equivalence between $\Sigma_{\mathcal{G}}$ and a wide variety of $\Sigma_{g}$. However, circuits which transform between $\Sigma_{g}$ may sometimes require measurements and global classical communication to be implemented in shallow depth, meaning there exist $\Sigma_{g}$ in different LC phases. This provides further evidence that the properties of a topological subsystem code cannot be accurately described by a single Hamiltonian.

\subsection{Gauge color codes and dimensional jumps} \label{gauge-color-codes}

Here we study gauge color codes in $D$ spatial dimensions which each encode one logical qubit. We will consider color codes in $D$ dimensions defined on families of lattices $\Gamma_{D}$, where each lattice tiles a $D$-simplex and has linear dimension $L$. We will suppress $L$ unless necessary here. We will only consider lattices which can support color codes and are compatible, namely lattices such that in $D$ spatial dimensions, for all $i \in [2,D-1]$, $\Gamma_{i} \subseteq \Gamma_{i+1}$ is an $i$-dimensional facet of $\Gamma_{i+1}$. This compatibility will allow us to code switch between color codes on $\Gamma_{i}$ and $\Gamma_{i+1}$. Here, $\Gamma_{i} \subseteq \Gamma_{D}$ are all sublattices of $\Gamma_{D}$. In Figure~\ref{fig:cc-dim-jump}(a-b), we give examples of compatible lattices and color codes with $L=3$ in $D=2$ and $D=3$.

Given such a lattice $\Gamma_{D}$, one can place a $D$-dimensional $(x,z)$ color code on it where qubits are placed on vertices of $\Gamma_{D}$ and integers $x,z \in [0,D-2]$ satisfying $x+z \leq D-2$ specify the dimension of cells which stabilizer and gauge generators are placed on. Explicitly, $X$ and $Z$ type stabilizer generators are placed on vertices of $(D-x)$-cells and $(D-z)$-cells, respectively, and $X$ and $Z$ type gauge generators are placed on vertices of $(z+2)$-cells and $(x+2)$-cells, respectively. This code is a stabilizer code for $x+z = D-2$ and a subsystem code for $x+z < D-2$. Denote the convex set corresponding to the full code space, Eq.~(\ref{eq:subsystem-codespace}), of this code $\Sigma_{\Gamma_{D}}^{(x,z)}$. For a detailed introduction to color codes, see Refs.~\cite{kubica2015transversal,bombin2006,bombin2007cc3d,bombin2007,bombin2015gauge}.

As described in Ref.~\cite{kubica2015transversal}, one can code switch between $(x,z)$ and $(x^\prime,z^\prime)$ color codes on the same lattice $\Gamma_{D}$ by measuring a subset of gauge generators and applying Pauli corrections if both $x \geq x^\prime$ and $z \geq z^\prime$, since in this case $\Sigma_{\Gamma_{D}}^{(x,z)} \subseteq \Sigma_{\Gamma_{D}}^{(x^\prime,z^\prime)}$ corresponds to partially fixing the gauge state of $\Sigma_{\Gamma_{D}}^{(x^\prime,z^\prime)}$. In Ref.~\cite{bombin2015singleshot}, it is proven that one can perform this code switching procedure fault tolerantly in finite depth in $D=3$ between the $(0,0)$ gauge color code and $(1,0)$ and $(0,1)$ stabilizer color codes.

\begin{figure*}[ht!]
    \centering
    \includegraphics[width=0.9\textwidth]{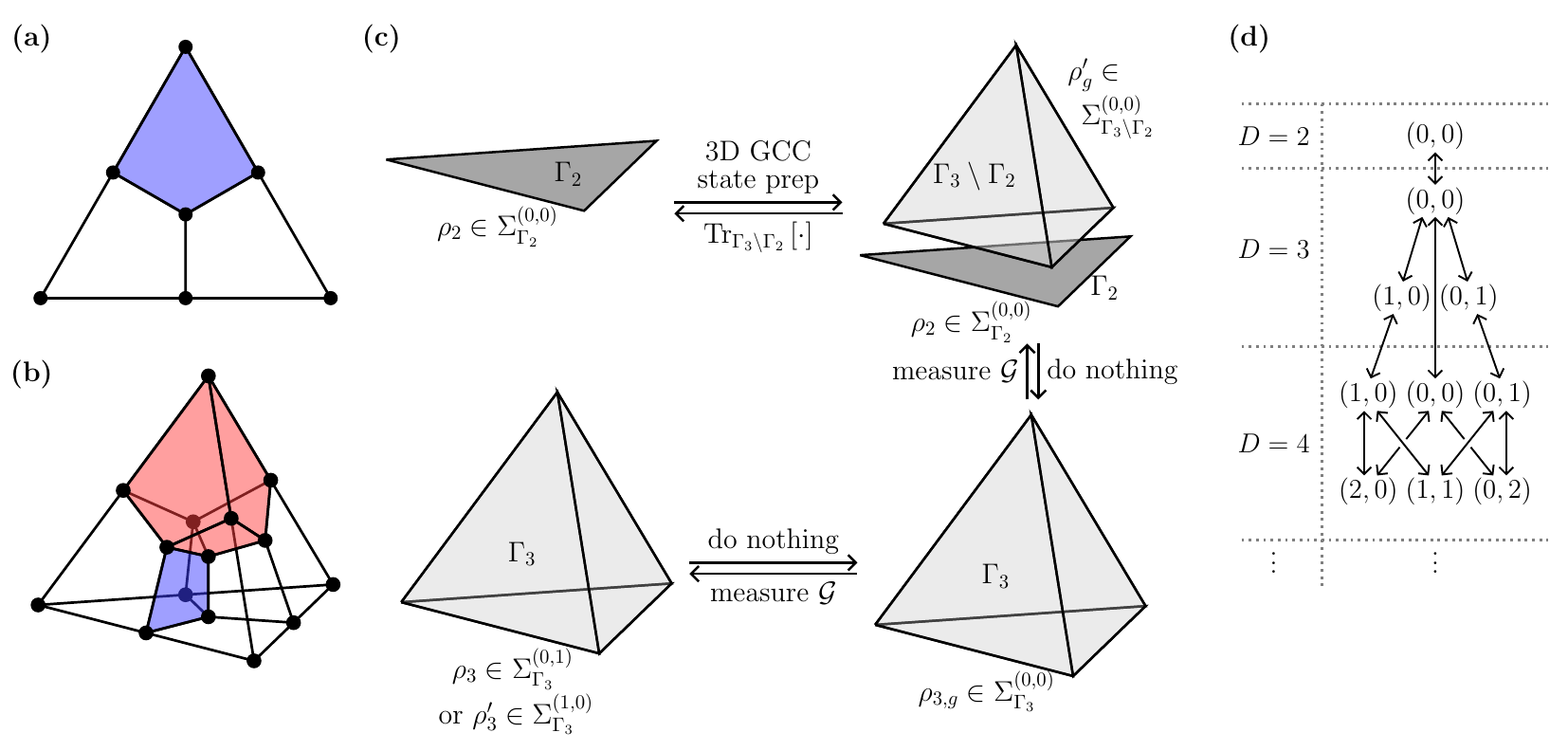}
    \caption{(a) Lattice supporting the 2D color code with linear size $L=3$. Qubits are placed on vertices and stabilizers are placed on plaquettes (blue). (b) Lattice supporting the 3D stabilizer and gauge color codes with $L=3$. Qubits are placed on vertices, stabilizers may be placed on plaquettes (blue) or cells (red), and gauge generators are placed on plaquettes. (c) Fault tolerant dimensional jump for code switching between the 2D and 3D color codes using the 3D gauge color code (GCC). (d) Dimensional ladder of color codes placed on $D$-simplices. In $D$ spatial dimensions on a lattice $\Gamma_{D}$, all color codes placed on $\Gamma_{D}$ or its sublattices $\Gamma_{i}$, $i \leq D$, are in the same FTQL phase, as one can code switch multiple times to reach any code from any other code on the ladder. Here, arrows indicate pairs of codes for which fault tolerance of code switching follows directly from applying results in Refs.~\cite{bombin2015singleshot,bombin2016}.}
    \label{fig:cc-dim-jump}
\end{figure*}

Additionally, as described in Ref.~\cite{bombin2016}, one can code switch between $(x,z)$ color codes on $\Gamma_{D}$ and one of its facets, $\Gamma_{D-1}$ in a ``dimensional jump'' procedure. This involves constructing a $(x,z)$ color code on $\Gamma_{D} \setminus \Gamma_{D-1}$ which encodes no logical qubits, $\Sigma_{\Gamma_{D} \setminus \Gamma_{D-1}}^{(x,z)}$, and is doable since the two codes on $\Gamma_{D-1}$ and $\Gamma_{D} \setminus \Gamma_{D-1}$ correspond to partially fixing the gauge state of $\Sigma_{\Gamma_{D}}^{(x,z)}$,
\begin{equation}
    \Sigma_{\Gamma_{D-1}}^{(x,z)} \otimes \Sigma_{\Gamma_{D} \setminus \Gamma_{D-1}}^{(x,z)} \subseteq \Sigma_{\Gamma_{D}}^{(x,z)} .
\end{equation}
Similar to code switching above, it is proven in Ref.~\cite{bombin2016} that one can perform this dimensional jump fault tolerantly in finite depth for $D=3$ and $(x,z) = (0,0)$. Also, one can fault tolerantly prepare the color code with no logical qubits, $\Sigma_{\Gamma_{3} \setminus \Gamma_{2}}^{(0,0)}$, from a product state $\dyad{0}{0}^{\otimes \vert \Gamma_{3} \setminus \Gamma_{2} \vert}$ on $\Gamma_{3} \setminus \Gamma_{2}$ in finite depth by measuring gauge generators and applying Pauli corrections, as shown in Ref.~\cite{bombin2015singleshot}.

While proven explicitly for $D=3$, the arguments in Refs.~\cite{bombin2015singleshot,bombin2016} generalize to any $D \geq 3$, $x \leq D-3$, and $z \leq D-3$ in a straightforward manner. This means one can fault tolerantly code switch between the $(x,z)$ gauge color code and the $(x,D-2-x)$ and $(D-2-z,z)$ stabilizer color codes in finite depth, perform the dimensional jump procedure fault tolerantly in finite depth, and fault tolerantly prepare $\Sigma_{\Gamma_{D} \setminus \Gamma_{D-1}}^{(x,z)}$ from a product state in finite depth, all for any $D \geq 3$, $x \leq D-3$, and $z \leq D-3$.

\subsection{FTQL phase equivalence of color codes} \label{cc-equivalence}

We now outline how results from Refs.~\cite{bombin2015singleshot,bombin2016}, which demonstrate single-shot code switching, also imply that code spaces of all color codes defined on $i$-simplices embedded in $D$ spatial dimensions, $i \leq D$, are in the same FTQL phase. We show this by using the procedures described above to construct a ``dimensional ladder'' which captures equivalences between codes in various dimensions. When we say an operation is \textit{single-shot} we mean it is implementable fault tolerantly with a finite-depth QL circuit. Here we demonstrate FTQL phase equivalence with respect to a class of Pauli noise, $\mathbb{N}^{\mathrm{loc}}_{\tau,\epsilon}$, used in proofs of single-shot QEC and code switching in Refs.~\cite{bombin2015singleshot,bombin2016}. Loosely, each channel $\mathcal{N} \in \mathbb{N}^{\mathrm{loc}}_{\tau,\epsilon}$ has the same distribution of stabilizer syndromes as some Pauli channel in $\mathbb{L}_{\tau}$ and has logical error rate upper bounded by $\epsilon$ when QEC is performed using a minimum weight decoder. See Ref.~\cite{bombin2015singleshot} for more details about $\mathbb{N}^{\mathrm{loc}}_{\tau,\epsilon}$, including the proof that it satisfies the properties of Definition~\ref{def:noise-class}. Notably, $\mathbb{L}_{\lambda} \subseteq \mathbb{N}^{\mathrm{loc}}_{\lambda,\epsilon(\lambda)}$ where $\epsilon(\lambda)$ is specified in Ref.~\cite{bombin2015singleshot}.

Since the dimensional jump can be performed fault tolerantly between $\Sigma_{\Gamma_{D}}^{(x,z)}$ and the pair $\Sigma_{\Gamma_{D-1}}^{(x,z)}$ and $\Sigma_{\Gamma_{D} \setminus \Gamma_{D-1}}^{(x,z)}$ for any $D \geq 3$, $x \leq D-3$, and $z \leq D-3$ using finite-depth measurement-feedback circuits~\cite{bombin2016}, and since this procedure does not change the code's logical state, $\Sigma_{\Gamma_{D}}^{(x,z)}$ and $\Sigma_{\Gamma_{D-1}}^{(x,z)} \otimes \Sigma_{\Gamma_{D} \setminus \Gamma_{D-1}}^{(x,z)}$ are in the same FTQL phase,
\begin{equation}
    \label{eq:gcc-dim-jump}
    \Sigma_{\Gamma_{D}}^{(x,z)} \sim \Sigma_{\Gamma_{D-1}}^{(x,z)} \otimes \Sigma_{\Gamma_{D} \setminus \Gamma_{D-1}}^{(x,z)} .
\end{equation}
Additionally, one can fault tolerantly prepare a state in $\Sigma_{\Gamma_{D} \setminus \Gamma_{D-1}}^{(x,z)}$ from the product state $\dyad{0}{0}^{\otimes \vert \Gamma_{D} \setminus \Gamma_{D-1} \vert}$ in finite depth, and one can fault tolerantly transform from any state in $\Sigma_{\Gamma_{D} \setminus \Gamma_{D-1}}^{(x,z)}$ to $\dyad{0}{0}^{\otimes \vert \Gamma_{D} \setminus \Gamma_{D-1} \vert}$ by measuring all single-qubit $Z_{i}$ and conditionally applying $X_{i}$. Since $\Sigma_{\Gamma_{D} \setminus \Gamma_{D-1}}^{(x,z)}$ encodes no information, it is in the same FTQL phase as the single state $\dyad{0}{0}^{\otimes \vert \Gamma_{D} \setminus \Gamma_{D-1} \vert}$ and
\begin{equation}
    \label{eq:gcc-state-prep}
    \Sigma_{\Gamma_{D}}^{(x,z)} \sim \Sigma_{\Gamma_{D-1}}^{(x,z)} \otimes \dyad{0}{0}^{\otimes \vert \Gamma_{D} \setminus \Gamma_{D-1} \vert}  \sim \Sigma_{\Gamma_{D-1}}^{(x,z)} ,
\end{equation}
where tracing out $\Gamma_{D} \setminus \Gamma_{D-1}$ yields the last equivalence. Lastly, since one can fault tolerantly code switch between the $(x,z)$ gauge color code and the $(x,D-2-x)$ and $(D-2-z,z)$ stabilizer color codes in finite depth, these three codes are all in the same FTQL phase,
\begin{equation}
    \label{eq:gcc-gauge-fix}
    \Sigma_{\Gamma_{D}}^{(x,D-2-x)} \sim \Sigma_{\Gamma_{D}}^{(x,z)} \sim \Sigma_{\Gamma_{D}}^{(D-2-z,z)} .
\end{equation}
Operations performed in the discussed code switching procedures are illustrated in Figure~\ref{fig:cc-dim-jump}(c) for $D=3$.

These relations induce a ladder-like structure of code switching protocols one can use to fault tolerantly transfer information between color codes on $D$-simplices and their facets. Specifically, in $D$ spatial dimensions, one can fault tolerantly switch between a $D$-dimensional color code on a $D$-simplex $\Gamma_{D}$ and a $(D-1)$-dimensional color code on one of its $(D-1)$-simplex facets, $\Gamma_{D-1}$, between this $(D-1)$-dimensional color code and a $(D-2)$-dimensional color code on one of its $(D-2)$-simplex facets, $\Gamma_{D-2}$, and so on until one reaches the 2D $(0,0)$ color code on a triangular sublattice of the original $D$-dimensional lattice, $\Gamma_{2}$. The code spaces of all color codes attainable in this way are then in the same FTQL phase, along with stabilizer color codes which can be obtained by code switching on the same lattice. This dimensional ladder and equivalences among different codes are illustrated in Figure~\ref{fig:cc-dim-jump}(d). Additionally, by successively applying the above results, one finds that in $D$ spatial dimensions, all stabilizer and gauge color codes defined on the sublattices $\Gamma_{i}$ of $\Gamma_{D}$, $i \in [2,D]$, are in the same FTQL phase. This is summarized in the following lemma.

\begin{lemma}
    \label{lemma:3dgcc-phases}
    Consider the lattice $\Gamma_{D}$ in $D$ spatial dimensions and its sublattices $\Gamma_{i}$, $i \in [2,D]$, defined above. All color codes $\Sigma_{\Gamma_{i}}^{(x,z)}$ are in the same FTQL phase with respect to the noise class $(\mathbb{N}^{\mathrm{loc}}_{\tau,\epsilon}, \mathbb{M}_{\zeta})$.
\end{lemma}

We emphasize that Lemma~\ref{lemma:3dgcc-phases} is proven explicitly for $D=3$ by directly applying results from Refs.~\cite{bombin2015singleshot,bombin2016}, but these arguments generalize to any $D \geq 3$, $x \leq D-3$, and $z \leq D-3$ in a straightforward manner.
Lastly, while we only discuss color codes here, Eq.~(\ref{eq:gcc-gauge-fix}) and the arguments in Ref.~\cite{bombin2015singleshot} also apply to gauge-fixed convex sets for $\Sigma_{\Gamma_{D}}^{(0,0)}$ where some gauge generators are fixed in $-1$ eigenstates, rather than fixing all $Z$-type or $X$-type generators to $+1$ as is done when code switching here.

\subsection{LC phase inequivalence of color codes} \label{cc-mixed-states}

Measurements and feedback are necessary for the above arguments establishing FTQL phase equivalence between color codes to hold. In fact, stabilizer color codes on compatible lattices $\Gamma_{D-1}$ and $\Gamma_{D}$ cannot be transformed into each other with any shallow, locally reversible LC circuit, as we show with the following theorem. To prove this, one can bound the growth of a logical operator when propagated through a locally reversible LC circuit, similar to the discussion in Section~\ref{lc-finite-depth}. This is schematically depicted in Figure~\ref{fig:cc-mixed-phases} for color codes on $\Gamma_{2}$ and $\Gamma_{3}$.

\begin{theorem}
\label{theorem:colorcode-lc-phases}
    Suppose $\Sigma_{\Gamma_{D}}^{(x,z)}$ and $\Sigma_{\Gamma_{D-1}}^{(x^\prime,z^\prime)}$ are stabilizer color codes, $\Gamma_{D}$ tiles a $D$-simplex and has linear dimension $L$, and $\Gamma_{D-1} \subseteq \Gamma_{D}$ is a $(D-1)$-dimensional facet of $\Gamma_{D}$.
    Then $\Sigma_{\Gamma_{D}}^{(x,z)}$ and $\Sigma_{\Gamma_{D-1}}^{(x^\prime,z^\prime)}$ are in different LC phases.
\end{theorem}
\begin{proof}
    First, since $x + z = D - 2$ and $x^\prime + z^\prime = D - 3$, we find $x + z = x^\prime + z^\prime + 1$. Thus, without loss of generality, we can assume $z > z^\prime$. Let $\overline{X}_{D}$ and $\overline{X}_{D-1}$ be minimum weight $X$-type logical operators for $\Sigma_{\Gamma_{D}}^{(x,z)}$ and $\Sigma_{\Gamma_{D-1}}^{(x^\prime,z^\prime)}$ respectively, meaning they each have the smallest weight of all equivalent logical operators. Their weights depend on their dimensionality on the lattice~\cite{bombin2015gauge}, specifically,
    \begin{equation}
        \label{eq:xop-weights}
        \wt \overline{X}_{D} = \Theta(L^{z+1}) , \quad 
        \wt \overline{X}_{D-1} = \Theta(L^{z^\prime+1}) .
    \end{equation}
    Since $z \geq z^\prime + 1$, $\wt \overline{X}_{D-1} = O(L^{z})$.

    \begin{figure}[t!]
        \centering
        \includegraphics[width=0.45\textwidth]{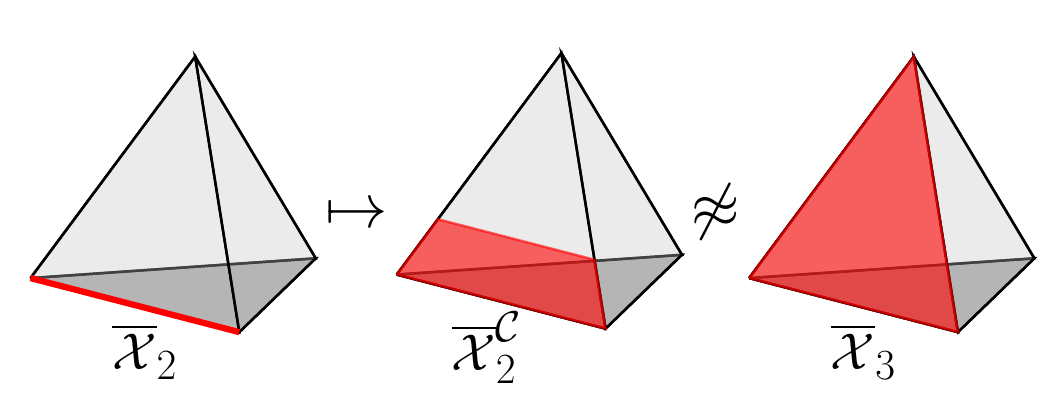}
        \caption{Shallow locally reversible LC circuits can only grow logical operators by a bounded amount in the 2D color code. This means no shallow LC circuit can enlarge all logical operators to those of the 3D color code, therefore the 2D and 3D color codes must be in different LC phases.}
        \label{fig:cc-mixed-phases}
    \end{figure}

    Consider encoding and decoding maps $\mathcal{E}_{D-1}$, $\mathcal{D}_{D-1}$ for $\Sigma_{\Gamma_{D-1}}^{(x^\prime,z^\prime)}$ and $\mathcal{E}_{D}$, $\mathcal{D}_{D}$ for $\Sigma_{\Gamma_{D}}^{(x,z)}$, where
    \begin{equation}
        \mathcal{D}_{D-1} \left[ \Sigma_{\Gamma_{D-1}}^{(x^\prime,z^\prime)} \right] = \mathcal{D}_{D} \left[ \Sigma_{\Gamma_{D}}^{(x,z)} \right] = \Omega .
    \end{equation}
    Suppose an $(r,\ell)$ locally reversible LC circuit $\mathcal{C}$ with range $r\ell = O(\polylog(L))$ exists such that for some fixed state $\rho_{Q}$ of qubits in $\Gamma_{D} \setminus \Gamma_{D-1}$ and for all $\rho \in \Omega$, $\mathcal{C}\left(\mathcal{E}_{D-1} \left( \rho \right) \otimes \rho_{Q}\right) \approx \mathcal{E}_{D} \left( \rho \right)$. Denote channels $\overline{\mathcal{X}}_{D-1} (\cdot) \coloneqq \overline{X}_{D-1} (\cdot) \overline{X}_{D-1}$ and $\overline{\mathcal{X}}_{D} (\cdot) \coloneqq \overline{X}_{D} (\cdot) \overline{X}_{D}$ for clarity. Since
    \begin{equation}
        \overline{\mathcal{X}}_{i} \circ \mathcal{E}_{i} \left( \rho \right) = \mathcal{E}_{i} \left( \overline{X} \rho \overline{X} \right) , \quad i = D-1,D , 
    \end{equation}
    where $\overline{X}$ acts on $\Omega$, we find that for all $\rho \in \Omega$,
    \begin{equation}
        \label{eq:cc-approx-eq}
        \mathcal{C} \circ \overline{\mathcal{X}}_{D-1}\left(\mathcal{E}_{D-1} \left( \rho \right) \otimes \rho_{Q}\right) \approx \overline{\mathcal{X}}_{D} \circ \mathcal{E}_{D} \left( \rho \right) .
    \end{equation}
    Since $\mathcal{C}$ is locally reversible, we also find
    \begin{align}
        &\mathcal{C} \circ \overline{\mathcal{X}}_{D-1}\left(\mathcal{E}_{D-1} \left( \rho \right) \otimes \rho_{Q}\right) \\
        &\approx \overline{\mathcal{X}}_{D-1}^{\mathcal{C}} \circ \mathcal{C} \left(\mathcal{E}_{D-1} \left( \rho \right) \otimes \rho_{Q}\right) \approx \overline{\mathcal{X}}_{D-1}^{\mathcal{C}} \circ \mathcal{E}_{D} \left( \rho \right) , \nonumber
    \end{align}
    where the dressed channel $\overline{\mathcal{X}}_{D-1}^{\mathcal{C}}$ is defined in Eq.~(\ref{eq:lc-conjugation}) and has weight upper bounded by Eq.~(\ref{eq:lc-bound}),
    \begin{equation}
        \wt \overline{\mathcal{X}}_{D-1}^{\mathcal{C}} \leq \gamma (r \ell)^{D} \wt \overline{\mathcal{X}}_{D-1}
    \end{equation}
    for some constant $\gamma$. Since $r\ell = O(\polylog(L))$, $\overline{\mathcal{X}}_{D-1}^{\mathcal{C}}$ cannot act as an $X$-type logical operator of $\Sigma_{\Gamma_{D}}^{(x,z)}$. This is because to implement the lowest weight logical operator $\overline{\mathcal{X}}_{D}$, one would need $\wt \overline{\mathcal{X}}_{D} \leq \wt \overline{\mathcal{X}}_{D-1}^{\mathcal{C}}$, which from Eq.~(\ref{eq:xop-weights}) would require $r\ell = \Omega(L^{1/D})$.
    
    Consider the specific code state $\mathcal{E}_{D} \left( \rho \right) = \dyad{\overline{0}}{\overline{0}}$, therefore $\overline{\mathcal{X}}_{D} \circ \mathcal{E}_{D} \left( \rho \right) = \dyad{\overline{1}}{\overline{1}}$. Since $\overline{\mathcal{X}}_{D-1}^{\mathcal{C}}$ cannot implement an $X$-type logical operator, we can write the state
    \begin{equation}
        \overline{\mathcal{X}}_{D-1}^{\mathcal{C}} \left( \dyad{\overline{0}}{\overline{0}} \right) = (1-q) \dyad{\overline{0}}{\overline{0}} + q \rho^{\prime}
    \end{equation}
     where $\rho^{\prime}$ is not in the code space and is orthogonal to both $\dyad{\overline{0}}{\overline{0}}$ and $\overline{\mathcal{X}}_{D} \left( \dyad{\overline{0}}{\overline{0}} \right) = \dyad{\overline{1}}{\overline{1}}$. This means
     \begin{equation}
         \norm*{\overline{\mathcal{X}}_{D-1}^{\mathcal{C}} \left( \dyad{\overline{0}}{\overline{0}} \right) - \overline{\mathcal{X}}_{D} \left( \dyad{\overline{0}}{\overline{0}} \right)}_{1} = 2, 
     \end{equation}
     therefore Eq.~(\ref{eq:cc-approx-eq}) cannot hold. This means the circuit $\mathcal{C}$ described above, with range $r\ell = O(\polylog(L))$, does not exist. Therefore, $\Sigma_{\Gamma_{D-1}}^{(x^\prime,z^\prime)}$ is in a different LC phase than $\Sigma_{\Gamma_{D}}^{(x,z)}$ according to Definition~\ref{def:lc-phase}. The case $x > x^\prime$ is handled identically by swapping $X \leftrightarrow Z$ in the above arguments.
\end{proof}

Together, Lemma~\ref{lemma:3dgcc-phases} and Theorem~\ref{theorem:colorcode-lc-phases} demonstrate the strict inclusion $LC \subsetneq FTQL$ given in Eq.~(\ref{eq:stable-phase-inclusions}). This is because from Lemma~\ref{lemma:3dgcc-phases}, the code spaces of the 2D and 3D stabilizer color codes are in the same FTQL phase, but from Theorem~\ref{theorem:colorcode-lc-phases}, they are in different LC phases.

\section{Robust state preparation and QEC} \label{robust-meas-feedback}

In this section, we study single-shot error correction and state preparation of qudit stabilizer codes to establish FTQL state equivalence between code states and product states. These examples will also highlight general properties which can make measurement-feedback circuits fault tolerant. The codes we consider here are those which exhibit self-correcting properties and allow passive stability at finite temperature~\cite{bombin2013}. In Ref.~\cite{bombin2015singleshot}, it was shown that qubit codes exhibiting these properties permit single-shot QEC for Pauli noise. Here, we show that these codes permit single-shot state preparation as well. We also generalize results to qudit codes and general channel noise, Eq.~(\ref{eq:qudit-noise}). Since results and techniques used to study these codes are very similar to those in Ref.~\cite{bombin2015singleshot}, we highlight key properties here and discuss technical details and subtleties which arise when treating qudits and general channel noise in Appendices~\ref{qudit-qec} and~\ref{self-corr-all}.

\subsection{Error correction for qudit stabilizer codes} \label{qudit-codes}

Let $\mathcal{S} \subset \mathcal{P}$ be a stabilizer group with code space $\mathcal{V}$. Given a generating set $\mathcal{S}_{0} = \{S_{i}\}$ of $\mathcal{S}$, we say the Pauli operator $P \in \mathcal{P}$ has syndrome $\sigma \in \mathbb{Z}_{d}^{\vert \mathcal{S}_{0} \vert}$ if $S_{i} P = \omega^{\sigma_{i}} P S_{i}$ for all $S_{i} \in \mathcal{S}_{0}$. If $P$ has syndrome $\sigma$, it can be decomposed as $P = R_{\sigma}^{\dagger} S_{P} L_{P}$, where $R_{\sigma}^{\dagger}$ has syndrome $\sigma$, $S_{P} \in \mathcal{S}$, and $L_{P} \in \mathcal{L}_{\mathcal{S}}$, where $\mathcal{L}_{\mathcal{S}} \subseteq \mathcal{Z}(\mathcal{S})$ is a representative set of logical operators. Here, $\{ R_{\sigma} \} \subset \mathcal{P}$ is a set of chosen Pauli correction operators. Syndromes are additive, so if $P_{1}$ has syndrome $\sigma_{1}$ and $P_{2}$ has syndrome $\sigma_{2}$, then $P_{1}P_{2}$ has syndrome $\sigma_{1} + \sigma_{2} \bmod d$.

If $P = R_{\sigma}^{\dagger} S_{P} L_{P}$ acts on a code state $\ket{\psi} \in \mathcal{V}$, one can perform QEC by first measuring all $S_{i} \in \mathcal{S}_{0}$, to obtain the syndrome $\sigma$, then applying the correction operator $R_{\sigma}$. After correction, the system will be in the state $S_{P} L_{P} \ket{\psi} \in \mathcal{V}$. Denoting $\Pi^{S_{i}}_{\sigma_{i}}$ as the projector onto the $\sigma_{i}$-eigenspace of $S_{i}$ and $\Pi^{\mathcal{S}}_{\sigma} = \prod_{i} \Pi^{S_{i}}_{\sigma_{i}}$, QEC may be implemented by the recovery channel
\begin{equation}
    \mathcal{R} = \left\{ R_{\sigma} \Pi_{\sigma}^{\mathcal{S}} \right\}_{\sigma} .
\end{equation}
Note that $\Pi^{\mathcal{S}}_{0}$ projects onto the code space $\mathcal{V}$. For a more detailed discussion of QEC, see Appendix~\ref{qudit-qec}.

If $\mathcal{S}_{0}$ is overcomplete, relations exist between eigenvalues of $S_{i}$ and not all arbitrary $\sigma \in \mathbb{Z}_{d}^{\vert \mathcal{S}_{0} \vert}$ are \textit{valid} syndromes. A syndrome $\sigma$ is \textit{valid} if and only if $\Pi_{\sigma}^{\mathcal{S}} \neq 0$. If $\sigma$ and $\sigma^\prime$ are valid, so is $\sigma + \sigma^\prime$. We will sometimes call valid syndromes \textit{excitations}, as they are associated with excited states of the Hamiltonian Eq.~(\ref{eq:stabilizer-hamiltonian}). Ideal stabilizer measurements will always yield valid syndromes, however measurement errors $\mu$ may lead to invalid observed syndromes $\sigma + \mu$.  In this situation, the measurement error $\mu$ decomposes into components $\mu = \mu_{v} + \mu_{e}$, where $\mu_{v}$ is itself valid and $\mu_{e}$ is invalid. One way to choose a correction operator based on an invalid syndrome $\sigma + \mu$ is to isolate the invalid component, $\mu_{e}$, and add a syndrome correction $\hat{\mu}(\mu_{e})$ which depends only on $\mu_{e}$~\cite{bombin2015singleshot,campbell2019}. This guarantees the corrected syndrome $\sigma + \mu + \hat{\mu}$, is valid and one may apply the correction operator $R_{\sigma + \mu + \hat{\mu}}$.

\subsection{Sufficient properties for single-shot state preparation} \label{self-correction-setup}

Consider a qudit stabilizer code $\mathcal{V}$ which has a geometrically local, LDPC generating set $\mathcal{S}_{0}$ of its stabilizer group, meaning every generator is supported within a ball of finite radius and every qudit is in the support of a finite number of generators. A set of sufficient conditions for such a code to permit single-shot QEC are loosely that (1) syndromes can be written as a sum of connected components, $\sigma = \sum_{c} \sigma_{c}$, where each component $\sigma_{c}$ is itself a valid syndrome, (2) correction operators $\{ R_{\sigma} \}$ may be chosen to correct each component $\sigma_{c}$ separately, and (3) the code has distance which scales with the system size, $d_{\mathcal{V}} = \Omega (n^{a})$ for some $a>0$. Precise conditions follow directly from those given in Ref.~\cite{bombin2015singleshot} for self-correcting codes, and we discuss them in detail in Appendix~\ref{self-corr-conditions}.

Given such a code, we can perform QEC using an $(r,\ell)$ QL circuit, $\Lambda_{\mathrm{EC}}$, which has $r\ell = O(1)$. We will choose $\Lambda_{\mathrm{EC}}$ which first measure each $S_{i} \in \mathcal{S}_{0}$ once, then apply correction operators given the observed noisy syndrome $\sigma + \mu$. The correction operator is determined by choosing a syndrome correction $\hat{\mu}(\mu)$, as discussed in Section~\ref{qudit-codes}, which has minimal weight among all possible corrections which yield a valid syndrome $\sigma + \mu + \hat{\mu}$, then applying $R_{\sigma + \mu + \hat{\mu}} \in \{ R_{\sigma} \}$ to the system. In the noiseless case, $\Lambda_{\mathrm{EC}}$ satisfies the properties of Definition~\ref{def:recovery-channel} and may serve as a recovery channel for $\Sigma_{\mathcal{V}}$.

We consider here a generalization of local stochastic noise, $\mathbb{N}_{\tau,\epsilon}^{\mathrm{exc}}$, which implicitly depends on the code $\mathcal{V}$ and recovery channel $\mathcal{R}$ used. Here, we choose $\mathcal{R} = \Lambda_{\mathrm{EC}}$. Loosely, $\mathbb{N}_{\tau,\epsilon}^{\mathrm{exc}}$ is defined such that for any channel $\mathcal{N} \in \mathbb{N}_{\tau,\epsilon}^{\mathrm{exc}}$, the distribution of syndromes is $\tau$-bounded and the logical error rate is bounded by $\epsilon$, however we explicitly define and discuss properties of $\mathbb{N}_{\tau,\epsilon}^{\mathrm{exc}}$ in Appendix~\ref{exc-noise-properties} and we remark it generalizes Definition~11 of Ref.~\cite{bombin2015singleshot} to general channel noise. Importantly, $\mathbb{N}_{\tau,\epsilon}^{\mathrm{exc}}$ is $\Lambda_{\mathrm{EC}}$-recoverable for $\epsilon = 1/\Omega(\poly(n))$, and together $\mathcal{V}$ and $\Lambda_{\mathrm{EC}}$ support single-shot QEC against $\mathbb{N}_{\tau,\epsilon}^{\mathrm{exc}}$.
\begin{theorem}
    \label{theorem:self-corr-qec}
    Consider a code $\mathcal{V}$ and a circuit $\Lambda_{\mathrm{EC}}$ satisfying the above properties. Then, $\Lambda_{\mathrm{EC}} : \Sigma_{\mathcal{V}} \rightarrow \Sigma_{\mathcal{V}}$ is fault tolerant against the noise class $(\mathbb{N}_{\tau,\epsilon}^{\mathrm{exc}},\mathbb{M}_{\zeta})$ for $\epsilon = 1/\Omega(\poly(n))$.
\end{theorem}
We prove Theorem~\ref{theorem:self-corr-qec} in Appendix~\ref{self-corr-proof}, which directly generalizes Theorem 14 of Ref.~\cite{bombin2015singleshot} from qubits to qudits and from Pauli noise to general channel noise.

If $\mathcal{V}$ is a CSS code and $\mathcal{S}_{0} = \mathcal{S}^{X}_{0} \oplus \mathcal{S}^{Z}_{0}$ decomposes into $X$-type generators $\mathcal{S}^{X}_{0}$ and $Z$-type generators $\mathcal{S}^{Z}_{0}$, then Theorem~\ref{theorem:self-corr-qec} may hold for correction of $Z$-type Pauli channels only, $\mathcal{N}^{Z} \in \mathbb{N}_{\tau,\epsilon}^{\mathrm{exc},Z} \subset \mathbb{N}_{\tau,\epsilon}^{\mathrm{exc}}$ or of $X$-type Pauli channels only, $\mathcal{N}^{X} \in \mathbb{N}_{\tau,\epsilon}^{\mathrm{exc},X} \subset \mathbb{N}_{\tau,\epsilon}^{\mathrm{exc}}$. For the codes discussed here, this allows one to fault tolerantly prepare the encoded state $\dyad{\overline{0}}{\overline{0}}^{\otimes k}$ from $\dyad{0}{0}^{\otimes n}$ or $\dyad{\overline{+}}{\overline{+}}^{\otimes k}$ from $\dyad{+}{+}^{\otimes n}$ in finite depth. Note that the ability to perform single-shot QEC does not always allow one to perform single-shot state preparation. This is because some codes rely on errors having sufficiently low weight or for high weight errors to occur with sufficiently low probability to perform single-shot QEC~\cite{quintavalle2021}, and this assumption does not apply to state preparation. Other codes, such as those discussed here, only depend on properties of syndromes and measurement errors~\cite{campbell2019,quintavalle2021}.

Here, the ability to perform single-shot state preparation from a product state is synonymous with FTQL state equivalence between $\dyad{\overline{0}}{\overline{0}}^{\otimes k}$ and $\dyad{0}{0}^{\otimes n}$ (or between $\dyad{\overline{+}}{\overline{+}}^{\otimes k}$ and $\dyad{+}{+}^{\otimes n}$). This is because one can fault tolerantly prepare $\dyad{0}{0}^{\otimes n}$ ($\dyad{+}{+}^{\otimes n}$) from any arbitrary state by measuring $Z_{q}$ ($X_{q}$) at each site and applying $X_{q}$ ($Z_{q}$) dependent on the measurement outcome.

\begin{theorem}
    \label{theorem:self-corr-state-prep}
    Suppose $\mathcal{V}$ is a CSS code which satisfies the above properties and supports single-shot QEC according to Theorem~\ref{theorem:self-corr-qec}, for $Z$-type Pauli channels with QEC circuit $\Lambda_{\mathrm{EC}}^{X}$. Then the states $\dyad{\overline{0}}{\overline{0}}^{\otimes k}$ and $\dyad{0}{0}^{\otimes n}$ are in the same FTQL phase with respect to $(\mathbb{N}_{\tau,\epsilon}^{\mathrm{exc}},\mathbb{M}_{\zeta})$. Similarly, if $\mathcal{V}$ supports single-shot QEC for $X$-type Pauli channels with QEC circuit $\Lambda_{\mathrm{EC}}^{Z}$, then $\dyad{\overline{+}}{\overline{+}}^{\otimes k}$ and $\dyad{+}{+}^{\otimes n}$ are in the same FTQL phase with respect to $(\mathbb{N}_{\tau,\epsilon}^{\mathrm{exc}},\mathbb{M}_{\zeta})$.
\end{theorem}
We prove Theorem~\ref{theorem:self-corr-state-prep} in Appendix~\ref{self-corr-prep-proof}. Here it is important to distinguish between state equivalence and convex set equivalence. This is because $\Sigma_{\mathcal{V}}$ is stable and in a nontrivial FTQL phase, however individual states $\ket{\overline{0}}^{\otimes k}$ or $\ket{\overline{+}}^{\otimes k}$ may be equivalent to product states as no encoded information is transferred in state preparation.

\subsection{Robust state preparation in \textit{D}-dimensional qudit toric codes} \label{q-form-tc}

One family of codes which permit single-shot QEC and state preparation according to Theorems~\ref{theorem:self-corr-qec} and~\ref{theorem:self-corr-state-prep} are $D$-dimensional qudit toric codes. We denote these codes $(D,q)$ $\mathbb{Z}_{d}$ toric codes, where $D$ is the spatial dimension, $d$ is the local qudit dimension, and $q$ specifies the structure of excitations in the system~\cite{bombin2013,schulz2012,hastings2014percolation,anwar2014,watanabe2023}. Consider a $D$-dimensional lattice $\mathcal{J}$. Assign an orientation to each $p$-cell in $\Delta_p$, $p \in [0,D]$. Let $\Omega_p$ be a vector space with basis $\Delta_p$. Define an inner product on $\Omega_p$ where $\langle \varphi , \varphi^\prime \rangle_p = \delta_{\varphi,\varphi^\prime}$ for $\varphi, \varphi^\prime \in \Delta_p$, a boundary map $\partial_{p} : \Omega_{p} \rightarrow \Omega_{p-1}$, where $\partial_{p} \vartheta$ takes a $p$-cell $\vartheta$ to a linear combination of $(p-1)$ cells in its boundary with coefficients given by their orientation relative to $\vartheta$, and its adjoint with respect to this inner product $\partial^{\dagger}_{p} : \Omega_{p-1} \rightarrow \Omega_{p}$, where $\partial^{\dagger}_{p}\varphi$ picks out the $p$-cells whose boundary the $(p-1)$-cell $\varphi$ appears in with coefficient given by their orientation relative to $\varphi$. Here, $\partial_{p-1}\partial_{p} \vartheta = 0$ always.

To construct a $(D,q)$ $\mathbb{Z}_{d}$ toric code $\mathcal{V}$, place a $d$-dimensional qudit on each $q$-cell. Associate $Z$-type stabilizer generators $A_{s}$ with each $(q-1)$-cell $s \in \Delta_{q-1}$ and $X$-type stabilizer generators $B_{m}$ with each $(q+1)$-cell $m \in \Delta_{q+1}$, defined as
\begin{equation}
    A_{s} = \prod_{k \varphi \in \partial^{\dagger}_{q} s} Z_{\varphi}^{k} \quad , \quad B_{m} = \prod_{k \varphi \in \partial_{q+1} m} X_{\varphi}^{k} .
\end{equation}
All $A_{s}$ and $B_{m}$ commute with each other, so $\mathcal{V}$ has stabilizer group $\mathcal{S} = \left\langle A_{s} , B_{m} \right\rangle$, and code states in $\Sigma_{\mathcal{V}}$ are mutual $+1$ eigenstates of all $A_{s}$ and $B_{m}$.

For $D \geq 4$ and $q \in [2,D-2]$, $(D,q)$ $\mathbb{Z}_{d}$ toric codes permit both single-shot error correction and single-shot state preparation, and for $D \geq 2$ and either $q \geq 2$ or $q \leq D-2$, they permit single-shot error correction of either $X$-type noise or $Z$-type noise, respectively, and single-shot state preparation of either $\dyad{\overline{0}}{\overline{0}}^{\otimes k}$ or $\dyad{\overline{+}}{\overline{+}}^{\otimes k}$. We summarize these results with the following corollaries, which we prove in Appendix~\ref{tc-single-shot-proof}.

\begin{corollary}
    \label{cor:tc-single-shot}
    Suppose $\mathcal{J}$ is a regular lattice which tiles a manifold in $D$ spatial dimensions and has minimum linear size $L = \Omega(n^{1/D})$. Consider a $(D,q)$ $\mathbb{Z}_{d}$ toric code $\mathcal{V}$ on $\mathcal{J}$, where $D \geq 4$. If $q \in [2,D-2]$, then Theorems~\ref{theorem:self-corr-qec} and~\ref{theorem:self-corr-state-prep} both apply to $\mathcal{V}$, meaning it exhibits both single-shot QEC and single-shot state preparation.
\end{corollary}

\begin{corollary}
    \label{cor:ghz-single-shot}
    Suppose $\mathcal{J}$ is a regular lattice which tiles a manifold in $D$ spatial dimensions and has minimum linear size $L = \Omega(n^{1/D})$. Consider a $(D,q)$ $\mathbb{Z}_{d}$ toric code $\mathcal{V}$ on $\mathcal{J}$, where $D \geq 2$. If either $q \geq 2$ and only Pauli $X$ errors occur, or $q \leq D-2$ and only Pauli $Z$ errors occur, then Theorems~\ref{theorem:self-corr-qec} and~\ref{theorem:self-corr-state-prep} both apply to $\mathcal{V}$.
\end{corollary}

Corollary~\ref{cor:tc-single-shot} applies to the $(4,2)$ $\mathbb{Z}_{d}$ toric code, which is the well-known 4D toric code with looplike syndromes~\cite{dennis2002}, as well as various other toric codes in $D>4$. In $D=3$, Corollary~\ref{cor:ghz-single-shot} applies to error correction of $Z$ errors and state preparation of $\dyad{\overline{0}}{\overline{0}}^{\otimes k}$, or to error correction of $X$ errors and state preparation of $\dyad{\overline{+}}{\overline{+}}^{\otimes k}$, but not both. In $D \geq 2$, the $(D,D)$ $\mathbb{Z}_{2}$ toric code is a classical repetition code and the $\dyad{\overline{+}}{\overline{+}}^{\otimes k}$ state of this code is the $D$ dimensional GHZ state. Therefore Corollary~\ref{cor:ghz-single-shot} implies that the GHZ state in $D \geq 2$ can be prepared fault tolerantly with respect to $X$-type noise in finite depth.

Together, Corollaries~\ref{cor:ldpc-code-nontrivial},~\ref{cor:tc-single-shot}, and~\ref{cor:ghz-single-shot} exemplify the difference between fault tolerant state preparation and fault tolerant information transfer. Since $(D,q)$ $\mathbb{Z}_{d}$ toric codes, and more generally the code families discussed in Section~\ref{self-correction-setup}, are LDPC and therefore stable, Corollary~\ref{cor:ldpc-code-nontrivial} implies their code spaces $\Sigma_{\mathcal{V}}$ are in different phases than convex sets of unencoded qudits. However, Corollaries~\ref{cor:tc-single-shot} and~\ref{cor:ghz-single-shot} demonstrate one can fault tolerantly prepare specific code states from product states. This is because there is no transfer of encoded information when preparing a specific state and therefore no information can be corrupted, as discussed in Section~\ref{ft-state-prep}.

\section{Non-robust quantum-local circuits} \label{non-robust-circuits}

In this section, we show that syndrome structure, specifically having pointlike syndromes or slight generalizations, can prohibit fault tolerance in certain classes of measurement-feedback circuits which have range at most $O(\log\log(n))$. We will study both stabilizer measurements in qudit codes and the more general case of measuring commuting terms of a frustration-free Hamiltonian with Abelian excitations, such as the Hamiltonian of a quantum double model~\cite{kitaev2003} or string-net model~\cite{levin2005}. Qualitatively, the reason pointlike syndromes can prohibit fault tolerance is because low weight measurement errors can mimic the syndrome of an arbitrarily high weight operator acting on the system. This is illustrated in Figure~\ref{fig:no-ft-strings} for the case of a stabilizer code, where the syndrome of a stringlike Pauli error $P$ has weight independent of $\wt P$. When performing QEC or state preparation, a measurement error which looks like this syndrome may cause one to falsely attempt to correct $P$ and instead apply a high weight error to the system.

\subsection{Properties of non-robust circuits} \label{non-robust-properties}

\begin{figure}[t!]
    \centering
    \includegraphics[width=0.45\textwidth]{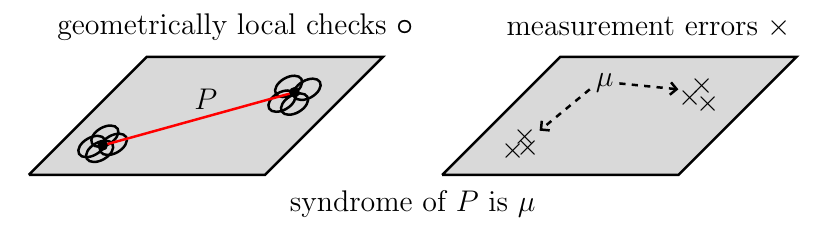}
    \caption{Illustrative example of when measurement-feedback circuits cannot be fault tolerant. If stabilizer measurement outcomes yield pointlike syndromes, then a low weight measurement error $\mu$ can reproduce the syndrome of a high weight Pauli error $P$. A particular channel of local stochastic measurement noise may implement this measurement error $\mu$ with finite probability, independent of $\wt P$.}
    \label{fig:no-ft-strings}
\end{figure}

Here, we enumerate properties which prohibit fault tolerance and are shared by a range of shallow circuits implementing QEC, state preparation, and sequential gauging of topological orders. To make statements precise, we use the following definition to discuss geometrically local, LDPC sets or multisets of operators to be measured.
\begin{definition}
    \label{def:geom-local}
    A set, or multiset, of operators $\mathcal{I}$ is $(v_{q},v_{s})$ geometrically local if each $O_{i} \in \mathcal{I}$ is supported within a ball of volume $v_{q}$ on the system lattice, and each qudit is in the support of at most $v_{s}$ operators $O_{i} \in \mathcal{I}$.
\end{definition}
All operators in $\mathcal{I}$ can be measured by a circuit with range $r\ell = \Theta( v_{q}^{1/D} v_{s} )$. Also, the multiplicity of an operator in a multiset $\mathcal{I}$ is upper bounded by $v_{s}$.

To study both stabilizer codes and local Hamiltonian models, such as quantum double models and string-net models, within the same framework, we will consider circuits which measure terms or products of terms in some frustration-free Hamiltonian. Consider a system of $\mathbb{Z}_{d}$ qudits on a lattice and a Hamiltonian of the form
\begin{equation}
    \label{eq:gapped-hamiltonian}
    H = - J_{\mathcal{S}} \sum_{S \in \mathcal{S}_{0}} (S + S^{\dagger}) - J_{\mathcal{T}} \sum_{T \in \mathcal{T}_{0}} (T + T^{\dagger}) .
\end{equation}
$\mathcal{S}_{0}$ and $\mathcal{T}_{0}$ are sets of unitary operators which satisfy
\begin{enumerate}[itemsep=0em,start=1,label=(\alph*)]
    \item For all $S \in \mathcal{S}_{0}$ and all $T \in \mathcal{T}_{0}$, $[S,T] = 0$.
    \item $\Pi_{+1}^{\mathcal{T}_{0}} = \prod_{i} \Pi_{+1}^{T_{i}} \neq 0$, namely $\mathcal{T}_{0}$ has a $+1$ eigenspace.
    \item For all $S_{i},S_{j} \in \mathcal{S}_{0}$, $\left[S_{i},S_{j}\right] \Pi_{+1}^{\mathcal{T}_{0}} = 0$.
\end{enumerate}
Here, $\mathcal{T}_{0}$ captures properties of the initial state of the system and $\mathcal{S}_{0}$ captures properties of the final state of the system and of the measurements made. Let $\mathcal{V}$ be the ground state subspace of $H$, namely the mutual $+1$ eigenspace of all $S \in \mathcal{S}_{0}$ and $T \in \mathcal{T}_{0}$, and call the convex set of all mixed ground states $\Sigma_{\mathcal{V}}$ as in Eq.~(\ref{eq:conv-set-stabilizer}).

To account for measurements of products of operators in $\mathcal{S}_{0}$ and for repeated measurements of the same operator, we let $\mathcal{S}_{\mathrm{meas}}$ be a multiset of unitary operators which satisfies $\left\langle \mathcal{S}_{\mathrm{meas}} \right\rangle = \left\langle \mathcal{S}_{0} \right\rangle$.
Suppose each $S_{j} \in \mathcal{S}_{\mathrm{meas}}$ has $d_{j}$ distinct eigenvalues $s_{j}(a) = e^{2 \pi i a / d_{j}}$, where $0 \leq a < d_{j}$. Consider a set of unitary operators $\mathcal{B}_{\mathcal{S}}$, where each $R_{\sigma}^{\dagger} \in \mathcal{B}_{\mathcal{S}}$ is labeled by a syndrome $\sigma$ with elements $\sigma_{j}$ corresponding to each $S_{j} \in \mathcal{S}_{\mathrm{meas}}$. Here, $\mathcal{B}_{\mathcal{S}}$ generalizes the set of correction operators used in stabilizer QEC codes. Suppose each $R_{\sigma}^{\dagger}$ satisfies
\begin{equation}
    \label{eq:hamiltonian-correction}
    S_{j} R_{\sigma}^{\dagger} \Pi_{+1}^{\mathcal{T}_{0}} = s_{j}(\sigma_{j}) R_{\sigma}^{\dagger} S_{j} \Pi_{+1}^{\mathcal{T}_{0}} \quad \forall S_{j} \in \mathcal{S}_{\mathrm{meas}} .
\end{equation}
Not all $\sigma_{j}$ may be independent, and there may exist constraints on which lists $\sigma$ are valid syndromes, analogous to the stabilizer case as discussed in Section~\ref{qudit-codes}. Additionally, let $\mathcal{L}_{\mathcal{S}}$ form a basis of unitary logical operators, where each $L \in \mathcal{L}_{\mathcal{S}}$ commutes with all operators in both $\mathcal{S}_{0}$ and $\mathcal{T}_{0}$ but may map between ground states of $H$. Defining $\Pi_{\sigma}^{\mathcal{S}}$ as the projector onto the mutual $\sigma$ eigenspace of all operators in $\mathcal{S}_{\mathrm{meas}}$, one obtains the relations on projectors and correction operators
\begin{align}
    \label{eq:correction-relations}
     &R_{\sigma}^{\dagger} \Pi_{\kappa}^{\mathcal{S}} \Pi_{+1}^{\mathcal{T}_{0}} = \Pi_{\kappa+\sigma}^{\mathcal{S}} R_{\sigma}^{\dagger} \Pi_{+1}^{\mathcal{T}_{0}} , \quad
     \Pi_{\sigma}^{\mathcal{S}} \Pi_{\kappa}^{\mathcal{S}} \Pi_{+1}^{\mathcal{T}_{0}} = \delta_{\sigma,\kappa} \Pi_{\sigma}^{\mathcal{S}} \Pi_{+1}^{\mathcal{T}_{0}} , \nonumber \\
     &R_{\sigma} R_{\kappa} R_{\sigma+\kappa}^{\dagger} \Pi_{+1}^{\mathcal{T}_{0}} = L_{\sigma,\kappa} S_{\sigma,\kappa} \Pi_{+1}^{\mathcal{T}_{0}} ,
\end{align}
where $L_{\sigma,\kappa} \in \left\langle\mathcal{L}_{\mathcal{S}}\right\rangle$ and $S_{\sigma,\kappa} \in \left\langle \mathcal{S}_{0} \right\rangle$.

Suppose that
\begin{enumerate}[itemsep=0em,start=1]
    \item $\mathcal{S}_{0}$ is $(v_{q},v_{s})$ geometrically local and $\mathcal{S}_{\mathrm{meas}}$ is $(w_{q},w_{s})$ geometrically local.
\end{enumerate}
Given $\mathcal{S}_{0}$, define a graph $\Gamma$, where each vertex corresponds to a qudit and two vertices are connected by an edge if some $S_{j} \in \mathcal{S}_{0}$ acts nontrivially on both qudits. $\Gamma$ has maximum vertex degree $z \leq (v_{q} - 1) v_{s}$. Suppose
\begin{enumerate}[itemsep=0em,start=2]
    \item There exist syndromes $\sigma^{\star}$ with bounded weight $\vert \sigma^{\star} \vert = O(w_{q} w_{s})$ and a constant $\nu > 0$ such that if any operator $K_{\sigma^{\star}}$ satisfies $\Pi_{\sigma^{\star}}^{\mathcal{S}} K_{\sigma^{\star}} \rho K_{\sigma^{\star}}^{\dagger} \Pi_{\sigma^{\star}}^{\mathcal{S}} \neq 0$ for some $\rho \in \Sigma_{\mathcal{V}}$, then $\supp K_{\sigma^{\star}}$ has at least one connected component on $\Gamma$ with weight $g_{\sigma^{\star}} = \Omega(n^\nu)$.
\end{enumerate}
Denote such a syndrome $\sigma^{\star}$ a terminal syndrome and an operator $K_{\sigma^{\star}}$ a terminal operator. Terminal syndromes often correspond to pointlike excitations of the Hamiltonian Eq.~(\ref{eq:gapped-hamiltonian}). For example, in the 2D toric code, anyons created at the endpoints of long, stringlike Pauli operators correspond to a terminal syndrome.

Consider a deterministic measurement-feedback circuit $\Lambda$ which acts as $\Lambda (\rho_{1}) = \rho_{2}$ for some initial state $\rho_{1}$ in the $+1$ eigenspace of $\mathcal{T}_{0}$, $\Pi_{+1}^{\mathcal{T}_{0}} \rho_{1} \Pi_{+1}^{\mathcal{T}_{0}} = \rho_{1}$, and $\rho_{2} \in \Sigma_{\mathcal{V}}$. Suppose $\Lambda$ involves measuring all operators in $\mathcal{S}_{\mathrm{meas}}$ and applying operators in $\mathcal{B}_{\mathcal{S}}$ contingent on the measurement outcomes $\sigma$. $\Lambda$ can be written explicitly as
\begin{equation}
    \Lambda = \left\{R_{\sigma} \Pi_{\sigma}^{\mathcal{S}}\right\}_{\sigma}, \quad R_{\sigma} \Pi_{\sigma}^{\mathcal{S}} \rho_{1} \Pi_{\sigma}^{\mathcal{S}} R_{\sigma}^{\dagger} = c_{\sigma} \rho_{2} ,
\end{equation}
where $\sum_{\sigma} c_{\sigma} = 1$. In the presence of measurement errors $\mu$, where $\sigma + \mu$ is observed instead of $\sigma$, we place no restrictions on the operators in $\mathcal{B}_{\mathcal{S}}$ applied except
\begin{enumerate}[itemsep=0em,start=3]
    \item if $\sigma + \mu = \omega$ is a valid syndrome, then $R_{\omega}$ is applied.
\end{enumerate}

Such a circuit $\Lambda$ captures QEC performed by measuring stabilizers and applying corrections, state preparation of stabilizer codes and topologically ordered systems by measurement and feedback, and sequential gauging for state preparation of non-Abelian topological order, where multiple such circuits are applied one after another.

\subsection{No-go theorem for non-robust circuits} \label{non-robust-proof}

We now prove that no circuit $\Lambda$ which satisfies all the properties of Section~\ref{non-robust-properties} and has range $O(\log\log(n))$ is fault tolerant against local stochastic noise, captured by the following theorem. The basic idea of this proof is to construct a distribution of local stochastic measurement errors which, with bounded probability, mimics a terminal syndrome $\sigma^{\star}$. Therefore, the correction operator $R_{\sigma^{\star}}$ applied will have unbounded weight. While we use this specific distribution for simplicity and to highlight how the existence of terminal syndromes prohibits fault tolerance, one could reach a similar conclusion by applying i.i.d.~noise to measurements clustered in local regions which are well-separated on the system lattice. Also, range $O(\log\log(n))$ circuits are used
rather than range $O(\polylog(n))$ circuits
to ensure the noiseless and noisy final states are distinct, $\Lambda(\rho_{1}) \not\approx \Lambda^{\mathcal{N}} (\rho_{1})$, and to lower bound the residual noise strength.

\begin{theorem}
    \label{theorem:state-prep-no-ft}
    Consider a deterministic measurement-feedback circuit $\Lambda$, acting as $\Lambda (\rho_{1}) = \rho_{2}$, where $\rho_{2}$ is a ground state of a Hamiltonian of the form Eq.~(\ref{eq:gapped-hamiltonian}). Suppose $\Lambda$, $\mathcal{S}_{\mathrm{meas}}$, $\mathcal{S}_{0}$, $\mathcal{T}_{0}$, $\mathcal{B}_{\mathcal{S}}$, and $\rho_{1}$ satisfy the properties of Section~\ref{non-robust-properties}, $v_{q},v_{s} = O(1)$, and $w_{q},w_{s} = O(\log\log(n))$. Then, $\Lambda$ is not fault tolerant against local stochastic noise, $(\mathbb{L}_{\lambda}, \mathbb{M}_{\zeta})$. 
\end{theorem}
\begin{proof}
    Consider a terminal syndrome $\sigma^{\star}$ satisfying property 2 in Section~\ref{non-robust-properties}. Consider a noisy implementation of the channel $\Lambda$, $\Lambda^{\mathcal{N}}$, where correlated measurement noise acts as
    \begin{align}
        \pr(\mu) = \begin{cases}
            (1-p^{\vert \sigma^{\star} \vert}) & \mu = 0 \\
            p^{\vert \sigma^{\star} \vert} & \mu =\sigma^{\star} \\
            0 & \mathrm{else}
        \end{cases} .
    \end{align}
    Note that $\left\{ \mu , \pr(\mu) \right\}_{\mu} \in \mathbb{M}_{p}$. $\Lambda^{\mathcal{N}}$ may be written as
    \begin{align}
        \Lambda^{\mathcal{N}} = \left\{ \sqrt{\pr(\mu)} R_{\sigma + \mu} \Pi_{\sigma}^{\mathcal{S}} \right\}_{\mu,\sigma} .
    \end{align}
    Using $R_{\sigma + \mu} = R_{\sigma + \mu}R_{\sigma}^{\dagger} R_{\sigma}$ and $R_{\sigma} \Pi_{\sigma}^{\mathcal{S}} \rho_{1} \Pi_{\sigma}^{\mathcal{S}} R_{\sigma}^{\dagger} = c_{\sigma} \rho_{2}$,
    \begin{equation}
        \Lambda^{\mathcal{N}} (\rho_{1}) = \sum_{\mu \in \{0, \sigma^{\star}\}} \sum_{\sigma} \pr(\mu) c_{\sigma} R_{\sigma + \mu} R_{\sigma}^{\dagger} \rho_{2} R_{\sigma} R_{\sigma + \mu}^{\dagger} .
    \end{equation}
    $\sigma^{\star}$ is a valid syndrome, thus property 3 and Eq.~(\ref{eq:correction-relations}) in Section~\ref{non-robust-properties} imply $R_{\sigma + \sigma^{\star}} R_{\sigma}^{\dagger} \Pi_{+1}^{\mathcal{T}_{0}} = R_{\sigma^{\star}} L_{\sigma,\sigma^{\star}} S_{\sigma,\sigma^{\star}} \Pi_{+1}^{\mathcal{T}_{0}}$ for all valid $\sigma$. Using this and $\rho_{2} = \Pi_{+1}^{\mathcal{T}_{0}} \rho_{2} \Pi_{+1}^{\mathcal{T}_{0}}$, we find
    \begin{align}
        \Lambda^{\mathcal{N}} (\rho_{1})
        &= (1-p^{\vert \sigma^{\star} \vert}) \sum_{\sigma} c_{\sigma} R_{\sigma} R_{\sigma}^{\dagger} \rho_{2} R_{\sigma} R_{\sigma}^{\dagger} \\
        &+ p^{\vert \sigma^{\star} \vert} \sum_{\sigma} c_{\sigma} R_{\sigma^{\star}} L_{\sigma,\sigma^{\star}} S_{\sigma,\sigma^{\star}} \rho_{2} S_{\sigma,\sigma^{\star}}^{\dagger} L_{\sigma,\sigma^{\star}}^{\dagger} R_{\sigma^{\star}}^{\dagger} . \nonumber 
    \end{align}
    Since $S_{\sigma,\sigma^{\star}} \rho_{2} S_{\sigma,\sigma^{\star}}^{\dagger} = \rho_{2}$ for all $S_{\sigma,\sigma^{\star}}$ and $\sum_{\sigma} c_{\sigma} = 1$,
    \begin{align}
        \Lambda^{\mathcal{N}} (\rho_{1}) &= (1-p^{\vert \sigma^{\star} \vert}) \rho_{2} + p^{\vert \sigma^{\star} \vert} \mathcal{W}_{\Lambda} (\rho_{2}) , \\
        \mathcal{W}_{\Lambda} (\rho_{2}) &= R_{\sigma^{\star}} \left(\sum_{\sigma} c_{\sigma} L_{\sigma,\sigma^{\star}} \rho_{2} L_{\sigma,\sigma^{\star}}^{\dagger} \right) R_{\sigma^{\star}}^{\dagger} . \nonumber
    \end{align}
    Projectors $\Pi_{0}^{\mathcal{S}}$ and $\Pi_{\sigma^{\star}}^{\mathcal{S}}$ act on $\Lambda^{\mathcal{N}} (\rho_{1})$ as
    \begin{align}
        \Pi_{0}^{\mathcal{S}} \Lambda^{\mathcal{N}} (\rho_{1}) \Pi_{0}^{\mathcal{S}} &= (1-p^{\vert \sigma^{\star} \vert}) \rho_{2} , \\
        \Pi_{\sigma^{\star}}^{\mathcal{S}} \Lambda^{\mathcal{N}} (\rho_{1}) \Pi_{\sigma^{\star}}^{\mathcal{S}} &= p^{\vert \sigma^{\star} \vert} \mathcal{W}_{\Lambda} (\rho_{2}). \nonumber
    \end{align}
    These operators are orthogonal since $\Pi_{0}^{\mathcal{S}} \Pi_{\sigma^{\star}}^{\mathcal{S}} = 0$. Also, $\Lambda(\rho_{1}) \not\approx \Lambda^{\mathcal{N}} (\rho_{1})$ as
    \begin{align}
        \norm{\Lambda(\rho_{1}) - \Lambda^{\mathcal{N}} (\rho_{1})}_{1} &= p^{\vert \sigma^{\star} \vert} \norm{\rho_{2} - \mathcal{W}_{\Lambda} (\rho_{2})}_{1} \\
        &= 2p^{\vert \sigma^{\star} \vert} = 1/O(\log(n)^{a}) , \nonumber
    \end{align}
    for some constant $a>0$ which depends only on $p$.
    
    Consider a residual noise channel $\widetilde{\mathcal{N}}$, which acts on $\rho_{2}$ as $\widetilde{\mathcal{N}} (\rho_{2}) \approx \Lambda^{\mathcal{N}} (\rho_{1})$. Suppose $\Lambda$ is fault tolerant against $(\mathbb{L}_{\lambda}, \mathbb{M}_{\zeta})$ with critical measurement noise strength $\zeta^{\star} > 0$. Since $\left\{ \mu , \pr(\mu) \right\}_{\mu} \in \mathbb{M}_{p}$, this means for all $p < \zeta^{\star}$, $\widetilde{\mathcal{N}} \in \mathbb{L}_{\eta(p)}$ where $\eta(p)$ is a continuous function with $\eta(0) = 0$. Therefore $\widetilde{\mathcal{N}}$ admits a decomposition $\widetilde{\mathcal{N}} (\cdot) = \sum_{\alpha} \pr(\alpha) \Phi_{\alpha} (\cdot)$ where for all $A \subseteq \mathcal{Q}$,
    \begin{equation}
        \sum_{\alpha : \supp \Phi_{\alpha} \supseteq A} \pr(\alpha) \leq \eta^{\vert A \vert} .
    \end{equation}
    Since $\vert \Tr K \vert \leq \norm{K}_{1}$ for any operator $K$ and since the trace distance is contractive under application of projectors, Theorem 6.3 of Ref.~\cite{khatri2024},
    \begin{align}
        &\left\vert \Tr\left[\Pi_{\sigma^{\star}}^{\mathcal{S}} \widetilde{\mathcal{N}} (\rho_{2}) \Pi_{\sigma^{\star}}^{\mathcal{S}} \right]  - \Tr\left[\Pi_{\sigma^{\star}}^{\mathcal{S}} \Lambda^{\mathcal{N}} (\rho_{1}) \Pi_{\sigma^{\star}}^{\mathcal{S}} \right] \right\vert \nonumber \\
        &\leq \norm{\Pi_{\sigma^{\star}}^{\mathcal{S}} \widetilde{\mathcal{N}} (\rho_{2}) \Pi_{\sigma^{\star}}^{\mathcal{S}} - \Pi_{\sigma^{\star}}^{\mathcal{S}} \Lambda^{\mathcal{N}} (\rho_{1}) \Pi_{\sigma^{\star}}^{\mathcal{S}}}_{1} \nonumber \\
        &\leq \norm{ \widetilde{\mathcal{N}} (\rho_{2}) - \Lambda^{\mathcal{N}} (\rho_{1}) }_{1} . 
    \end{align}
    Since $\Tr\left[\Pi_{\sigma^{\star}}^{\mathcal{S}} \Lambda^{\mathcal{N}} (\rho_{1}) \Pi_{\sigma^{\star}}^{\mathcal{S}} \right] = p^{\vert \sigma^{\star} \vert}$, this means
    \begin{align}
        &\Tr\left[\Pi_{\sigma^{\star}}^{\mathcal{S}} \widetilde{\mathcal{N}} (\rho_{2}) \Pi_{\sigma^{\star}}^{\mathcal{S}} \right] = \sum_{\alpha} \pr(\alpha) \Tr\left[\Pi_{\sigma^{\star}}^{\mathcal{S}} \Phi_{\alpha} (\rho_{2}) \Pi_{\sigma^{\star}}^{\mathcal{S}}\right] \nonumber \\
        &= p^{\vert \sigma^{\star} \vert} + \varepsilon ,
    \end{align}
    where $\vert \varepsilon \vert \leq \norm{ \widetilde{\mathcal{N}} (\rho_{2}) - \Lambda^{\mathcal{N}} (\rho_{1}) }_{1} = 1/\Omega(\poly(n))$.
    
    If $\Pi_{\sigma^{\star}}^{\mathcal{S}} \Phi_{\alpha} (\rho_{2}) \Pi_{\sigma^{\star}}^{\mathcal{S}} \neq 0$, at least one Kraus operator of $\Phi_{\alpha}$ must be a terminal operator, meaning $\supp \Phi_{\alpha}$ has at least one connected component $A \subseteq \supp \Phi_{\alpha}$ on $\Gamma$ with weight $\vert A \vert = g_{\sigma^{\star}}$. Defining $C(g_{\sigma^{\star}})$ as the collection of all connected subsets of $\Gamma$ of weight $g_{\sigma^{\star}}$, $A \in C(g_{\sigma^{\star}})$. The cardinality of $C(g_{\sigma^{\star}})$ is upper bounded by~\cite{gottesman2014}
    \begin{equation}
        \vert C(g_{\sigma^{\star}}) \vert \leq \vert \Gamma \vert \left( ez \right)^{g_{\sigma^{\star}}-1} ,
    \end{equation}
    where $z \leq (v_{q} - 1) v_{s}$ from property 1 in Section~\ref{non-robust-properties}. We also derive the above bound in Appendix~\ref{combinatorics}. Since $\Tr\left[\Pi_{\sigma^{\star}}^{\mathcal{S}} \Phi_{\alpha} (\rho_{2}) \Pi_{\sigma^{\star}}^{\mathcal{S}}\right] \leq 1$, we may bound $p^{\vert \sigma^{\star} \vert} + \varepsilon$ by
    \begin{align}
        &p^{\vert \sigma^{\star} \vert} + \varepsilon \leq \sum_{\substack{\alpha : \exists A \in C(g_{\sigma^{\star}}) : \\ \supp \Phi_{\alpha} \supseteq A}} \pr(\alpha) \Tr\left[\Pi_{\sigma^{\star}}^{\mathcal{S}} \Phi_{\alpha} (\rho_{2}) \Pi_{\sigma^{\star}}^{\mathcal{S}}\right] \nonumber \\
        &\leq \sum_{A \in C(g_{\sigma^{\star}})} \sum_{\alpha : \supp \Phi_{\alpha} \supseteq A} \pr(\alpha) \leq \sum_{A \in C(g_{\sigma^{\star}})} \eta^{\vert A \vert} 
        \nonumber \\ &\leq \vert C(g_{\sigma^{\star}}) \vert \eta^{g_{\sigma^{\star}}} \leq n \left( ez \right)^{g_{\sigma^{\star}}-1} \eta^{g_{\sigma^{\star}}} .
    \end{align}
    Rearranging, we can instead lower bound $\eta(p)$ by
    \begin{equation}
        \eta(p) \geq \frac{1}{ez} \left[\frac{ez}{n} \left( p^{\vert \sigma^{\star} \vert} + \varepsilon \right)\right]^{1/g_{\sigma^{\star}}} .
    \end{equation}
    For any constant $p>0$, since $g_{\sigma^{\star}} = \Omega(n^\nu)$, $\vert \sigma^{\star} \vert = O(w_{q} w_{s}) = O(\log\log(n))$, $z \leq (v_{q} - 1) v_{s} = O(1)$, and $\vert \varepsilon \vert = 1/\Omega(\poly(n))$, we find    
    \begin{equation}
        \lim_{n \rightarrow \infty} \eta(p) \geq \lim_{n \rightarrow \infty} \left[\frac{ez}{n} \left( p^{\vert \sigma^{\star} \vert} + \varepsilon \right)\right]^{1/g_{\sigma^{\star}}} = \left( ez \right)^{-1} .
    \end{equation}
    Since $\eta(0) = 0$, $\eta(p)$ is not continuous. Thus by contradiction, $\Lambda$ is not fault tolerant against $(\mathbb{L}_{\lambda}, \mathbb{M}_{\zeta})$.
\end{proof}

\subsection{Examples: QEC and state preparation} \label{no-go-pointlike}

Theorem~\ref{theorem:state-prep-no-ft} directly applies to QEC and state preparation for stabilizer codes. If $\Lambda$ has range $O(\log\log(n))$ and consists of measuring stabilizer operators and applying corrections based on measurement outcomes where valid syndromes are always corrected, and if $\rho_{1} = \rho_{2}$, Theorem~\ref{theorem:state-prep-no-ft} states that $\Lambda$ is not fault tolerant. This does not prohibit QEC in shallow depth entirely, but rules out many circuits conventionally used. If $\Lambda$ is a state preparation circuit, $\Lambda (\rho_{1}) = \rho_{2}$, which consists of measuring stabilizer operators and applying corrections where valid syndromes are always corrected, then Theorem~\ref{theorem:state-prep-no-ft} also implies $\Lambda$ is not fault tolerant. This does not explicitly prove inequivalence, $\rho_{1} \nsim \rho_{2}$ or $\Sigma_{1} \nsim \Sigma_{2}$ for $\rho_{1} \in \Sigma_{1}$, but it places constraints on how $\rho_{2}$ can be prepared, as one must either apply a deeper circuit, such as repeating measurements $\Omega (\log(n))$ times, or apply multiple layers of measurement-feedback circuits and LC circuits. Note that the state preparation method used here differs from those presented in Refs.~\cite{bravyi2020,bergamaschi2025}, which use a resource state in $D+1$ dimensions to fault tolerantly prepare a code state in $D$ dimensions and use this extra spatial dimension to do so with a finite-depth circuit.

Examples of stabilizer codes which satisfy the properties of Section~\ref{non-robust-properties} include those whose syndromes can be viewed as pointlike excitations. These include the $\mathbb{Z}_{d}$ toric code in $D=2$, the 2D color code, and 1D GHZ states. For these examples, excitations are explicitly the endpoints of stringlike operators or string-net operators, thus Theorem~\ref{theorem:state-prep-no-ft} directly applies to QEC and to state preparation of $\dyad{\overline{0}}{\overline{0}}^{\otimes k}$ or $\dyad{\overline{+}}{\overline{+}}^{\otimes k}$ from product states $\dyad{0}{0}^{\otimes n}$ or $\dyad{+}{+}^{\otimes n}$ by measuring stabilizers. Additionally, Theorem~\ref{theorem:state-prep-no-ft} applies to correction of $X$ errors and preparation of $\dyad{\overline{+}}{\overline{+}}^{\otimes k}$ for $(D,1)$ $\mathbb{Z}_{d}$ toric codes in $D \geq 3$, or for correction of $Z$ errors and preparation of $\dyad{\overline{0}}{\overline{0}}^{\otimes k}$ for $(D,D-1)$ $\mathbb{Z}_{d}$ toric codes in $D \geq 3$.

\begin{table}[t!] 
\centering
\begin{ruledtabular}
\begin{tabular}{c c c}
Model & QEC not FT & SP not FT \\
\hline
\begin{tabular}{c}
$(D \geq 2,q)$ $\mathbb{Z}_{d}$ toric code, \\
$q=1$ ($q=D-1$)
\end{tabular}
& $X$ ($Z$) errors & $\ket{\overline{+}}^{\otimes k}$ ($\ket{\overline{0}}^{\otimes k}$) \\
2D color code & $Z,X$ errors & $\ket{\overline{0}}^{\otimes k} , \ket{\overline{+}}^{\otimes k}$ \\
X-cube model & $Z,X$ errors & $\ket{\overline{0}}^{\otimes k} , \ket{\overline{+}}^{\otimes k}$ \\
1D repetition code & $X$ errors & $\ket{\mathrm{GHZ}} = \ket{\overline{+}}$ \\
Double semion model & N/A & $\ket{\mathrm{DS}}$ \\
\end{tabular}
\end{ruledtabular}
\caption{
Summary of models for which QEC and state preparation (SP) are not fault tolerant (FT) when implemented by a circuit with range $O(\log\log(n))$ which measures stabilizers or Hamiltonian terms and applies corrections. The GHZ state is the $\ket{\overline{+}}$ state of the repetition code, and the double semion model state $\ket{\mathrm{DS}}$ is the one prepared from $\ket{+}^{\otimes n}$.}
\label{table:state-prep-models}
\end{table}

Another example where Theorem~\ref{theorem:state-prep-no-ft} applies is the X-cube model, described in Ref.~\cite{vijay2016}, which is defined on a 3D cubic lattice and has pointlike excitations with restricted mobility. Explicitly, $X$-type stabilizers detect excitations at the four corners of a rectangular, sheetlike $Z$ operator, and $Z$-type stabilizers detect excitations at each corner and endpoint of a stringlike $X$ operator. Despite restricted mobility, i.e., excitations cannot always be moved independently on the system lattice without creating additional excitations, arbitrarily large sheetlike $Z$ operators only create $4$ excitations, and arbitrarily long linelike $X$ operators only create $2$ excitations. Therefore, Theorem~\ref{theorem:state-prep-no-ft} applies to error correction of all errors and to state preparation of both $\dyad{\overline{0}}{\overline{0}}^{\otimes k}$ and $\dyad{\overline{+}}{\overline{+}}^{\otimes k}$.

Theorem~\ref{theorem:state-prep-no-ft} applies more broadly to state preparation of frustration-free lattice models with pointlike excitations, such as quantum double models and string-net models which host Abelian anyons. Similar to the stabilizer code case, if $\Lambda$ is a circuit with range $O(\log\log(n))$ which consists of measuring Hamiltonian terms of a model and pairing up excitations to annihilate them, Theorem~\ref{theorem:state-prep-no-ft} implies $\Lambda$ is not fault tolerant. One notable non-stabilizer model for which this applies is the double semion model~\cite{levin2005}, which hosts pointlike semionic excitations. Details about this model can be found in Ref.~\cite{levin2005}, however since it has pointlike excitations, state preparation of a ground state $\ket{\mathrm{DS}}$ from the product state $\ket{+}^{\otimes n}$, where one applies a shallow circuit which measures Hamiltonian terms or products of them, is not fault tolerant. We summarize the models discussed in this section in Table~\ref{table:state-prep-models}.

\subsection{Examples: Sequential gauging for non-Abelian topological orders} \label{sequential-gauging}

Another application of Theorem~\ref{theorem:state-prep-no-ft} is to state preparation of non-Abelian topological orders by sequential gauging. For models such as non-Abelian quantum-double models or string-net models, one cannot naively measure Hamiltonian terms and apply corrections, as operators which deterministically pair up and annihilate non-Abelian anyons have range scaling with the distance between anyons and are generally not shallow. However, for quantum-double and string-net models corresponding to solvable anyon theories, one can circumvent this issue by applying a finite number of alternating rounds of measurement and feedback, where in each round one first measures a set of local operators then applies corrections~\cite{tantivasadakarn2024,tantivasadakarn2023,bravyi2022,ren2025}. Doing this, one can measure and pair up Abelian anyons in each round, and the overall state preparation circuit takes the form $\Lambda = \Lambda_{b} \circ \ldots \circ \Lambda_{1}$, where $b$ is finite and each $\Lambda_{i}$ is a finite-depth measurement-feedback circuit. These procedures generally involve sequentially gauging quotients $N_{i}/N_{i+1}$ of subgroups of a solvable group $G$ in a derived series, $N_{i} \triangleleft G$, or a normal series, $N_{i+1} \triangleleft N_{i}$, where $N_{i}/N_{i+1} \cong \mathbb{Z}_{d_{i}}$ is Abelian. Each step of gauging is implemented by the circuit $\Lambda_{i}$ which measures local operators and pairs up Abelian anyons based on measurement outcomes. While useful for noiseless state preparation and for classifying QL phases, these procedures face the same drawbacks in the fault tolerant setting as state preparation of 2D models with Abelian anyons, since they essentially compose together multiple measurement-feedback circuits which pair up Abelian anyons like those discussed in Section~\ref{no-go-pointlike}. Just as before, Theorem~\ref{theorem:state-prep-no-ft} implies no circuit $\Lambda_{i}$ with range $O(\log\log(n))$ which measures these operators or products of them is fault tolerant, thus the overall circuit $\Lambda$ is not fault tolerant either. To be precise, if one considers measurement noise acting on the final layer $\Lambda_{b}$, then Theorem~\ref{theorem:state-prep-no-ft} implies that $\Lambda$ is not fault tolerant.

We do not describe any sequential gauging procedure in detail here and instead refer readers to Refs.~\cite{tantivasadakarn2024,tantivasadakarn2023,bravyi2022,ren2025}. However, the class of models which can be prepared this way in the noiseless case, but which do not admit fault tolerant state preparation with a shallow circuit $\Lambda$ of the form described above, includes quantum-double, twisted quantum double, and string-net models for the groups $S_3$, $S_4$, $A_4$, $Q_8$, $\mathbb{Z}_{k}$, and $D_k$ for all $k$~\cite{bravyi2022,tantivasadakarn2023,ren2025}.

\section{Conclusions} \label{conclusions}

We have proposed a classification of fault tolerant quantum phases based on the ability to robustly transfer encoded information between convex sets of states using shallow quantum-local circuits. Our work makes explicit the connection between quantum fault tolerance and quantum phases of matter which are robust to noise. While it has been known that topological orders can encode logical information and serve as quantum memories, we extend the notion of preserving information encoded within a phase to a set of active operations which can be implemented on quantum devices, notably including circuits with measurements and feedback.

In our classification, stability, i.e., the existence of a nonzero recovery threshold for local noise, is a property of an entire phase. Moreover, stable FTQL phases are coarser than LC phases when one restricts to finite depth. Fault tolerant phases also allow one to classify subsystem codes, and we show that code spaces of stabilizer and gauge color codes defined on $\leq D$-dimensional lattices are all in the same fault tolerant phase when embedded in $D$-dimensional space. Lastly, we studied fault tolerance in state preparation circuits with measurements and feedback and showed that stabilizer codes with connected excitations permit single-shot state preparation, while codes and other topologically ordered systems with pointlike excitations do not.

We focused here on quantum systems embedded in $D$ spatial dimensions, however our framework may be extended to classify phases of systems which live on an arbitrary graph or which cannot be embedded locally in Euclidean space. This includes families of quantum LDPC codes such as quantum expander codes, which have recently been studied as representative phases of matter with nonlocal interactions~\cite{rakovszky2023ldpc,rakovszky2024ldpc,placke2024spinglass}.

An open question is whether information-theoretic diagnostics exist which distinguish FTQL phases, especially phases which encode the same information. While entanglement measures used to study LC phases may not all carry over to the FTQL setting, one entropy measure which may still be useful is the coherent information~\cite{schumacher1996,barnum1998}, which can quantify how much information is recoverable after application of a noisy circuit. Coherent information depends on the specific noisy circuit applied, however it would be useful to investigate whether general statements about different classes of circuits can be made, or whether any other information-theoretic quantities are useful for analyzing FTQL phases.

An interesting example to study through the lens of fault tolerance and FTQL phases is that of intrinsic mixed-state topological orders~\cite{wang2025,sohal2025}, LC phases which host no pure states. For example, some of these phases, such as that of the $ZX$-decohered toric code, can only encode classical information but are long-range entangled and in different LC phases than classical topological orders~\cite{wang2025}. Whether the FTQL phase of an intrinsic mixed-state order remains intrinsically mixed and distinct from that of any classical topological order is not currently known, but would illuminate how classical memories fit into the classification of FTQL phases.

Since fault tolerant phases are useful for studying subsystem codes, one may also be able to classify phases of Floquet codes in a similar manner~\cite{hastings2021}. Floquet codes and other dynamical codes can often be viewed as subsystem codes and rely on sequences of measurements to encode quantum information, therefore one might be able to view a Floquet code as a particular path traversing through a phase and associate the fault tolerance properties of this path with those of the Floquet code.

Lastly, since robust information transfer is the defining principle of our framework, FTQL phases may be relevant for identifying shallow circuits which implement logical operations or for viewing fault tolerant quantum computation as a transformation within a phase.

\begin{acknowledgments}
We thank Sergey Blinov, Katie Chang, Meng Cheng, Zhongling Lu, Shruti Puri, Daniel Qenani, and Nathan Wiebe for helpful discussions. We acknowledge support from the U.S. Department of Energy, Office of Science, National Quantum Information Science Research Centers, Co-design Center for Quantum Advantage (Contract No. DE-SC0012704), IARPA and the U.S. Army Research Office (ELQ Program, Cooperative Agreement No. W911NF-23-2-0219).
We thank the Kavli Institute for Theoretical Physics (supported in part by grant NSF PHY-2309135) for hosting us during the Noise-robust Phases of Quantum Matter Program.
\end{acknowledgments}

\appendix

\section{Noisy 2D toric code state preparation and mixed-state renormalization} \label{rg-derivation}

In this section, we study the LC phase of the mixed state obtained after noisy preparation of a 2D toric code ground state from a product state. The prepared state closely resembles a Gibbs state, and in fact we use the mixed-state renormalization procedure developed in Ref.~\cite{sang2024rg} for the toric code thermal state to show that this state is not in the same LC phase as the toric code.

The 2D qubit toric code can be defined on a 2D square lattice tiling a torus, with qubits on its edges. If the lattice has linear size $L$, it has $n=2L^2$ total qubits. The toric code is stabilized by products of $Z_j$ surrounding each vertex $v \in \Delta_{0}$, $A_{v} = \prod_{j \ni v} Z_j$, and products of $X_j$ surrounding each plaquette $w \in \Delta_{2}$, $B_{w} = \prod_{j \in w} X_j$,
\begin{equation}
    A_{v} = 
    \vcenter{\hbox{
    \begin{tikzpicture}[scale=0.7,thick]
    \draw[color=black] (-1,0) -- (0,0);
    \draw[color=black] (0,0) -- (0,1);
    \draw[color=black] (0,0) -- (0,-1);
    \draw[color=black] (0,0) -- (1,0);
    \node[fill=white,inner sep=1pt] at (-0.5,0) {$Z$};
    \node[fill=white,inner sep=1pt] at (0.5,0) {$Z$};
    \node[fill=white,inner sep=1pt] at (0,0.5) {$Z$};
    \node[fill=white,inner sep=1pt] at (0,-0.5) {$Z$};
    \end{tikzpicture}
    }}
    ,
    \quad
    B_{w} = 
    \vcenter{\hbox{
    \begin{tikzpicture}[scale=0.7,thick]
    \draw[color=black] (-0.5,0.5) -- (0.5,0.5);
    \draw[color=black] (-0.5,-0.5) -- (0.5,-0.5);
    \draw[color=black] (-0.5,0.5) -- (-0.5,-0.5);
    \draw[color=black] (0.5,0.5) -- (0.5,-0.5);
    \node[fill=white,inner sep=1pt] at (-0.5,0) {$X$};
    \node[fill=white,inner sep=1pt] at (0.5,0) {$X$};
    \node[fill=white,inner sep=1pt] at (0,0.5) {$X$};
    \node[fill=white,inner sep=1pt] at (0,-0.5) {$X$};
    \end{tikzpicture}
    }}
    .
\end{equation}
There are $L^2$ distinct $A_{v}$ and $L^2$ distinct $B_{w}$, however only $L^2-1$ of each are independent, as
\begin{equation}
    \label{eq:tc-global-constraint}
    \prod_{v \in \Delta_{0}} A_{v} = \prod_{w \in \Delta_{2}} B_{w} = I .
\end{equation}
$Z$-type logical operators are products of $Z_j$ along noncontractible horizontal $C_1$ and vertical $C_2$ closed loops around the torus on the dual lattice, $\bar{Z}_a = \prod_{j \in C_a} Z_j$. A complete basis for the system Hilbert space is given by eigenstates $\left\{ \ket{\vec{m},\vec{e},\bar{\ell}_1 \bar{\ell}_2} \right\}$ of all $A_{v}$, $B_{w}$, and $\bar{Z}_a$, where 
\begin{align}
    B_{w} \ket{\vec{m},\vec{e},\bar{\ell}_1 \bar{\ell}_2} &= (-1)^{e_w} \ket{\vec{m},\vec{e},\bar{\ell}_1 \bar{\ell}_2} \nonumber \\
    A_{v} \ket{\vec{m},\vec{e},\bar{\ell}_1 \bar{\ell}_2} &= (-1)^{m_v} \ket{\vec{m},\vec{e},\bar{\ell}_1 \bar{\ell}_2} \\
    \bar{Z}_{a} \ket{\vec{m},\vec{e},\bar{\ell}_1 \bar{\ell}_2} &= (-1)^{\bar{\ell}_a} \ket{\vec{m},\vec{e},\bar{\ell}_1 \bar{\ell}_2} . \nonumber
\end{align}
The global constraint Eq.~(\ref{eq:tc-global-constraint}) ensures $\vert \vec{e} \vert$ and $\vert \vec{m} \vert$ are both even. To prepare the encoded state $\ket{\overline{0}\overline{0}} = \ket{\vec{0},\vec{0},\overline{0}\overline{0}}$ from a product state $\ket{0}^{\otimes n}$, one can measure all $B_{w}$ to obtain a configuration of excitations $\vec{e}$, then apply strings $\xi(\vec{e})$ of $Z_j$ on the dual lattice with endpoints at plaquettes $w$ corresponding to $e_{w} = 1$ to pair up and annihilate $e$-type excitations,
\begin{align}
    &\sum_{\vec{e}} R_{\vec{e}} \Pi_{\vec{e}}^{\mathcal{S}_{X}} \dyad{0}{0}^{\otimes n} \Pi_{\vec{e}}^{\mathcal{S}_{X}} R_{\vec{e}}^{\dagger} = \dyad{\overline{0}\overline{0}}{\overline{0}\overline{0}} , \\ &\Pi_{\vec{e}}^{\mathcal{S}_{X}} = \prod_{w \in \Delta_{2}} \left( \frac{I + (-1)^{e_w} B_{w}}{2} \right) , \quad R_{\vec{e}} = \prod_{j \in \xi(\vec{e})} Z_{j} . \nonumber
\end{align}
Here, recovery operators $R_{\vec{e}}$ act as
\begin{equation}
    R_{\vec{s}} \ket{\vec{m},\vec{e},\bar{\ell}_1 \bar{\ell}_2} \propto \ket{\vec{m},\vec{e}+\vec{s},\bar{\ell}_1 \bar{\ell}_2} ,
\end{equation}
where $\vec{e}+\vec{s}$ is modulo $2$.
Also, one finds
\begin{equation}
    \Pi_{\vec{e}}^{\mathcal{S}_{X}} \ket{0}^{\otimes n} = \frac{1}{\sqrt{2^{L^2-1}}} \ket{\vec{0},\vec{e},\bar{0} \bar{0}} .
\end{equation}
Since $\sum_{\vec{e}} \Pi_{\vec{e}}^{\mathcal{S}_{X}} = I$, 
$\ket{0}^{\otimes n}$ can be rewritten as
\begin{equation}
    \ket{0}^{\otimes n} = \frac{1}{\sqrt{2^{L^2-1}}} \sum_{\substack{\vec{e} : \vert \vec{e} \vert = 0 \bmod 2}} \ket{\vec{0},\vec{e},\bar{0} \bar{0}} ,
\end{equation}
so measuring $B_{w}$ will yield $e_{w} = 0$ or $e_{w} = 1$ with equal probability, aside from Eq.~(\ref{eq:tc-global-constraint}) ensuring an even number of plaquettes have $e_{w} = 1$.

Suppose measurement noise modifies the observed measurement outcomes, acting as $\vec{e} \rightarrow \vec{e} + \vec{\mu}$ with $\mu_{w} \in \{0,1\}$. Suppose an even number of measurement errors always occur, but errors occur independently with probability $p$ otherwise. This is qualitatively similar to i.i.d.~noise where one remeasures all $B_{w}$ if an odd number of $e_{w} = 1$ are observed. Then, $\vec{\mu}$ has even parity $\pi(\vec{\mu})=0$ and the normalized probability of error $\vec{\mu}$ occurring is
\begin{align}
    \label{eq:tc-noise-distr}
    &\pr\left(\vec{\mu}\right) = N_{p} \delta\left(\pi(\vec{\mu})=0\right) \prod_{w \in \Delta_{2}} p^{\mu_{w}} (1-p)^{1 - \mu_{w}} \\
    &= \left[ \frac{2}{1+(1-2p)^{L^2}} \right] \delta\left(\pi(\vec{\mu})=0\right) \left[ (1-p)^{L^2 - \vert \vec{\mu} \vert} p^{\vert \vec{\mu} \vert} \right] , \nonumber
\end{align}
where $N_{p} = 2 / \left( 1+(1-2p)^{L^2} \right)$ is a normalization constant. Noisy state preparation then yields
\begin{equation}
    \rho_{p} = \sum_{\vec{\mu}} \pr\left(\vec{\mu}\right) \sum_{\vec{e}} R_{\vec{e} + \vec{\mu}} \Pi_{\vec{e}}^{\mathcal{S}_{X}} \dyad{0}{0}^{\otimes n} \Pi_{\vec{e}}^{\mathcal{S}_{X}} R_{\vec{e} + \vec{\mu}}^{\dagger} .
\end{equation}
$\vec{\mu}$ is always a valid syndrome here, so $R_{\vec{e} + \vec{\mu}}$ can be decomposed as a product of $Z$ operators. Since $R_{\vec{e} + \vec{\mu}}$ and $R_{\vec{\mu}} R_{\vec{e}}$ have the same syndrome, we will write $R_{\vec{e} + \vec{\mu}}$ as
\begin{equation}
    R_{\vec{e} + \vec{\mu}} = R_{\vec{\mu}} \left( \prod_{v : \vec{e},\vec{\mu}} A_{v} \right) \left( \prod_{a : \vec{e},\vec{\mu}} \bar{Z}_a \right) R_{\vec{e}} .
\end{equation}
Since $R_{\vec{e}} \Pi_{\vec{e}}^{\mathcal{S}_{X}} \ket{0}^{\otimes n} \propto \ket{\overline{0}\overline{0}}$ and $A_{v} \ket{\overline{0}\overline{0}} = \bar{Z}_a \ket{\overline{0}\overline{0}} = \ket{\overline{0}\overline{0}}$, the final state becomes
\begin{align}
    \label{eq:tc-noisy-state}
    &\rho_{p} = \sum_{\vec{\mu}} \pr\left(\vec{\mu}\right) R_{\vec{\mu}}\dyad{\vec{0},\vec{0},\overline{0}\overline{0}}{\vec{0},\vec{0},\overline{0}\overline{0}} R_{\vec{\mu}}^{\dagger} \\
    &= N_{p} \sum_{\substack{\vec{\mu}: \\ \vert \vec{\mu} \vert = 0 \bmod 2}} (1-p)^{L^2}\left(\frac{p}{1-p}\right)^{\vert \vec{\mu} \vert} \dyad{\vec{0},\vec{\mu},\overline{0} \overline{0}}{\vec{0},\vec{\mu},\overline{0} \overline{0}} . \nonumber
\end{align}

To classify the LC phase of $\rho_{p}$, we can apply the mixed-state renormalization procedure developed in Ref.~\cite{sang2024rg} for the toric code thermal state to $\rho_{p}$. We will provide a full derivation here for completeness. However, the channels used are exactly the same as those constructed in Ref.~\cite{sang2024rg}. Assuming the thermodynamic limit $L \rightarrow \infty$, we coarse grain the lattice into $2 \times 2$ blocks $\overline{w}$ containing 4 plaquettes $w_{i}$ and 12 qubits,
\begin{equation}
    \overline{w} = 
    \vcenter{\hbox{
    \begin{tikzpicture}[scale=1,thick]
    
    \foreach \x in {0, ..., 2}{
    \draw[color=black] (-0.5,0.5-\x) -- (1.5,0.5-\x);
    \draw[color=black] (-0.5+\x,0.5) -- (-0.5+\x,-1.5);
    }

    \foreach \x in {0,1,2}{
    \foreach \y in {0,1}{
    \filldraw[color=black, fill=black] (0+\y,0.5-\x) circle (0.05cm);
    \filldraw[color=black, fill=black] (-0.5+\x,0-\y) circle (0.05cm);
    }
    }
    
    \node at (0,0) {$w_{1}$};
    \node at (1,0) {$w_{2}$};
    \node at (0,-1) {$w_{3}$};
    \node at (1,-1) {$w_{4}$};
    \end{tikzpicture}
    }}
    .
\end{equation}
Each $\overline{w}$ contains distinct plaquettes $w_{i}$ but may share qubits with other blocks. Consider a local channel gate acting on a single $\overline{w}$ which first measures $B_{w_{i}}$ for each $w_{i}$, then performs the unitary $U_{\vec{e}_{\overline{w}}}$ dependent on the excitation configuration $\vec{e}_{\overline{w}}$ in $\overline{w}$, moving all $e_{i}$ to the top left plaquette $w_{1}$,
\begin{align}
    &\mathcal{C}_{\overline{w}} (\cdot) = \sum_{\vec{e}_{\overline{w}}} U_{\vec{e}_{\overline{w}}} \Pi_{\vec{e}_{\overline{w}}} (\cdot) \Pi_{\vec{e}_{\overline{w}}} U_{\vec{e}_{\overline{w}}}^{\dagger} , \nonumber \\
    &\Pi_{\vec{e}_{\overline{w}}} \ket{\vec{m},\vec{e},\bar{\ell}_1 \bar{\ell}_2} = \delta(\vec{e}_{\overline{w}}\in e)\ket{\vec{m},\vec{e},\bar{\ell}_1 \bar{\ell}_2} , \nonumber \\
    &U_{\vec{e}_{\overline{w}}} \ket{(e_{w_{1}}, e_{w_{2}}, e_{w_{3}}, e_{w_{4}})} = \ket{(\pi(\vec{e}_{\overline{w}}), 0, 0, 0)},
\end{align}
where $\pi(\vec{e}_{\overline{w}}) \in \{0,1\}$ is the parity of $\vec{e}_{\overline{w}}$. Importantly, $\pi(\vec{e}_{\overline{w}})$ is conserved within each $\overline{w}$, thus $\pi(\vec{\mu})$ is conserved overall, and $\mathcal{C}_{\overline{w}}$ is locally reversible when applied to any mixture of toric code eigenstates $\dyad{\vec{m},\vec{e},\bar{\ell}_1 \bar{\ell}_2}{\vec{m},\vec{e},\bar{\ell}_1 \bar{\ell}_2}$. This is because one can construct a reversal channel
\begin{align}
    &\widetilde{\mathcal{C}}_{\overline{w}} (\cdot) = \nonumber \\
    &\sum_{\vec{e}_{\overline{w}}} \pr\left( \vec{e}_{\overline{w}} \vert \pi(\vec{e}_{\overline{w}}) \right) U_{\vec{e}_{\overline{w}}}^{\dagger} \Pi_{\pi(\vec{e}_{\overline{w}})}^{1} (\cdot) \Pi_{\pi(\vec{e}_{\overline{w}})}^{1} U_{\vec{e}_{\overline{w}}} 
\end{align}
where $\Pi_{\pi(\vec{e}_{\overline{w}})}^{1}$ projects onto the $e_{w_{1}} = \pi(\vec{e}_{\overline{w}})$ subspace and $\pr\left( \vec{e}_{\overline{w}} \vert \pi(\vec{e}_{\overline{w}}) \right)$ can be computed from Eq.~(\ref{eq:tc-noise-distr}). 
After applying $\mathcal{C}_{\overline{w}}$ to all $\overline{w}$, one can apply a series of CNOT gates $U_{CX}$, outlined in Ref.~\cite{sang2024rg}, to decouple qubits from the system and obtain a renormalized plaquette
\begin{equation}
    \overline{w} = 
    \vcenter{\hbox{
    \begin{tikzpicture}[scale=1,thick]
    
    \foreach \x in {0, ..., 2}{
    \draw[color=black] (-0.5,0.5-\x) -- (1.5,0.5-\x);
    \draw[color=black] (-0.5+\x,0.5) -- (-0.5+\x,-1.5);
    }

    \foreach \x in {0,1,2}{
    \foreach \y in {0,1}{
    \filldraw[color=black, fill=black] (0+\y,0.5-\x) circle (0.05cm);
    \filldraw[color=black, fill=black] (-0.5+\x,0-\y) circle (0.05cm);
    }
    }
    
    \node at (0,0) {$\pi(\vec{e}_{\overline{w}})$};
    \node at (1,0) {$0$};
    \node at (0,-1) {$0$};
    \node at (1,-1) {$0$};
    \end{tikzpicture}
    }}
    \Rightarrow
    \vcenter{\hbox{
    \begin{tikzpicture}[scale=1,thick]
    
    \foreach \x in {0,2}{
    \draw[color=black] (-0.5,0.5-\x) -- (1.5,0.5-\x);
    \draw[color=black] (-0.5+\x,0.5) -- (-0.5+\x,-1.5);
    }

    \filldraw[color=black, fill=black] (1,0.5) circle (0.05cm);
    \filldraw[color=black, fill=black] (1,-1.5) circle (0.05cm);
    \filldraw[color=black, fill=black] (-0.5,-1) circle (0.05cm);
    \filldraw[color=black, fill=black] (1.5,-1) circle (0.05cm);

    \node at (0.5,-0.5) {$\pi(\vec{e}_{\overline{w}})$};
    \end{tikzpicture}
    }}
    .
\end{equation}
Each renormalized plaquette $\overline{w}$ now has $\pi(\vec{e}_{\overline{w}})$ excitations on it and the system is in a renormalized state with $L^2 / 4$ sites. Crucially, since the probability distribution over excitations Eq.~(\ref{eq:tc-noise-distr}) factorizes into a distribution of independent probabilities $\pr\left(\mu_{w}\right) = p^{\mu_{w}} (1-p)^{1 - \mu_{w}}$ on each plaquette, and since each $\overline{w}$ contains an excitation if and only if $\pi(\vec{e}_{\overline{w}}) = 1$, the state
\begin{equation}
    \label{eq:tc-rg-channel}
    \rho_{p^{\prime}} = U_{CX} \left( \bigotimes_{\overline{w}} \mathcal{C}_{\overline{w}} (\rho_{p}) \right) U_{CX}^{\dagger}
\end{equation}
has the same form as $\rho_{p}$ in Eq.~(\ref{eq:tc-noisy-state}) but with renormalized probability
\begin{align}
    \label{eq:renorm-prob}
    p^{\prime} &= \sum_{\substack{\vec{e}_{\overline{w}} : \pi(\vec{e}_{\overline{w}}) = 1}} \prod_{w \in \overline{w}} p^{e_{w}} (1-p)^{1 - e_{w}} \nonumber \\
    &= 4p(1-p)^3 + 4p^3(1-p) .
\end{align}
For $p=0$, $p^\prime = p = 0$, and for $p=1/2$, $p^\prime = 1/2$, so this renormalization procedure has fixed points at $p=0,1/2$. However, $p^{\prime} \geq p$ for $p\in (0 ,1/2]$, thus any state in this range flows to $p=1/2$ after sufficiently many renormalization steps. As mentioned, since $U_{CX}$ is a unitary operator and $\mathcal{C}_{\overline{w}}$ are locally reversible, one can reverse each renormalization step with
\begin{equation}
    \rho_{p} = \bigotimes_{\overline{w}} \widetilde{\mathcal{C}}_{\overline{w}} \left( U_{CX}^{\dagger}\rho_{p^{\prime}}U_{CX} \right) .
\end{equation}
Therefore, by successively applying the locally reversible circuit Eq.~(\ref{eq:tc-rg-channel}) to each renormalized system, one obtains a locally reversible LC circuit $\mathcal{C}$ which maps
\begin{equation}
    \mathcal{C}\left(\rho_{p}\right) \approx \lim_{q \rightarrow 1/2^{-}} \rho_{q} \quad \Leftrightarrow \quad \widetilde{\mathcal{C}}\left( \lim_{q \rightarrow 1/2^{-}} \rho_{q} \right) \approx \rho_{p}.
\end{equation}
Furthermore, one only needs to apply $O(\log \log (L^2 / \varepsilon ))$ renormalization steps to obtain a renormalized state $\rho_{q}$ which has trace distance
\begin{equation}
    \label{eq:tc-renorm-dist}
    \frac{1}{2} \norm{\rho_{q} - \rho_{1/2}}_{1} < \sqrt{2\varepsilon}
\end{equation}
with $\rho_{1/2}$, which we will show shortly. To obtain $\varepsilon = 1 / \Omega(\poly(L))$, one can apply an $(r,\ell)$ locally reversible LC circuit with $r\ell = O(\polylog(L))$, thus $\rho_{q}$ and $\rho_{1/2}$ are in the same LC phase. However, the fixed point state
\begin{equation}
    \label{eq:tc-fixed-point}
    \rho_{1/2} = \sum_{\substack{\vec{\mu}:  \vert \vec{\mu} \vert = 0 \bmod 2}} \frac{1}{2^{L^2-1}} \dyad{\vec{0},\vec{\mu},\bar{0} \bar{0}}{\vec{0},\vec{\mu},\bar{0} \bar{0}} 
\end{equation}
is not in the same LC phase as $\dyad{\overline{0}\overline{0}}{\overline{0}\overline{0}}$. This can be seen as $\rho_{1/2} = \mathcal{W} (\dyad{0}{0}^{\otimes n})$, where $\mathcal{W}$ is the LC circuit
\begin{align}
    &\mathcal{W} (\cdot) = \bigotimes_{w \in \Delta_{2}} \mathcal{W}_{w} (\cdot) , \; \mathcal{W}_{w} (\rho) =  \frac{1}{2}\left( \rho + B_{w} \rho B_{w} \right) ,
\end{align}
and one can obtain $\dyad{0}{0}^{\otimes n}$ from $\rho_{1/2}$ by applying the channel $\Tr_{i} (\cdot) \dyad{0}{0}_{i}$ to each individual qubit $i$. $\mathcal{W}$ is not necessarily locally reversible, however it still helps show that $\rho_{1/2}$ is not in the same LC phase as $\dyad{\overline{0}\overline{0}}{\overline{0}\overline{0}}$. This is because no shallow LC circuit can prepare any toric code state, such as $\dyad{\overline{0}\overline{0}}{\overline{0}\overline{0}}$, from a product state~\cite{konig2014}. Since one can convert between $\dyad{0}{0}^{\otimes n}$ and $\rho_{1/2}$ using the two finite-depth LC circuits above, it is impossible to prepare $\dyad{\overline{0}\overline{0}}{\overline{0}\overline{0}}$ from $\rho_{1/2}$ using any shallow LC circuit, including any shallow locally reversible LC circuit. Therefore, $\dyad{\overline{0}\overline{0}}{\overline{0}\overline{0}}$ and $\rho_{1/2}$ are in different LC phases.

Lastly, it remains to verify Eq.~(\ref{eq:tc-renorm-dist}) holds for $\varepsilon = 1 / \Omega(\poly(L))$ and a renormalization circuit with $r\ell = O(\polylog(L))$. First, we will compute the fidelity,
\begin{equation}
    F(\rho_{p} , \rho_{1/2}) = \Tr \sqrt{\sqrt{\rho_{1/2}} \rho_{p} \sqrt{\rho_{1/2}}} .
\end{equation}
Since $\rho_{p}$ and $\rho_{1/2}$ are diagonal in the $\ket{\vec{m},\vec{e},\bar{\ell}_1 \bar{\ell}_2}$ basis,
\begin{equation}
    \sqrt{\rho_{1/2}} = \sum_{\substack{\vec{\mu} : \vert \vec{\mu} \vert = 0 \bmod 2}} \frac{1}{\sqrt{2^{L^2-1}}} \dyad{\vec{0},\vec{\mu},\bar{0} \bar{0}}{\vec{0},\vec{\mu},\bar{0} \bar{0}} ,
\end{equation}
and one obtains, after simplifying and taking the trace,
\begin{align}
    &F(\rho_{p} , \rho_{1/2}) = \sqrt{\frac{N_{p} (1-p)^{L^2}}{2^{L^2-1}}} \sum_{\substack{\vec{\mu}: \\ \vert \vec{\mu} \vert = 0 \bmod 2}} \left(\frac{p}{1-p}\right)^{\vert \vec{\mu} \vert/2}  \nonumber \\
    &= 2 \frac{2^{-L^2 / 2} \left(1-p\right)^{L^2/2}}{\sqrt{1+(1-2p)^{L^2}}} \sum_{\substack{\vec{\mu}: \\ \vert \vec{\mu} \vert = 0 \bmod 2}} \left( \frac{p}{1-p} \right)^{\vert \vec{\mu} \vert/2} .
\end{align}
Define a parameter $\beta$ similar to inverse temperature via
\begin{equation}
    e^{-2 \beta} = \frac{p}{1-p},
\end{equation}
where $p \in [0,1/2]$ and $\beta \in [0,\infty]$. This gives
\begin{align}
    &\sum_{\substack{\vec{\mu}: \vert \vec{\mu} \vert = 0 \bmod 2}} \left( \frac{p}{1-p} \right)^{\vert \vec{\mu} \vert/2} = \sum_{\substack{m:  m = 0 \bmod 2}} \binom{L^2}{m} e^{-\beta m} \nonumber \\
    &= \frac{1}{2} \left( (1+e^{-\beta})^{L^2} + (1-e^{-\beta})^{L^2} \right), 
\end{align}
and yields the fidelity in terms of $\beta$,
\begin{align}
    F(\rho_{p} , \rho_{1/2}) =& \frac{2^{-L^2 / 2}}{\sqrt{1+(\tanh \beta)^{L^2}}} \left( \frac{e^{\beta}}{e^{\beta} + e^{-\beta}} \right)^{L^2 / 2} \nonumber \\
    &\times \left( (1+e^{-\beta})^{L^2} + (1-e^{-\beta})^{L^2} \right) .
\end{align}
One finds for any $p \in [0,1/2]$,
\begin{equation}
    p^{L^2} \left(1+(1-2p)^{L^2}\right) \leq 2^{-L^2} ,
\end{equation}
which can be rewritten in terms of $\beta$ to yield the bound
\begin{equation}
    \frac{2^{-L^2 / 2}}{\sqrt{1+(\tanh \beta)^{L^2}}} \geq \left( \frac{e^{-\beta}}{e^{\beta} + e^{-\beta}} \right)^{L^2 / 2} .
\end{equation}
Noting $(1-e^{-\beta})^{L^2} \geq 0$, we can lower bound the fidelity,
\begin{align}
    &F(\rho_{p} , \rho_{1/2}) \geq \nonumber \\
    &\left( \frac{e^{-\beta}}{e^{\beta} + e^{-\beta}} \right)^{L^2 / 2} \left( \frac{e^{\beta}}{e^{\beta} + e^{-\beta}} \right)^{L^2 / 2} \left( (1+e^{-\beta})^{L^2} \right) \nonumber \\
    &= \left( \frac{e^{\beta}+1}{e^{2\beta} + 1} \right)^{L^2} \geq \left( \frac{e^{\beta}}{e^{2\beta}} \right)^{L^2} = e^{-\beta L^2} ,
\end{align}
where the last inequality holds since $e^{2\beta} \geq e^{\beta}$. Moreover, the renormalized probability, Eq.~(\ref{eq:renorm-prob}), can be rewritten in terms of $\beta$ as $\tanh \beta^\prime = \tanh^4 \beta$, or rather $\beta^\prime = \tanh^{-1}\left( \tanh^4 \beta\right)$. As $p \rightarrow 1/2$, $\beta \rightarrow 0$. Also,
\begin{equation}
    \label{eq:rg-beta-bound}
    \beta^\prime = \tanh^{-1}\left( \tanh^4 \beta\right) \leq \beta^{4} 
\end{equation}
for all $\beta \geq 0$. Denoting $\beta^{(k)}$ as $\beta$ after $k$ renormalization steps and the renormalized system size $L^{(k)} = L/2^{k}$, one eventually obtains $\beta^{(k_{0})} < 1$ for some constant $k_{0}$. Also, from Eq.~(\ref{eq:rg-beta-bound}), $\beta^{(k+h)} \leq \left(\beta^{(k)}\right)^{4^{h}}$ is always true. Suppose at step $k > k_{0}$, one has
\begin{equation}
    \beta^{(k)} \leq \left(\beta^{(k_{0})}\right)^{4^{k-k_{0}}} < \varepsilon / L^2
\end{equation}
for some $\varepsilon$. Rearranging, we see that
\begin{equation}
    k > k_{0} + \log_{4} \log_{{\beta^{(k_{0})}}^{-1}} \left(L^2/\varepsilon \right)
\end{equation}
and the fidelity is lower bounded by
\begin{equation}
    e^{-\beta^{(k)} \left(L^{(k)}\right)^2} \geq e^{-(\varepsilon / L^2) (L^2 / 4^{k})} \geq \left( 1 - \varepsilon \right)^{1/4^{k}}.
\end{equation}
Since $\left( 1 - \varepsilon \right)^{1/4^{k}} \geq 1 - \varepsilon$ for $k \geq 1$, we find that after $k = O(\log \log (L^2 / \varepsilon ))$ renormalization steps, $F(\rho_{q} , \rho_{1/2}) \geq 1 - \varepsilon$. We can then bound the trace distance as~\cite{fuchs1999}
\begin{equation}
    \frac{1}{2} \norm{\rho_{q} - \rho_{1/2}}_{1} \leq \sqrt{1 - F(\rho_{q} , \rho_{1/2})^2} \leq \sqrt{2\varepsilon} .
\end{equation}
Each renormalization step requires operations with range at most $O\left(2^{\log \log (L^2 / \varepsilon )}\right) = O(\log (L^2 / \varepsilon ))$, therefore to achieve $\varepsilon = 1 / \Omega(\poly(L))$, the channel needs at most $r\ell = \polylog(L)$. This verifies Eq.~(\ref{eq:tc-renorm-dist}), therefore $\rho_{p}$ and $\rho_{1/2}$ are in the same LC phase for $0 < p \leq 1/2$.

To close this section, we point out that the fixed point state, $\rho_{1/2}$, is locally indistinguishable from the classical loop state discussed in Ref.~\cite{sang2025reversibility}. As shown in Ref.~\cite{sang2025reversibility}, this state is neither in the same phase as the 2D toric code nor in the trivial phase, rather it exhibits long-range classical correlations but is not long-range entangled.

\section{Short derivations and proofs} \label{short-proofs}

First, we present derivations of Eq.~(\ref{eq:decoding-equivalence-alternate}) and Eq.~(\ref{eq:ft-implications}) in Sections~\ref{encoded-info} and~\ref{ft-phase-defs}.

\begin{proof}[Derivation of Eq.~(\ref{eq:decoding-equivalence-alternate})]
    Using the triangle inequality,
    \begin{align}
        &\norm{\mathcal{D}_{1} \circ \Lambda_{2} \circ \Lambda_{1} (\rho_{1}) - \mathcal{D}_{1} (\rho_{1})}_{1} \\
        &\leq \norm{\mathcal{D}_{1} \circ \Lambda_{2} \circ \Lambda_{1} (\rho_{1}) - \mathcal{D}_{1} \circ \Lambda_{2} (\rho_{2})}_{1} \nonumber \\
        &+ \norm{\mathcal{D}_{1} \circ \Lambda_{2} (\rho_{2}) - \mathcal{D}_{2} (\rho_{2})}_{1} + \norm{\mathcal{D}_{2} (\rho_{2}) - \mathcal{D}_{2} \circ \Lambda_{1} (\rho_{1})}_{1} \nonumber \\
        &+ \norm{\mathcal{D}_{2} \circ \Lambda_{1} (\rho_{1}) - \mathcal{D}_{1} (\rho_{1})}_{1} \nonumber 
    \end{align}
    for some $\rho_{2}$. Contractivity of the trace distance yields
    \begin{align}
        &\norm{\mathcal{D}_{1} \circ \Lambda_{2} \circ \Lambda_{1} (\rho_{1}) - \mathcal{D}_{1} (\rho_{1})}_{1} \leq 2 \norm{\Lambda_{1} (\rho_{1}) - \rho_{2}}_{1}  \\
        &+ \norm{\mathcal{D}_{1} \circ \Lambda_{2} (\rho_{2}) - \mathcal{D}_{2} (\rho_{2})}_{1} + \norm{\mathcal{D}_{2} \circ \Lambda_{1} (\rho_{1}) - \mathcal{D}_{1} (\rho_{1})}_{1} . \nonumber
    \end{align}
    Choosing $\rho_{2} \in \Sigma_{2}$ such that $\Lambda_{1} (\rho_{1}) \approx \rho_{2}$, all terms are $= 1/\Omega(\poly(n))$, therefore
    \begin{equation}
        \norm{\mathcal{D}_{1} \circ \Lambda_{2} \circ \Lambda_{1} (\rho_{1}) - \mathcal{D}_{1} (\rho_{1})}_{1} = 1/\Omega(\poly(n)) .
    \end{equation}
    A similar argument yields
    \begin{equation}
        \norm{\mathcal{D}_{2} \circ \Lambda_{1} \circ \Lambda_{2} (\rho_{2}) - \mathcal{D}_{2} (\rho_{2})}_{1} = 1/\Omega(\poly(n)) .
    \end{equation}
\end{proof}

\begin{proof}[Derivation of Eq.~(\ref{eq:ft-implications})]
    Using contractivity of the trace distance,
    \begin{equation}
        \norm{\widetilde{\mathcal{N}} \circ \Lambda (\rho_1) - \widetilde{\mathcal{N}}(\rho_{2})}_{1} \leq \norm{\Lambda (\rho_{1}) - \rho_{2}}_{1} ,
    \end{equation}
    therefore $\widetilde{\mathcal{N}} \circ \Lambda (\rho_1) \approx \widetilde{\mathcal{N}}(\rho_{2})$ if $\rho_{2} \approx \Lambda (\rho_1)$.
    Using the triangle inequality and contractivity of the trace distance,
    \begin{align}
        &\norm{\mathcal{D}_{2}(\rho_{2}) - \mathcal{D}_{2} \circ \mathcal{R}_{2} \circ \Lambda^{\mathcal{N}} (\rho_1)}_{1} \\
        &\leq \norm{\mathcal{D}_{2}(\rho_{2}) - \mathcal{D}_{2} \circ \mathcal{R}_{2} \circ \widetilde{\mathcal{N}}(\rho_{2})}_{1} \nonumber \\
        &+ \norm{\mathcal{D}_{2} \circ \mathcal{R}_{2} \circ \widetilde{\mathcal{N}}(\rho_{2}) - \mathcal{D}_{2} \circ \mathcal{R}_{2} \circ \Lambda^{\mathcal{N}} (\rho_1)}_{1} \nonumber \\
        &\leq \norm{\mathcal{D}_{2}(\rho_{2}) - \mathcal{D}_{2} \circ \mathcal{R}_{2} \circ \widetilde{\mathcal{N}}(\rho_{2})}_{1} + \norm{\widetilde{\mathcal{N}}(\rho_{2}) - \Lambda^{\mathcal{N}} (\rho_1)}_{1} . \nonumber
    \end{align}
    Since $\widetilde{\mathcal{N}} \circ \Lambda (\rho_1) \approx \widetilde{\mathcal{N}}(\rho_{2})$ and $\Lambda^{\mathcal{N}} (\rho_1) \approx \widetilde{\mathcal{N}} \circ \Lambda (\rho_1)$ if $\rho_{2} \approx \Lambda (\rho_1)$, we obtain $\mathcal{D}_{2}(\rho_{2}) \approx \mathcal{D}_{2} \circ \mathcal{R}_{2} \circ \Lambda^{\mathcal{N}} (\rho_1)$.
\end{proof}

Next, we prove Lemma~\ref{lemma:ft-composition} in Section~\ref{ft-phase-defs}.

\begin{proof}[Proof of Lemma~\ref{lemma:ft-composition}]
    Consider a noisy channel $(\Lambda_{2} \circ \Lambda_{1})^{\mathcal{N}}$, where $\mathcal{N} \in (\mathbb{N}_{\tau}, \mathbb{M}_{\zeta})$. Suppose $\zeta < \zeta^{\star}$ and $\tau \leq \tau + \eta_{1} < \tau^{\star}$ for $\eta_{1} = \eta_{1}(\tau,\zeta)$ above. $\mathcal{N}$ can be split into $\mathcal{N}_1 \in (\mathbb{N}_{\tau}, \mathbb{M}_{\zeta})$ acting on $\Lambda_{1}$ and $\mathcal{N}_2 \in (\mathbb{N}_{\tau}, \mathbb{M}_{\zeta})$ acting on $\Lambda_{2}$, $(\Lambda_{2} \circ \Lambda_{1})^{\mathcal{N}} = \Lambda_{2}^{\mathcal{N}_2} \circ \Lambda_{1}^{\mathcal{N}_1}$. Since $\Lambda_{1}$ is fault tolerant, property 1 of Definition~\ref{def:ft-channel} implies
    \begin{align}
        &(\Lambda_{2} \circ \Lambda_{1})^{\mathcal{N}} (\rho_1) = \Lambda_{2}^{\mathcal{N}_2} \circ \Lambda_{1}^{\mathcal{N}_1} (\rho_1) \\
        &\approx \Lambda_{2}^{\mathcal{N}_2} \circ \widetilde{\mathcal{N}}_1 (\rho_2)
        = \Lambda_{2}^{\mathcal{N}_2^{\prime}} (\rho_2) \quad \forall \rho_1 \in \Sigma_{1}, \nonumber
    \end{align}
    where $\widetilde{\mathcal{N}}_1 \in \mathbb{N}_{\eta_{1}}$, $\Lambda_{2}^{\mathcal{N}_2^{\prime}} = \Lambda_{2}^{\mathcal{N}_2} \circ \widetilde{\mathcal{N}}_1$ has $\widetilde{\mathcal{N}}_1$ absorbed into its first layer of qudit noise, and $\rho_2 \approx \Lambda_{1}(\rho_{1})$, $\rho_{2} \in \Sigma_{2}$. Since each layer of qudit noise in $\Lambda_{2}^{\mathcal{N}_2}$ is in $\mathbb{N}_{\tau}$ and $\widetilde{\mathcal{N}}_1 \in \mathbb{N}_{\eta_{1}}$, the first qudit noise layer of $\Lambda_{2}^{\mathcal{N}_2^{\prime}}$ must be in $\mathbb{N}_{\tau + \eta_{1}}$. Therefore, $\mathcal{N}_2^{\prime} \in (\mathbb{N}_{\tau + \eta_{1}}, \mathbb{M}_{\zeta})$ for $\Lambda_{2}^{\mathcal{N}_2^{\prime}}$. Since $\Lambda_{2}$ is fault tolerant and $\tau + \eta_{1} < \tau^{\star}$, Definition~\ref{def:ft-channel} similarly implies
    \begin{equation}
        \Lambda_{2}^{\mathcal{N}_2^{\prime}} (\rho_2) \approx \widetilde{\mathcal{N}}_2^{\prime} (\rho_3) \quad \forall \rho_2 \in \Sigma_{2}
    \end{equation}
    where $\widetilde{\mathcal{N}}_2^{\prime} \in \mathbb{N}_{\eta_{2}}$ for $\eta_{2} = \eta_{2}(\tau + \eta_{1},\zeta)$ and $\rho_3 \approx \Lambda_{2}(\rho_{2})$, $\rho_3 \in \Sigma_{3}$. This means
    \begin{equation}
        (\Lambda_{2} \circ \Lambda_{1})^{\mathcal{N}} (\rho_1) \approx  \widetilde{\mathcal{N}}_2^{\prime} \circ \Lambda_{2} \circ \Lambda_{1} (\rho_1) \quad \forall \rho_1 \in \Sigma_{1} .
    \end{equation}
    Also, property 2 of Definition~\ref{def:ft-channel} implies there exists a recovery channel $\mathcal{R}_{3}$ for $\Sigma_{3}$ such that $\widetilde{\mathcal{N}}_2^{\prime}$ is $\mathcal{R}_{3}$-recoverable. Therefore, $\Lambda_{2} \circ \Lambda_{1}$ is fault tolerant.
\end{proof}

Next, we present a result used repeatedly in later appendices. This result has been previously used to bound Pauli error probability distributions in Refs.~\cite{bombin2015singleshot,kubica2022}, but we provide a full derivation for completeness.

\begin{proposition}
    \label{proposition:bounded-composition}
    Consider two probability distributions $\pr_{1}(\alpha)$ and $\pr_{2}(\beta)$ and functions $g$ and $h$ which map $\alpha$ and $\beta$ to some set $V$. Suppose probability distributions over $g(\alpha) \subseteq V$ and $h(\beta) \subseteq V$ are $\tau_{1}$ and $\tau_{2}$-bounded, namely for all $A \subseteq V$,
    \begin{equation}
        \sum_{\alpha : g(\alpha) \supseteq A} \pr_{1}(\alpha) \leq \tau_{1}^{\vert A \vert} , \quad \sum_{\beta : h(\beta) \supseteq A} \pr_{2}(\beta) \leq \tau_{2}^{\vert A \vert} .
    \end{equation}
    Then, the composite distribution over $g(\alpha) \cup h(\beta) \subseteq V$ is $(\tau_{1}+\tau_{2})$-bounded, for all $A \subseteq V$,
    \begin{equation}
        \sum_{\alpha , \beta : g(\alpha) \cup h(\beta) \supseteq A} \pr_{1}(\alpha) \pr_{2}(\beta) \leq \left(\tau_1 + \tau_2\right)^{\vert A \vert} .
    \end{equation}
\end{proposition}
\begin{proof}
    \begin{align}
        &\sum_{\alpha , \beta : g(\alpha) \cup h(\beta) \supseteq A} \pr_{1}(\alpha) \pr_{2}(\beta)
        \nonumber \\
        &= \sum_{E : E \subseteq A} \sum_{\substack{\alpha : g(\alpha) \supseteq E}} \sum_{\substack{\beta : h(\beta) \supseteq (A \setminus E) : \\ h(\beta) \cap E = \emptyset}} \pr_{1}(\alpha) \pr_{2}(\beta)
        \nonumber \\
        &\leq \sum_{E : E \subseteq A}
        \left(\sum_{\alpha : g(\alpha) \supseteq E} \pr_{1}(\alpha)\right) \left(\sum_{\beta : h(\beta) \supseteq (A \setminus E)} \pr_{2}(\beta)\right) \nonumber \\
        &\leq \sum_{E : E \subseteq A} \tau_{1}^{\vert E \vert} \tau_{2}^{\vert A \setminus E \vert}
        = \sum_{\vert E \vert = 0}^{\vert A \vert} \binom{\vert A \vert}{\vert E \vert} \tau_{1}^{\vert E \vert} \tau_{2}^{\vert A \vert - \vert E \vert} \nonumber \\ &
        = \left(\tau_1 + \tau_2\right)^{\vert A \vert} ,
    \end{align}
    where $A$ decomposes into $E \subseteq g(\alpha) \setminus \left( g(\alpha) \cap h(\beta) \right)$ and $(A \setminus E) \subseteq h(\beta)$ in the second line.
\end{proof}

We can now demonstrate that local stochastic noise $\mathbb{L}_{\lambda}$, described in Section~\ref{noise-char}, satisfies Definition~\ref{def:noise-class}. This result is given by the following proposition.

\begin{proposition}
    \label{proposition:stochastic-consistent}
    The class of local stochastic noise, $\mathbb{L}_{\lambda}$, satisfies Definition~\ref{def:noise-class}.
\end{proposition}
\begin{proof}
    Property 1 of Definition~\ref{def:noise-class} follows trivially from Definition~\ref{def:tau-bound} and Eq.~(\ref{eq:local-stochastic-pauli}), as $\kappa^{\vert A \vert} \leq \tau^{\vert A \vert}$ for all $0\leq \kappa \leq \tau$ therefore $\mathbb{L}_{\kappa} \subseteq \mathbb{L}_{\tau}$ for all $\kappa \leq \tau$. Property 2 is also satisfied as $\mathbb{L}_{\lambda} \subseteq \mathbb{L}_{\lambda}$ is a subset of itself. To demonstrate property 3, suppose $\mathcal{N} \in \mathbb{L}_{\tau_1}$ and $\mathcal{M} \in \mathbb{L}_{\tau_2}$, where
    \begin{align}
        &\mathcal{N} (\cdot) = \sum_{\alpha} \pr_{\mathcal{N}}(\alpha) \Phi_{\alpha} (\cdot) , \; \mathcal{M} (\cdot) = \sum_{\beta} \pr_{\mathcal{M}}(\beta) \Phi_{\beta} (\cdot) , \nonumber \\
        &\mathcal{N} \circ \mathcal{M}= \sum_{\nu} \pr_{\mathcal{N} \circ \mathcal{M}}(\nu) \Phi_{\nu} (\cdot) , \\
        &\pr_{\mathcal{N} \circ \mathcal{M}}(\nu) = \sum_{\alpha,\beta : \Phi_{\alpha} \circ \Phi_{\beta} = \Phi_{\nu}} \pr_{\mathcal{N}}(\alpha) \pr_{\mathcal{M}}(\beta) . \nonumber
    \end{align}
    Then, we can show $\pr_{\mathcal{N} \circ \mathcal{M}}(Q)$ is $(\tau_1 + \tau_2)$-bounded,
    \begin{align}
        &\sum_{\nu: \supp \Phi_{\nu} \supseteq A} \pr_{\mathcal{N} \circ \mathcal{M}}(\nu) \nonumber \\
        &= \sum_{\substack{\alpha,\beta : \supp \left( \Phi_{\alpha} \circ \Phi_{\beta} \right) \supseteq A}} \pr_{\mathcal{N}}(\alpha) \pr_{\mathcal{M}}(\beta)  \\
        &\leq \sum_{\substack{\alpha,\beta : \supp \Phi_{\alpha} \cup \supp \Phi_{\beta} \supseteq A}} \pr_{\mathcal{N}}(\alpha) \pr_{\mathcal{M}}(\beta) \leq \left(\tau_1 + \tau_2\right)^{\vert A \vert} . \nonumber
    \end{align}
    where the last inequality follows from Proposition~\ref{proposition:bounded-composition}, setting $g(\alpha) = \supp \Phi_{\alpha}$, and $h(\beta) = \supp \Phi_{\beta}$.
    Therefore, $\mathcal{N} \circ \mathcal{M} \in \mathbb{L}_{\tau_1 + \tau_2}$, and $\mathbb{L}_{\tau_1} \circ \mathbb{L}_{\tau_2} \subseteq \mathbb{L}_{\tau_1 + \tau_2}$.
\end{proof}

Lastly, we prove the following bound on all possible decompositions of Pauli channels. This result is useful for analyzing whether any decomposition of a channel is $\tau$-bounded. Specifically, if the channel $\mathcal{N}$ admits a Pauli representation and this representation is not $\tau$-bounded, then no decomposition of $\mathcal{N}$ is $\tau$-bounded.

\begin{lemma}
    \label{lemma:unitary-channel-bound}
    Consider a channel which admits a Pauli representation, $\mathcal{N} = \{\sqrt{\pr(P)} P\}$ with $P \in \mathcal{P}$. Then, any decomposition $\mathcal{N} (\cdot) = \sum_{\alpha} \pr_{\mathcal{N}}(\alpha) \Phi_{\alpha} (\cdot)$ satisfies for all $A \subseteq \mathcal{Q}$,
    \begin{equation}
        \label{eq:unitary-bound}
        \sum_{P : \supp P \supseteq A} \pr(P) \leq
        \sum_{\alpha : \supp \Phi_{\alpha} \supseteq A} \pr_{\mathcal{N}}(\alpha) .
    \end{equation}
\end{lemma}
\begin{proof}
    First, note that one may choose the Pauli representation $\mathcal{N} = \{\sqrt{\pr(P)} P\}$ such that each $P$ is linearly independent, as $P \rho P^{\dagger} = Q \rho Q^{\dagger}$ if $Q = e^{i\theta} P$. Doing so has no effect on the left hand side of Eq.~(\ref{eq:unitary-bound}), and we will use this representation for the rest of the proof. Since these Pauli operators form a basis for the space of linear operators on $\mathcal{H}$, for any two basis operators $P_{1} \neq P_{2}$,
    \begin{equation}
        \supp \left(\alpha P_{1} + \beta P_{2} \right) = \supp P_{1} \cup \supp P_{2} .
    \end{equation}
    Rewriting the channel $\mathcal{N} (\cdot) = \sum_{K} K (\cdot) K^{\dagger}$,
    \begin{align}
        &\Phi_{\alpha} (\cdot) = \frac{1}{\pr_{\mathcal{N}}(\alpha)} \sum_{K : \alpha} K (\cdot) K^{\dagger} , \nonumber \\
        &\sum_{K : \alpha} K^{\dagger} K = \pr_{\mathcal{N}}(\alpha) I .
    \end{align}
    Since two sets of Kraus operators representing $\mathcal{N}$ can be related by a unitary transformation, we may write $K$ as
    \begin{equation}
        K = \sum_{P \in \mathcal{P}} u_{K,P} \sqrt{\pr(P)} P
    \end{equation}
    where $u_{K,P}$ are elements of a unitary matrix. This means
    \begin{equation}
        \sum_{P \in \mathcal{P}} \vert u_{K,P} \vert^2 = \sum_{K} \vert u_{K,P} \vert^2 = 1 , 
    \end{equation}
    and one may rewrite
    \begin{align}
        &\sum_{K : \alpha} K^{\dagger} K = \sum_{K : \alpha} \sum_{P,Q \in \mathcal{P}} u_{K,P} u_{K,Q}^{*} \sqrt{\pr(P) \pr(Q)} Q^{\dagger} P  \nonumber \\
        &= \pr_{\mathcal{N}}(\alpha) I = \sum_{K : \alpha} \sum_{P \in \mathcal{P}} \vert u_{K,P} \vert^2 \pr(P) I ,
    \end{align}
    where we have used $Q^{\dagger}P = I$ only if $P=Q$ and linear independence of $P$ and $Q$. Also,
    \begin{align}
        &\vert u_{K,P} \vert^2 \pr(P) \delta(A \subseteq \supp P) \nonumber \\
        &\leq \vert u_{K,P} \vert^2 \pr(P) \delta(A \subseteq \supp K)
    \end{align}
    since $\supp P \subseteq \supp K$ if $u_{K,P} \neq 0$ and $\pr(P) \neq 0$. Also, $\delta(A \subseteq \supp K) \leq \delta(A \subseteq \supp \Phi_{\alpha})$ if $K / \sqrt{\pr_{\mathcal{N}}(\alpha)}$ is a Kraus operator of $\Phi_{\alpha}$. Now, write the sum
    \begin{align}
        &\sum_{P \in \mathcal{P} : \supp P \supseteq A} \pr(P) = \sum_{P \in \mathcal{P}} \pr(P) \delta(A \subseteq \supp P) \\
        &= \sum_{P \in \mathcal{P}} \pr(P) \delta(A \subseteq \supp P) \sum_{K} \vert u_{K,P} \vert^2 \nonumber \\
        &= \sum_{\alpha} \sum_{K : \alpha} \sum_{P \in \mathcal{P}} \vert u_{K,P} \vert^2 \pr(P) \delta(A \subseteq \supp P) \nonumber \\
        &\leq \sum_{\alpha} \sum_{K : \alpha} \sum_{P \in \mathcal{P}} \vert u_{K,P} \vert^2 \pr(P) \delta(A \subseteq \supp \Phi_{\alpha}) \nonumber \\
        &= \sum_{\alpha} \pr_{\mathcal{N}}(\alpha) \delta(A \subseteq \supp \Phi_{\alpha}) = \sum_{\alpha : \supp \Phi_{\alpha} \supseteq A} \pr_{\mathcal{N}}(\alpha) , \nonumber 
    \end{align}
    which proves Eq.~(\ref{eq:unitary-bound}) always holds.
\end{proof}

\section{Graph combinatorics} \label{combinatorics}

This section mainly contains technical results which will be used in later appendices. Many of these results have been used previously to study single-shot QEC of qubit codes in Ref.~\cite{bombin2015singleshot}. Here, we provide all details needed for completeness and generalize results needed to study qudit codes when necessary.

Let $G=(V,E)$ be a graph with vertices $V$ and edges $E$, where each edge is a subset of $V$ containing two vertices, $\{v_1,v_2\} \in E$, $v_1,v_2 \in V$. A subset of vertices $A \subseteq V$ is called connected if and only if one can start from an arbitrary vertex $v \in A$ and traverse edges in $E$ to reach every other vertex in $A$. Two subsets of vertices $A_{1},A_{2} \subseteq V$ are called mutually disconnected if they are disjoint, $A_{1} \cap A_{2} = \emptyset$, and if no edges connect the two subsets, specifically if $v_{1} \in A_{1}$ and $v_{2} \in A_{2}$, then $\{v_1,v_2\} \notin E$. For any $A \subseteq V$, we can uniquely decompose $A$ into connected components $A_c$, $A = \bigcup_c A_c$, where each $A_c$ is a maximal connected subset of $A$. A subset $B \subseteq V$ is called a connected cover of $A$ if and only if $A \subseteq B$ and for each connected component $B_c$ of $B = \bigcup_c B_c$, for all $B_c$, $B_c \cap A \neq \emptyset$, $B$ contains $A$ and every connected component of $B$ intersects $A$. The collection of connected covers $B$ of $A$ which contain $\vert B \vert = s \geq \vert A \vert$ total vertices is denoted $C(s,A)$ and is a subset of the power set of $V$, $2^V$. Also, Lemma 5 of Ref.~\cite{aliferis2008} states that if the graph $G$ has maximum vertex degree $z$, then
\begin{equation}
    \label{eq:connected-covers}
    \vert C(s,A) \vert \leq \frac{(e z)^s}{e z^{\vert A \vert}} .
\end{equation}

Consider a graph $G = (V,E)$, a subset of vertices $A \subseteq V$, and some $k \geq 0$. Let $C_{k}(s,A) \subseteq 2^V$, where a subset $W \subseteq V$ is an element of $C_{k}(s,A)$ if there exists some subset $W_0 \subseteq V$ such that
\begin{equation}
    \vert W_0 \vert = s , \quad A \cup W \subseteq W_0 ,
\end{equation}
and every maximal connected component $W_{0,c}$ of $W_0 = \bigcup_c W_{0,c}$ satisfies
\begin{equation}
    A \cap W_{0,c} \neq \emptyset , \quad \vert W \cap W_{0,c} \vert \geq k \vert W_{0,c} \vert .
\end{equation}
In other words, $W_0$ is a connected cover of $A$, as it uniquely decomposes into maximal connected subsets $W_{0,c}$ where $A \cap W_{0,c} \neq \emptyset$ for all $W_{0,c}$, and $W$ is a subset of $W_0$ satisfying the above. Since $\vert W \cap W_{0,c} \vert \leq \vert W_{0,c} \vert$, $W$ intersects a fraction $k \in[0,1]$ of all connected components $W_{0,c}$ of $W_0$. Also, $W_0 \in C(s,A)$, and for $k=1$,
\begin{equation}
    C_{1}(s,A) = C(s,A) .
\end{equation}
This means that for every $W \in C_{k}(s,A)$, there exists at least one $W_0 \in C(s,A)$ such that $W \subseteq W_0$. Thus, one can bound $\vert C_{k}(s,A) \vert$ by observing that since each $W_0 \in C(s,A)$ has $\vert W_0 \vert = s$, there are $2^s$ possible subsets of $W_0$ and thus $\leq 2^s$ possible $W \in C_{k}(s,A)$ where $W \subseteq W_0$,
\begin{equation}
    \vert C_{k}(s,A) \vert \leq 2^s \vert C(s,A) \vert .
\end{equation}
Moreover, if $G$ has maximum degree $z$, Eq.~(\ref{eq:connected-covers}) states
\begin{equation}
    \label{eq:cover-size-bound}
    \vert C(s,A) \vert \leq \frac{(e z)^s}{e z^{\vert A \vert}} \leq (e z)^s  \Rightarrow \vert C_{k}(s,A) \vert \leq (2 e z)^s .
\end{equation}
Lastly, it is useful to define the collection of all connected subsets of size $s$,
\begin{equation}
    C(s) = \bigcup_{v \in V} C(s,\{v\}) .
\end{equation}
We can bound $\vert C(s) \vert$ using Eq.~(\ref{eq:connected-covers}) as
\begin{equation}
    \label{eq:all-connected-subsets}
    \vert C(s) \vert \leq \sum_{v \in V} \vert C(s,\{v\}) \vert \leq \vert V \vert \frac{\left( ez \right)^s}{ez} .
\end{equation}
Using the above, we can prove the following lemma, which generalizes Lemma 10 of Ref.~\cite{bombin2015singleshot}.

\begin{lemma}
    \label{lemma:tau-bounded-proof}
    Consider a graph $G = (V,E)$ with maximum vertex degree $z$ and a constant $k > 0$. Suppose one can associate a value $a_{i} \in \{0,\ldots,d-1\} \cong \mathbb{Z}_{d}$ with each vertex $v_{i} \in V$, and let $T$ be the set of all vectors $a$ with elements $a_{i}$ above. Define the support of $a \in T$ to be the subset of vertices $v_{i}$ such that $a_{i}$ is nonzero, $\supp a = \{v_{i} \in V : a_{i} \neq 0\} \subseteq V$. Let $V$ have two subsets of vertices $V_1 \subseteq V$ and $V_2 \subseteq V$, and let $T_1 \subseteq T$ and $T_2 \subseteq T$ correspond to elements of $T$ with supports only in $V_1$ and only in $V_2$ respectively. Consider a function $f : T_1 \rightarrow T$, such that for all $a \in T_1$ and for every maximal connected component $\left( \supp f(a) \right)_{c}$ of $\supp f(a)$,
    \begin{equation}
        \label{eq:support-condition}
        \vert \supp a \cap \left( \supp f(a) \right)_{c} \vert \geq k \vert \left( \supp f(a) \right)_{c} \vert .
    \end{equation}
    Let $\tau_0 = (2ez)^{-1/k}$. Then, for every $\tau \in (0,\tau_0)$, if the distribution $\pr_1(a)$ is $\tau$-bounded over the supports of vectors $a \in T_1$ with $\supp a \subseteq V_1$,
    \begin{align}
        \label{eq:general-proof-original-bound}
        \sum_{\substack{A^\prime \in V_1 : \\ A^\prime \supseteq A}} \sum_{\substack{a : \\ \supp a = A^\prime}} \pr_1(a) = \sum_{\substack{a : \\ \supp a \supseteq A}} \pr_1(a) \leq \tau^{\vert A \vert} ,
    \end{align}
    then the distribution $\pr_2(B)$ over vertices $B \subseteq V_2$,
    \begin{equation}
        \pr_2(B) = \sum_{\substack{a : a \in T_1 , \\ \left(\supp f(a)\right) \cap V_2 = B}} \pr_1(a) ,
    \end{equation}
    is $\eta$-bounded for $\eta = (\tau / \tau_0)^{k} / \left(1 - (\tau / \tau_0)^{k}\right)$. Namely,
    \begin{align}
        \label{eq:general-proof-eta-bound}
        \sum_{\substack{B^\prime \subseteq V_2 : \\ B^\prime \supseteq B}} \pr_2(B^\prime) &= \sum_{\substack{B^\prime \subseteq V_2 : \\ B^\prime \supseteq B}} \sum_{\substack{a : a \in T_1 , \\ \left(\supp f(a)\right) \cap V_2} = B^\prime} \pr_1(a) \nonumber \\
        &= \sum_{a : \supp f(a) \supseteq B} \pr_1(a) \leq \eta^{\vert B \vert} .
    \end{align}
\end{lemma}
\begin{proof}
    Consider a subset of vertices $B \subseteq V_2$ and a vector $a \in T_1$ such that $B \subseteq \supp f(a)$. Decompose $\supp f(a) = \bigcup_{c} \left( \supp f(a) \right)_{c}$, where $\left( \supp f(a) \right)_{c}$ are maximal connected components of $\supp f(a)$. Define $W_0 \subseteq V$ as
    \begin{equation}
        W_0 = \bigcup_{c : \left( \supp f(a) \right)_{c} \cap B \neq \emptyset} \left( \supp f(a) \right)_{c} ,
    \end{equation}
    namely $W_0$ is the union of connected components of $\supp f(a)$ which each intersect $B$. Here, connected components of $W_0$ are $W_{0,c} = \left( \supp f(a) \right)_{c}$, where each $W_{0,c}$ intersects $B$, $W_{0,c} \cap B \neq \emptyset$. For completeness, this means $B \subseteq W_0$. For brevity, call $\supp a = A \subseteq V_1$ and define
    \begin{equation}
        s = \vert W_0 \vert , \quad W = W_0 \cap A ,
    \end{equation}
    so $W \subseteq A$ and $W \subseteq W_0$. Since $B \subseteq W_0$, $B \cup W \subseteq W_0$ as well. Also, for every connected component $W_{0,c}$ of $W_0$,
    \begin{equation}
        \vert W \cap W_{0,c} \vert = \vert W_0 \cap A \cap W_{0,c} \vert \geq k \vert W_{0,c} \vert ,
    \end{equation}
    where the inequality follows from Eq.~(\ref{eq:support-condition}). This implies $W \in C_{k}(s,B)$. We can also lower bound $\vert W \vert \geq ks$, as
    \begin{align}
        W &= W \cap W_0 = \bigcup_{c} W \cap W_{0,c} , \nonumber \\
        \vert W \vert &= \sum_c \vert W \cap W_{0,c} \vert \geq \sum_c k \vert W_{0,c} \vert = k \vert W_0 \vert .
    \end{align}
    Since $\pr_1(a) \geq 0$ for all $a \in T_1$, we can obtain the bound in Eq.~(\ref{eq:general-proof-eta-bound}) by summing over $\pr_1(a)$ for a subset of all possible $a \in T_1$ which contains at least all $a$ with $B \subseteq \supp f(a)$. So far, we have shown that for some fixed $B$ and $f$, we can associate to each $A = \supp a$ some $W$ and $W_0$ defined above such that $W \in C_{k}(s,B)$ where $s \geq \vert B \vert$, thus for all $A \subseteq V_1$, there exists $W \in C_{k}(s,B)$ as described above. Since $W \subseteq A \subseteq V_1$, the set of all vectors $x \in T_1$ with $\supp x = K \subseteq V_1$ such that there exists $W \in C_{k}(s,B)$ where $W \subseteq K$ must necessarily contain the set of all $a$ with $B \subseteq \supp f(a)$ as a subset. Namely,
    \begin{align}
        &\{ a : a \in T_1 , B \subseteq \supp f(a) \} \subseteq \\
        &\left\{ a : a \in T_1 , \exists W \in C_{k}(s,B) : W \subseteq \supp a , s \geq \vert B \vert \right\} \nonumber .
    \end{align}
    We can then upper bound
    \begin{align}
        &\sum_{\substack{B^\prime : B^\prime \supseteq B}} \pr_2(B^\prime) = \sum_{a : \supp f(a) \supseteq B} \pr_1(a) \\
        &\leq \sum_{s : s \geq \vert B \vert} \sum_{\substack{W \subseteq V_1 : \\ W \in C_{k}(s,B)}} \sum_{\substack{a : A \in T_1 , \\ \supp a \supseteq W}} \pr_1(a) . \nonumber
    \end{align}
    Using Eq.~(\ref{eq:general-proof-original-bound}) and $\vert W \vert \geq ks$, we obtain
    \begin{align}
        &\sum_{\substack{B^\prime : B^\prime \supseteq B}} \pr_2(B^\prime) \leq \sum_{s : s \geq \vert B \vert} \sum_{\substack{W \subseteq V_1 : \\ W \in C_{k}(s,B)}} \tau^{\vert W \vert} \\
        &\leq \sum_{s : s \geq \vert B \vert} \sum_{\substack{W \subseteq V_1 : \\ W \in C_{k}(s,B)}} \tau^{ks} \leq \sum_{s : s \geq \vert B \vert} \vert C_{k}(s,B) \vert \tau^{ks} . \nonumber
    \end{align}
    Finally, using Eq.~(\ref{eq:cover-size-bound}), writing $2ez = {\tau_0}^{-k}$, and using $\tau < \tau_0$, we find $\pr_2(B)$ is bounded by
    \begin{align}
        &\sum_{\substack{B^\prime : B^\prime \supseteq B}} \pr_2(B^\prime) \leq \sum_{s : s \geq \vert B \vert} (2ez)^s \tau^{ks} \\
        &= \sum_{s : s \geq \vert B \vert} \left(\frac{\tau}{\tau_0}\right)^{ks} = \left(\frac{\tau}{\tau_0}\right)^{k\vert B \vert} \left( \frac{1}{1 - (\tau / \tau_0)^{k}} \right) \nonumber \\
        &\leq \left( \frac{(\tau / \tau_0)^{k}}{1 - (\tau / \tau_0)^{k}} \right)^{\vert B \vert} = \eta^{\vert B \vert} . \nonumber
    \end{align}
    For $B=\emptyset$, the bound holds as the sum of all probabilities $\pr_1(a)$ is $1 = \eta^{0}$.
\end{proof}

When vectors $a$ are defined over $\mathbb{Z}_{2}$, each $a$ is isomorphic to a subset of $A \subseteq V$ where $a_{i} = 1$ if and only if the corresponding vertex $v_{i} \in A$, thus $\supp a \cong a$. In this case, $f$ induces a function $F$ on $V$,
\begin{equation}
    F(\supp a) = \supp b \Leftrightarrow f(a) = b
\end{equation}
and one recovers Lemma 10 of Ref.~\cite{bombin2015singleshot} as a special case of Lemma~\ref{lemma:tau-bounded-proof}.

\section{Proof of Lemma~\ref{lemma:lc-localnoise}} \label{lc-proof}

\begin{proof}
    Suppose $\Lambda = \mathcal{C}_{\ell} \circ \ldots \circ \mathcal{C}_{1}$, and errors $\mathcal{N}_{i} \in \mathbb{L}_{\tau}$ occur before layer $\mathcal{C}_{i}$. Thus $\Lambda^{\mathcal{N}} = \mathcal{C}_{\ell} \circ \mathcal{N}_{\ell} \circ \ldots \circ \mathcal{C}_{1} \circ \mathcal{N}_{1}$. Each gate in $\Lambda$ and in its reversal channel $\widetilde{\Lambda}$ has range $\leq r$, and we can propagate noise through all $\ell$ layers of $\Lambda$ to obtain an effective noise channel $\mathcal{N}_{\mathrm{eff}} = \Lambda^{\mathcal{N}} \circ \widetilde{\Lambda}$,
    \begin{align}
        &\Lambda^{\mathcal{N}} (\rho_{1}) \approx \Lambda^{\mathcal{N}} \circ \widetilde{\Lambda} \circ \Lambda (\rho_{1}) = \mathcal{N}_{\mathrm{eff}} \circ \Lambda (\rho_{1}) , \\
        &\mathcal{N}_{\mathrm{eff}}
        = \left( \mathcal{C}_{\ell} \circ \mathcal{N}_{\ell} \ldots \circ \left( \mathcal{C}_{1} \circ \mathcal{N}_{1} \circ \widetilde{\mathcal{C}}_{1} \right) \circ \ldots \circ \widetilde{\mathcal{C}}_{\ell} \right) . \nonumber
    \end{align}
    Writing each noise channel
    \begin{equation}
        \mathcal{N}_{i} (\cdot) = \sum_{\alpha_{i}} \pr_{\mathcal{N}_{i}}(\alpha_{i}) \Phi_{\alpha_{i}} (\cdot) ,
    \end{equation}
    we can express $\mathcal{N}_{\mathrm{eff}}$ as a sum over channels corresponding to each specific error pattern $\vec{\alpha}$, weighted by probabilities $\pr_{\mathcal{N}}(\vec{\alpha}) = \prod_{i=1}^{\ell} \pr_{\mathcal{N}_{i}}(\alpha_{i})$,
    \begin{align}
        &\mathcal{N}_{\mathrm{eff}} (\cdot) = \sum_{\vec{\alpha}} \pr_{\mathcal{N}}(\vec{\alpha}) C_{\vec{\alpha}} ( \cdot ) , \\
        C_{\vec{\alpha}} ( \cdot ) &= \mathcal{C}_{\ell} \circ \Phi_{\alpha_{\ell}} \circ \ldots \circ \left[ \mathcal{C}_{1} \circ \Phi_{\alpha_{1}} \circ \widetilde{\mathcal{C}}_{1} \right] \circ \ldots \circ \widetilde{\mathcal{C}}_{\ell} ( \cdot ) , \nonumber
    \end{align}
    where the structure of each $C_{\vec{\alpha}}$ mirrors that of $\mathcal{N}_{\mathrm{eff}}$. We can now recursively replace layers of the channel $C_{\vec{\alpha}}$ with dressed channels, as discussed in Section~\ref{lc-finite-depth}, to obtain a residual channel $\mathcal{J}_{\vec{\alpha}}^{(\ell)}$ which acts approximately the same on all $\Lambda (\rho_{1})$ as $C_{\vec{\alpha}}$ but has bounded support. To do so, first construct the dressed channel $\Phi_{\alpha_{1}}^{\mathcal{C}_{1}}$ given $\Phi_{\alpha_{1}}$ and the LC circuit $\mathcal{C}_{1}$, then the dressed channel $\left(\Phi_{\alpha_{2}} \circ \Phi_{\alpha_{1}}^{\mathcal{C}_{1}}\right)^{\mathcal{C}_{2}}$ given $\Phi_{\alpha_{i+1}} \circ \Phi_{\alpha_{1}}^{\mathcal{C}_{1}}$ and LC circuit $\mathcal{C}_{2}$, and so on. We can write this recursively as
    \begin{equation}
        \mathcal{J}_{\vec{\alpha}}^{(1)} = \Phi_{\alpha_{1}}^{\mathcal{C}_{1}} , \quad
        \mathcal{J}_{\vec{\alpha}}^{(i+1)} = \left(\Phi_{\alpha_{i+1}} \circ \mathcal{J}_{\vec{\alpha}}^{(i)}\right)^{\mathcal{C}_{i+1}} .
    \end{equation}
    By construction, $\mathcal{J}_{\vec{\alpha}}^{(\ell)} \circ \Lambda (\rho_{1}) \approx C_{\vec{\alpha}} \circ \Lambda (\rho_{1})$ for all $\rho_{1} \in \Sigma_{1}$. We may also bound the support of each $\mathcal{J}_{\vec{\alpha}}^{(i)}$. Define
    \begin{equation}
        F_{\alpha,i} = \supp \Phi_{\alpha_{i}} , \quad F_{\vec{\alpha}} = \bigcup_{i \in [1,\ell]} F_{\alpha,i} .
    \end{equation}
    Since each gate in $\mathcal{C}$ has range $\leq r$, the support of $\mathcal{J}_{\vec{\alpha}}^{(i)}$ is bounded by the union of balls of radius $r(i-t+1)$ surrounding each vertex $v \in F_{\alpha,t}$, for all $t \in [1,i]$. This can be seen by observing that there are $(i-t+1)$ layers of $\mathcal{C}$ applied after each channel $\Phi_{\alpha_{t}}$ to construct $\mathcal{J}_{\vec{\alpha}}^{(i)}$. Let $W_{t,v}$ be a ball of radius $rt$ on the system lattice surrounding site $v$. Then, we find
    \begin{equation}
        \supp \mathcal{J}_{\vec{\alpha}}^{(\ell)} \subseteq \bigcup_{t \in [1,\ell]} \bigcup_{v \in F_{\alpha,t}} W_{\ell-t+1,v} \subseteq \bigcup_{v \in F_{\vec{\alpha}}} W_{\ell,v} ,
    \end{equation}
    where we have used $W_{\ell-t+1,v} \subseteq W_{\ell,v}$ to upper bound the support. Since each ball of radius $rt$ has cardinality $\vert W_{t,v} \vert \leq \gamma (rt)^{D}$ where $\gamma$ is a constant dependent on the lattice geometry, we can also bound
    \begin{equation}
        \wt \mathcal{J}_{\vec{\alpha}}^{(\ell)} \leq \left\vert \bigcup_{v \in F_{\vec{\alpha}}} W_{\ell,v} \right\vert  \leq \gamma (r\ell)^{D} \left\vert \bigcup_{i \in [1,\ell]} \supp \Phi_{\alpha_{i}} \right\vert
    \end{equation}
    Note that each ball $W_{\ell,v}$ forms a connected region on the system lattice.
       
    Consider the residual noise channel $\widetilde{\mathcal{N}}^{\Lambda}$, defined as
    \begin{equation}
        \widetilde{\mathcal{N}}^{\Lambda} (\cdot) = \sum_{\vec{\alpha}} \pr_{\mathcal{N}}(\vec{\alpha}) \mathcal{J}_{\vec{\alpha}}^{(\ell)} ( \cdot ) .
    \end{equation}
    By construction, $\widetilde{\mathcal{N}}^{\Lambda}$ acts approximately the same as $\mathcal{N}_{\mathrm{eff}}$ on all states $\Lambda (\rho_{1})$, namely $\widetilde{\mathcal{N}}^{\Lambda} \circ \Lambda (\rho_{1}) \approx \mathcal{N}_{\mathrm{eff}} \circ \Lambda (\rho_{1})$ for all $\rho_{1} \in \Sigma_{1}$. We can now show that $\widetilde{\mathcal{N}}^{\Lambda} \in \mathbb{L}_{\eta}$ by proving that the probability distribution over the supports of $\mathcal{J}_{\vec{\alpha}}^{(\ell)}$ in $\widetilde{\mathcal{N}}^{\Lambda}$ is $\eta$-bounded. Specifically, for all $B \subseteq \mathcal{Q}$,
    \begin{equation}
        \sum_{\vec{\alpha} : \supp \mathcal{J}_{\vec{\alpha}}^{(\ell)} \supseteq B} \pr_{\mathcal{N}}(\vec{\alpha}) \leq \eta^{\vert B \vert} .
    \end{equation}
    To prove this, we must first consider a probability distribution over vertices and all noise channels $\mathcal{N}_{i}$. We already know that all $\mathcal{N}_{i} \in \mathbb{L}_{\tau}$, and consider the probability distribution
    \begin{equation}
        \sum_{\substack{\vec{\alpha} : A = \\ \bigcup_{i} \supp \Phi_{\alpha_{i}}}} \pr_{\mathcal{N}}(\vec{\alpha}) = \sum_{\substack{\vec{\alpha} : A = \\\bigcup_{i} \supp \Phi_{\alpha_{i}}}}\prod_{i=1}^{\ell} \pr_{\mathcal{N}_{i}}(\alpha_{i}) .
    \end{equation}
    We will show that this distribution is $\ell \tau$-bounded inductively. To do so, first, consider the distribution
    \begin{equation}
        \label{eq:lc-induction-basecase}
        \sum_{\substack{\alpha_{1},\alpha_{2} : \\ \supp \Phi_{\alpha_{1}} \cup \supp \Phi_{\alpha_{2}} = A}} \pr_{\mathcal{N}_{1}}(\alpha_{1}) \pr_{\mathcal{N}_{2}}(\alpha_{2}) .
    \end{equation}
    Applying Proposition~\ref{proposition:bounded-composition}, setting $\alpha = \alpha_{1}$, $\beta = \alpha_{2}$, $g(\alpha_{1}) = \supp \Phi_{\alpha_{1}}$, $h(\alpha_{2}) = \supp \Phi_{\alpha_{2}}$, and $\tau_{1} = \tau_{2} = \tau$, we find
    \begin{equation}
        \sum_{\substack{\alpha_{1},\alpha_{2} : \\ \supp \Phi_{\alpha_{1}} \cup \supp \Phi_{\alpha_{2}} \supseteq A}} \pr_{\mathcal{N}_{1}}(\alpha_{1}) \pr_{\mathcal{N}_{2}}(\alpha_{2})  \leq (2\tau)^{\vert A \vert} ,
    \end{equation}
    the distribution Eq.~(\ref{eq:lc-induction-basecase}) is $2\tau$-bounded. Now, suppose
    \begin{equation}
        \sum_{\substack{\alpha_{1},\ldots, \alpha_{k} : \\ \bigcup_{i \leq k} \supp \Phi_{\alpha_{i}} = A}} \prod_{i=1}^{k} \pr_{\mathcal{N}_{i}}(\alpha_{i}) 
    \end{equation}
    is $k\tau$-bounded for arbitrary $k$, and consider
    \begin{align}
        \label{eq:lc-induction-indstep}
        &\sum_{\substack{\alpha_{1},\ldots, \alpha_{k}, \alpha_{k+1} : \\ \left(\bigcup_{i \leq k} \supp \Phi_{\alpha_{i}}\right) \\ \bigcup \supp \Phi_{\alpha_{k+1}} = A}} \left(\prod_{i=1}^{k} \pr_{\mathcal{N}_{i}}(\alpha_{i}) \right) \pr_{\mathcal{N}_{k+1}}(\alpha_{k+1}) \nonumber \\
        &= \sum_{\substack{\alpha_{1},\ldots, \alpha_{k+1} : \\ \bigcup_{i \leq k+1} \supp \Phi_{\alpha_{i}} = A}} \prod_{i=1}^{k+1} \pr_{\mathcal{N}_{i}}(\alpha_{i})
    \end{align}
    Applying Proposition~\ref{proposition:bounded-composition}, setting $\alpha = \{\alpha_{1},\ldots, \alpha_{k}\}$, $\beta = \alpha_{k+1}$, $g(\{\alpha_{1},\ldots, \alpha_{k}\}) = \bigcup_{i \leq k} \supp \Phi_{\alpha_{i}}$, $h(\alpha_{k+1}) = \supp \Phi_{\alpha_{k+1}}$, $\tau_{1} = k\tau$, and $\tau_{2} = \tau$, we find
    \begin{equation}
        \sum_{\substack{\alpha_{1},\ldots, \alpha_{k+1} : \\ \bigcup_{i \leq k+1} \supp \Phi_{\alpha_{i}} = A}} \prod_{i=1}^{k+1} \pr_{\mathcal{N}_{i}}(\alpha_{i}) \leq \left[(k+1)\tau\right]^{\vert A \vert}
    \end{equation}
    the distribution Eq.~(\ref{eq:lc-induction-indstep}) is $(k+1)\tau$-bounded. By induction, we conclude that
    \begin{equation}
        \sum_{\substack{\alpha_{1},\ldots, \alpha_{k} : \\ \bigcup_{i \leq k} \supp \Phi_{\alpha_{i}} = A}} \prod_{i=1}^{k} \pr_{\mathcal{N}_{i}}(\alpha_{i}) 
    \end{equation}
    is $k\tau$-bounded for arbitrary $k$, which means the distribution over $\vec{\alpha}$ and $1 \leq i \leq \ell$ is $\ell \tau$-bounded,
    \begin{align}
        &\sum_{A^\prime : A^\prime \supseteq A} \sum_{\vec{\alpha} : \bigcup_{i} \supp \Phi_{\alpha_{i}} = A^\prime} \pr_{\mathcal{N}}(\vec{\alpha}) \nonumber \\
        &= \sum_{\vec{\alpha} : \bigcup_{i} \supp \Phi_{\alpha_{i}} \supseteq A} \pr_{\mathcal{N}}(\vec{\alpha}) \leq (\ell \tau)^{\vert A \vert} .
    \end{align}
    
    Now, we invoke Lemma~\ref{lemma:tau-bounded-proof}. Let $\Gamma = (V,E)$ be a graph associated with the system lattice, where vertices $V$ are system qudits and edges $E$ link any two geometrically adjacent vertices. This way, $\Gamma$ mirrors the connectivity of the physical system and vertices $V$ have finite maximum degree $z$ dependent on the specific lattice and its spatial dimension $D$. Associate a vector $a$ over $\mathbb{Z}_{2}$ with each subset $A \subseteq V$, where $a_{v} = 1$ if and only if vertex $v \in A$, thus $\supp a = A$. Here, all subsets of vertices $A$ are isomorphic to vectors $a$, and we can use the two interchangeably. Define a function $f$ which maps a vector $a$ to $f(a) = b$, where if $\supp a = A$, then $\supp b = B = \bigcup_{v \in A} W_{\ell,v}$ is the union of balls of radius $r\ell$ surrounding each vertex $v$ in $A$ as discussed earlier. This means for every vertex $v \in A$, $\supp f(a_{v}) = W_{\ell,v}$. Also, decomposing $B = \supp f(a)$ into maximal connected components, $B = \bigcup_{c} B_{c}$, every maximal connected component $B_{c}$ can be written as $B_{c} = \supp f(u)$ for some vector $u$, where $\supp u \subseteq \supp a = A$. Since $v \in W_{\ell,v}$, $\supp u \subseteq B_{c}$. This means, $A \cap B_{c} \supseteq \supp u$, and since $\vert W_{\ell,v} \vert \leq \gamma (r \ell)^{D}$,
    \begin{equation}
        \vert B_{c} \vert \leq \gamma (r \ell)^{D} \vert \supp u \vert \leq \gamma (r \ell)^{D} \vert A \cap B_{c} \vert .
    \end{equation}
    Noting $A = \supp a$, $B = \supp f(a)$, rearranging gives
    \begin{align}
        &\vert \supp a \cap \left( \supp f(a) \right)_{c} \vert \geq k \vert \left( \supp f(a) \right)_{c} \vert , \nonumber \\
        &\Leftrightarrow \vert A \cap B_{c} \vert \geq k \vert B_{c} \vert , \quad k = \frac{1}{\gamma (r \ell)^{D}} 
    \end{align}
    for all $B_{c}$. Applying Lemma~\ref{lemma:tau-bounded-proof} with subsets $V = V_{1} = V_{2}$ and $\tau_{0} = \left(2ez\right)^{-\gamma (r \ell)^{D}}$, since for all $A \subseteq V$,
    \begin{equation}
        \sum_{A^\prime : A^\prime \supseteq A} \sum_{\vec{\alpha} : \bigcup_{i} \supp \Phi_{\alpha_{i}} = A^\prime} \pr_{\mathcal{N}}(\vec{\alpha}) \leq (\ell \tau)^{\vert A \vert} ,
    \end{equation}
    the distribution
    \begin{equation}
        \sum_{A^\prime : f(A^\prime) = B} \left( \sum_{\vec{\alpha} : \bigcup_{i} \supp \Phi_{\alpha_{i}} = A^\prime} \pr_{\mathcal{N}}(\vec{\alpha}) \right)
    \end{equation}
    is $\eta$-bounded with
    \begin{equation}
        \eta = \frac{(\ell \tau / \tau_0)^{k}}{1 - (\ell \tau / \tau_0)^{k}} = \frac{(\ell \tau)^{1 / \gamma (r \ell)^{D}}}{\left(2ez\right)^{-1} - (\ell \tau)^{1 / \gamma (r \ell)^{D}}} .
    \end{equation}
    Specifically, for all $B \subseteq \mathcal{Q}$,
    \begin{equation}
        \sum_{A^\prime : f(A^\prime) \supseteq B} \left( \sum_{\vec{\alpha} : \bigcup_{i} \supp \Phi_{\alpha_{i}} = A^\prime} \pr_{\mathcal{N}}(\vec{\alpha}) \right) \leq \eta^{\vert B \vert} .
    \end{equation}
    Moreover, since $f\left(F_{\vec{\alpha}}\right) = \bigcup_{v \in F_{\vec{\alpha}}} W_{\ell,v} \supseteq \supp \mathcal{J}_{\vec{\alpha}}^{(\ell)}$ and
    $F_{\vec{\alpha}} = \bigcup_{i} \supp \Phi_{\alpha_{i}}$,
    \begin{align}
        &\sum_{\vec{\alpha} : \supp \mathcal{J}_{\vec{\alpha}}^{(\ell)} \supseteq B} \pr_{\mathcal{N}}(\vec{\alpha}) \leq \sum_{\vec{\alpha} : f\left(\bigcup_{i} \supp \Phi_{\alpha_{i}}\right) \supseteq B} \pr_{\mathcal{N}}(\vec{\alpha}) \nonumber \\
        &= \sum_{A^\prime : f(A^\prime) \supseteq B} \left( \sum_{\vec{\alpha} : \bigcup_{i} \supp \Phi_{\alpha_{i}} = A^\prime} \pr_{\mathcal{N}}(\vec{\alpha}) \right) \leq \eta^{\vert B \vert} ,
    \end{align}
    and we can conclude that the distribution associated with the supports of $\mathcal{J}_{\vec{\alpha}}^{(\ell)}$ is $\eta$-bounded and $\widetilde{\mathcal{N}}^{\Lambda} \in \mathbb{L}_{\eta}$, for all $B \subseteq \mathcal{Q}$
    \begin{equation}
        \sum_{\vec{\alpha} : \supp \mathcal{J}_{\vec{\alpha}}^{(\ell)} \supseteq B} \pr_{\mathcal{N}}(\vec{\alpha}) \leq \eta^{\vert B \vert} .
    \end{equation}

    To summarize, we have shown that for $\ell$ layers of range-$r$ locally reversible LC gates, $\Lambda = \mathcal{C}_{\ell} \circ \ldots \circ \mathcal{C}_{1}$, with noise $\mathcal{N}_{i} \in \mathbb{L}_{\tau}$ occurring before layer $\mathcal{C}_{i}$,
    \begin{align}
        \Lambda^{\mathcal{N}} (\rho_{1}) &= \mathcal{C}_{\ell} \circ \mathcal{N}_{\ell} \circ \ldots \circ \mathcal{C}_{1} \circ \mathcal{N}_{1} (\rho_{1}) \nonumber \\
        &\approx \mathcal{N}_{\mathrm{eff}} \circ \Lambda (\rho_{1}) \approx \widetilde{\mathcal{N}}^{\Lambda} \circ \Lambda (\rho_{1})
    \end{align}
    for all $\rho_{1} \in \Sigma_{1}$, where $\widetilde{\mathcal{N}}^{\Lambda} \in \mathbb{L}_{\eta}$ with
    \begin{equation}
        \eta = \frac{(\ell\tau)^{1 / \gamma (r\ell)^{D}}}{\left(2ez\right)^{-1} - (\ell\tau)^{1 / \gamma (r\ell)^{D}}} .
    \end{equation}
\end{proof}

\section{Error correction for general channel noise} \label{qudit-qec}

Error correction for qudit stabilizer codes subject to general channel noise, Eq.~(\ref{eq:qudit-noise}), is similar to that for qubit codes subject to Pauli noise with a few generalizations, and proving fault tolerance qualitatively mirrors the qubit case. Techniques used here are somewhat similar to those presented in Refs.~\cite{bombin2015singleshot,kubica2022}, however we will outline all derivations here for completeness and will point out where generalizations are needed to treat qudits and non-Pauli noise. While we exclusively analyze stabilizer codes here to prove Theorems~\ref{theorem:self-corr-qec} and~\ref{theorem:self-corr-state-prep}, this procedure generalizes to subsystem codes in a rather straightforward manner. For a treatment of fault tolerance in qubit subsystem and stabilizer codes with Pauli noise, we refer readers to the main text of Ref.~\cite{bombin2015singleshot} and the supplementary information of Ref.~\cite{kubica2022}.

Let $\mathcal{S} \subset \mathcal{P}$ be a stabilizer group with code space $\mathcal{V}$. Given a generating set $\mathcal{S}_{0} = \{S_{i}\}$ of $\mathcal{S}$, we define a map
\begin{equation}
    \partial : P \in \mathcal{P} \rightarrow \partial P = \sigma \in \mathbb{Z}_{d}^{\vert \mathcal{S}_{0} \vert} ,
\end{equation}
such that $S_{i} P = \omega^{\sigma_{i}} P S_{i}$. Here, $\partial P = \sigma$ is the syndrome of $P$. As discussed in Section~\ref{qudit-codes}, syndromes satisfy $\partial(P_1 P_2) = \partial P_1 + \partial P_2 \bmod d$. Also, $\partial(P^{\dagger}) = - \partial P \bmod d$. Also as discussed in Section~\ref{qudit-codes}, any Pauli operator may be decomposed as $P = R_{\partial P}^{\dagger} S_{P} L_{P}$ where $R_{\partial P}^{\dagger}$ is a correction operator, $L_{P} \in \mathcal{L}_{\mathcal{S}}$ is a logical operator, and $S_{P} \in \mathcal{S}$ is a stabilizer. By definition, $\partial P = \partial R_{\partial P}^{\dagger}$. Also, a product of $R_{\sigma}$ with zero syndrome yields
\begin{equation}
    \label{eq:correction-composition}
    R_{\sigma} R_{\kappa} R_{\sigma+\kappa}^{\dagger} = L_{\sigma,\kappa} S_{\sigma,\kappa} ,
\end{equation}
where $L_{\sigma,\kappa} \in \mathcal{L}_{\mathcal{S}}$ and $S_{\sigma,\kappa} \in \mathcal{S}$.

Consider QEC implemented by the recovery channel $\mathcal{R} = \left\{ R_{\sigma} \Pi_{\sigma}^{\mathcal{S}} \right\}_{\sigma}$, where $\Pi^{S_{i}}_{\sigma_{i}}$ is the projector onto the $\sigma_{i}$-eigenspace of $S_{i}$ and $\Pi^{\mathcal{S}}_{\sigma} = \prod_{i} \Pi^{S_{i}}_{\sigma_{i}}$ is the projector onto the mutual $\sigma$-eigenspace of $\mathcal{S}_{0}$. One finds that correction operators and projectors satisfy
\begin{align}
    \label{eq:projector-identities}
     &R_{\sigma}^{\dagger} \Pi_{\kappa}^{\mathcal{S}} = \Pi_{\kappa+\sigma}^{\mathcal{S}} R_{\sigma}^{\dagger} , \quad \Pi_{\sigma}^{\mathcal{S}} \Pi_{\kappa}^{\mathcal{S}} = \delta_{\sigma,\kappa} \Pi_{\sigma}^{\mathcal{S}} .
\end{align}
There are many possible choices of correction operators $\{ R_{\sigma} \}$, however in the following sections we will choose operators such that $R_{-\sigma} = R_{\sigma}^{\dagger}$ always. This choice is not necessary to perform QEC, but it simplifies some calculations and makes proofs of fault tolerance easier.

Suppose general channel noise $\mathcal{N}$, given by Eq.~(\ref{eq:qudit-noise}), acts on the code space $\Sigma_{\mathcal{V}}$ of the code $\mathcal{V}$ discussed above. $\mathcal{N}$ takes the form
\begin{align}
    \mathcal{N} (\cdot) &= \sum_{\alpha} \pr_{\mathcal{N}}(\alpha) \Phi_{\alpha} (\cdot) ,
\end{align}
where each $\Phi_{\alpha}$ has Kraus decomposition $\Phi_{\alpha} = \left\{ K_{\alpha} \right\}_{K_{\alpha}}$. Since Pauli operators form a basis for the space of linear operators on $\mathcal{H}$, each Kraus operator $K_{\alpha}$ can be written as a sum over basis Pauli operators,
\begin{align}
    \label{eq:kraus-pauli-decomp}
    \Phi_{\alpha} (\cdot) &= \sum_{K_{\alpha}} K_{\alpha} (\cdot) K_{\alpha}^{\dagger} , \quad  K_{\alpha} = \sum_{P \in \mathcal{P}} c_{P}^{K_{\alpha}} P .
\end{align}
Here we write $\supp \Phi_{\alpha} = \bigcup_{K_{\alpha}} \supp K_{\alpha}$ and $\supp K_{\alpha} = \bigcup_{P : c_{P}^{K_{\alpha}} \neq 0} \supp P$. Each $\Phi_{\alpha}$ acts on a state $\rho$ as
\begin{align}
    \label{eq:noise-pauli-decomp}
    \Phi_{\alpha} (\rho) &= \sum_{K_{\alpha}} \left( \sum_{P \in \mathcal{P}} c_{P}^{K_{\alpha}} P \right) \rho \left( \sum_{Q \in \mathcal{P}} \left(c_{Q}^{K_{\alpha}}\right)^{*} Q^{\dagger} \right) \nonumber \\
    &= \sum_{K_{\alpha}} \sum_{P,Q \in \mathcal{P}} c_{P}^{K_{\alpha}} \left(c_{Q}^{K_{\alpha}}\right)^{*} P \rho Q^{\dagger} .
\end{align}
Suppose $\mathcal{N}$ acts on the state $\rho \in \Sigma_{\mathcal{V}}$ and apply the recovery channel $\mathcal{R} = \left\{ R_{\sigma} \Pi_{\sigma}^{\mathcal{S}} \right\}_{\sigma}$ from above to the system,
\begin{align}
    \mathcal{R} \circ \mathcal{N} (\rho) &= \sum_{\alpha} \pr_{\mathcal{N}}(\alpha) \mathcal{R} \circ \Phi_{\alpha} (\rho) , \\
    \mathcal{R} \circ \Phi_{\alpha} (\rho) &= \sum_{\sigma , K_{\alpha}} \sum_{P,Q \in \mathcal{P}} c_{P}^{K_{\alpha}} \left(c_{Q}^{K_{\alpha}}\right)^{*} R_{\sigma} \Pi_{\sigma}^{\mathcal{S}} P \rho Q^{\dagger} \Pi_{\sigma}^{\mathcal{S}} R_{\sigma}^{\dagger} . \nonumber
\end{align}
Since $\rho \in \Sigma_{\mathcal{V}}$, $\Pi^{S}_{\sigma} \rho = \delta_{\sigma,0} \rho$, and commuting $\Pi_{\sigma}^{\mathcal{S}} P = P \Pi_{\sigma - \partial P}^{\mathcal{S}}$ using Eq.~(\ref{eq:projector-identities}),
\begin{align}
    \label{eq:corrected-noise-pauli}
    &\Pi_{\sigma}^{\mathcal{S}} P \rho Q^{\dagger} \Pi_{\sigma}^{\mathcal{S}} = \delta_{\sigma,\partial P} \delta_{\sigma,\partial Q} \left(R_{\partial P}^{\dagger} L_{P} S_{P}\right) \rho \left(S_{Q}^{\dagger} L_{Q}^{\dagger} R_{\partial Q}\right) , \nonumber \\
    &\mathcal{R} \circ \Phi_{\alpha} (\rho) = \sum_{\sigma , K_{\alpha}} \sum_{P,Q \in \mathcal{P}} c_{P}^{K_{\alpha}} \left(c_{Q}^{K_{\alpha}}\right)^{*} \delta_{\sigma,\partial P} \delta_{\sigma,\partial Q} L_{P} \rho L_{Q}^{\dagger} \nonumber \\
    &= \sum_{K_{\alpha}} \sum_{P,Q \in \mathcal{P}} c_{P}^{K_{\alpha}} \left(c_{Q}^{K_{\alpha}}\right)^{*} \delta_{\partial P , \partial Q} L_{P} \rho L_{Q}^{\dagger} ,
\end{align}
where we write $ \mathcal{R} \circ \Phi_{\alpha} (\rho)$ for $\rho \in \Sigma_{\mathcal{V}}$
in terms of logical Pauli operators.
Equivalently, defining effective logical Kraus operators
\begin{equation}
    \label{eq:effective-kraus}
    \widetilde{K}_{\alpha,\sigma} = \sum_{P \in \mathcal{P}} \left(c_{P}^{K_{\alpha}} \delta_{\sigma,\partial P}\right) L_{P} S_{P} ,
\end{equation}
$ \mathcal{R} \circ \Phi_{\alpha} (\rho)$ may be written as
\begin{equation}
    \label{eq:corrected-noise-kraus}
    \mathcal{R} \circ \Phi_{\alpha} (\rho) = \sum_{\sigma,K_{\alpha}} \widetilde{K}_{\alpha,\sigma} \rho \widetilde{K}_{\alpha,\sigma}^{\dagger} .
\end{equation}
As an aside, to see that $\widetilde{K}_{\alpha,\sigma}$ are in fact Kraus operators, first observe that $\sum_{K_{\alpha}} K_{\alpha}^{\dagger} K_{\alpha} = I$ implies
\begin{align}
    &\sum_{K_{\alpha}} \sum_{P,Q \in \mathcal{P}} c_{P}^{K_{\alpha}} \left(c_{Q}^{K_{\alpha}}\right)^{*} Q^{\dagger} P = I \\
    &= \sum_{K_{\alpha}} \sum_{P} \left\vert c_{P}^{K_{\alpha}} \right\vert^{2} I + \sum_{K_{\alpha}} \sum_{P \neq Q} c_{P}^{K_{\alpha}} \left(c_{Q}^{K_{\alpha}}\right)^{*} Q^{\dagger} P . \nonumber
\end{align}
Since $Q^{\dagger} P \neq I$ if $P \neq Q$, and Pauli operators $T \in \mathcal{P}$ are linearly independent,
\begin{align}
    &\sum_{K_{\alpha},P} \left\vert c_{P}^{K_{\alpha}} \right\vert^{2} = 1 , \quad \sum_{K_{\alpha}} \sum_{P \neq Q} c_{P}^{K_{\alpha}} \left(c_{Q}^{K_{\alpha}}\right)^{*} Q^{\dagger} P = 0 , \nonumber \\
    &\Rightarrow \sum_{\substack{T \in \mathcal{P} : \\
    T \neq I}} \sum_{K_{\alpha}} \sum_{\substack{P \neq Q : \\ Q^{\dagger} P \propto T}} c_{P}^{K_{\alpha}} \left(c_{Q}^{K_{\alpha}}\right)^{*} Q^{\dagger} P = \sum_{\substack{T \in \mathcal{P} : \\
    T \neq I}} \omega_{T} T = 0 . \nonumber
\end{align}
This means the sums over $K_{\alpha}$, $P$, and $Q$ yield $\omega_{T} T = 0$ independent of $T$. However, this also means
\begin{align}
    &\sum_{\sigma,K_{\alpha}} \widetilde{K}_{\alpha,\sigma}^{\dagger} \widetilde{K}_{\alpha,\sigma} \\
    &= \sum_{\sigma,K_{\alpha}} \sum_{P,Q \in \mathcal{P}} c_{P}^{K_{\alpha}} \left(c_{Q}^{K_{\alpha}}\right)^{*}\delta_{\sigma,\partial Q} \delta_{\sigma,\partial P} S_{Q}^{\dagger} L_{Q}^{\dagger} L_{P} S_{P} \nonumber \\
    &= \sum_{K_{\alpha}} \sum_{P,Q \in \mathcal{P}} c_{P}^{K_{\alpha}} \left(c_{Q}^{K_{\alpha}}\right)^{*} \delta_{\partial Q,\partial P} S_{Q}^{\dagger} L_{Q}^{\dagger} R_{\partial Q} R_{\partial P}^{\dagger} L_{P} S_{P} \nonumber \\
    &= \sum_{K_{\alpha},P} \left\vert c_{P}^{K_{\alpha}} \right\vert^{2} I + \sum_{K_{\alpha}} \sum_{P \neq Q} c_{P}^{K_{\alpha}} \left(c_{Q}^{K_{\alpha}}\right)^{*} \delta_{\partial Q,\partial P} Q^{\dagger} P \nonumber \\
    &= I + \sum_{\substack{T \in \mathcal{P} : \\
    T \neq I,\partial T = 0}} \sum_{K_{\alpha}} \sum_{\substack{P \neq Q : \\ Q^{\dagger} P \propto T}} c_{P}^{K_{\alpha}} \left(c_{Q}^{K_{\alpha}}\right)^{*} Q^{\dagger} P = I , \nonumber
\end{align}
where we have used $\partial T = \partial P - \partial Q$ for $Q^{\dagger} P \propto T$ and $\omega_{T} T = 0$ from above. This confirms $\widetilde{K}_{\alpha,\sigma}$ are valid Kraus operators. We can then use either representation of $\mathcal{R} \circ \mathcal{N}$ above to upper bound the trace distance $\norm{\rho - \mathcal{R} \circ \mathcal{N}(\rho)}_{1}$ independent of $\rho \in \Sigma_{\mathcal{V}}$. Using the triangle inequality,
\begin{equation}
    \norm{\rho - \mathcal{R} \circ \mathcal{N}(\rho)}_{1} \leq \sum_{\alpha} \pr_{\mathcal{N}}(\alpha) \norm{\rho - \mathcal{R} \circ \Phi_{\alpha}(\rho)}_{1} .
\end{equation}
Now, define the quantity
\begin{equation}
    F_{\mathcal{R}} (\Phi_{\alpha}) \coloneqq \begin{cases}
        1 & \exists K_{\alpha} , L_{P} \neq I : c_{P}^{K_{\alpha}} \neq 0 \\
        0 & \mathrm{else}
    \end{cases} 
\end{equation}
which indicates whether $\mathcal{R} \circ \Phi_{\alpha}$ acts nontrivially on the code space $\Sigma_{\mathcal{V}}$ in any way, namely whether any Kraus operator $K_{\alpha}$ of $\Phi_{\alpha}$ contains at least one $P = R_{\partial P}^{\dagger} L_{P} S_{P}$ with $L_{P} \neq I$ and nonzero coefficient $c_{P}^{K_{\alpha}}$. Also, define the failure rate of a channel $\mathcal{N}$ given $\mathcal{R}$,
\begin{equation}
    \label{eq:fail-prob}
    \mathrm{fail}\left( \mathcal{N} \right) \coloneqq \sum_{\alpha} \pr_{\mathcal{N}}(\alpha) F_{\mathcal{R}} (\Phi_{\alpha}) .
\end{equation}
Note that $\norm{\rho - \mathcal{R} \circ \Phi_{\alpha}(\rho)}_{1} \leq 2 F_{\mathcal{R}} (\Phi_{\alpha})$ independent of $\rho \in \Sigma_{\mathcal{V}}$,
therefore
\begin{equation}
    \label{eq:fail-bound}
    \sup_{\left\{\rho \in \Sigma_{\mathcal{V}} \right\}} \norm{\rho - \mathcal{R} \circ \mathcal{N}(\rho)}_{1} \leq 2 \cdot \mathrm{fail}\left( \mathcal{N} \right) .
\end{equation}
When discussing noisy recovery and fault tolerance, it will be useful to compute the probability distribution of syndromes $\sigma$ for a channel $\Phi_{\alpha}$ and $\mathcal{R}$, which we define as
\begin{equation}
    \label{eq:syndrome-dist}
    \pr_{\Phi_{\alpha}}(\sigma) \coloneqq \sum_{K_{\alpha}} \sum_{P \in \mathcal{P} : \partial P = \sigma} \left\vert c_{P}^{K_{\alpha}} \right\vert^{2} .
\end{equation}
One may interpret $\pr_{\Phi_{\alpha}}(\sigma)$ as capturing the total content of each Kraus operator with a given syndrome distribution, summed over all Kraus operators. $\pr_{\Phi_{\alpha}}(\sigma)$ induces a Pauli channel $\overline{\Phi}_{\alpha}$ which has the same syndrome distribution as $\Phi_{\alpha}$, $\pr_{\Phi_{\alpha}}(\sigma) = \pr_{\overline{\Phi}_{\alpha}}(\sigma)$, but has $\norm{\rho - \mathcal{R} \circ \overline{\Phi}_{\alpha}(\rho)}_{1} = 0$ for all $\rho \in \Sigma_{\mathcal{V}}$,
\begin{equation}
    \overline{\Phi}_{\alpha} (\cdot) = \sum_{\sigma} \pr_{\Phi_{\alpha}}(\sigma) R_{\sigma}^{\dagger} (\cdot) R_{\sigma} .
\end{equation}
Since $\partial P = \partial R_{\partial P}^{\dagger}$, $\overline{\Phi}_{\alpha}$ has the same syndrome distribution as $\Phi_{\alpha}$ but applies only correctable, Pauli errors. Given $\overline{\Phi}_{\alpha}$, we can also construct a conjugate channel $\overline{\Phi}_{\alpha}^{*}$ which has $\norm{\rho - \mathcal{R} \circ \overline{\Phi}_{\alpha}^{*}(\rho)}_{1} = 0$ for all $\rho \in \Sigma_{\mathcal{V}}$ and is related to $\overline{\Phi}_{\alpha}$ by conjugating $R_{\sigma}^{\dagger} \rightarrow R_{\sigma}$,
\begin{equation}
    \overline{\Phi}_{\alpha}^{*} (\cdot) = \sum_{\sigma} \pr_{\Phi_{\alpha}}(\sigma) R_{\sigma} (\cdot) R_{\sigma}^{\dagger} .
\end{equation}
$\overline{\Phi}_{\alpha}$ and $\overline{\Phi}_{\alpha}^{*}$ have related properties as $R_{\sigma} = R_{-\sigma}^{\dagger}$ and $\supp \sigma = \supp (-\sigma)$. Importantly, probability distributions over the supports of errors and over the supports of syndromes are equal,
\begin{align}
    \sum_{R_{\sigma}^{\dagger} : A = \supp R_{\sigma}^{\dagger}} \pr_{\overline{\Phi}_{\alpha}}(\sigma) &= \sum_{R_{\sigma} : A = \supp R_{\sigma}} \pr_{\overline{\Phi}_{\alpha}}(\sigma) , \nonumber \\
    \sum_{\sigma : B = \supp \sigma} \pr_{\overline{\Phi}_{\alpha}}(\sigma) &= \sum_{\sigma : B = \supp \sigma} \pr_{\overline{\Phi}_{\alpha}}(-\sigma) .
\end{align}
One may then compute syndrome distributions and induced channels $\overline{\mathcal{N}}$ and $\overline{\mathcal{N}}^{*}$ for any general channel noise $\mathcal{N}$ by linearity, with
\begin{align}
    \pr_{\mathcal{N}}(\sigma) &= \sum_{\alpha} \pr_{\mathcal{N}}(\alpha) \pr_{\Phi_{\alpha}}(\sigma) , \\
    \overline{\mathcal{N}} (\cdot) &= \sum_{\alpha} \pr_{\mathcal{N}}(\alpha) \overline{\Phi}_{\alpha} (\cdot) =
    \sum_{\sigma} \pr_{\mathcal{N}}(\sigma) R_{\sigma}^{\dagger} (\cdot) R_{\sigma} , \nonumber \\
    \overline{\mathcal{N}}^{*} (\cdot) &= \sum_{\alpha} \pr_{\mathcal{N}}(\alpha) \overline{\Phi}_{\alpha}^{*} (\cdot) = \sum_{\sigma} \pr_{\mathcal{N}}(\sigma) R_{\sigma} (\cdot) R_{\sigma}^{\dagger} . \nonumber
\end{align}

When composing two channels $\mathcal{N}$ and $\mathcal{M}$, one obtains
\begin{align}
    &\mathcal{N} \circ \mathcal{M} (\cdot) = \sum_{\alpha,\beta} \pr_{\mathcal{N}}(\alpha) \pr_{\mathcal{M}}(\beta) \Phi_{\alpha} \circ \Phi_{\beta} (\cdot) , \\
    &\Phi_{\alpha} \circ \Phi_{\beta} (\cdot) = \sum_{K_{\alpha}, J_{\beta}} K_{\alpha} J_{\beta} (\cdot) J_{\beta}^{\dagger} K_{\alpha}^{\dagger} . \nonumber
    \\ &= \sum_{K_{\alpha}, J_{\beta}} \sum_{P,Q,E,F \in \mathcal{P}} c_{P}^{K_{\alpha}} c_{E}^{J_{\beta}} \left(c_{Q}^{K_{\alpha}} c_{F}^{J_{\beta}}\right)^{*} P E (\cdot) F^{\dagger} Q^{\dagger} . \nonumber
\end{align}
We will now compute $\mathcal{R} \circ \Phi_{\alpha} \circ \Phi_{\beta} (\rho)$ for $\rho \in \Sigma_{\mathcal{V}}$. From Eq.~(\ref{eq:correction-composition}), one finds
\begin{equation}
    R_{\partial (PE)} P E = e^{i\phi(P,E)} L_{\partial P,\partial E} L_{P} L_{E} S_{\partial P,\partial E} S_{P} S_{E}
\end{equation}
where $\phi(P,E)$ is a phase. Computing $\mathcal{R} \circ \Phi_{\alpha} \circ \Phi_{\beta} (\rho)$ for $\rho \in \Sigma_{\mathcal{V}}$,
\begin{align}
    &\mathcal{R} \circ \Phi_{\alpha} \circ \Phi_{\beta} (\rho) = \\
    &\sum_{\substack{\sigma , K_{\alpha}, J_{\beta}, \\ P,Q,E,F \in \mathcal{P}}} c_{P}^{K_{\alpha}} c_{E}^{J_{\beta}} \left(c_{Q}^{K_{\alpha}} c_{F}^{J_{\beta}}\right)^{*} R_{\sigma} \Pi_{\sigma}^{\mathcal{S}} P E \rho F^{\dagger} Q^{\dagger} \Pi_{\sigma}^{\mathcal{S}} R_{\sigma}^{\dagger} \nonumber \\
    &= \sum_{\substack{K_{\alpha}, J_{\beta}, \\ P,Q,E,F \in \mathcal{P}}} \Big( c_{P}^{K_{\alpha}} c_{E}^{J_{\beta}} \left(c_{Q}^{K_{\alpha}} c_{F}^{J_{\beta}}\right)^{*} \delta_{\partial (PE) , \partial (QF)} \nonumber \\
    &\quad \quad \times e^{i\phi(P,E)-i\phi(Q,F)} L_{\partial P,\partial E} L_{P} L_{E} \rho L_{F}^{\dagger} L_{Q}^{\dagger} L_{\partial Q,\partial F}^{\dagger} \Big) . \nonumber
\end{align}
Similarly, computing $\mathcal{R} \circ \overline{\Phi}_{\alpha} \circ \overline{\Phi}_{\beta}$,
\begin{align}
    &\mathcal{R} \circ \overline{\Phi}_{\alpha} \circ \overline{\Phi}_{\beta} (\rho) = \\
    &\sum_{\xi} R_{\xi} \Pi_{\xi}^{\mathcal{S}} \left( \sum_{\sigma,\kappa} \pr_{\Phi_{\alpha}}(\sigma) \pr_{\Phi_{\beta}}(\kappa) R_{\sigma}^{\dagger} R_{\kappa}^{\dagger} (\cdot) R_{\kappa} R_{\sigma} \right) \Pi_{\xi}^{\mathcal{S}} R_{\xi}^{\dagger} \nonumber \\
    &= \sum_{\xi} \sum_{\sigma,\kappa : \sigma + \kappa = \xi} \pr_{\Phi_{\alpha}}(\sigma) \pr_{\Phi_{\beta}}(\kappa) L_{\sigma,\kappa} \rho L_{\sigma,\kappa}^{\dagger}  \nonumber \\
    &= \sum_{\xi} \sum_{\sigma,\kappa : \sigma + \kappa = \xi} \sum_{\substack{K_{\alpha}, J_{\beta},P,E : \\ \partial P = \sigma , \partial E = \kappa}}  \left\vert c_{P}^{K_{\alpha}} c_{E}^{J_{\beta}} \right\vert^{2} L_{\sigma,\kappa} \rho L_{\sigma,\kappa}^{\dagger} \nonumber \\
    &= \sum_{\xi} \sum_{\substack{K_{\alpha}, J_{\beta},P,E : \\ \partial (PE) = \xi}}  \left\vert c_{P}^{K_{\alpha}} c_{E}^{J_{\beta}} \right\vert^{2} L_{\partial P,\partial E} \rho L_{\partial P,\partial E}^{\dagger} \nonumber \\
    &= \sum_{\substack{K_{\alpha}, J_{\beta}, \\ P,E \in \mathcal{P}}}  \left\vert c_{P}^{K_{\alpha}} c_{E}^{J_{\beta}} \right\vert^{2} L_{\partial P,\partial E} \rho L_{\partial P,\partial E}^{\dagger} . \nonumber
\end{align}
Observe that $F_{\mathcal{R}}$ for the above channels $\Phi_{\alpha} \circ \Phi_{\beta}$ and $\overline{\Phi}_{\alpha} \circ \overline{\Phi}_{\beta}$ can then be written as
\begin{align}
    F_{\mathcal{R}} (\Phi_{\alpha} \circ \Phi_{\beta}) &= \begin{cases}
        1 & \exists K_{\alpha}, J_{\beta} , L_{\partial P,\partial E} L_{P} L_{E} \neq I : \\ & c_{P}^{K_{\alpha}} c_{E}^{J_{\beta}} \neq 0 \\
        0 & \mathrm{else}
    \end{cases}
    , \nonumber \\
    F_{\mathcal{R}} (\overline{\Phi}_{\alpha} \circ \overline{\Phi}_{\beta}) &= \begin{cases}
        1 & \exists K_{\alpha}, J_{\beta} , L_{\partial P,\partial E} \neq I : c_{P}^{K_{\alpha}} c_{E}^{J_{\beta}} \neq 0 \\
        0 & \mathrm{else}
    \end{cases} 
    . \nonumber
\end{align}
Since $L_{\partial P,\partial E} L_{P} L_{E} \neq I$ only if $L_{P} \neq I$ or $L_{E} \neq I$ or $L_{\partial P,\partial E} \neq I$, we can bound
\begin{equation}
    F_{\mathcal{R}} (\Phi_{\alpha} \circ \Phi_{\beta}) \leq F_{\mathcal{R}} (\Phi_{\alpha}) + F_{\mathcal{R}} (\Phi_{\beta}) + F_{\mathcal{R}} (\overline{\Phi}_{\alpha} \circ \overline{\Phi}_{\beta}) .
\end{equation}
Rewriting $L_{\partial P,\partial E} = \left(L_{\partial P,\partial E} L_{P} L_{E} \right) L_{E}^{\dagger} L_{P}^{\dagger}$, we also find $L_{\partial P,\partial E} \neq I$ only if $L_{P} \neq I$ or $L_{E} \neq I$ or $L_{\partial P,\partial E} L_{P} L_{E} \neq I$, meaning we can bound
\begin{equation}
    F_{\mathcal{R}} (\overline{\Phi}_{\alpha} \circ \overline{\Phi}_{\beta}) \leq F_{\mathcal{R}} (\Phi_{\alpha}) + F_{\mathcal{R}} (\Phi_{\beta}) + F_{\mathcal{R}} (\Phi_{\alpha} \circ \Phi_{\beta})
\end{equation}
as well. These results extend by linearity to the failure probabilities of general channel noise $\mathcal{N}$ and $\mathcal{M}$ and their composition. The above observations can all be encapsulated in the following lemma, which may be viewed as a generalization of parts of Lemma 1 of Ref.~\cite{bombin2015singleshot} and Lemma 4 of Ref.~\cite{kubica2022} to general channel noise.
\begin{lemma}
    \label{lemma:pauli-composition}
    For any general channel noise $\mathcal{N}$ and $\mathcal{M}$,
    \begin{align}
        &\mathrm{fail}\left( \mathcal{N} \circ \mathcal{M} \right) \leq \mathrm{fail}\left( \mathcal{N} \right) + \mathrm{fail}\left( \mathcal{M} \right) + \mathrm{fail}\left( \overline{\mathcal{N}} \circ \overline{\mathcal{M}} \right) , \nonumber \\
        &\mathrm{fail}\left( \overline{\mathcal{N}} \circ \overline{\mathcal{M}} \right) \leq \mathrm{fail}\left( \mathcal{N} \right) + \mathrm{fail}\left( \mathcal{M} \right) + \mathrm{fail}\left( \mathcal{N} \circ \mathcal{M} \right) .
    \end{align}
\end{lemma}

\subsection{Fault tolerant error correction} \label{noisy-qudit-qec}

To study fault tolerant QEC, we will use noise classes $\mathbb{N}_{\tau,\epsilon}$ which depend on the specific recovery channel $\mathcal{R}$ and set of recovery operators $\{ R_{\sigma}^{\dagger} \}$ used. Here, $\tau$ characterizes noise strength and $\epsilon$ bounds $\mathrm{fail}\left( \mathcal{N} \right)$ for all channels $\mathcal{N} \in \mathbb{N}_{\tau,\epsilon}$ with respect to $\mathcal{R}$.
\begin{definition}
    \label{def:recovery-dep-class}
    A recovery-dependent class of noise $\mathbb{N}_{\tau,\epsilon}$ is a parametrized class of noise with strength $\tau$, as outlined in Definition~\ref{def:noise-class}, which depends on the recovery channel $\mathcal{R}$ and satisfies the following additional properties.
    \begin{enumerate}
        \item $\mathcal{N} \in \mathbb{N}_{\tau,\epsilon}$ if and only if $\mathrm{fail}\left( \mathcal{N} \right) \leq \epsilon$ and $\overline{\mathcal{N}} \in \mathbb{N}_{\tau,0}$.
        \item For any $\tau > 0$, there exists $\lambda > 0$ and $\epsilon = f(\lambda ; n)$ such that $\mathbb{L}_{\lambda} \subseteq \mathbb{N}_{\tau,\epsilon}$.
        \item For any $\mathbb{N}_{\tau_1,\epsilon_1}$ and $\mathbb{N}_{\tau_2,\epsilon_2}$, there exists $\delta$ such that $\mathbb{N}_{\tau_1,\epsilon_1} \circ \mathbb{N}_{\tau_2,\epsilon_2} \subseteq \mathbb{N}_{\tau_1 + \tau_2,\epsilon_1 + \epsilon_2 + \delta}$.
    \end{enumerate}
\end{definition}

Consider QEC in the presence of measurement errors. Since the recovery channel $\mathcal{R}$ is an ideal channel, we will denote a potentially noisy implementation of $\mathcal{R}$ as the QL circuit $\Lambda$, where if measurement errors $\mu$ occur with probability $\pr_{\Lambda}(\mu)$, a noisy realization of $\Lambda$ is
\begin{equation}
    \label{eq:noisy-qec-channel}
     \Lambda^{K} (\cdot) = \sum_{\mu} \pr_{\Lambda}(\mu) \Lambda_{\mu} (\cdot) , \; \Lambda_{\mu} = \left\{ R_{\sigma + \mu} \Pi^{\mathcal{S}}_{\sigma} \right\}_{\sigma} .
\end{equation}
Here, we choose recovery operators $R_{\sigma + \mu}$ such that measurement errors which produce invalid syndromes $\mu$ lead to corrections with valid syndromes, using the method outlined in Section~\ref{qudit-codes}. Explicitly, $R_{\sigma + \mu} = R_{\sigma + \mu + \hat{\mu}}$, where $\mu + \hat{\mu}$ is a valid syndrome and $\hat{\mu} = \hat{\mu}(\mu)$ is independent of $\sigma$. From this choice, $\partial R_{\sigma + \mu}^{\dagger} = \sigma + \mu + \hat{\mu}$. Noisy error correction for $\rho \in \Sigma_{\mathcal{V}}$ then takes the form
\begin{align}
    &\Lambda^{K} \circ \mathcal{N} (\rho) = \sum_{\alpha} \pr_{\mathcal{N}}(\alpha) \sum_{\mu} \pr_{\Lambda}(\mu) \Lambda_{\mu} \circ \Phi_{\alpha} (\rho) , \nonumber \\
    &\Lambda_{\mu} \circ \Phi_{\alpha} (\rho) = \sum_{\substack{\sigma , K_{\alpha}, \\ P,Q \in \mathcal{P}}} c_{P}^{K_{\alpha}} \left(c_{Q}^{K_{\alpha}}\right)^{*} R_{\sigma + \mu} \Pi_{\sigma}^{\mathcal{S}} P \rho Q^{\dagger} \Pi_{\sigma}^{\mathcal{S}} R_{\sigma + \mu}^{\dagger} \nonumber \\
    &= \sum_{\sigma , K_{\alpha},P,Q} c_{P}^{K_{\alpha}} \left(c_{Q}^{K_{\alpha}}\right)^{*} \delta_{\sigma,\partial Q} \delta_{\sigma,\partial Q} R_{\sigma + \mu} P \rho Q R_{\sigma + \mu}^{\dagger} \nonumber \\
    &= \sum_{\sigma , K_{\alpha}} \widetilde{K}_{\alpha ; \mu,\sigma} \rho \widetilde{K}_{\alpha ; \mu,\sigma}^{\dagger} = \mathcal{J}_{\alpha ; \mu} (\rho) ,
\end{align}
where effective Kraus operators of $\mathcal{J}_{\alpha ; \mu}$ are
\begin{align}
    \widetilde{K}_{\alpha ; \mu,\sigma} &= R_{\mu} L_{\mu,\sigma} S_{\mu,\sigma} \sum_{P \in \mathcal{P}} \left(c_{P}^{K_{\alpha}} \delta_{\sigma,\partial P}\right) L_{P} S_{P} \nonumber \\
    &= R_{\mu} L_{\mu,\sigma} S_{\mu,\sigma} \widetilde{K}_{\alpha,\sigma} ,
\end{align}
where $\widetilde{K}_{\alpha,\sigma}$ are the effective Kraus operators given in Eq.~(\ref{eq:effective-kraus}). One may verify $\widetilde{K}_{\alpha ; \mu,\sigma}$ are Kraus operators in the same way as $\widetilde{K}_{\alpha,\sigma}$. We can then construct an effective channel $\mathcal{F}$, where $\mathcal{F} (\rho) = \Lambda^{K} \circ \mathcal{N} (\rho)$ for all $\rho \in \Sigma_{\mathcal{V}}$, as
\begin{equation}
    \mathcal{F} (\cdot) = \sum_{\alpha,\mu} \pr_{\mathcal{N}}(\alpha) \pr_{\Lambda}(\mu) \mathcal{J}_{\alpha ; \mu} (\cdot) .
\end{equation}
Recovery-dependent channels for $\mathcal{J}_{\alpha ; \mu}$ and $\mathcal{F}$ are
\begin{align}
    \overline{\mathcal{J}}_{\alpha ; \mu} (\cdot) = R_{\mu} (\cdot) R_{\mu}^{\dagger} , &\quad \overline{\mathcal{F}} (\cdot) = \sum_{\mu} \pr_{\Lambda}(\mu) \overline{\mathcal{J}}_{\alpha ; \mu} (\cdot) , \\
    \overline{\mathcal{J}}_{\alpha ; \mu}^{*} (\cdot) = R_{\mu}^{\dagger} (\cdot) R_{\mu} , &\quad \overline{\mathcal{F}}^{*} (\cdot) = \sum_{\mu} \pr_{\Lambda}(\mu) \overline{\mathcal{J}}_{\alpha ; \mu}^{*} (\cdot) . \nonumber
\end{align}
Composing $\overline{\mathcal{F}}^{*}$ and $\mathcal{N}$ yields
\begin{align}
    &\overline{\mathcal{F}}^{*} \circ \mathcal{N} (\rho) = \sum_{\alpha,\mu} \pr_{\mathcal{N}}(\alpha) \pr_{\Lambda}(\mu) \overline{\mathcal{J}}_{\alpha ; \mu}^{*} \circ \Phi_{\alpha} (\rho) , \\
    &\overline{\mathcal{J}}_{\alpha ; \mu}^{*} \circ \Phi_{\alpha} (\rho) = \sum_{K_{\alpha},P,Q} c_{P}^{K_{\alpha}} \left(c_{Q}^{K_{\alpha}}\right)^{*} R_{\mu}^{\dagger} P \rho Q^{\dagger} R_{\mu} . \nonumber
\end{align}
Since $R_{\mu}^{\dagger} P = R_{\mu + \partial P}^{\dagger} L_{\mu,\partial P} S_{\mu,\partial P} L_{P} S_{P}$, one finds
\begin{equation}
    F_{\mathcal{R}} (\mathcal{J}_{\alpha ; \mu}) = F_{\mathcal{R}} \left(\overline{\mathcal{J}}_{\alpha ; \mu}^{*} \circ \Phi_{\alpha}\right) .
\end{equation}
Furthermore, this implies $\mathrm{fail}\left( \mathcal{F} \right) = \mathrm{fail}\left( \overline{\mathcal{F}}^{*} \circ \mathcal{N} \right)$. Suppose $\mathcal{N} \in \mathbb{N}_{\tau,\epsilon}$ and $\overline{\mathcal{F}} , \overline{\mathcal{F}}^{*} \in \mathbb{N}_{\eta,0}$. Definition~\ref{def:noise-class} implies
\begin{equation}
    \overline{\mathcal{F}}^{*} \circ \mathcal{N} \in \mathbb{N}_{\eta,0} \circ \mathbb{N}_{\tau,\epsilon} \subseteq \mathbb{N}_{\tau + \eta,\epsilon + \delta}
\end{equation}
for some $\delta$, meaning $\mathrm{fail}\left( \mathcal{F} \right) = \mathrm{fail}\left( \overline{\mathcal{F}}^{*} \circ \mathcal{N} \right) \leq \epsilon + \delta$. Therefore, $\mathcal{F} \in \mathbb{N}_{\eta,\epsilon + \delta}$.

To summarize, if $\Lambda^{K} \circ \mathcal{N} (\rho) = \mathcal{F} (\rho)$ for all $\rho \in \Sigma_{\mathcal{V}}$ and if $\mathcal{N} \in \mathbb{N}_{\tau,\epsilon}$ and $\overline{\mathcal{F}} , \overline{\mathcal{F}}^{*} \in \mathbb{N}_{\eta,0}$, then $\mathcal{F} \in \mathbb{N}_{\eta,\epsilon + \delta}$. This means the noise strength $\eta$ and maximum failure rate $\epsilon + \delta$ of the residual noise are known and can potentially be controlled. If $\eta$ is continuous function which depends on the strength of measurement noise and $\epsilon + \delta = 1/\Omega(\poly(n))$, the QL circuit $\Lambda$ is fault tolerant.

\subsection{Proof of Lemma~\ref{lemma:ldpc-stability}} \label{gen-stochastic-stability}

Here, we prove Lemma~\ref{lemma:ldpc-stability} and show that the code space $\Sigma_{\mathcal{V}}$ of any LDPC qudit stabilizer code with distance $d_{\mathcal{V}} = \Omega(n^{a})$ for some $a > 0$ is stable. More specifically, we show there exists $\lambda^{\star} > 0$ such that for all $\lambda < \lambda^{\star}$, $\mathbb{L}_{\lambda}$ acting on $\Sigma_{\mathcal{V}}$ is $\mathcal{R}$-recoverable, where $\mathcal{R}$ is an error correction channel which applies the minimum weight Pauli correction for every observed syndrome. Much of this section follows closely the proof of Theorem 3 in Ref.~\cite{gottesman2014}, however we include all details for completeness.

\begin{proof}[Proof of Lemma~\ref{lemma:ldpc-stability}]
    Consider a qudit stabilizer code $\mathcal{V}$ with stabilizer group $\mathcal{S}$ and LDPC generating set $\mathcal{S}_{0} = \{ S_{i} \}$, where each generator $S_{i}$ is supported on $\wt S_{i} \leq v_{q}$ qudits and each qudit is itself in the support of $\leq v_{s}$ generators $S_{i}$. Suppose the recovery channel $\mathcal{R} = \left\{ R_{\sigma} \Pi_{\sigma}^{\mathcal{S}} \right\}_{\sigma}$ measures generators in $\mathcal{S}_{0}$, then applies operators $R_{\sigma}$ which have minimum weight among all Pauli operators with syndrome $-\sigma$, $\wt R_{\sigma} \leq \wt P$ for all $P \in \mathcal{P}$ with $\partial P = -\sigma$. Define a graph $\Gamma = (V,E)$, where each vertex $v \in V$ corresponds to a system qudit, $V \cong \mathcal{Q}$, and two vertices are connected by an edge if and only if a generator in $\mathcal{S}_{0}$ acts nontrivially on both qudits, $\{q_{1},q_{2}\} \in E$ if and only if there exists $S_{i} \in \mathcal{S}_{0}$ such that $q_{1},q_{2} \in \supp S_{i}$. This graph has maximum degree $z \leq (v_{q} - 1) v_{s}$.
    
    Consider a noise channel $\mathcal{N} \in \mathbb{L}_{\lambda}$. From Eq.~(\ref{eq:fail-bound}), to show $\mathcal{N}$ is $\mathcal{R}$-recoverable it suffices to bound $\mathrm{fail}\left( \mathcal{N} \right)$ by bounding when any $\Phi_{\alpha}$ may have $F_{\mathcal{R}} (\Phi_{\alpha}) = 1$ based on the support of $\Phi_{\alpha}$. The remainder of this section follows almost exactly the proof of Theorem 3 in Ref.~\cite{gottesman2014}, but we provide all steps for completeness. Suppose a Kraus operator of $\Phi_{\alpha}$ contains the Pauli operator $P$ with nonzero coefficient $c_{P}^{K_{\alpha}} \neq 0$. Noting $R_{\partial P} P = L_{P} S_{P}$ is the remaining Pauli operator after correction, each connected component of $\supp L_{P} S_{P}$ on the graph $\Gamma$ must either itself implement a logical operator or a stabilizer. This is because no two maximal connected components $\left(\supp L_{P} S_{P}\right)_{c}$ and $\left(\supp L_{P} S_{P}\right)_{d}$ of $\supp L_{P} S_{P}$ on $\Gamma$ can intersect the same generators $S_{i}$ by construction, namely
    \begin{align}
        \neg\exists S_{i} \in \mathcal{S}_{0} :& \supp S_{i} \cap \left(\supp L_{P} S_{P}\right)_{c} \neq \emptyset \nonumber \\
        &\text{and }
        \supp S_{i} \cap \left(\supp L_{P} S_{P}\right)_{d} \neq \emptyset 
    \end{align}
    for $d \neq c$.
    Therefore, $L_{P} S_{P}$ restricted to $\left(\supp L_{P} S_{P}\right)_{c}$ ,
    \begin{equation}
        \left(L_{P} S_{P}\right)_{c} \coloneqq \Tr_{V \setminus \left(\supp L_{P} S_{P}\right)_{c}} L_{P} S_{P} ,
    \end{equation}
    must have trivial syndrome $\partial \left(L_{P} S_{P}\right)_{c} = 0$. This means any $\left(L_{P} S_{P}\right)_{c} \notin \mathcal{S}$ is a logical operator and must have
    \begin{equation}
        \left\vert\left(\supp L_{P} S_{P} \right)_{c}\right\vert = \left\vert\left(\supp R_{\partial P} P \right)_{c}\right\vert \coloneqq w \geq d_{\mathcal{V}} .
    \end{equation}
    Since $\left(\supp R_{\partial P} P\right)_{c} \subseteq \left(\supp R_{\partial P}\right)_{c} \cup \left(\supp P\right)_{c}$,
    \begin{equation}
        w \leq \left\vert \left(\supp R_{\partial P}\right)_{c} \right\vert + \left\vert \left(\supp P\right)_{c} \right\vert .
    \end{equation}
    However, since $R_{\partial P}$ is minimum weight,
    \begin{equation}
        \left\vert \left(\supp R_{\partial P}\right)_{c} \right\vert \leq \left\vert \left(\supp P^{\dagger}\right)_{c} \right\vert =  \left\vert \left(\supp P\right)_{c} \right\vert ,
    \end{equation}
    as otherwise one could replace $\left(R_{\partial P}\right)_{c}$ with $\left(P^{\dagger}\right)_{c}$
    on $\left(\supp L_{P} S_{P}\right)_{c}$ to obtain a lower weight correction. Thus,
    \begin{equation}
        w \leq 2\left\vert \left(\supp P\right)_{c} \right\vert \Rightarrow \left\vert \left(\supp P\right)_{c} \right\vert \geq w/2 .
    \end{equation}
    This means any connected $\left(\supp L_{P} S_{P}\right)_{c}$ which supports a nontrivial logical operator $\left(L_{P} S_{P}\right)_{c} \notin \mathcal{S}$ must contain $w$ vertices, and $P$ must act on $\geq w/2$ of them. Since $\supp P \subseteq \supp K_{\alpha} \subseteq \supp \Phi_{\alpha}$, if $F_{\mathcal{R}} (\Phi_{\alpha}) = 1$, then $\Phi_{\alpha}$ must be supported on $\geq w/2$ vertices within a connected subset $A$ with $\vert A \vert = w \geq d_{\mathcal{V}}$. We can then bound $\mathrm{fail}\left( \mathcal{N} \right)$ by the probability that any $\Phi_{\alpha}$ is supported on $\geq \vert A \vert/2$ sites in any connected subset $A \subseteq V$ with $\vert A \vert \geq d_{\mathcal{V}}$. From Appendix~\ref{combinatorics}, the collection of all connected subsets $A \subseteq V$ with size $\vert A \vert = w$ is $C(w)$. Since $\pr_{\mathcal{N}}(\alpha)$ is $\lambda$-bounded, and $\vert C(w) \vert$ is bounded as discussed in Eq.~(\ref{eq:all-connected-subsets}),
    \begin{align}
        &\mathrm{fail}\left( \mathcal{N} \right) \leq \sum_{w \geq d_{\mathcal{V}}} \sum_{A \in C(w)} \sum_{\substack{B \subseteq A : \\ \vert B \vert \geq \vert A \vert / 2}} \sum_{\substack{\alpha : \\ \supp \Phi_{\alpha} \supseteq B}} \pr_{\mathcal{N}}(\alpha) \nonumber \\
        &\leq \sum_{w \geq d_{\mathcal{V}}} \sum_{A \in C(w)} \sum_{\substack{B \subseteq A : \\ \vert B \vert \geq w/ 2}} \lambda^{\vert B \vert} \leq  \sum_{w \geq d_{\mathcal{V}}} \sum_{A \in C(w)} 2^{w} \lambda^{w/2} \nonumber \\
        &\leq \sum_{w \geq d_{\mathcal{V}}} \vert C(w) \vert 2^{w} \lambda^{w/2} \leq \sum_{w \geq d_{\mathcal{V}}} n (ez)^{w-1} 2^{w} \lambda^{w/2} \nonumber \\
        &=\frac{n}{ez} \frac{1}{1-2ez\sqrt{\lambda}} \left(2ez\sqrt{\lambda}\right)^{d_{\mathcal{V}}} .
    \end{align}
    Rearranging, this gives
    \begin{equation}
        \mathrm{fail}\left( \mathcal{N} \right) \leq
        \frac{n/ez}{1-2ez\sqrt{\lambda}} \left(\frac{\lambda}{(2ez)^{-2}}\right)^{d_{\mathcal{V}} / 2} ,
    \end{equation}
    and since $d_{\mathcal{V}} = \Omega(n^{a})$, $\mathrm{fail}\left( \mathcal{N} \right) < 1/\Omega(\poly(n))$ for any $\lambda < (2ez)^{-2}$, where $z = (v_{q} - 1) v_{s}$. Therefore, every $\mathcal{N} \in \mathbb{L}_{\lambda}$ is $\mathcal{R}$-recoverable according to Definition~\ref{def:recoverable} and $\Sigma_{\mathcal{V}}$ is stable according to Definition~\ref{def:stable-phase}.
\end{proof}

\section{Fault tolerance for self-correcting qudit codes} \label{self-corr-all}

In this section, we prove Theorems~\ref{theorem:self-corr-qec} and~\ref{theorem:self-corr-state-prep}. We will start by elaborating on the properties discussed in Section~\ref{self-correction-setup} which a qudit stabilizer code must satisfy for these theorems to apply, then discuss the class of noise used, $\mathbb{N}_{\tau,\epsilon}^{\mathrm{exc}}$, and explicitly prove Theorems~\ref{theorem:self-corr-qec} and~\ref{theorem:self-corr-state-prep}. Similar to Appendix~\ref{qudit-qec}, many techniques used here are generalizations of those presented in Refs.~\cite{bombin2015singleshot,kubica2022}. There are a few subtleties which arise when dealing with non-Pauli, qudit noise, but proof techniques are all qualitatively similar to those developed in Ref.~\cite{bombin2015singleshot}.

\subsection{Properties of self-correcting codes} \label{self-corr-conditions}

As discussed in Section~\ref{self-correction-setup}, consider a qudit stabilizer code $\mathcal{V}$ which has a $(v_{q},v_{s})$ geometrically local generating set of its stabilizer $\mathcal{S}_{0}$ according to Definition~\ref{def:geom-local}, where both $v_{q},v_{s} = O(1)$. We say a generator $S_{i} \in \mathcal{S}_{0}$ is excited if the corresponding syndrome element $\sigma_{i} \neq 0$. To study properties of syndromes, we can construct an excitation graph $\Gamma = (V_{\Gamma},E_{\Gamma})$, whose vertices $V_{\Gamma}$ are associated with generators $S_{i} \in \mathcal{S}_{0}$. Edges $E_{\Gamma}$ may depend on the specific code $\mathcal{V}$, but we require that each vertex in $V_{\Gamma}$ is part of at least one edge, and an edge only connects two vertices if the generators they correspond to are supported a distance $O(1)$ apart on the system lattice. Note that $\Gamma$ then has bounded maximum degree, $z_{\Gamma}$.

We define the support of a syndrome $\sigma$ on $\Gamma$ to be the subset of vertices corresponding to excited generators,
\begin{equation}
    \supp \sigma \cong \left\{ S_{i} \in \mathcal{S}_{0} : \sigma_{i} \neq 0 \right\} .
\end{equation}
$\supp \sigma$ can be decomposed uniquely into maximal connected components on $\Gamma$, $\supp \sigma = \bigsqcup_{c} (\supp \sigma)_{c}$, where subsets $(\supp \sigma)_{c}$ are all mutually disconnected on $\Gamma$, as discussed in Appendix~\ref{combinatorics}. We can then define $\sigma_{c}$ be the elements of $\sigma$ restricted to the subset $(\supp \sigma)_{c}$,
\begin{equation}
    \label{eq:connected-synd-comp}
    \sigma_{c} : \sigma_{c,i} =
    \begin{cases}
        \sigma_{i} , \quad S_{i} \in (\supp \sigma)_c
        \\
        0 , \quad S_{i} \notin (\supp \sigma)_c
    \end{cases} .
\end{equation}
Therefore, $\supp \sigma_{c} = (\supp \sigma)_{c}$ and $\sigma$ can then be rewritten as a sum over its connected components, $\sigma = \sum_{c} \sigma_{c}$.

$\Gamma$ induces a second graph $\Theta = (V_{\Theta},E_{\Theta})$, whose vertices $V_{\Theta}$ correspond to qudits $\mathcal{Q}$, and for all pairs of qudits $q_1,q_2 \in V_{\Theta}$, there is an edge $\{q_1,q_2\} \in E_{\Theta}$ if and only if there exists an edge $\{S_1,S_2\} \in E_{\Gamma}$ such that $q_1,q_2 \in \supp S_1 \cup \supp S_2$. A direct consequence of the construction of $\Theta$ is the following proposition.
\begin{proposition}
    \label{proposition:mutually-disconnected}
    Suppose the supports of two Pauli operators $P_{1}$ and $P_{2}$ are mutually disconnected on $\Theta$. Then, the supports of their syndromes $\partial P_{1}$ and $\partial P_{2}$ are mutually disconnected on $\Gamma$.
\end{proposition}
\begin{proof}
    Since the supports of $P_{1}$ and $P_{2}$ are mutually disconnected, $\supp P_{1} \cap \supp P_{2} = \emptyset$ and no edges in $E_{\Theta}$ connect $\supp P_{1}$ and $\supp P_{2}$, namely $\{q_1,q_2\} \notin E_{\Theta}$ for any pair of qudits $q_1 \in \supp P_{1}$ and $q_2 \in \supp P_{2}$.

    To show $\partial P_{1}$ and $\partial P_{2}$ are mutually disconnected on $\Gamma$, first suppose $\supp \partial P_{1} \cap \supp \partial P_{2} \neq \emptyset$. This means there exists at least one generator $S_a$ in the support of both syndromes, $S_a \in \supp \partial P_{1}$ and $S_a \in \supp \partial P_{2}$. Since $S_a$ is part of at least one edge in $E_{\Gamma}$, the edge $\{q_i,q_j\} \in E_{\Theta}$ for all pairs of qudits $q_i,q_j \in \supp S_{a}$. Since $\supp P_{1} \cap \supp S_a \neq \emptyset$ and $\supp P_{2} \cap \supp S_a \neq \emptyset$, either $P_{1}$ and $P_{2}$ share a qudit, $q_i \in \supp P_{1} \cap \supp P_{2}$, or there exists a pair of qudits $q_i \in \supp P_{1} \cap \supp S_a$ and $q_j \in \supp P_{2} \cap \supp S_a$ such that $\{q_i,q_j\} \in E_{\Theta}$. Since neither are true, $\supp \partial P_{1} \cap \supp \partial P_{2} = \emptyset$ by contradiction.

    Second, suppose there exists $S_1 \in \supp \partial P_{1}$ and $S_2 \in \supp \partial P_{2}$ such that $\{S_1,S_2\} \in E_{\Gamma}$. This means for all pairs of qudits $q_1,q_2 \in \supp S_1 \cup \supp S_2$, the edge $\{q_1,q_2\} \in E_{\Theta}$. However, since $\supp P_{1} \cap \supp S_1 \neq \emptyset$ and $\supp P_{2} \cap \supp S_2 \neq \emptyset$, there exists $ q_1 \in \supp P_{1} \cap \supp S_1$ and $q_2 \in \supp P_{2} \cap \supp S_2$ such that $\{q_1,q_2\} \in E_{\Theta}$. This is not true, so by contradiction $\{S_1,S_2\} \notin E_{\Gamma}$ for any pair of generators $S_1 \in \supp \partial P_{1}$ and $S_2 \in \supp \partial P_{2}$.

    This means the syndromes $\partial P_{1}$ and $\partial P_{2}$ have support on mutually disconnected sets of stabilizers in $\Gamma$, namely $\supp \partial P_{1} \cap \supp \partial P_{2} = \emptyset$ and $\{S_1,S_2\} \notin E_{\Gamma}$ for all pairs of generators $S_1 \in \supp \partial P_{1}$ and $S_2 \in \supp \partial P_{2}$.
\end{proof}

The graphs $\Gamma$ and $\Theta$ discussed above may be constructed for any qudit stabilizer code $\mathcal{V}$ and the properties described above apply generally. To prove Theorems~\ref{theorem:self-corr-qec} and~\ref{theorem:self-corr-state-prep}, we must assume $\mathcal{V}$, $\mathcal{S}_{0}$, $\Gamma$, $\Theta$, and the set of chosen correction operators $\{R_{\sigma}\}$ satisfy the following additional properties.
\begin{enumerate}[itemsep=0em,start=1]
    \item Each maximal connected component $\sigma_{c}$ of a valid syndrome $\sigma$ is itself a valid syndrome.
    \item Correction operators $R_{\sigma}$ for valid syndromes $\sigma$ are chosen such that each maximal connected component $\sigma_{c}$ of $\sigma$ is corrected separately, $R_{\sigma} = \prod_{c} R_{\sigma_{c}}$ where $R_{\sigma_{c}}$ is the correction for the syndrome $\sigma_{c}$.
    \item There exist constants $b,s > 0$ such that if a Pauli operator $P \in \mathcal{P}$ has support on a single connected component of $\Theta$ with $\wt P \leq bn^s$, then $R_{\partial P} P \in \mathcal{S}$.
    \item There exist constants $c,t > 0$ such that given two syndromes $\sigma$ and $\kappa$, if every connected component of $\supp \sigma \cup \supp \kappa$, has $\leq cn^t$ elements, then $R_{\sigma + \kappa} R_{\sigma}^{\dagger} R_{\kappa}^{\dagger} \in \mathcal{S}$.
\end{enumerate}
From Proposition~\ref{proposition:mutually-disconnected}, property 3 generalizes to any Pauli operator whose maximal connected components each have weight $\leq bn^s$. Also, property 2 is a choice of correction operators which is possible if property 1 holds. Often it is useful to choose $R_{\sigma_{c}}$ as geometrically local as possible, namely as the Pauli operator with syndrome $-\sigma_{c}$ which is supported entirely within the smallest possible connected region on $\Theta$. Defining $\Xi(P) \supseteq \supp P$ as the smallest connected region containing $\supp P$, $R_{\sigma_{c}}$ then satisfies $\vert \Xi(R_{\sigma_{c}}) \vert \leq \vert \Xi(P) \vert$ for all $P \in \mathcal{P}$ such that $\partial P = -\sigma_{c}$. This does not necessarily mean that $R_{\sigma_{c}}$ has minimum weight, $\wt R_{\sigma_{c}} \leq \wt P$ for all $P \in \mathcal{P}$ such that $\partial P = -\sigma_{c}$.
Choosing $R_{\sigma_{c}}$ this way may ensure property 4 holds if properties 1-3 all hold.
\begin{proposition}
    \label{proposition:syndrome-supports}
    Suppose $\mathcal{V}$ is defined on a regular lattice tiling a manifold in $D$ spatial dimensions and $\mathcal{S}_{0}$, $\Gamma$, $\Theta$, and $\{R_{\sigma}\}$ satisfy properties 1, 2, and 3 described above. Suppose $R_{\sigma_{c}}$ are chosen as geometrically local as possible. Suppose there exist $\alpha,\beta > 0$ such that $\alpha n^\beta \leq b n^s$ and if all generators $S_{i} \in \supp \sigma_{c}$ corresponding to a connected syndrome component $\sigma_{c}$ are contained within a hypercubic region of the lattice $\Upsilon$ containing $G_{\Upsilon} \leq \alpha n^\beta$ qudits, then there exists a Pauli error $P_{\Upsilon}$ supported within $\Upsilon$ such that $\partial P_{\Upsilon} = \sigma_{c}$. This means property 4 holds with $c = \alpha^{1/D} / (h_{c} v_{q})$ and $t = \beta / D$, where $h_{c} = O(1)$ and depends on the specific lattice structure.
\end{proposition}
\begin{proof}
    Consider two syndromes $\sigma$ and $\kappa$. Decomposing the support of each syndrome into the union of maximal connected components, $\supp \sigma = \bigcup_{a} \supp \sigma_{a}$, $\supp \kappa = \bigcup_{b} \supp \kappa_{b}$, where $\sigma_{a}$ and $\kappa_{b}$ are constructed as in Eq.~(\ref{eq:connected-synd-comp}) and are themselves syndromes. Each maximal connected component of $\supp \sigma \cup \supp \kappa$ can then be decomposed into the union of supports of $\sigma_{a}$ and $\kappa_{b}$,
    \begin{align}
        \label{eq:conn-comp-decomp}
        &\supp \sigma \cup \supp \kappa = \bigcup_{c} \left(\supp \sigma \cup \supp \kappa \right)_{c} , \\
        &\left(\supp \sigma \cup \supp \kappa \right)_{c} = \left( \bigcup_{a:c} \supp \sigma_{a} \right) \cup \left( \bigcup_{b:c} \supp \kappa_{b} \right) . \nonumber
    \end{align}
    Also, since $(\sigma + \kappa)_{i} \neq 0$ only if $\sigma_{i} \neq 0$ or $\kappa_{i} \neq 0$, one finds $\supp (\sigma + \kappa) \subseteq \supp \sigma \cup \supp \kappa$, thus each connected component of $\supp (\sigma + \kappa)$ must be contained within a connected component of $\supp \sigma \cup \supp \kappa$. We can group components $\supp (\sigma + \kappa)_{d}$ of $\supp (\sigma + \kappa)$ as
    \begin{align}
        \label{eq:sum-comp-decomp}
        &\supp (\sigma + \kappa) = \bigcup_{c} \left[ \bigcup_{d:c} \supp (\sigma + \kappa)_{d} \right] , \nonumber \\
        &\bigcup_{d:c} \supp (\sigma + \kappa)_{d} \subseteq \left(\supp \sigma \cup \supp \kappa \right)_{c} .
    \end{align}
    Since $\sigma_{a}$, $\kappa_{b}$, and $(\sigma + \kappa)_{d}$ are all valid syndromes, the above decompositions impose relations between $\sigma$, $\kappa$, and $\sigma+\kappa$. Namely, since all maximal connected components $\left(\supp \sigma \cup \supp \kappa \right)_{c}$ are mutually disconnected,
    \begin{equation}
        \label{eq:comp-syndrome-constraint}
        \sum_{a:c} \sigma_{a} + \sum_{b:c} \kappa_{b} = \sum_{d:c} (\sigma + \kappa)_{d} \quad \forall c .
    \end{equation}
    Therefore, we can decompose correction operators for $\sigma$, $\kappa$, and $\sigma+\kappa$ into products of correction operators for each maximal connected component, $R_{\sigma} = \prod_{a} R_{\sigma_{a}}$, $R_{\kappa} = \prod_{b} R_{\kappa_{b}}$, $R_{\sigma + \kappa} = \prod_{d} R_{(\sigma + \kappa)_{d}}$. Using Eqs.~(\ref{eq:conn-comp-decomp}-\ref{eq:comp-syndrome-constraint}),
    \begin{align}
        &R_{\sigma}^{\dagger} R_{\kappa}^{\dagger} = \prod_{c} \left[ \prod_{a:c} R_{\sigma_{a}}^{\dagger} \prod_{b:c} R_{\kappa_{b}}^{\dagger} \right] , \nonumber \\
        &R_{\sigma + \kappa} = \prod_{c} \left[ \prod_{d:c} R_{(\sigma + \kappa)_{d}} \right] , \\
        &\partial \left( \prod_{d:c} R_{(\sigma + \kappa)_{d}} \prod_{a:c} R_{\sigma_{a}}^{\dagger} \prod_{b:c} R_{\kappa_{b}}^{\dagger} \right) = 0 \quad \forall c . \nonumber
    \end{align}
    We can then rewrite $R_{\sigma + \kappa} R_{\sigma}^{\dagger} R_{\kappa}^{\dagger}$ by grouping operators according to $c$
    \begin{equation}
        R_{\sigma + \kappa} R_{\sigma}^{\dagger} R_{\kappa}^{\dagger} = \theta \prod_{c} \left[ \prod_{d:c} R_{(\sigma + \kappa)_{d}} \prod_{a:c} R_{\sigma_{a}}^{\dagger} \prod_{b:c} R_{\kappa_{b}}^{\dagger} \right] ,
    \end{equation}
    where $\theta = e^{i\alpha}$ is a phase from reordering Pauli operators, and we can analyze each $c$ separately.
    
    Suppose each connected component $\left(\supp \sigma \cup \supp \kappa \right)_{c}$ has weight $\leq w$. Since each stabilizer generator is supported on $\leq v_{q}$ qudits, the generators excited by any subset of $\left(\supp \sigma \cup \supp \kappa \right)_{c}$ are supported on $\leq v_{q} w$ total qudits. Moreover, since the system lattice is regular, each $\left(\supp \sigma \cup \supp \kappa \right)_{c}$ excites stabilizers supported entirely within a hypercube $\Upsilon$ of side length $h_{c} v_{q} w$ containing $\leq (h_{c} v_{q} w)^{D}$ qudits, where $h_{c} = O(1)$ depends on the lattice structure. This means each connected component $\supp \sigma_{a}$ or $\supp \kappa_{b}$ in Eq.~(\ref{eq:conn-comp-decomp}) is also contained within $\Upsilon$. Suppose $(h_{c} v_{q} w)^{D} \leq \alpha n^\beta$. Rearranging, this means $w \leq c n^t$, where $c = \alpha^{1/D} / (h_{c} v_{q})$ and $t = \beta / D$. Then, for each $\sigma_{a}$ or $\kappa_{b}$, there exists a Pauli operator $P_{a}$ or $P_{b}$ supported within $\Upsilon$ with syndrome $\partial P_{a} = \sigma_{a}$ or $\partial P_{b} = \kappa_{b}$. However, this means corrections $R_{\sigma_{a}}^{\dagger}$ and $R_{\kappa_{b}}^{\dagger}$ must all be contained within $\Upsilon$ as well since they are applied as geometrically local as possible. Therefore, $\prod_{a:c} R_{\sigma_{a}}^{\dagger} \prod_{b:c} R_{\kappa_{b}}^{\dagger}$ resides within $\Upsilon$ and is supported on $\leq (h_{c} v_{q} w)^{D}$ qudits. Since $(h_{c} v_{q} w)^{D} \leq\alpha n^\beta \leq b n^s$, applying the correction
    \begin{equation}
        R_{\sum_{a:c} \sigma_{a} + \sum_{b:c} \kappa_{b}} = R_{\sum_{d:c} (\sigma + \kappa)_{d}} = \prod_{d:c} R_{(\sigma + \kappa)_{d}}
    \end{equation}
    to $\prod_{a:c} R_{\sigma_{a}}^{\dagger} \prod_{b:c} R_{\kappa_{b}}^{\dagger}$ yields a stabilizer,
    \begin{equation}
        \prod_{d:c} R_{(\sigma + \kappa)_{d}} \prod_{a:c} R_{\sigma_{a}}^{\dagger} \prod_{b:c} R_{\kappa_{b}}^{\dagger} = S_{c} \in \mathcal{S} \quad \forall c.
    \end{equation}
    Therefore, $R_{\sigma + \kappa} R_{\sigma}^{\dagger} R_{\kappa}^{\dagger} = \theta \prod_{c} S_{c} \in \mathcal{S}$.
\end{proof}

Lastly, the circuit we will use to perform single-shot QEC and state preparation is the QL circuit $\Lambda_{\mathrm{EC}}$, discussed in Section~\ref{self-correction-setup}, which first measures each $S_{i} \in \mathcal{S}_{0}$, corrects the noisy observed syndrome $\sigma + \mu$ by applying a minimum weight syndrome correction $\hat{\mu}(\mu)$, then applies $R_{\sigma + \mu + \hat{\mu}} \in \{ R_{\sigma} \}$ to the system. If measurement errors $\mu$ occur with probability $\pr_{\Lambda}(\mu)$, then a noisy implementation of $\Lambda_{\mathrm{EC}}$ is
\begin{equation}
    \label{eq:noisy-correction-channel}
    \Lambda_{\mathrm{EC}}^{K} = \left\{ \sqrt{\pr_{\Lambda}(\mu)} R_{\sigma + \mu + \hat{\mu}} \Pi^{\mathcal{S}}_{\sigma} \right\}_{\sigma,\mu} .
\end{equation}
We will also use a noiseless version of $\Lambda_{\mathrm{EC}}$ as a recovery channel for $\Sigma_{\mathcal{V}}$ throughout the following sections, namely the recovery channel $\mathcal{R} = \Lambda_{\mathrm{EC}} = \left\{ R_{\sigma} \Pi^{\mathcal{S}}_{\sigma} \right\}_{\sigma}$.

\subsection{Classes of excitation noise} \label{exc-noise-properties}

The class of noise we consider here is a generalization of local stochastic noise. First, since a generic operator $W$ can be written as a linear combination of Pauli operators $P \in \mathcal{P}$, define the syndrome support of $W$ on the excitation graph $\Gamma$, given $\mathcal{S}_{0}$ and the recovery channel $\mathcal{R} = \Lambda_{\mathrm{EC}}$ for the code space $\Sigma_{\mathcal{V}}$, as the union of the supports of syndromes of Pauli operators which $W$ may be constructed from,
\begin{equation}
    W = \sum_{P \in \mathcal{P}} w_{P} P  : \quad \supp \partial W \coloneqq 
    \bigcup_{P : w_{P} \neq 0} \supp \partial P .
\end{equation}
Given a Kraus decomposition $\Phi_{\alpha} = \left\{ K_{\alpha} \right\}_{K_{\alpha}}$ of a channel $\Phi_{\alpha}$, we define the syndrome support of $\Phi_{\alpha}$ to be
\begin{equation}
    \label{eq:syndrome-indicator}
    \supp \partial \Phi_{\alpha} \coloneqq 
    \bigcup_{K_{\alpha}} \supp \partial K_{\alpha} .
\end{equation}
Now consider general channel noise $\mathcal{N}$ . Define the syndrome support distribution of $\mathcal{N}$ with respect to the decomposition Eq.~(\ref{eq:qudit-noise}) to be
\begin{equation}
    \label{eq:syndrome-supp-dist}
    \pr_{\mathrm{syn}\mathcal{N}}(A) \coloneqq \sum_{\alpha : \supp \partial \Phi_{\alpha} = A} \pr_{\mathcal{N}}(\alpha) .
\end{equation}

Consider a qudit stabilizer code $\mathcal{V}$ and recovery channel $\mathcal{R} = \Lambda_{\mathrm{EC}}$ which satisfy the properties of Appendix~\ref{self-corr-conditions}. We define the recovery-dependent class of noise $\mathbb{N}_{\tau,\epsilon}^{\mathrm{exc}}$, dependent on recovery channel $\mathcal{R}$, such that a channel $\mathcal{N} \in \mathbb{N}_{\tau,\epsilon}^{\mathrm{exc}}$ if and only if $\mathrm{fail}\left( \mathcal{N} \right) \leq \epsilon$ and the syndrome support distribution $\pr_{\mathrm{syn}\mathcal{N}}(A)$ is $\tau$-bounded, for all $A \subseteq V_{\Gamma}$
\begin{equation}
    \sum_{\alpha : \supp \partial \Phi_{\alpha} \supseteq A} \pr_{\mathcal{N}}(\alpha) \leq \tau^{\vert A \vert} .
\end{equation}
We will now demonstrate that $\mathbb{N}_{\tau,\epsilon}^{\mathrm{exc}}$ satisfies the properties of Definitions~\ref{def:noise-class} and~\ref{def:recovery-dep-class}. Note that the second and third properties of Definition~\ref{def:noise-class} are implied by those of Definition~\ref{def:recovery-dep-class}, and $\mathbb{N}_{\tau,\epsilon}^{\mathrm{exc}}$ trivially satisfies the first property of Definition~\ref{def:noise-class}, therefore we just need to show the three properties of Definition~\ref{def:recovery-dep-class} are satisfied. First, $\pr_{\mathrm{syn}\mathcal{N}}(A) = \pr_{\mathrm{syn}\overline{\mathcal{N}}}(A) = \pr_{\mathrm{syn}\overline{\mathcal{N}}^{*}}(A)$, which means $\pr_{\mathrm{syn}\mathcal{N}}(A)$ is $\tau$-bounded if and only if $\overline{\mathcal{N}} \in \mathbb{N}_{\tau,0}^{\mathrm{exc}}$, and $\overline{\mathcal{N}} \in \mathbb{N}_{\tau,0}^{\mathrm{exc}}$ if and only if $\overline{\mathcal{N}}^{*} \in \mathbb{N}_{\tau,0}^{\mathrm{exc}}$ as well. Therefore, $\mathcal{N} \in \mathbb{N}_{\tau,\epsilon}^{\mathrm{exc}}$ if and only if $\mathrm{fail}\left( \mathcal{N} \right) \leq \epsilon$ and $\overline{\mathcal{N}} \in \mathbb{N}_{\tau,0}^{\mathrm{exc}}$, and $\mathbb{N}_{\tau,\epsilon}^{\mathrm{exc}}$ satisfies the first property of Definition~\ref{def:recovery-dep-class}. 

To show $\mathbb{N}_{\tau_1,\epsilon_1}^{\mathrm{exc}} \circ \mathbb{N}_{\tau_2,\epsilon_2}^{\mathrm{exc}} \subseteq \mathbb{N}_{\tau_1 + \tau_2,\epsilon_1 + \epsilon_2 + \delta}^{\mathrm{exc}}$, consider $\mathcal{N} \in \mathbb{N}_{\tau_1,\epsilon_1}^{\mathrm{exc}}$ and $\mathcal{M} \in \mathbb{N}_{\tau_2,\epsilon_2}^{\mathrm{exc}}$. First, Eq.~(\ref{eq:syndrome-indicator}) yields
\begin{align}
    &\supp \partial (\Phi_{\alpha} \circ \Phi_{\beta}) = 
    \bigcup_{K_{\alpha},J_{\beta}} \bigcup_{\substack{P,E : c_{P}^{K_{\alpha}} c_{E}^{J_{\beta}} \neq 0}} \supp \partial (PE) \nonumber \\
    &\subseteq \bigcup_{K_{\alpha},J_{\beta}} \bigcup_{\substack{P,E : c_{P}^{K_{\alpha}} c_{E}^{J_{\beta}} \neq 0}} \supp \partial P \cup \supp \partial E  \nonumber \\
    &\subseteq \supp \partial \Phi_{\alpha} \cup \supp \partial \Phi_{\beta} ,
\end{align}
where we have used $\supp (\sigma + \kappa ) \subseteq \supp \sigma \cup \supp \kappa$ and $c_{P}^{K_{\alpha}} c_{E}^{J_{\beta}} \neq 0$ only if $c_{P}^{K_{\alpha}} \neq 0$ and $c_{E}^{J_{\beta}} \neq 0$. We then find $\pr_{\mathrm{syn}\mathcal{N} \circ \mathcal{M}}(A)$ is $\left(\tau_1 + \tau_2\right)$-bounded, as
\begin{align}
    &\sum_{\alpha,\beta : \supp \partial (\Phi_{\alpha} \circ \Phi_{\beta}) \supseteq A} \pr_{\mathcal{N}}(\alpha) \pr_{\mathcal{M}}(\beta) \leq \\
    &\sum_{\alpha,\beta : \supp \partial \Phi_{\alpha} \cup \supp \partial \Phi_{\beta} \supseteq A} \pr_{\mathcal{N}}(\alpha) \pr_{\mathcal{M}}(\beta) \leq \left(\tau_1 + \tau_2\right)^{\vert A \vert} , \nonumber
\end{align}
where we set $g(\nu) = h(\nu) = \supp \partial \Phi_{\nu}$ and use Proposition~\ref{proposition:bounded-composition}. Since $\pr_{\mathrm{syn}\mathcal{N} \circ \mathcal{M}}(A) = \pr_{\mathrm{syn}\overline{\mathcal{N} \circ \mathcal{M}}}(A)$, the syndrome support distribution for $\overline{\mathcal{N} \circ \mathcal{M}}$ is also $\left(\tau_1 + \tau_2\right)$-bounded, thus $\overline{\mathcal{N} \circ \mathcal{M}} \in \mathbb{N}_{\tau_1 + \tau_2,0}^{\mathrm{exc}}$. We can then bound $\mathrm{fail}\left( \mathcal{N} \circ \mathcal{M} \right)$ using Lemma~\ref{lemma:pauli-composition}. We know $\mathrm{fail}\left( \mathcal{N} \right) \leq \epsilon_1$ and $\mathrm{fail}\left( \mathcal{M} \right) \leq \epsilon_2$, and can bound $\mathrm{fail}\left( \overline{\mathcal{N}} \circ \overline{\mathcal{M}} \right)$ by summing over necessary conditions for $F_{\mathcal{R}} (\overline{\Phi}_{\alpha} \circ \overline{\Phi}_{\beta}) = 1$. Here, $L_{\partial P,\partial E} \neq I$ means $R_{\partial (PE)} R_{\partial P}^{\dagger} R_{\partial PE}^{\dagger} \notin \mathcal{S}$. However, property 4 in Appendix~\ref{self-corr-conditions} states that if $R_{\sigma + \kappa} R_{\sigma}^{\dagger} R_{\kappa}^{\dagger} \notin \mathcal{S}$, then $\supp \sigma \cup \supp \kappa$ has at least one connected component with $\geq cn^t$ elements. Thus, $\supp \sigma \cup \supp \kappa \subseteq V_{\Gamma}$ must contain at least one element of $C(cn^t)$, where $\vert C(cn^t) \vert$ is given by Eq.~(\ref{eq:all-connected-subsets}) with $V=\mathcal{S}_{0}$. This means if $F_{\mathcal{R}} (\overline{\Phi}_{\alpha} \circ \overline{\Phi}_{\beta}) = 1$, then there exists $A \in C(cn^t)$, $K_{\alpha}$, $J_{\beta}$, and $P,E \in \mathcal{P}$ with $c_{P}^{K_{\alpha}} c_{E}^{J_{\beta}} \neq 0$ such that $A \subseteq \supp \partial P \cup \supp \partial E$, which also means $A \subseteq \supp \partial (\Phi_{\alpha} \circ \Phi_{\beta})$. We then bound $\mathrm{fail}\left( \overline{\mathcal{N}} \circ \overline{\mathcal{M}} \right)$ as
\begin{align}
    \label{eq:fail-composition-bound}
    &\mathrm{fail}\left( \overline{\mathcal{N}} \circ \overline{\mathcal{M}} \right) \leq \sum_{\substack{\alpha,\beta : \exists A \in C(cn^t) : \\ \supp \partial (\Phi_{\alpha} \circ \Phi_{\beta}) \supseteq A}} \pr_{\mathcal{N}}(\alpha) \pr_{\mathcal{M}}(\beta) \nonumber \\
    &= \sum_{A \in C(cn^t)} \sum_{\substack{\alpha,\beta : \supp \partial (\Phi_{\alpha} \circ \Phi_{\beta}) \supseteq A}} \pr_{\mathcal{N}}(\alpha) \pr_{\mathcal{M}}(\beta) \nonumber \\
    &\leq \sum_{A \in C(cn^t)} \left(\tau_1 + \tau_2\right)^{\vert A \vert} \leq \vert C(cn^t) \vert \left(\tau_1 + \tau_2\right)^{cn^t} \nonumber \\
    &\leq \frac{\vert \mathcal{S}_{0} \vert}{ez_{\Gamma}} \left( \frac{\tau_1 + \tau_2}{(ez_{\Gamma})^{-1}} \right)^{cn^t} = \delta ,
\end{align}
where we use Proposition~\ref{proposition:bounded-composition} in the third line, setting $g(\nu) = h(\nu) = \supp \partial \Phi_{\nu}$, and $\Gamma$ has maximum vertex degree $z_{\Gamma}$. Applying Lemma~\ref{lemma:pauli-composition} yields $\mathrm{fail}\left( \mathcal{N} \circ \mathcal{M} \right) \leq \epsilon_1 + \epsilon_2 + \delta$, therefore $\mathcal{N} \circ \mathcal{M} \in \mathbb{N}_{\tau_1 + \tau_2,\epsilon_1 + \epsilon_2 + \delta}^{\mathrm{exc}}$ and $\mathbb{N}_{\tau_1,\epsilon_1}^{\mathrm{exc}} \circ \mathbb{N}_{\tau_2,\epsilon_2}^{\mathrm{exc}} \subseteq \mathbb{N}_{\tau_1 + \tau_2,\epsilon_1 + \epsilon_2 + \delta}^{\mathrm{exc}}$ and $\mathbb{N}_{\tau,\epsilon}^{\mathrm{exc}}$ satisfies the second property of Definition~\ref{def:recovery-dep-class}.

To show that for any $\tau > 0$, $\exists \lambda > 0$ and $\epsilon = f(\lambda ; n)$ such that $\mathbb{L}_{\lambda} \subseteq \mathbb{N}_{\tau,\epsilon}^{\mathrm{exc}}$, we must show that any $\mathcal{N} \in \mathbb{L}_{\lambda}$ admits a decomposition given by Eq.~(\ref{eq:qudit-noise}) where for any subset of generators $B$,
\begin{equation}
    \label{eq:Nexc-stochastic-bound}
    \sum_{\alpha : \supp \partial \Phi_{\alpha} \supseteq B} \pr_{\mathcal{N}}(\alpha) \leq \tau^{\vert B \vert} , \; \mathrm{fail}\left( \mathcal{N} \right) \leq f(\lambda ; n) .
\end{equation}
Recall from Definition~\ref{def:geom-local} that each generator $S_{i} \in \mathcal{S}_{0}$ has $\leq v_{q}$ qudits in its support and each qudit $q_j \in \mathcal{Q}$ is in the support of $\leq v_{s}$ generators. To obtain the bound in Eq.~(\ref{eq:Nexc-stochastic-bound}), consider some subset of generators $B \subseteq V_{\Gamma} \cong \mathcal{S}_{0}$ and channel $\Phi_{\alpha}$ with support on qudits $\supp \Phi_{\alpha} \subseteq \mathcal{Q}$ such that $B \subseteq \supp \partial \Phi_{\alpha}$. $\Phi_{\alpha}$ must be supported on at least one qudit $q_j \in \supp S_{i}$ for all vertices corresponding to generators $S_{i} \in B$, namely $\supp S_{i} \cap \supp \Phi_{\alpha} \neq \emptyset$ for all $S_{i} \in B$, as otherwise $S_{i}$ commutes with all Kraus operators $K_{\alpha}$ of $\Phi_{\alpha}$ and therefore all Pauli operators $P$ with $c_{P}^{K_{\alpha}} \neq 0$. If $B \subseteq \supp \partial \Phi_{\alpha}$, there exists at least one set of qudits $\{q \} = \bigcup_{i} q_{i} \subseteq \supp \Phi_{\alpha}$ such that each $q_{i} \in \supp S_{i}$ for all $S_{i} \in B$. We can sum over all possible $\{q \} = \bigcup_{i} q_{i}$ such that $q_{i} \in \supp S_{i}$ to bound
\begin{align}
    &\sum_{\alpha : \supp \partial \Phi_{\alpha} \supseteq B} \pr_{\mathcal{N}}(\alpha) \nonumber \\
    &\leq \sum_{\substack{ \{q \} = \bigcup_{i} q_{i}: \\ q_{i} \in \supp S_{i} \forall S_{i} \in B}} \sum_{\alpha : \supp \Phi_{\alpha} \supseteq \{q \}} \pr_{\mathcal{N}}(\alpha) .
\end{align}
To constrain the sets $\{q \}$, first observe $\vert \{q \} \vert \leq \vert B \vert$ as each $S_{i} \in B$ corresponds to one $q_{i} \in \{q \}$. Second, since $\supp S_{i} \cap \{q \} \neq \emptyset$ for all $S_{i}$, $\vert B \vert$ is bounded by the maximum number of generators with support on any single qudit, $v_{s}$, times the total number of qudits in $\{q \}$, $\vert B \vert \leq v_{s} \vert \{q \} \vert$. Third, we can bound the number of distinct subsets $\{q \}$ by noting that each generator $S_{i}$ is supported on $\leq v_{q}$ different qudits, so for $\vert B \vert$ distinct generators $S_{i}$, there are $\leq v_{q}^{\vert B \vert}$ possible distinct subsets $\{q \} = \bigcup_{i} q_{i}$. Using these constraints and the fact that $\pr_{\mathcal{N}}(\alpha)$ is $\lambda$-bounded over the supports of $\Phi_{\alpha}$ for all $\mathcal{N} \in \mathbb{L}_{\lambda}$, $\lambda \in [0,1)$, we obtain
\begin{align}
    &\sum_{\alpha : \supp \partial \Phi_{\alpha} \supseteq B} \pr_{\mathcal{N}}(\alpha) \nonumber \\
    &\leq \sum_{\substack{ \{q \} = \bigcup_{i} q_{i}: \\ q_{i} \in \supp S_{i} \forall S_{i} \in B}} \sum_{\alpha : \supp \Phi_{\alpha} \supseteq \{q \}} \pr_{\mathcal{N}}(\alpha) \nonumber \\
    &\leq \sum_{\substack{ \{q \} = \bigcup_{i} q_{i}: \\ q_{i} \in \supp S_{i} \forall S_{i} \in B}} \lambda^{\vert \{q \} \vert} \leq \sum_{\substack{ \{q \} = \bigcup_{i} q_{i}: \\ q_{i} \in \supp S_{i} \forall S_{i} \in B}} \lambda^{\vert B \vert / v_{s}} \nonumber \\
    &\leq v_{q}^{\vert B \vert} \lambda^{\vert B \vert / v_{s}} = \left(v_{q} \lambda^{1 / v_{s}}\right)^{\vert B \vert} = \tau^{\vert B \vert} .
\end{align}
Thus, $\pr_{\mathrm{syn}\mathcal{N}}(A)$ is $\tau$-bounded for $\tau = v_{q} \lambda^{1 / v_{s}}$. To show
\begin{equation}
    \mathrm{fail}\left( \mathcal{N} \right) = \sum_{\alpha} \pr_{\mathcal{N}}(\alpha) F_{\mathcal{R}} (\Phi_{\alpha}) \leq f(\lambda ; n) ,
\end{equation}
recall that property 3 in Appendix~\ref{self-corr-conditions} states that if $R_{\partial P} P = L_{P} S_{P} \notin \mathcal{S}$, then $P$ has at least one connected component with $\geq bn^s$ elements. In this case $\supp P \subseteq V_{\Theta}$ must contain at least one element of $C(bn^s)$, where $\vert C(bn^s) \vert$ is given by Eq.~(\ref{eq:all-connected-subsets}) with $V = \mathcal{Q}$. This means if $F_{\mathcal{R}} (\Phi_{\alpha}) = 1$, then there exists $A \in C(bn^s)$, $K_{\alpha}$, and $P \in \mathcal{P}$ with $c_{P}^{K_{\alpha}} \neq 0$ such that $A \subseteq \supp P$, which also means $A \subseteq \supp \Phi_{\alpha}$. We can then bound $\mathrm{fail}\left( \mathcal{N} \right)$ as
\begin{align}
    &\mathrm{fail}\left( \mathcal{N} \right) \leq \sum_{\substack{\alpha : \exists A \in C(bn^s) : \supp \Phi_{\alpha} \supseteq A}} \pr_{\mathcal{N}}(\alpha) \nonumber \\
    &= \sum_{A \in C(bn^s)} \sum_{\substack{\alpha : \supp \Phi_{\alpha} \supseteq A}} \pr_{\mathcal{N}}(\alpha) \leq \sum_{A \in C(bn^s)} \lambda^{\vert A \vert} \nonumber \\
    &= \vert C(bn^s) \vert \lambda^{bn^s} \leq \frac{n}{ez_{\Theta}} \left( \frac{\lambda}{(ez_{\Theta})^{-1}} \right)^{bn^s} = f(\lambda ; n) .
\end{align}
Here, we note that $\Theta$ has maximum vertex degree $z_{\Theta}$.
Thus, $\mathbb{L}_{\lambda} \subseteq \mathbb{N}_{\tau,\epsilon}^{\mathrm{exc}}$ for $\tau = v_{q} \lambda^{1 / v_{s}}$ and $\epsilon = f(\lambda ; n)$ given above, and we conclude that $\mathbb{N}_{\tau,\epsilon}^{\mathrm{exc}}$ satisfies the third property of Definition~\ref{def:recovery-dep-class}. Therefore, $\mathbb{N}_{\tau,\epsilon}^{\mathrm{exc}}$ satisfies all properties of Definitions~\ref{def:noise-class} and~\ref{def:recovery-dep-class}.

\subsection{Proof of Theorem~\ref{theorem:self-corr-qec}} \label{self-corr-proof}

In this section, we will prove Theorem~\ref{theorem:self-corr-qec}. Specifically, we will prove that if a code $\mathcal{V}$ and a circuit $\Lambda_{\mathrm{EC}}$ satisfy the properties of Section~\ref{self-correction-setup} and Appendix~\ref{self-corr-conditions}, then for all $\zeta < (2ez_{\Gamma})^{-2}$, $\Lambda_{\mathrm{EC}}^{K} \in \Lambda_{\mathrm{EC}}^{\mathbb{M}_{\zeta}}$, and $\mathcal{N} \in \mathbb{N}_{\tau,\epsilon}^{\mathrm{exc}}$,
\begin{align}
    \Lambda_{\mathrm{EC}}^{K} \circ \mathcal{N} (\rho) = \widetilde{\mathcal{N}} (\rho)  \quad \forall \rho \in \Sigma_{\mathcal{V}}
\end{align}
where $\widetilde{\mathcal{N}} \in \mathbb{N}_{\eta,\epsilon + \delta}^{\mathrm{exc}}$ and
\begin{equation*}
    \eta = \frac{\left(\zeta / (2ez_{\Gamma})^{-2}\right)^{1/2}}{1 - \left(\zeta / (2ez_{\Gamma})^{-2}\right)^{1/2}} , \; \delta = \frac{\vert \mathcal{S}_{0} \vert}{ez_{\Gamma}} \left( \frac{\tau + \eta}{(e z_{\Gamma})^{-1}} \right)^{cn^t} .
\end{equation*}
This means if $\epsilon = 1/\Omega(\poly(n))$, then $\Lambda_{\mathrm{EC}} : \Sigma_{\mathcal{V}} \rightarrow \Sigma_{\mathcal{V}}$ is fault tolerant against the class $(\mathbb{N}_{\tau,\epsilon}^{\mathrm{exc}},\mathbb{M}_{\zeta})$ with critical noise strengths $\zeta^{\star} = (2ez_{\Gamma})^{-2}$ and $\tau^{\star}(\zeta) = (ez_{\Gamma})^{-1} - \eta$.

Suppose the code $\mathcal{V}$ and circuit $\Lambda_{\mathrm{EC}}$ satisfy the properties described in Section~\ref{self-correction-setup} and Appendix~\ref{self-corr-conditions}. Suppose qudit noise $\mathcal{N} \in \mathbb{N}_{\tau,\epsilon}^{\mathrm{exc}}$ acts on the state $\rho \in \Sigma_{\mathcal{V}}$ and local stochastic measurement noise in $\mathbb{M}_{\zeta}$ acts on the error correction circuit $\Lambda_{\mathrm{EC}}$. This is equivalent to subjecting the circuit $\Lambda_{\mathrm{EC}} : \Sigma_{\mathcal{V}} \rightarrow \Sigma_{\mathcal{V}}$ to noise in the class $(\mathbb{N}_{\tau,\epsilon}^{\mathrm{exc}},\mathbb{M}_{\zeta})$, as $\Lambda_{\mathrm{EC}}^{(\mathbb{N}_{\tau,\epsilon}^{\mathrm{exc}},\mathbb{M}_{\zeta})} = \Lambda_{\mathrm{EC}}^{\mathbb{M}_{\zeta}} \circ \mathbb{N}_{\tau,\epsilon}^{\mathrm{exc}}$. Each specific instance of $\Lambda_{\mathrm{EC}}^{K} \in \Lambda_{\mathrm{EC}}^{\mathbb{M}_{\zeta}}$ expressed in Eq.~(\ref{eq:noisy-qec-channel}) can be rewritten after correcting the noisy syndrome. Since $\sigma + \mu \rightarrow \sigma + \mu + \hat{\mu}$ and $\hat{\mu} = \hat{\mu}(\mu)$, $R_{\sigma + \mu} = R_{\sigma + (\mu + \hat{\mu})}$ and
\begin{align}
    &\Lambda_{\mathrm{EC}}^{K} = \left\{ \sqrt{\pr_{\Lambda}(\omega)} R_{\sigma + \omega} \Pi^{\mathcal{S}}_{\sigma} \right\}_{\sigma,\omega} , \nonumber \\
    &\pr_{\Lambda}(\omega) = \sum_{\mu : \mu + \hat{\mu} = \omega} \pr_{\Lambda}(\mu) ,
\end{align}
where $\omega = \omega(\mu) = \mu + \hat{\mu}$ is a valid syndrome. As discussed in Section~\ref{noisy-qudit-qec}, noisy error correction induces an effective noise channel $\mathcal{F}$, where $\Lambda_{\mathrm{EC}}^{K} \circ \mathcal{N} (\rho) = \mathcal{F} (\rho)$ for all $\rho \in \Sigma_{\mathcal{V}}$. We may also rewrite $\overline{\mathcal{F}}$ and $\overline{\mathcal{F}}^{*}$ as
\begin{equation}
    \label{eq:ftqec-effective-noise}
    \overline{\mathcal{F}} = \left\{ \sqrt{\pr_{\Lambda}(\omega)} R_{\omega} \right\}_{\omega} , \overline{\mathcal{F}}^{*} = \left\{ \sqrt{\pr_{\Lambda}(\omega)} R_{\omega}^{\dagger} \right\}_{\omega} .
\end{equation}
If $\mathcal{N} \in \mathbb{N}_{\tau,\epsilon}^{\mathrm{exc}}$ and $\overline{\mathcal{F}} , \overline{\mathcal{F}}^{*} \in \mathbb{N}_{\eta,0}^{\mathrm{exc}}$, then $\mathcal{F} \in \mathbb{N}_{\eta,\epsilon + \delta}^{\mathrm{exc}}$ where $\delta$ is given in Eq.~(\ref{eq:fail-composition-bound}). This means to prove $\Lambda_{\mathrm{EC}}$ is fault tolerant against $(\mathbb{N}_{\tau,\epsilon}^{\mathrm{exc}},\mathbb{M}_{\zeta})$ and enables single-shot, fault tolerant QEC, we just need to prove that for all $\Lambda_{\mathrm{EC}}^{K} \in \Lambda_{\mathrm{EC}}^{\mathbb{M}_{\zeta}}$ and $\mathcal{N} \in \mathbb{N}_{\tau,\epsilon}^{\mathrm{exc}}$, $\overline{\mathcal{F}} , \overline{\mathcal{F}}^{*} \in \mathbb{N}_{\eta,0}^{\mathrm{exc}}$. This is because when $\mathcal{F} \in \mathbb{N}_{\eta,\epsilon + \delta}^{\mathrm{exc}}$ and when $\tau$, $\zeta$, and $\eta$ are small enough, $\delta$ decays superpolynomially with $n$ as shown in Eq.~(\ref{eq:fail-composition-bound}), therefore when $\epsilon$ decays at least polynomially with $n$, the above guarantees both properties of Definition~\ref{def:ft-channel} are satisfied. Since $\overline{\mathcal{F}} \in \mathbb{N}_{\tau,0}^{\mathrm{exc}}$ if and only if $\overline{\mathcal{F}}^{*} \in \mathbb{N}_{\tau,0}^{\mathrm{exc}}$, we just need to find $\eta$ such that for all $B \subseteq V_{\Gamma}$
\begin{equation}
    \sum_{\omega : B \subseteq \supp \omega} \pr_{\Lambda}(\omega) \leq \eta^{\vert B \vert} ,
\end{equation}
for all $\Lambda_{\mathrm{EC}}^{K} \in \Lambda_{\mathrm{EC}}^{\mathbb{M}_{\zeta}}$. To do so, consider some $\Lambda_{\mathrm{EC}}^{K} \in \Lambda_{\mathrm{EC}}^{\mathbb{M}_{\zeta}}$ where $\pr_{\Lambda}(\mu)$ is $\zeta$-bounded, which induces $\overline{\mathcal{F}}$, $\overline{\mathcal{F}}^{*}$, and $\pr_{\Lambda}(\omega)$ as described above. To find a bound on $\pr_{\Lambda}(\omega)$, we will analyze $\supp \omega$. First, decompose $\supp \omega$ into maximal connected components $\supp \omega = \bigcup_c (\supp \omega)_c = \bigcup_c \supp \omega_c$, where $\omega_{c}$ are constructed as in Eq.~(\ref{eq:connected-synd-comp}) and $\supp \omega_c = (\supp \omega)_c$. Since $\omega = \mu + \hat{\mu}$, $\omega_{c,i} = 0$ if and only if $\mu_{i} + \hat{\mu}_{i} = 0$. This lets us rewrite
\begin{align}
    \supp \omega_c &= \supp \omega \cap \supp \omega_c = \supp (\mu + \hat{\mu}) \cap \supp \omega_c \nonumber \\
    &= \left(\supp \mu \cup \supp \hat{\mu}\right) \cap \supp \omega_c ,
\end{align}
from which we can compute
\begin{align}
    \label{eq:synd-comp-weight}
    \vert \supp \omega_c \vert =& \vert \supp \mu \cap \supp \omega_c \vert + \vert \supp \hat{\mu} \cap \supp \omega_c \vert \nonumber \\
    &- \vert \supp \mu \cap \supp \hat{\mu} \cap \supp \omega_c \vert .
\end{align}
Since $\omega_c$ is a maximal connected component of $\omega$, it is itself a valid syndrome due to property 1 in Appendix~\ref{self-corr-conditions}. Therefore, $\omega- \omega_c$ is a valid syndrome as well, corresponding to applying the correction $\hat{\mu} - \omega_c$ instead of $\hat{\mu}$. Since $\hat{\mu}$ has minimum weight among all corrections which yield a valid syndrome, $\vert \supp \hat{\mu} \vert \leq \vert \supp (\hat{\mu} - \omega_c) \vert$. We can decompose $\supp (\hat{\mu} - \omega_c)$ as
\begin{align}
    &\supp (\hat{\mu} - \omega_c) \nonumber \\
    &= \left[\supp \hat{\mu} \cup \supp \omega_c\right] \setminus \supp ( \omega_c : \omega_{c,i} = \hat{\mu}_{i}) ,
\end{align}
where $\omega_{c,i} = \hat{\mu}_{i}$ means $\mu_{i} = 0$ but $\hat{\mu}_{i} \neq 0$. This gives
\begin{align}
    &\supp ( \omega_c : \omega_{c,i} = \hat{\mu}_{i}) = \supp ( \omega_c : \mu_{i} = 0) \nonumber \\
    &= \supp \omega_c \setminus \left(\supp \mu \cap \supp \omega_c\right) ,
\end{align}
and we obtain the expression
\begin{align}
    &\vert \supp (\hat{\mu} - \omega_c) \vert \\
    &= \vert \supp \hat{\mu} \cup \supp \omega_c \vert - \vert \supp \omega_c \setminus \left(\supp \mu \cap \supp \omega_c\right) \vert \nonumber \\
    &= \vert \supp \hat{\mu} \vert - \vert \supp \hat{\mu} \cap \supp \omega_c \vert + \vert \supp \mu \cap \supp \omega_c \vert . \nonumber
\end{align}
Since $\hat{\mu}$ is minimum weight, we find
\begin{align}
    &\vert \supp \hat{\mu} \vert \leq \vert \supp (\hat{\mu} - \omega_c) \vert \nonumber \\
    &= \vert \supp \hat{\mu} \vert - \vert \supp \hat{\mu} \cap \supp \omega_c \vert + \vert \supp \mu \cap \supp \omega_c \vert \nonumber \\
    &\Leftrightarrow \vert \supp \hat{\mu} \cap \supp \omega_c \vert \leq \vert \supp \mu \cap \supp \omega_c \vert .
\end{align}
Substituting into Eq.~(\ref{eq:synd-comp-weight}) gives
\begin{align}
    &\vert \supp \omega_c \vert = \vert \supp \mu \cap \supp \omega_c \vert + \vert \supp \hat{\mu} \cap \supp \omega_c \vert \nonumber \\
    &\quad \quad - \vert \supp \mu \cap \supp \hat{\mu} \cap \supp \omega_c \vert \nonumber \\
    &\leq 2\vert \supp \mu \cap \supp \omega_c \vert - \vert \supp \mu \cap \supp \hat{\mu} \cap \supp \omega_c \vert \nonumber \\
    &\leq 2\vert \supp \mu \cap \supp \omega_c \vert ,
\end{align}
and rearranging yields the inequality
\begin{equation}
    \vert \supp \mu \cap (\supp \omega)_c \vert \geq \frac{1}{2} \vert (\supp \omega)_c \vert .
\end{equation}
Now, invoke Lemma~\ref{lemma:tau-bounded-proof}, with $G = \Gamma$, subsets $V = V_1 = V_2 = V_{\Gamma}$, the function $f(\mu) = \mu + \hat{\mu} = \omega$, and $k = 1/2$. Since $\pr_{\Lambda}(\mu)$ is $\zeta$-bounded, and the above inequality is satisfied for any connected component $(\supp \omega)_c$ of $\supp \omega = \bigcup_c (\supp \omega)_c$, Lemma~\ref{lemma:tau-bounded-proof} implies that the distribution over the supports of $\omega$, 
\begin{equation}
    \sum_{\substack{\mu : \supp (\mu + \hat{\mu}) \supseteq B}} \pr_{\Lambda}(\mu) = \sum_{\omega : \supp \omega \supseteq B} \pr_{\Lambda}(\omega) \leq \eta^{\vert B \vert} ,
\end{equation}
is $\eta$-bounded with
\begin{equation}
    \eta = \frac{(\zeta / \zeta_0)^{1/2}}{1 - (\zeta / \zeta_0)^{1/2}} , \quad \zeta_0 = (2ez_{\Gamma})^{-2} .
\end{equation}
This means $\overline{\mathcal{F}}, \overline{\mathcal{F}}^{*} \in \mathbb{N}_{\eta,0}^{\mathrm{exc}}$, which proves $\mathcal{F} \in \mathbb{N}_{\eta,\epsilon + \delta}^{\mathrm{exc}}$. $\Lambda_{\mathrm{EC}}$ is then fault tolerant as both properties of Definition~\ref{def:ft-channel} are satisfied if $\zeta < (2ez_{\Gamma})^{-2}$, $\tau + \eta < (ez_{\Gamma})^{-1}$, and $\epsilon = 1/\Omega(\poly(n))$.

We comment that the above derivations strictly generalize the treatment of Pauli noise, and one can follow an identical procedure to prove fault tolerance against $\mathbb{N}_{\tau,\epsilon}^{\mathrm{exc}}$ when $\Phi_{\alpha} (\cdot) = P (\cdot) P^{\dagger}$ are restricted to be Pauli operators. Additionally, the above also applies to CSS stabilizer codes when noise is further restricted to consist of either only Pauli $Z$ errors or Pauli $X$ errors, $\Phi_{\alpha} (\cdot) = P (\cdot) P^{\dagger}$ where $P = \bigotimes_{j} Z_{j}^{a_{j}}$ or $P = \bigotimes_{j} X_{j}^{a_{j}}$. Specifically, when channels with only $Z$-type Pauli errors $\mathcal{N}^{Z} \in \mathbb{N}_{\tau,\epsilon}^{\mathrm{exc},Z} \subset \mathbb{N}_{\tau,\epsilon}^{\mathrm{exc}}$ are applied and the circuit $\Lambda_{\mathrm{EC}}^{X}$ consists of measuring $X$-type generators only, or when channels with only $X$-type Pauli errors $\mathcal{N}^{X} \in \mathbb{N}_{\tau,\epsilon}^{\mathrm{exc},X} \subset \mathbb{N}_{\tau,\epsilon}^{\mathrm{exc}}$ are applied and the circuit $\Lambda_{\mathrm{EC}}^{Z}$ consists of measuring $Z$-type generators only. Here $\mathcal{S}_{0} = \mathcal{S}^{X}_{0} \oplus \mathcal{S}^{Z}_{0}$ decomposes into $X$-type and $Z$-type stabilizer generators, $\mathcal{S}^{X}_{0} = \left\{ S^{X} \right\}$ and $\mathcal{S}^{Z}_{0} = \left\{ S^{Z} \right\}$ with $S^{X}=\bigotimes_{j}X_{j}^{a_{j}}$ and $S^{Z}=\bigotimes_{j}Z_{j}^{a_{j}}$. The pair $\mathcal{S}^{X}_{0}$ and $\Lambda_{\mathrm{EC}}^{X}$ and the pair $\mathcal{S}^{Z}_{0}$ and $\Lambda_{\mathrm{EC}}^{Z}$ must each individually satisfy the properties of Section~\ref{self-correction-setup} and Appendix~\ref{self-corr-conditions} for $Z$-type and $X$-type errors respectively. Following the same derivations above for $\mathbb{N}_{\tau,\epsilon}^{\mathrm{exc},Z}$ or $\mathbb{N}_{\tau,\epsilon}^{\mathrm{exc},X}$ rather than for $\mathbb{N}_{\tau,\epsilon}^{\mathrm{exc}}$, one finds in this case,
\begin{align}
    \label{eq:css-single-shot}
    \forall \Lambda_{\mathrm{EC}}^{X,K} &\in \Lambda_{\mathrm{EC}}^{X,\mathbb{M}_{\zeta}}, \mathcal{N}^{Z} \in \mathbb{N}_{\tau,\epsilon}^{\mathrm{exc},Z}, \rho \in \Sigma_{\mathcal{V}} : \nonumber \\
    &\Lambda_{\mathrm{EC}}^{X,K} \circ \mathcal{N}^{Z} (\rho) = \mathcal{F}^{Z} (\rho) , \quad \mathcal{F}^{Z} \in \mathbb{N}_{\eta,\epsilon + \delta}^{\mathrm{exc},Z} , \nonumber \\
    \forall \Lambda_{\mathrm{EC}}^{Z,K} &\in \Lambda_{\mathrm{EC}}^{Z,\mathbb{M}_{\zeta}}, \mathcal{N}^{X} \in \mathbb{N}_{\tau,\epsilon}^{\mathrm{exc},X}, \rho \in \Sigma_{\mathcal{V}} : \nonumber \\
    &\Lambda_{\mathrm{EC}}^{Z,K} \circ \mathcal{N}^{X} (\rho) = \mathcal{F}^{X} (\rho) , \quad \mathcal{F}^{X} \in \mathbb{N}_{\eta,\epsilon + \delta}^{\mathrm{exc},X} . \nonumber
\end{align}

\subsection{Proof of Theorem~\ref{theorem:self-corr-state-prep}} \label{self-corr-prep-proof}

In this section, we will show that one can perform single-shot state preparation if $\mathcal{V}$ is a CSS code which supports single-shot QEC according to Theorem~\ref{theorem:self-corr-qec} and as discussed in Appendix~\ref{self-corr-proof} for (1) $Z$-type Pauli channels with the circuit $\Lambda_{\mathrm{EC}}^{X}$ or (2) $X$-type Pauli channels with the circuit $\Lambda_{\mathrm{EC}}^{Z}$. Specifically, we will show that one can prepare (1) $\dyad{\overline{0}}{\overline{0}}^{\otimes k}$ from $\dyad{0}{0}^{\otimes n}$ or (2) $\dyad{\overline{+}}{\overline{+}}^{\otimes k}$ from $\dyad{+}{+}^{\otimes n}$ fault tolerantly against $(\mathbb{N}_{\tau,\epsilon}^{\mathrm{exc}},\mathbb{M}_{\zeta})$ by applying (1) $\Lambda_{\mathrm{EC}}^{X}$ or (2) $\Lambda_{\mathrm{EC}}^{Z}$ to a product state. Explicitly, we use results in Appendix~\ref{self-corr-proof} to show
\begin{align}
    &(1) : \forall \Lambda_{\mathrm{EC}}^{X,K} \in \Lambda_{\mathrm{EC}}^{X,\mathbb{M}_{\zeta}}, \mathcal{N} \in \mathbb{N}_{\tau,\epsilon}^{\mathrm{exc}} : \\
    &\Lambda_{\mathrm{EC}}^{X,K} \circ \mathcal{N} \left(\dyad{0}{0}^{\otimes n}\right) = \widetilde{\mathcal{N}} \left(\dyad{\overline{0}}{\overline{0}}^{\otimes k}\right) , \; \widetilde{\mathcal{N}} \in \mathbb{N}_{\tau+\eta,\epsilon}^{\mathrm{exc}} , \nonumber \\
    &(2) : \forall \Lambda_{\mathrm{EC}}^{Z,K} \in \Lambda_{\mathrm{EC}}^{Z,\mathbb{M}_{\zeta}}, \mathcal{N} \in \mathbb{N}_{\tau,\epsilon}^{\mathrm{exc}} : \\
    &\Lambda_{\mathrm{EC}}^{Z,K} \circ \mathcal{N} \left(\dyad{+}{+}^{\otimes n}\right) = \widetilde{\mathcal{N}} \left(\dyad{\overline{+}}{\overline{+}}^{\otimes k}\right) , \; \widetilde{\mathcal{N}} \in \mathbb{N}_{\tau+\eta,\epsilon}^{\mathrm{exc}} . \nonumber
\end{align}

For a CSS code $\mathcal{V}$, syndromes $\sigma$ split into components $\sigma = \sigma^{X} \oplus \sigma^{Z}$, and projectors $\Pi^{\mathcal{S}}_{\sigma}$ decompose as
\begin{equation}
    \Pi^{\mathcal{S}}_{\sigma} = \Pi^{X}_{\sigma^{X}} \Pi^{Z}_{\sigma^{Z}} .
\end{equation}
QEC circuits for $Z$-type and $X$-type errors are
\begin{equation}
    \Lambda_{\mathrm{EC}}^{X} = \left\{ R^{Z}_{\sigma^{X}} \Pi^{X}_{\sigma^{X}} \right\}_{\sigma^{X}} , \quad
    \Lambda_{\mathrm{EC}}^{Z} = \left\{ R^{X}_{\sigma^{Z}} \Pi^{Z}_{\sigma^{Z}} \right\}_{\sigma^{Z}} .
\end{equation}
Not only do $\Lambda_{\mathrm{EC}}^{X}$ and $\Lambda_{\mathrm{EC}}^{Z}$ correct $Z$-type and $X$-type errors, but can also be used to prepare the encoded states
\begin{align}
    \Lambda_{\mathrm{EC}}^{X} \left( \dyad{0}{0}^{\otimes n} \right) &= \dyad{\overline{0}}{\overline{0}}^{\otimes k} , \quad \ket{\overline{0}}^{\otimes k} = \sqrt{d^{r_X}} \Pi^{X}_{0} \ket{0}^{\otimes n} , \nonumber \\
    \Lambda_{\mathrm{EC}}^{Z} \left( \dyad{+}{+}^{\otimes n} \right) &= \dyad{\overline{+}}{\overline{+}}^{\otimes k} , \quad \ket{\overline{+}}^{\otimes k} = \sqrt{d^{r_Z}} \Pi^{Z}_{0} \ket{+}^{\otimes n} , \nonumber
\end{align}
where $r_X$ and $r_Z$ are the number of independent $X$-type and $Z$-type stabilizer generators respectively. This is because $\dyad{0}{0}^{\otimes n} = \Pi_{\mathcal{V}_{Z}} \in \mathcal{V}_{Z}$ is a stabilizer state with stabilizer given by all single-qudit $Z_j$, $\mathcal{S}_{\mathcal{V}_{Z}} = \left\langle Z_j \right\rangle$, and $\dyad{+}{+}^{\otimes n} = \Pi_{\mathcal{V}_{X}} \in \mathcal{V}_{X}$ is a stabilizer state with stabilizer given by all single-qudit $X_j$, $\mathcal{S}_{\mathcal{V}_{X}} = \left\langle X_j \right\rangle$. All groups of $Z$-type operators form subgroups of $\mathcal{S}_{\mathcal{V}_{Z}}$ and all groups of $X$-type operators form subgroups of $\mathcal{S}_{\mathcal{V}_{X}}$, thus $\dyad{0}{0}^{\otimes n}$ is already in the $+1$ eigenstate of all $Z$-type stabilizer generators and $\dyad{+}{+}^{\otimes n}$ is already in the $+1$ eigenstate of all $X$-type stabilizer generators. This means projectors satisfy the identities $\Pi^{Z}_{0} \Pi_{\mathcal{V}_{Z}} = \Pi_{\mathcal{V}_{Z}}$ and $\Pi^{X}_{0} \Pi_{\mathcal{V}_{X}} = \Pi_{\mathcal{V}_{X}}$.

If $\mathcal{S}^{X}_{0}$ and $\Lambda_{\mathrm{EC}}^{X}$ admit single-shot QEC as outlined in Section~\ref{self-corr-proof}, then one can fault tolerantly prepare $\ket{\overline{0}}^{\otimes k}$, and if $\mathcal{S}^{Z}_{0}$ and $\Lambda_{\mathrm{EC}}^{Z}$ admit single-shot QEC, then one can fault tolerantly prepare $\ket{\overline{+}}^{\otimes k}$. Here we will focus on applying $\Lambda_{\mathrm{EC}}^{X}$ to prepare $\ket{\overline{0}}^{\otimes k}$ from $\ket{0}^{\otimes n}$, however preparation of $\ket{\overline{+}}^{\otimes k}$ follows from the exact same steps but swapping $X \leftrightarrow Z$ everywhere. In the noiseless case,
\begin{align}
    \Lambda_{\mathrm{EC}}^{X} \left( \dyad{0}{0}^{\otimes n} \right) &= \Lambda_{\mathrm{EC}}^{X} \circ \Pi_{\mathcal{V}_{Z}} = \left\{ R^{Z}_{\sigma^{X}} \Pi^{X}_{\sigma^{X}} \Pi_{\mathcal{V}_{Z}} \right\}_{\sigma^{X}} \nonumber \\
    &= \left\{ \Pi^{X}_{0} R^{Z}_{\sigma^{X}} \Pi_{\mathcal{V}_{Z}} \right\}_{\sigma^{X}} = \left\{ \Pi^{X}_{0} \Pi_{\mathcal{V}_{Z}} \right\}_{\sigma^{X}} \nonumber \\
    &= d^{r_X} \Pi^{X}_{0} \Pi_{\mathcal{V}_{Z}} \Pi^{X}_{0} = \dyad{\overline{0}}{\overline{0}}^{\otimes k} .
\end{align}
Here, $d^{r_X}$ comes from summing over all valid $\sigma^{X}$. Identities in Eqs.~(\ref{eq:correction-composition}-\ref{eq:projector-identities}) all apply here,
\begin{align}
    \label{eq:x-proj-identities}
    &R^{Z}_{\sigma} \Pi^{X}_{\kappa} = \Pi^{X}_{\kappa-\sigma} R^{Z}_{\sigma} , \quad \Pi^{X}_{\sigma} \Pi^{X}_{\kappa} = \delta_{\sigma , \kappa} \Pi^{X}_{\sigma} , \nonumber \\
    &R^{Z}_{\sigma + \kappa} R^{Z \dagger}_{\sigma} R^{Z \dagger}_{\kappa} = L^{Z}_{\sigma,\kappa} S^{Z}_{\sigma,\kappa} ,
\end{align}
where $L^{Z}_{\sigma,\kappa}$ is a $Z$-type logical operator and $S^{Z}_{\sigma,\kappa}$ is a $Z$-type stabilizer operator. Also, all $Z$-type operators stabilize $\Pi_{\mathcal{V}_{Z}}$, $Z_{j} \Pi_{\mathcal{V}_{Z}} = \Pi_{\mathcal{V}_{Z}}$ for all $Z_{j}$, therefore so do all $R^{Z}_{\sigma^{X}}$, $L^{Z}$, and $S^{Z}$. For noisy state preparation, suppose noise in the class $\Lambda_{\mathrm{EC}}^{X , (\mathbb{N}_{\tau,\epsilon}^{\mathrm{exc}},\mathbb{M}_{\zeta})} = \Lambda_{\mathrm{EC}}^{X,\mathbb{M}_{\zeta}} \circ \mathbb{N}_{\tau,\epsilon}^{\mathrm{exc}}$ acts on $\Lambda_{\mathrm{EC}}^{X}$. For completeness, define $\mathbb{N}_{\tau,\epsilon}^{\mathrm{exc}}$ with respect to the overall recovery channel $\mathcal{R} = \Lambda_{\mathrm{EC}}^{X} \circ \Lambda_{\mathrm{EC}}^{Z}$. Each specific instance $\Lambda_{\mathrm{EC}}^{X,K} \in \Lambda_{\mathrm{EC}}^{X,\mathbb{M}_{\zeta}}$ takes the form
\begin{equation}
    \Lambda_{\mathrm{EC}}^{X,K} = \left\{ \sqrt{\pr_{\Lambda_{\mathrm{EC}}^{X}}(\mu)} R^{Z}_{\sigma^{X} + \mu} \Pi^{X}_{\sigma^{X}} \right\}_{\sigma^{X},\mu} .
\end{equation}
Given some channel $\mathcal{N} \in \mathbb{N}_{\tau,\epsilon}^{\mathrm{exc}}$ and partitioning each Pauli basis operator $P \in \mathcal{P}$ into components $P = X_{P} Z_{P}$,
\begin{align}
    &\Lambda_{\mathrm{EC}}^{X,K} \circ \mathcal{N} (\cdot) = \sum_{\alpha} \pr_{\mathcal{N}}(\alpha) \sum_{\mu} \pr_{\Lambda}(\mu) \Lambda_{\mathrm{EC},\mu}^{X} \circ \Phi_{\alpha} (\cdot) , \nonumber \\
    &\Lambda_{\mathrm{EC},\mu}^{X} \circ \Phi_{\alpha} \left(\dyad{0}{0}^{\otimes n}\right) = \nonumber \\
    &\sum_{\substack{\sigma^{X} , K_{\alpha}, \\ P,Q \in \mathcal{P}}} c_{P}^{K_{\alpha}} \left(c_{Q}^{K_{\alpha}}\right)^{*} R^{Z}_{\sigma^{X} + \mu} \Pi^{X}_{\sigma^{X}} P \dyad{0}{0}^{\otimes n} Q^{\dagger} \Pi^{X}_{\sigma^{X}} R^{Z \dagger}_{\sigma^{X} + \mu} \nonumber \\
    &= \sum_{\sigma^{X} , K_{\alpha}} R^{Z}_{\sigma^{X} + \mu} \left( \sum_{P \in \mathcal{P}} c_{P}^{K_{\alpha}} X_{P} \right) \Pi^{X}_{\sigma^{X}} \dyad{0}{0}^{\otimes n} \Pi^{X}_{\sigma^{X}} \nonumber \\
    &\times \left( \sum_{Q \in \mathcal{P}} c_{Q}^{K_{\alpha}} X_{Q} \right)^{\dagger} R^{Z \dagger}_{\sigma^{X} + \mu} .
\end{align}
Here, we used $Z_{P} \ket{0}^{\otimes n} = \ket{0}^{\otimes n}$ and $\Pi^{X}_{\sigma^{X}} X_{P} = X_{P} \Pi^{X}_{\sigma^{X}}$. One may rewrite
\begin{equation}
    R^{Z}_{\sigma^{X} + \mu} = R^{Z}_{\mu} L^{Z}_{\sigma^{X},\mu} S^{Z}_{\sigma^{X},\mu} R^{Z}_{\sigma^{X}} ,
\end{equation}
where implicitly $\partial R^{Z \dagger}_{\sigma^{X} + \mu} = \sigma^{X} + \mu + \hat{\mu}$ and $\hat{\mu} = \hat{\mu}(\mu)$. Now, define the conjugated operator
\begin{align}
    \widetilde{K}_{\mu;\sigma^{X}} &= R^{Z \dagger}_{\mu} R^{Z}_{\sigma^{X} + \mu} \left( \sum_{P \in \mathcal{P}} c_{P}^{K_{\alpha}} X_{P} \right) R^{Z \dagger}_{\sigma^{X} + \mu} R^{Z}_{\mu} \nonumber \\
    &= \left( \sum_{P \in \mathcal{P}} c_{P}^{K_{\alpha}} e^{i \phi(P,\mu,\sigma^{X})} X_{P} \right) ,
\end{align}
and since $R^{Z \dagger}_{\mu} R^{Z}_{\sigma^{X} + \mu} \Pi^{X}_{\sigma^{X}} \ket{0}^{\otimes n} = \ket{\overline{0}}^{\otimes k} / \sqrt{d^{r_X}}$,
\begin{align}
    &\Lambda_{\mathrm{EC},\mu}^{X} \circ \Phi_{\alpha} \left(\dyad{0}{0}^{\otimes n}\right) = \\
    &\sum_{\sigma^{X} , K_{\alpha}} \frac{1}{d^{r_X}} R^{Z}_{\mu} \widetilde{K}_{\mu;\sigma^{X}}   \dyad{\overline{0}}{\overline{0}}^{\otimes k} \widetilde{K}_{\mu;\sigma^{X}}^{\dagger} R^{Z \dagger}_{\mu} . \nonumber
\end{align}
Notice that $\Lambda_{\mathrm{EC},\mu}^{X} \circ \Phi_{\alpha} \left(\dyad{0}{0}^{\otimes n}\right) = \mathcal{G}_{\alpha ; \mu} \left(\dyad{\overline{0}}{\overline{0}}^{\otimes k}\right)$ for some channel $\mathcal{G}_{\alpha ; \mu}$. One such $\mathcal{G}_{\alpha ; \mu}$ may be defined by
\begin{align}
    \mathcal{G}_{\alpha ; \mu} (\cdot)& = \sum_{\sigma^{X} , K_{\alpha}} F_{\sigma,K_{\alpha}} (\cdot) F_{\sigma,K_{\alpha}}^{\dagger} + \left(I-\Pi^{\mathcal{S}}_{0}\right) (\cdot) \left(I-\Pi^{\mathcal{S}}_{0}\right) , \nonumber \\
    F_{\sigma,K_{\alpha}} &= \frac{1}{\sqrt{d^{r_X}}} R^{Z}_{\mu} \widetilde{K}_{\mu;\sigma^{X}}  \Pi^{\mathcal{S}}_{0} .
\end{align}
Therefore, the system state after noisy state preparation may be mimicked by the channel $\mathcal{\widetilde{N}}$,
\begin{align}
    &\Lambda_{\mathrm{EC}}^{X,K} \circ \mathcal{N} \left(\dyad{0}{0}^{\otimes n}\right) = \mathcal{\widetilde{N}} \left(\dyad{\overline{0}}{\overline{0}}^{\otimes k}\right) , \\
    &\mathcal{\widetilde{N}} (\cdot) = \sum_{\alpha,\mu} \pr_{\mathcal{N}}(\alpha) \pr_{\Lambda}(\mu) \mathcal{G}_{\alpha ; \mu} (\cdot) . \nonumber
\end{align}
Observe that $F_{\mathcal{R}} (\mathcal{G}_{\alpha ; \mu}) \leq F_{\mathcal{R}} (\Phi_{\alpha})$, as $F_{\mathcal{R}} (\mathcal{G}_{\alpha ; \mu}) = 1$ if some $L^{X}_{P} \neq I$ but $F_{\mathcal{R}} (\Phi_{\alpha}) = 1$ if $L^{X}_{P} \neq I$ or $L^{Z}_{P} \neq I$, therefore $\mathrm{fail}\left( \mathcal{\widetilde{N}} \right) \leq \mathrm{fail}\left( \mathcal{N} \right) \leq \epsilon$. Furthermore,
\begin{equation}
    \supp \partial \mathcal{G}_{\alpha ; \mu} = \supp (\mu + \hat{\mu}) \cup \supp \partial_{X} \Phi_{\alpha}
\end{equation}
where $\partial_{X}$ denotes only the syndrome of $X$-type components of an operator, $\partial_{X} P = \partial X_{P}$. This means $\supp \partial_{X} \Phi_{\alpha} \subseteq \supp \partial \Phi_{\alpha}$. Noting that
\begin{equation}
    \pr_{\mathrm{syn}\mathcal{\widetilde{N}}}(A) = \sum_{\alpha,\mu : \supp \partial \mathcal{G}_{\alpha ; \mu} = A} \pr_{\mathcal{N}}(\alpha) \pr_{\Lambda}(\mu) ,
\end{equation}
we find
\begin{align}
    \label{eq:noisy-prep-bound}
    &\sum_{\alpha,\mu : \supp \partial \mathcal{G}_{\alpha ; \mu} \supseteq A} \pr_{\mathcal{N}}(\alpha) \pr_{\Lambda}(\mu) \leq \nonumber \\
    &\sum_{\alpha,\mu : \supp \partial \Phi_{\alpha} \cup \supp (\mu + \hat{\mu}) \supseteq A} \pr_{\mathcal{N}}(\alpha) \pr_{\Lambda}(\mu) .
\end{align}
Since $\mathcal{N} \in \mathbb{N}_{\tau,\epsilon}^{\mathrm{exc}}$, $\pr_{\mathrm{syn}\mathcal{N}}(A)$ is $\tau$-bounded
by definition. Additionally, if $\mathcal{S}^{X}_{0}$ and $\Lambda_{\mathrm{EC}}^{X}$ admit single-shot error correction as described in Section~\ref{self-corr-proof}, then
\begin{equation}
    \sum_{\mu : \supp (\mu + \hat{\mu}) \supseteq A} \pr_{\Lambda}(\mu) \leq \eta^{\vert A \vert}
\end{equation}
for $\eta$ given in the previous section. Using Proposition~\ref{proposition:bounded-composition}, setting $g(\alpha) = \supp \partial \Phi_{\alpha}$ and $h(\mu) = \supp (\mu + \hat{\mu})$, we obtain from Eq.~(\ref{eq:noisy-prep-bound})
\begin{align}
    &\sum_{\alpha,\mu : \supp \partial \mathcal{G}_{\alpha ; \mu} \supseteq A} \pr_{\mathcal{N}}(\alpha) \pr_{\Lambda}(\mu) \leq \left(\tau + \eta\right)^{\vert A \vert} .
\end{align}
Therefore $\pr_{\mathrm{syn}\mathcal{\widetilde{N}}}(A)$ is $\left(\tau + \eta\right)$-bounded and we conclude that $\mathcal{\widetilde{N}} \in \mathbb{N}_{\tau+\eta,\epsilon}^{\mathrm{exc}}$. This means for all $\Lambda_{\mathrm{EC}}^{X,K} \in \Lambda_{\mathrm{EC}}^{X,\mathbb{M}_{\zeta}}$ and $\mathcal{N} \in \mathbb{N}_{\tau,\epsilon}^{\mathrm{exc}}$,
\begin{equation}
    \Lambda_{\mathrm{EC}}^{X,K} \circ \mathcal{N} \left(\dyad{0}{0}^{\otimes n}\right) = \mathcal{\widetilde{N}} \left(\dyad{\overline{0}}{\overline{0}}^{\otimes k}\right)
\end{equation}
where $\mathcal{\widetilde{N}} \in \mathbb{N}_{\tau+\eta,\epsilon}^{\mathrm{exc}}$. $\Lambda_{\mathrm{EC}}^{X}$ is then fault tolerant as both properties of Definition~\ref{def:ft-channel} are satisfied if $\zeta < (2ez_{\Gamma})^{-2}$, $\tau + \eta < 1$, and $\epsilon = 1/\Omega(\poly(n))$. This shows that single-shot state preparation of $\ket{\overline{0}}^{\otimes k}$ from $\ket{0}^{\otimes n}$ can be done if $\mathcal{S}^{X}_{0}$ and $\Lambda_{\mathrm{EC}}^{X}$ admit single-shot error correction. By applying the same arguments in this section but swapping $X \leftrightarrow Z$ everywhere, one can also show single-shot state preparation of $\ket{\overline{+}}^{\otimes k}$ from $\ket{+}^{\otimes n}$ can be done if $\mathcal{S}^{Z}_{0}$ and $\Lambda_{\mathrm{EC}}^{Z}$ admit single-shot error correction, namely for all $\Lambda_{\mathrm{EC}}^{Z,K} \in \Lambda_{\mathrm{EC}}^{Z,\mathbb{M}_{\zeta}}$ and $\mathcal{N} \in \mathbb{N}_{\tau,\epsilon}^{\mathrm{exc}}$,
\begin{equation}
    \Lambda_{\mathrm{EC}}^{Z,K} \circ \mathcal{N} \left(\dyad{+}{+}^{\otimes n}\right) = \mathcal{\widetilde{N}} \left(\dyad{\overline{+}}{\overline{+}}^{\otimes k}\right)
\end{equation}
where $\mathcal{\widetilde{N}} \in \mathbb{N}_{\tau+\eta,\epsilon}^{\mathrm{exc}}$.

\subsection{Proof of Corollaries~\ref{cor:tc-single-shot} and~\ref{cor:ghz-single-shot}} \label{tc-single-shot-proof}

\begin{proof}[Proof of Corollary~\ref{cor:tc-single-shot}]
For the $(D,q)$ $\mathbb{Z}_{d}$ toric codes discussed in Section~\ref{q-form-tc}, the Pauli group can be generated by operators defined by $q$-chains $\Psi \in \Omega_{q}$ on the primal lattice. For any
\begin{equation}
    \Psi = \sum_{\varphi \in \Delta_p} k_\varphi \varphi , \quad k_\varphi \in \{0,\pm 1\} ,
\end{equation}
define membrane operators
\begin{equation}
    X \left[ \Psi \right] = \prod_{k\varphi \in \Psi} X_{\varphi}^{k}  , \quad
    Z \left[ \Psi \right] = \prod_{k\varphi \in \Psi} Z_{\varphi}^{k} .
\end{equation}
These operators do not commute with all stabilizers, as
\begin{align}
    \label{eq:tc-membrane-commutation}
    A_{s} X \left[ \Psi \right] &= X \left[ \Psi \right] A_{s} \omega^{\langle s , \partial_{q} \Psi \rangle_{q-1}} , \\
    B_{m} Z \left[ \Psi \right] &= Z \left[ \Psi \right] B_{m} \omega^{-\langle \partial^{\dagger}_{q+1} \Psi , m \rangle_{q+1}} . \nonumber
\end{align}
Any Pauli operator can be written as a product of membrane operators $X \left[ \Psi \right]$ and $Z \left[ \Psi \right]$ for various $\Psi$. For the membrane operator $X \left[ \Psi \right]$, $\Psi$ can be viewed as a $q$-dimensional object on the primal lattice, and the $(q-1)$-dimensional boundary $\partial_{q} \Psi$ correspond to the syndrome of $X \left[ \Psi \right]$. Similarly, for $Z \left[ \Psi \right]$, $\Psi$ can be viewed as a $(D-q)$-dimensional object on the dual lattice, and the $(D-q-1)$-dimensional dual boundary $\partial^{\dagger}_{q+1} \Psi$ corresponds to the syndrome of $Z \left[ \Psi \right]$. Moreover, $B_{m} = X \left[ \Psi_{m} \right]$ and $A_{s} = Z \left[ \Psi_{s} \right]$ are membrane operators where $\Psi_{m}$ and $\Psi_{s}$ have no boundary or dual boundary and can be viewed as closed manifolds on the primal and dual lattice respectively. Similarly, logical operators can be constructed from products of $X \left[ \Psi \right]$ and $Z \left[ \Psi \right]$, where all such $\Psi$ can be viewed as non-contractible objects which have no boundary or dual boundary, respectively. The structure of such $\Psi$ depends on the topology of the manifold which the lattice $\mathcal{J}$ tiles. However, since $\mathcal{J}$ has minimum linear size $L = \Omega(n^{1/D})$ in any direction, all nontrivial logical operators have weight $\Omega(L)$.

We will now show that Theorems~\ref{theorem:self-corr-qec} and~\ref{theorem:self-corr-state-prep} apply to $(D,q)$ $\mathbb{Z}_{d}$ toric codes
$\mathcal{V}$ defined on such a lattice $\mathcal{J}$ for $D \geq 4$ and $q \in [2,D-2]$, therefore these codes exhibit both single-shot QEC and single-shot state preparation. Since  $(D,q)$ $\mathbb{Z}_{d}$ toric codes are CSS codes, one can correct $X$ and $Z$ errors independently and construct an excitation graph $\Gamma = \Gamma^{Z} \sqcup \Gamma^{X}$, where $\Gamma^{Z}$ corresponds to $A_{s}$ generators and $\Gamma^{X}$ corresponds to $B_{m}$ generators. To construct $\Gamma^{Z}$, associate vertices with $(q-1)$-cells $s \in \Delta_{q-1}$, and connect two vertices with an edge if the two cells share a $(q-2)$-cell in their boundaries. To construct $\Gamma^{X}$, associate vertices with $(q+1)$-cells $m \in \Delta_{q+1}$, and connect two vertices with an edge if the two cells are both in the boundary of a $(q+2)$-cell.

To verify the conditions in Appendix~\ref{self-corr-conditions}, first note that since $\mathcal{J}$ is a regular lattice which tiles a $D$-dimensional manifold, the sets $\mathcal{S}^{Z}_{0} = \{A_{s}\}$ and $\mathcal{S}^{X}_{0} = \{B_{m}\}$ are both $(v_{q},v_{s})$ geometrically local with finite $v_{q}$ and $v_{s}$, therefore both $\Gamma^{Z}$ and $\Gamma^{X}$ have bounded maximum degree. Qudit graphs, $\Theta^{Z}$ and $\Theta^{X}$ can also be constructed as described in Appendix~\ref{self-corr-conditions}. The connectivity of $\Gamma^{Z}$ and $\Gamma^{X}$ mirror that of closed $(q-1)$-dimensional extended objects on the primal lattice and $(D-q-1)$-dimensional extended objects on the dual lattice. Assuming the stabilizer group has no global constraints when defined on the lattice $\mathcal{J}$, this means each maximal connected component of a syndrome on either $\Gamma^{Z}$ and $\Gamma^{X}$ is itself a valid syndrome, so property 1 of Appendix~\ref{self-corr-conditions} holds. Moreover, property 2 holds from the construction of $\Lambda_{\mathrm{EC}}$.

In the presence of global constraints, the above holds if each connected component of a syndrome is topologically trivial. As discussed in Section IV of Ref.~\cite{bombin2015singleshot}, connected components of syndromes may not always be valid if they are topologically nontrivial due to global constraints, however these syndromes will have weight $\Omega(L)$. In practice, such syndromes indicate the presence of either a large measurement error or a large qudit error and can be dealt with on a case-by-case basis, as outlined in Ref.~\cite{bombin2015singleshot}, where one can either still infer a correction to a measurement error or in the worst case one just repeats all stabilizer measurements. Because of this and for simplicity of analysis, we will just assume each maximal connected component of a syndrome is itself a valid syndrome for the rest of this section. All arguments extend to the more complicated setting discussed in this paragraph in practice.

As discussed above, all nontrivial logical operators have weight $\Omega(L)$. To see property 3 holds, consider an arbitrary error $P$ with a single connected component which extends $\leq L / 2$ in any direction. Choosing the correction $R_{\partial P}$ to be as geometrically local as possible, $R_{\partial P} P \in \mathcal{S}$, as otherwise $R_{\partial P}$ must extend $\geq L / 2$ in some direction. Lastly, to see property 4 holds, note that any closed $(q-1)$-dimensional extended objects on the primal lattice ($(D-q-1)$-dimensional extended objects on the dual lattice) supported entirely within a contractible hypercubic region $\Upsilon$ of linear size $O\left(n^{1/D}\right)$ form the boundaries of $q$-dimensional extended objects on the primal lattice ($(D-q)$-dimensional extended objects on the dual lattice) contained within $\Upsilon$. If one chooses correction operators for connected components $\sigma_{c}$ and $\kappa_{c}$ of syndromes $\sigma$ and $\kappa$ supported on such $(q-1)$ or $(D-q-1)$-dimensional objects to be contained within $\Upsilon$, then $R_{\sigma + \kappa} R_{\sigma}^{\dagger} R_{\kappa}^{\dagger} \in \mathcal{S}$, as proven in Proposition~\ref{proposition:syndrome-supports}.

Since all conditions in Appendix~\ref{self-corr-conditions} are satisfied for both $X$ and $Z$ stabilizers, excitation graphs, qudit graphs, and correction circuits separately, we conclude that Theorems~\ref{theorem:self-corr-qec} and~\ref{theorem:self-corr-state-prep} apply to $(D,q)$ $\mathbb{Z}_{d}$ toric codes for $q \in [2,D-2]$ and $D \geq 4$. Therefore, these codes exhibit single-shot, finite-depth, fault tolerant error correction for both Pauli $X$ and $Z$ errors independently and permit finite-depth, fault tolerant state preparation of $\dyad{\overline{0}}{\overline{0}}^{\otimes k}$ and $\dyad{\overline{+}}{\overline{+}}^{\otimes k}$ from product states $\dyad{0}{0}^{\otimes n}$ and $\dyad{+}{+}^{\otimes n}$.
\end{proof}

\begin{proof}[Proof of Corollary~\ref{cor:ghz-single-shot}]
If only Pauli $X$ errors occur, the arguments in the above proof apply to $(D,q)$ $\mathbb{Z}_{d}$ toric codes with $2 \leq q$, $D \geq 2$, and if only Pauli $Z$ errors occur, arguments in the above proof apply to $(D,q)$ $\mathbb{Z}_{d}$ toric codes with $q \leq D-2$, $D \geq 2$.
\end{proof}

\bibliography{bibliography}

\end{document}